\documentclass[acmsmall,nonacm,screen]{acmart}
\renewcommand\footnotetextcopyrightpermission[1]{} 

\setcopyright{none}
\makeatletter
\@ifundefined{lhs2tex.lhs2tex.sty.read}%
  {\@namedef{lhs2tex.lhs2tex.sty.read}{}%
   \newcommand\SkipToFmtEnd{}%
   \newcommand\EndFmtInput{}%
   \long\def\SkipToFmtEnd#1\EndFmtInput{}%
  }\SkipToFmtEnd

\newcommand\ReadOnlyOnce[1]{\@ifundefined{#1}{\@namedef{#1}{}}\SkipToFmtEnd}
\usepackage{amstext}
\usepackage{amssymb}
\usepackage{stmaryrd}
\DeclareFontFamily{OT1}{cmtex}{}
\DeclareFontShape{OT1}{cmtex}{m}{n}
  {<5><6><7><8>cmtex8
   <9>cmtex9
   <10><10.95><12><14.4><17.28><20.74><24.88>cmtex10}{}
\DeclareFontShape{OT1}{cmtex}{m}{it}
  {<-> ssub * cmtt/m/it}{}

\DeclareFontShape{OT1}{cmtt}{bx}{n}
  {<5><6><7><8>cmtt8
   <9>cmbtt9
   <10><10.95><12><14.4><17.28><20.74><24.88>cmbtt10}{}
\DeclareFontShape{OT1}{cmtex}{bx}{n}
  {<-> ssub * cmtt/bx/n}{}

\newcommand{\Conid}[1]{\mathit{#1}}
\newcommand{\Varid}[1]{\mathit{#1}}
\newcommand{\anonymous}{\kern0.06em \vbox{\hrule\@width.5em}}

\newcommand{\bind}{\mathbin{>\!\!\!>\mkern-6.7mu=}}

\usepackage{polytable}

\@ifundefined{mathindent}%
  {\newdimen\mathindent\mathindent\leftmargini}%
  {}%

\def\resethooks{%
  \global\let\SaveRestoreHook\empty
  \global\let\ColumnHook\empty}
\newcommand*{\savecolumns}[1][default]%
  {\g@addto@macro\SaveRestoreHook{\savecolumns[#1]}}
\newcommand*{\restorecolumns}[1][default]%
  {\g@addto@macro\SaveRestoreHook{\restorecolumns[#1]}}
\newcommand*{\aligncolumn}[2]%
  {\g@addto@macro\ColumnHook{\column{#1}{#2}}}

\resethooks

\newcommand{\onelinecommentchars}{\quad-{}- }
\newcommand{\commentbeginchars}{\enskip\{-}
\newcommand{\commentendchars}{-\}\enskip}

\newcommand{\visiblecomments}{%
  \let\onelinecomment=\onelinecommentchars
  \let\commentbegin=\commentbeginchars
  \let\commentend=\commentendchars}

\newcommand{\invisiblecomments}{%
  \let\onelinecomment=\empty
  \let\commentbegin=\empty
  \let\commentend=\empty}

\visiblecomments

\newlength{\blanklineskip}
\newcommand{\hsindent}[1]{\quad}
\let\hspre\empty
\let\hspost\empty

\EndFmtInput
\makeatother
\ReadOnlyOnce{polycode.fmt}%
\makeatletter

\newcommand{\hsnewpar}[1]%
  {{\parskip=0pt\parindent=0pt\par\vskip #1\noindent}}

\newcommand{\hscodestyle}{}

\newcommand{\sethscode}[1]%
  {\expandafter\let\expandafter\hscode\csname #1\endcsname
   \expandafter\let\expandafter\endhscode\csname end#1\endcsname}

  {\par\noindent
   \advance\leftskip\mathindent
   \hscodestyle
   \let\\=\@normalcr
   \let\hspre\(\let\hspost\)%
   \pboxed}%
  {\endpboxed\)%
   \par\noindent
   \ignorespacesafterend}

  {\hsnewpar\abovedisplayskip
   \advance\leftskip\mathindent
   \hscodestyle
   \let\hspre\(\let\hspost\)%
   \pboxed}%
  {\endpboxed%
   \hsnewpar\belowdisplayskip
   \ignorespacesafterend}

  {\hsnewpar\abovedisplayskip
   \advance\leftskip\mathindent
   \hscodestyle
   \let\\=\@normalcr
   \(\pboxed}%
  {\endpboxed\)%
   \hsnewpar\belowdisplayskip
   \ignorespacesafterend}

\newcommand{\plainhs}{\sethscode{plainhscode}}

\plainhs

  {\hsnewpar\abovedisplayskip
   \advance\leftskip\mathindent
   \hscodestyle
   \let\\=\@normalcr
   \(\parray}%
  {\endparray\)%
   \hsnewpar\belowdisplayskip
   \ignorespacesafterend}

  {\parray}{\endparray}

  {\(\parray}{\endparray\)}

\def\codeframewidth{\arrayrulewidth}
\RequirePackage{calc}

  {\parskip=\abovedisplayskip\par\noindent
   \hscodestyle
   \arrayrulewidth=\codeframewidth
   \tabular{@{}|p{\linewidth-2\arraycolsep-2\arrayrulewidth-2pt}|@{}}%
   \hline\framedhslinecorrect\\{-1.5ex}%
   \let\endoflinesave=\\
   \let\\=\@normalcr
   \(\pboxed}%
  {\endpboxed\)%
   \framedhslinecorrect\endoflinesave{.5ex}\hline
   \endtabular
   \parskip=\belowdisplayskip\par\noindent
   \ignorespacesafterend}

\newcommand{\framedhslinecorrect}[2]%
  {#1[#2]}

  {\(\def\column##1##2{}%
   \let\>\undefined\let\<\undefined\let\\\undefined
   \newcommand\>[1][]{}\newcommand\<[1][]{}\newcommand\\[1][]{}%
   \def\fromto##1##2##3{##3}%
   }{\) }%

  {\let\orighscode=\hscode
   \let\origendhscode=\endhscode
   \def\endhscode{\def\hscode{\endgroup\def\@currenvir{hscode}\\}\begingroup}
   \orighscode\def\hscode{\endgroup\def\@currenvir{hscode}}}%
  {\origendhscode
   \global\let\hscode=\orighscode
   \global\let\endhscode=\origendhscode}%

\makeatother
\EndFmtInput
\ReadOnlyOnce{forall.fmt}%
\makeatletter

\let\HaskellResetHook\empty
\newcommand*{\AtHaskellReset}[1]{%
  \g@addto@macro\HaskellResetHook{#1}}
\newcommand*{\HaskellReset}{\HaskellResetHook}

\newcommand\hsforall{\global\let\hsdot=\hsperiodonce}
\newcommand*\hsperiodonce[2]{#2\global\let\hsdot=\hscompose}
\newcommand*\hscompose[2]{#1}

\AtHaskellReset{\global\let\hsdot=\hscompose}

\HaskellReset

\makeatother
\EndFmtInput

\renewcommand{\anonymous}{\_}

\newenvironment{plainhscodecentered}%
  {\hsnewpar\abovedisplayskip\centering
   \hscodestyle
   \let\hspre\(\let\hspost\)%
   \pboxed}%
  {\endpboxed%
   \hsnewpar\belowdisplayskip
   \ignorespacesafterend}

  {\relax\ifmmode\expandafter\pmboxed\else\expandafter\plainhscodecentered\fi}%
  {\relax\ifmmode\expandafter\endpmboxed\else\expandafter\endplainhscodecentered\fi}
\sethscode{autohscode}

\usepackage[bb=px]{mathalpha}
\usepackage{fancyvrb}
\usepackage{mathpartir}
\usepackage{subcaption}
\usepackage{tikz}
\usepackage{tikz-cd}
\usepackage{verbatim}
\usepackage{mathtools}
\usepackage{enumitem}

\usepackage{dingbat}
\usepackage{textgreek}
\usepackage{xifthen}
\usepackage{wasysym}
\usepackage{doi}
\usepackage{wrapfig}

\usepackage[nameinlink]{cleveref}

\DeclarePairedDelimiter{\classify}{\lceil}{\rceil}
\DeclarePairedDelimiter{\aset}{\{}{\}}
\DeclarePairedDelimiter{\atuple}{\langle}{\rangle}

\DeclarePairedDelimiter{\abrace}{\{}{\}}

\DeclarePairedDelimiter{\sem}{\llbracket}{\rrbracket}

\DeclarePairedDelimiter{\avert}{\vert}{\vert}

\newcommand{\bN}{\mathbb{N}}
\newcommand{\bT}{\mathbb{T}}

\newcommand{\bA}{\mathbb{A}}

\newcommand{\sP}{\mathscr{P}}

\newcommand{\fS}{\mathfrak{S}}

\newcommand{\MLTT}{Martin-L\"{o}f type theory}

\newcommand{\Jdg}{\mathbb{J}}
\newcommand{\impfun}[1]{\abrace{#1} \to}
\newcommand{\impapp}[1]{ \; \abrace{#1}}
\newcommand{\imparg}[1]{\abrace{#1}}
\newcommand{\impabs}[1]{\lambda \abrace{#1}.\;}
\newcommand{\extty}[3]{\aset{#1 \mid #2 \hookrightarrow #3}}
\newcommand{\Omod}{\mathop{\normalfont\Circle}}
\newcommand{\Cmod}{\mathop{\normalfont\CIRCLE}}
\newcommand{\Omodal}{\Circle\mathit{\hyp{}\mkern-1mu modal}}
\newcommand{\Cmodal}{\CIRCLE\mathit{\hyp{}\mkern-1mu modal}}
\newcommand{\isprop}{\mathit{is\hyp{}prop}}
\newcommand{\Oeta}{\eta^{\circ}}
\newcommand{\Ceta}{\eta^{\bullet}}
\newcommand{\Ceps}{\epsilon^{\bullet}}
\newcommand{\sop}{{\color{brown}\mathfrak{ob}}} 
\newcommand{\iso}{\cong}
\newcommand{\unitty}{\mathrm{1}}
\newcommand{\unitel}{{*}}
\newcommand{\emptyty}{\mathrm{0}}
\newcommand{\gluetysymbol}{\ltimes}
\newcommand{\gluety}[2]{(#1) \mathbin{\gluetysymbol} #2}
\newcommand{\glue}{\mathit{glue}}
\newcommand{\unglue}{\mathit{unglue} \; }
\newcommand{\gluetm}[2]{[\sop \hookrightarrow #2 \mid  #1]}
\newcommand{\TTSTC}{\textsc{StcTT}}
\newcommand{\TTASM}{\textsc{AsmTT}}
\newcommand{\TTSTCLR}{\textsc{2-TTstc}}
\newcommand{\lpre}{{\mathit{pre}}\;}
\newcommand{\lprf}{{\mathit{prf}}\;}

\newcommand{\inl}{\mathit{inl}}
\newcommand{\inr}{\mathit{inr}}
\newcommand{\refl}{\mathit{refl}}
\newcommand{\ifft}{\quad\text{iff}\quad}
\newcommand{\diverges}{\uparrow}
\newcommand{\converges}{\downarrow}

\newcommand{\hypo}[2]{#1 \vdash #2}

\newcommand{\iskind}[2]{\hypo{#1}{#2 \; \mathrm{kind}}}

\newcommand{\fst}{\pi_1}
\newcommand{\snd}{\pi_2}

\newcommand{\Hexc}{H_{\mathit{exc}}}
\newcommand{\hexc}{h_{\mathit{exc}}}

\newcommand{\CatJdgOf}[1]{\mathop{\mathsc{Jdg}}#1}
\newcommand{\CatGlueOf}[1]{\mathop{\mathsc{Gl}}#1}
\newcommand{\GFha}{\CatGlueOf{\Fha}}

\newcommand{\TyC}{U_{\CatC}}

\newcommand{\LFMod}[2]{#1\mathsc{-Mod}\mkern1mu(#2)}

\newcommand{\COmega}{\Omega^\bullet}

\newcommand{\Uir}{U^{\mathrm{ir}}}
\newcommand{\CU}{U^\bullet}

\newsavebox{\pullbackbox}
\sbox\pullbackbox{%
\begin{tikzpicture}%
  \draw[line width=rule_thickness] (0,0) -- (1.5ex,0ex) -- (1.5ex,1.5ex);%
\end{tikzpicture}}

\newsavebox{\pushoutbox}
\sbox\pushoutbox{%
\begin{tikzpicture}%
  \draw[line width=rule_thickness] (0,0) -- (-1.5ex,0ex) -- (-1.5ex,-1.5ex);%
\end{tikzpicture}}

\newcommand{\vcomp}{\mathbin{\cdot}}
\newcommand{\hcomp}{\mathbinscr{\circ}}

\newcommand{\Asm}{\mathsc{Asm}}
\newcommand{\AsmK}{\mathsc{Asm}(\mathbb{K})}
\newcommand{\Eff}{\mathsc{Eff}}

\newcommand{\PDom}{\Varid{PDom}}
\newcommand{\DOM}{\Varid{Dom}}
\renewcommand{\Pr}[1]{\mathop{\mathsc{PSh}}{#1}}
\newcommand{\PrC}{\Pr{\CatC}}
\newcommand{\PrD}{\Pr{\CatD}}

\newcommand{\Set}{\mathsc{Set}}

\newcommand{\CatSet}{\Set}

\newcommand{\Comma}[2]{#1 \downarrow #2}
\newcommand{\blank}{{-}}
\newcommand{\tuple}[1]{\langle #1 \rangle}

\newcommand{\yo}{\mathscr{y}}

\newcommand{\identity}{\mathit{id}}
\newcommand{\IdF}{\mathrm{Id}}

\newcommand{\Hom}{\mathsc{Hom}}

\newcommand{\LCCC}{\mathsc{LCCC}}

\newcommand{\realizes}[3][]{#2 \models_{#1} #3}

\newcommand{\oomega}{\bar{\omega}}

\newcommand{\Sier}{Sierpi\'{n}ski}
\newcommand{\pto}{\rightharpoonup}

\newcommand{\Pow}{\sP}

\newcommand{\lhskeyword}[1]{\textsc{\textls*[60]{#1}}}
\newcommand{\mathsc}[1]{\textnormal{\scshape #1}}
\newcommand{\Fha}{\text{\rmfamily F\rlap{\textsubscript{\textomega}}{\textsuperscript{ha}}}}
\newcommand{\Fhar}{\text{\rmfamily rF\rlap{\textsubscript{\textomega}}{\textsuperscript{ha}}}}
\newcommand{\Fom}{\text{\rmfamily F\textsubscript{\textomega}}}
\newcommand{\Spc}{\quad\quad}

\newcommand{\SpcAnd}{\quad\text{and}\quad}

\mathchardef\hyp="2D
\newcommand{\defeq}{\vcentcolon=}
\newcommand{\reason}[1]{\quad\abrace{\text{#1}}}

\newcommand{\reasonMult}[1]{\quad%
\left\{
\begin{tabular}{l}%
#1%
\end{tabular}%
\right\}
}

\newcommand{\spreadStar}[1]{\makebox[\linewidth]{\hfill\math #1 \endmath\hfill}}
\newcommand{\spreadNoStar}[1]{\makebox[\linewidth-\widthof{(\theequation)}-1.0001em]{\hfill\math #1 \endmath\hfill}
}
\makeatletter
\def\spread{\@ifstar\spreadStar\spreadNoStar}
\makeatother
\newcommand{\also}{\hfill}
\newcommand{\spreadTag}[2]{\makebox[\linewidth-\widthof{(#1)}-1.0001em]{\hfill\math #2 \endmath\hfill\tag{#1}}}

\newcommand{\mathbinscr}[1]{%
 \mathbin{\mskip1mu\nonscript\mskip-1mu%
          #1%
          \mskip1mu\nonscript\mskip-1mu}}

\newcommand{\LCCLF}{\textsc{LccLF}}

\newcommand{\allfor}{\bar{\forall}}

\newenvironment{lgather*}{%
 \[\array{l}%
}{\endarray\]\ignorespacesafterend}

\newenvironment{lgather}[1]{%
  \begin{equation}\label{#1}\begin{array}{l}%
  }{\end{array}\end{equation}\ignorespacesafterend}

  {\hsnewpar{0pt}
   \advance\leftskip\mathindent
   \hscodestyle
   \let\hspre\(\let\hspost\)%
   \pboxed}%
  {\endpboxed%
   \hsnewpar{0pt}
   \ignorespacesafterend}

\newcommand{\plainhsnoskip}{\sethscode{plainhscodenoskip}}

\theoremstyle{plain}
\theoremstyle{acmplain}
\newtheorem{theorem}{Theorem}[section]
\newtheorem{lemma}[theorem]{Lemma}
\newtheorem{conjecture}[theorem]{Conjecture}

\newtheorem*{lemma*}{Lemma}
\newtheorem*{conjecture*}{Conjecture}
\newtheorem*{proposition*}{Proposition}

\theoremstyle{definition}
\theoremstyle{acmdefinition}
\newtheorem{definition}[theorem]{Definition}

\newtheorem{question}[theorem]{Question}
\newtheorem{example}[theorem]{Example}
\newtheorem{axiom}[theorem]{Axiom}
\newtheorem{remark}[theorem]{Remark}
\newtheorem{notation}[theorem]{Notation}

\newtheorem{corollary}[theorem]{Corollary}
\newtheorem{typetheory}[theorem]{Language}

\newtheorem*{definition*}{Definition}
\newtheorem*{example*}{Example}
\newtheorem*{problem*}{Problem}
\newtheorem*{question*}{Question}
\newtheorem*{axiom*}{Axiom}
\newtheorem*{remark*}{Remark}
\newtheorem*{notation*}{Notation}
\newtheorem*{assumption*}{Assumption}
\newtheorem*{convention*}{Convention}
\newtheorem*{terminology*}{Terminology}
\newtheorem*{corollary*}{Corollary}

\NewDocumentEnvironment{definitionE}{O{}O{}+b}{%
\begin{theoremEnd}[restate, no link to proof, #2]{definition}[#1]%
#3%
\end{theoremEnd}%
}{}

\NewDocumentEnvironment{theoremE}{O{}O{}+b}{%
\begin{theoremEnd}[restate, no link to proof, #2]{theorem}[#1]%
#3%
\end{theoremEnd}%
}{}

\NewDocumentEnvironment{lemmaE}{O{}O{}+b}{%
\begin{theoremEnd}[restate, no link to proof, #2]{lemma}[#1]%
#3%
\end{theoremEnd}%
}{}

\NewDocumentEnvironment{propositionE}{O{}O{}+b}{%
\begin{theoremEnd}[restate, no link to proof, #2]{proposition}[#1]%
#3%
\end{theoremEnd}%
}{}

\NewDocumentEnvironment{corollaryE}{O{}O{}+b}{%
\begin{theoremEnd}[restate, no link to proof, #2]{corollary}[#1]%
#3%
\end{theoremEnd}%
}{}

\NewDocumentEnvironment{proofE}{O{Proof}+b}{%
\begin{proofEnd}[text proof={#1 of \string\Cref{thm:prAtEnd\pratendcountercurrent}}]%
#2%
\end{proofEnd}%
}{}

\crefname{equation}{}{}
\creflabelformat{equation}{#2(#1)#3}

\newlist{proofcase}{enumerate}{2}
\setlist[proofcase,1]{wide, parsep=0pt, topsep=\itemsep, partopsep=0pt, 
  label={\textit{Case \arabic*.}},ref=\arabic*, itemsep=0pt}
\setlist[proofcase,2]{wide, parsep=0pt, topsep=0pt, partopsep=0pt, 
  label={\textit{Case \theproofcasei.\arabic*.}},ref=\theproofcasei.\arabic*, itemsep=0pt}

\newenvironment{proofcase*}
  {\proofcase[label={\textit{Case \theproofcasei.\alph*.}},ref=\theproofcasei.\alph*]}
  {\endproofcase}

\crefname{proofcasei}{case}{cases}
\Crefname{proofcasei}{Case}{Cases}
\crefname{proofcaseii}{case}{cases}
\Crefname{proofcaseii}{Case}{Cases}

  {\fleqn[\mathindent]\allowdisplaybreaks\csname align*\endcsname}%
  {\csname endalign*\endcsname}

\newcommand{\CatC}{\mathscr{C}}

\newcommand{\CatG}{\mathscr{G}}

\newcommand{\CatD}{\mathscr{D}}

\newcommand{\MM}{S}
\newcommand{\rM}{\mathit{R'}}
\newcommand{\RM}{R}

\newcommand\mbind{%
  \ensuremath{\mathbin{>\mkern-6.8mu>\mkern-6.9mu=}}
}

\newcommand{\typing}[2]{#1 \vdash #2}
\newcommand{\typingSub}[3]{#1 \vdash_{#3} #2}

\newenvironment{element}{}{}

\newcommand{\ponders}[3]{\ignorespaces}
\newcommand{\TODO}[1]{}

\definecolor{light-gray}{gray}{0.9}
\newcommand{\hole}[1]{\colorbox{light-gray}{\ensuremath{#1}}}

\renewcommand{\Varid}[1]{\mathit{#1}}
\renewcommand{\Conid}[1]{\mathit{#1}}

\DefineVerbatimEnvironment%
  {core}{Verbatim}
  {xleftmargin=\mathindent}

\newcommand\numberthis{\addtocounter{equation}{1}\tag{\theequation}}

\newcommand{\varcitet}[3][]{\citeauthor{#2}#3~[\ifthenelse{\isempty{#1}}{\citeyear{#2}}{\citeyear[#1]{#2}}]}

\DeclareUnicodeCharacter{2080}{_0}
\DeclareUnicodeCharacter{2081}{_1}
\DeclareUnicodeCharacter{2082}{_2}
\DeclareUnicodeCharacter{2083}{_3}
\DeclareUnicodeCharacter{2200}{\forall}
\DeclareUnicodeCharacter{03B1}{\alpha}
\DeclareUnicodeCharacter{03B2}{\beta}
\DeclareUnicodeCharacter{03BB}{\lambda}
\DeclareUnicodeCharacter{03BC}{\mu}
\DeclareUnicodeCharacter{1D538}{\mathbb{A}}
\DeclareUnicodeCharacter{1D53D}{\mathbb{F}}
\DeclareUnicodeCharacter{1D54B}{\mathbb{T}}
\DeclareUnicodeCharacter{220B}{\ni}
\DeclareUnicodeCharacter{25A1}{\box}
\DeclareUnicodeCharacter{03A3}{\ensuremath{\mathrm{\Sigma}}}
\DeclareUnicodeCharacter{03B7}{\ensuremath{\eta}}

\allowdisplaybreaks

\begin{document}

\title{Handling Higher-Order Effectful Operations with Judgemental Monadic Laws}
\titlenote{A version of this paper without appendices will appear in POPL 2026.
The appendices contain detailed proofs and additional explanation of the recursive 
variant of the language studied in this paper.}

\author{Zhixuan Yang}
\orcid{0000-0001-5573-3357}             
\affiliation{
  \department{Department of Computing}    
  \institution{Imperial College London}            
  \country{United Kingdom}                    
}
\email{s.yang20@imperial.ac.uk}          

\author{Nicolas Wu}
\orcid{0000-0002-4161-985X}             
\affiliation{
  \department{Department of Computing}    
  \institution{Imperial College London}            
  \country{United Kingdom}                    
}
\email{n.wu@imperial.ac.uk}          

\begin{abstract}
This paper studies the design of programming languages with handlers of
\emph{higher-order effectful operations} -- effectful operations that may take
in computations as arguments or return computations as output.
We present and analyse a core calculus with higher-kinded impredicative
polymorphism, handlers of higher-order effectful operations, and optionally
general recursion.
The distinctive design choice of this calculus is that handlers are carried by lawless
raw monads, while the computation judgements still satisfy the monadic laws
judgementally.
We present the calculus with a logical framework and give denotational models of the calculus using realizability semantics.
We prove closed-term \emph{canonicity}  and \emph{parametricity} for the
recursion-free fragment of the language using synthetic Tait computability
and a novel form of the $\top\top$-lifting technique.
\end{abstract}

\maketitle

\section{Introduction}\label{sec:intro}

\subsection{What Are Higher-Order Effects and Handlers?}

\subsubsection{Motivating Higher-Order Effects}
One view of  Plotkin and Pretnar's [\citeyear{PlotPret09Hand,PlotPret13Hand}]
\emph{effect handlers} is that they are a language feature that empowers
the programmer to freely extend the programming language with new
syntax, and to interpret the syntax compositionally where needed using
effect handlers. 
The syntax that is possible to be added is restricted to the form of
\emph{generic operations}, each of which takes in a
parameter of some type $P$ and returns a result of some type $A$.
For example, the operation of outputting has 
\ensuremath{\Conid{String}} as the parameter type and the unit type as the return type;
the operation of making a nondeterministic choice has the unit type as its 
parameter type and the Boolean type as the return type.
In a call-by-value calculus parameterised by a set $\Sigma$ of operations, the
typing rule for invoking a generic operation $o$ with a parameter $v$ is simply
\[
\inferrule{
      o : (P, A) \in \Sigma 
    \\ \typingSub{\Gamma}{v : P}{\Sigma}
}{
  \typingSub{\Gamma}{o_v : A}{\Sigma}
}
\]
Generic operations are equivalent to \emph{algebraic operations}
\citep{PP03Alg}, which have the syntax $\typingSub{\Gamma}{o_v(a. \; t) : B}{\Sigma}$, where the return value of an operation is
bound as a variable $a$ in a `continuation' $\typingSub{\Gamma,a:A}{t : B}{\Sigma}$, and $o_v(a. \; t)$ judgementally commutes with sequential
composition: 
\begin{equation}\label{eq:algebraicity}
\ensuremath{\lhskeyword{let}\;\Varid{x}} = o_v(a.\; t)\; \ensuremath{\lhskeyword{in}\;\Varid{k}} \;=\; o_v(a.\;\ensuremath{\lhskeyword{let}\;\Varid{x}\mathrel{=}\Varid{t}\;\lhskeyword{in}\;\Varid{k}})
\end{equation}

In the original calculus of \citeauthor{PlotPret09Hand}, the parameter type $P$
and the return type $A$ must be both \emph{ground types}, such as integers or
Booleans, which do not involve \emph{computations} directly or indirectly.
The reason for this restriction is that in the denotational model, the
semantics of computations depends on the signatures of
the operations, so if the signatures of the operations also
involve computations, they would form a mutual recursion and 
greatly complicate the semantics.

Nonetheless, it is common for programming languages to have
effectful operations that take in computations as arguments or return
computations as results.
Typical examples include {exception catching} \ensuremath{\Varid{try}\;\{\mskip1.5mu \Varid{p}\mskip1.5mu\}\;\Varid{catch}\;\{\mskip1.5mu \Varid{h}\mskip1.5mu\}}, which
has as arguments the exception-raising program $p$ and the exception-handling
program $h$; {parallel composition} \ensuremath{\Varid{par}\;\{\mskip1.5mu \Varid{p}\mskip1.5mu\}\;\{\mskip1.5mu \Varid{q}\mskip1.5mu\}}, which takes in programs
$p$ and $q$ to be executed in parallel;
{scoped resource acquisition} \ensuremath{\Varid{with}\;\Varid{r}\;\{\mskip1.5mu \Varid{p}\mskip1.5mu\}}, which opens/closes 
a resource $r$, e.g.\ a file, before/after entering the scope of a program $p$.
Such `higher-order operations' may be implemented as effect handlers, and indeed
exception handling was a primary motivation for
\citeauthor{PlotPret09Hand}'s proposal of effect handlers.

However, as analysed by \citet{Wu2014}, implementing higher-order operations as
handlers causes the loss of compositionality that is enjoyed by
ordinary operations.
For example, if the programmer decides to write a program $q$ using
exception throwing and catching, and if exception catching 
\ensuremath{\Varid{try}\;\{\mskip1.5mu \Varid{p}\mskip1.5mu\}\;\Varid{catch}\;\{\mskip1.5mu \Varid{h}\mskip1.5mu\}} is implemented as effect handling \ensuremath{\lhskeyword{hdl}\;\Varid{p}\;\lhskeyword{with}\;\{\mskip1.5mu \Varid{throw}\mapsto\Varid{h}\mskip1.5mu\}}
in the program $q$, then the programmer cannot 
give alternative semantics to exception catching in $q$
(e.g.\ after $p$ throws an exception and $h$ is executed, $p$ gets \emph{re-tried}
again), because there is no mechanism for the programmer to reinterpret the
effect handling construct \ensuremath{\lhskeyword{hdl}\;\Varid{p}\;\lhskeyword{with}\;\Varid{h}}. 

To give better treatment of higher-order operations in the framework of
algebraic effects and handlers, a number of authors have studied
\emph{higher-order algebraic effects} and their handlers
\citep{Wu2014,PirogSWJ18,YPWvS2021,vSPW21,vS24,Bach_Poulsen_vd_Rest_2023,Frumin2024,YangWu23}.
Note, however, the precise technical meaning of `higher-order (algebraic) effects'
varies in the cited papers, although they share the connotation of \emph{effectful
operations that may take in computations as arguments and/or return
computations as results}.

\subsubsection{The Quick-and-Dirty Approach}
To begin with, there is \emph{nothing} stopping the language designer simply 
removing
the restriction on the parameter types $P$ and return types $A$ 
to be ground types,
as far as only 
the {type system} and {operational semantics} are concerned.
Indeed, most follow-up work on effect handlers does not have this restriction.
When denotational semantics is concerned, this relaxation necessitates solving
\emph{mixed-variant} recursive equations between the semantics of computation types
and operation signatures, which can be done using techniques from \emph{domain
theory}, as demonstrated by \citet{BauerPretnar2014} and
\citet{Kiselyov_Mu_Sabry_2021} using {classical domain theory}, and more
recently by \citet{Frumin2024} using {synthetic guarded domain theory}.

Simply removing the ground-type restriction in the type system is
a `quick-and-dirty' approach to higher-order algebraic effects, which does
not reveal the inherent structure in higher-order operations.
For example, suppose that \ensuremath{\Varid{try}\;\{\mskip1.5mu \Varid{p}\mskip1.5mu\}\;\Varid{catch}\;\{\mskip1.5mu \Varid{h}\mskip1.5mu\}} is implemented as an
operation \ensuremath{\Varid{catch}} with two computation parameters in this way, and \ensuremath{\Varid{throw}} is
implemented as a nullary operation with no parameters
as usual, a (deep) handler \ensuremath{\{\mskip1.5mu \Varid{catch}\;\Varid{p}\;\Varid{h}\;\Varid{k}\mapsto\cdots;\Varid{throw}\mapsto\cdots\mskip1.5mu\}} for
them would accept \emph{unhandled} computations \ensuremath{\Varid{p}} and \ensuremath{\Varid{h}} as arguments
that may invoke \ensuremath{\Varid{catch}} and \ensuremath{\Varid{throw}}, in contrast to the continuation $k$ 
for which \ensuremath{\Varid{catch}} and \ensuremath{\Varid{throw}} are already handled.
This gives the programmer full flexibility in how to deal with the computation
parameters \ensuremath{\Varid{p}} and \ensuremath{\Varid{h}} but
undermines handlers as a \emph{structured programming construct} that can be
reasoned about effectively.
After all, what makes (deep) handlers stand out among powerful control
operators is their simple mental model for programmers -- a
native form of \emph{catamorphism/fold} that replaces operation calls in the 
program being handled with the corresponding handler clauses.

\subsubsection{The Structured Approach}
To expose the inherent structure in higher-order operations more sharply,
several authors have proposed a number of refined definitions of (subsets of)
higher-order operations: 
\emph{scoped operations}
\citep{PirogSWJ18,YPWvS2021,Matache2025,BvTT2024}, \emph{latent operations}
\citep{vSPW21}, \emph{hefty operations} \citep{Bach_Poulsen_vd_Rest_2023}, and
the general frameworks by \citet{Wu2014}, \citet{YangWu23} and \citet{vS24}.
Regardless of the technical differences in their proposals, the common idea is
that the signature of a higher-order operation can take a different form from
ordinary operations, and the carrier of a handler (usually called an algebra in
this line of work) for a higher-order operation does not have to be a type
anymore, and usually is a type constructor of kind \ensuremath{\Conid{Type}\to \Conid{Type}}.

Taking scoped operations as an example, in the formulation of
\citet{BvTT2024}, the signature of a scoped operation $s : (P,  S)$ consists
of two types: $P$ still means that the operation $s$ takes in
a parameter of type $P$, but the type $S$ no longer means that the operation
returns a value of type $S$.
Instead, it means that the operation $s$ delimits $S$-many scopes;
for example, $S$ would be the two-element type for \ensuremath{\Varid{try}\;\{\mskip1.5mu \cdots\mskip1.5mu\}\;\Varid{catch}\;\{\mskip1.5mu \cdots\mskip1.5mu\}} because 
it delimits two scopes.
The typing rule for invoking scoped operations in a
call-by-value programming language is
\begin{equation}\label{eq:scp:op:call}
\inferrule{
       s : (P, S) \in \Sigma 
    \\ \typingSub{\Gamma}{v : P}{\Sigma}
    \\ \typingSub{\Gamma, x : S}{ c : A }{\Sigma}
}{
  \typingSub{\Gamma}{s_v\; \{x.\; c\} : A}{\Sigma}
}
\end{equation}
where the term $c$ represents \ensuremath{\Conid{S}}-many computations that the operation $s$
takes in as arguments; these computations can return an \emph{arbitrary} type
$A$, which will also be the type of the whole operation call.
This rule here is a simplified
version of the calculus by \citet{BvTT2024}, but it is the essence.

The reader might have noted that scoped operations $s_v\; \{x.\; c\}$ have
essentially the same syntax as algebraic operations $o_v(a.\; t)$.
Their difference is that scoped operations do not have to satisfy the equation
$(\ensuremath{\lhskeyword{let}\;\Varid{y}} = s_v\;\{x.\; c\}\; \ensuremath{\lhskeyword{in}\;\Varid{k}}) = s_v\;\{ x.\;\ensuremath{\lhskeyword{let}\;\Varid{y}\mathrel{=}\Varid{c}\;\lhskeyword{in}\;\Varid{k}} \}$, while
algebraic operations have to satisfy \Cref{eq:algebraicity}.
Informally, this means that $s_c\{x.\; c\}$ genuinely delimits a scope $\{x.\; c\}$ that it acts on, whereas the scope $(a.\; t)$ delimited by an algebraic operation $o_v(a. t)$
is superfluous.

Because a call to a scoped operation $s : (P,S)$ is polymorphic in the return
type $A$, a handler for the scoped operation $s$ will also need to be carried
by a type constructor $\ensuremath{\Conid{M}\mathbin{:}\Conid{Type}\to \Conid{Type}}$ rather than just a type.
A possible typing rule for handling a scoped operation $s$ is
\begin{equation}\label{eq:scp:op:hdl}
\inferrule{
\typingSub{\Gamma}{p : A}{\Sigma} \\
\typingSub{\Gamma, a : \alpha}{r : M\;\alpha}{\Sigma} \\
\typingSub{\Gamma,\; v :  P,\; c : S \to M\; \alpha}{h : M \; \alpha}{\Sigma}
}{
\typingSub{\Gamma}{\ensuremath{\lhskeyword{hdl}\;\Varid{p}\;\lhskeyword{with}\;\{\mskip1.5mu \lhskeyword{val}\;\Varid{a}\mapsto\Varid{r};\Varid{s}\;\Varid{v}\;\Varid{c}\mapsto\Varid{h}\mskip1.5mu\}} : M\;A}{\Sigma}
}
\end{equation}
in which $\alpha$ is a free type variable.
The practical difference between this approach and the quick-and-dirty approach
mentioned above is that (1) the types $P$ and $S$ in the signature of a scoped
operation can still be ground types, so that
the denotational semantics involves no mixed-variance recursive equations, and
reasoning principles such as fusion laws are still available \citep{YPWvS2021}, 
and (2) a handler $h$ of $s$ always works with recursively handled computations
$c : S \to M \; \alpha$,
so the programmer's mental model of handling scoped operations $s$ can still
be a fold that replaces every return value \ensuremath{\lhskeyword{val}\;\Varid{x}} with the term $r$ and
every operation call to $s$ with the
corresponding handler clause $h$, including those
calls nested in the scopes of other calls such as $s_v\;\{x.\fbox{$s_u\;\{y. \cdots\}$} \}$.

\subsection{Interaction of Effect Handlers and Sequential Composition}

\subsubsection{A Problem with Sequential Composition}
The rules above for invoking \Cref{eq:scp:op:call} and handling
\Cref{eq:scp:op:hdl} scoped operations are still not the end of the story. 
Any effectful language should have a construct for sequential
composition of computations, such as \ensuremath{\lhskeyword{let}}-bindings. 
A question then
is \emph{how handling should interact with sequential composition}. 
Concretely, suppose that we have a computation 
\[
\inferrule{
\typingSub{\Gamma}{s_v\{x. \; c\} : A}{\Sigma}
\\
\typingSub{\Gamma, a : A}{d : B}{\Sigma}
}{
\typingSub{\Gamma}{\ensuremath{\lhskeyword{let}}\; a = s_v\{x. \; c\} \;\ensuremath{\lhskeyword{in}}\; d : B}{\Sigma}
}
\]
that is the sequential composition of a scoped operation $s_v\{x. \; c\}$
followed by another computation $d$.
What should be the result of applying the handling construct
\Cref{eq:scp:op:hdl} to this computation?
Applying the handler recursively to $s_v\{x. \; c\}$ and $d$ would
give us two terms of type 
\begin{equation}\label{eq:handling:recursive:results}
\spread{
\Gamma \vdash_\Sigma \cdots : M\; A
\also
\text{and}
\also
\Gamma, a : A\vdash_\Sigma \cdots : M\; B
}
\end{equation}
and our goal is to have a term of type $M\;B$.
Now we have two ways to proceed:
\begin{enumerate}
\item[(i)]\label{approach:monad}
asking $M$ to additionally come with a monadic bind
$\bind : M\;\alpha \to (\alpha \to M\;\beta) \to M\;\beta$,
using which we can combine the recursively handled results
\Cref{eq:handling:recursive:results} into $M\;B$.

\item[(ii)]\label{approach:kont}
modifying the handler rule \Cref{eq:scp:op:hdl} so that the handler of $s$ takes in a
continuation parameter
\[
\inferrule{
\cdots \\
\typingSub{\Gamma,\; v :  P,\; c : S \to M\; \alpha, \;
  \fbox{$k : \alpha \to M\; \beta$}}
  {h' : M \,\fbox{$\beta$}}{\Sigma}
}{
\typingSub{\Gamma}{\ensuremath{\lhskeyword{hdl}\;\Varid{p}\;\lhskeyword{with}\;\{\mskip1.5mu \lhskeyword{val}\;\Varid{a}\mapsto\Varid{r};\Varid{s}\;\Varid{v}\;\Varid{c}\;\fbox{$k$}\mapsto\Varid{h'}\mskip1.5mu\}} : M\;A}{\Sigma}
}
\]
and the recursively handled result $\typingSub{\Gamma, a : A}{\cdots : M \; B}{\Sigma}$ of $d$ is supplied as the continuation parameter $k$ to the handler
when handling the operation $s_v\{x.\;c\}$.
\end{enumerate}
To programming language designers familiar with effect handlers,
\hyperref[approach:kont]{approach (ii)} may appear as the more natural one, since there is already a
similar continuation parameter for handlers of generic/algebraic operations in
\citeauthor{PlotPret09Hand}'s design. 
Indeed, \citet{BvTT2024}'s calculus of scoped effects and handlers follows
this approach, but let us not commit to this choice too quickly; instead,
let us first analyse the connections and differences between these two
approaches.

From an algebraic point of view, the difference in these two approaches lies in
their views of the \emph{universal property} of effectful computations.
Crudely speaking, \hyperref[approach:monad]{approach (i)} views computations with operation $s$ as the initial object
among \emph{monads} $M$ on types equipped with operations of type $\ensuremath{\forall \Varid{\alpha}\hsforall \hsdot{\circ }{.\ }\Conid{P}\to (\Conid{S}\to \Conid{M}\;\Varid{\alpha})\to \Conid{M}\;\Varid{\alpha}}$,
whereas \hyperref[approach:kont]{approach (ii)} views computations with $s$ as the initial
object among type constructors \ensuremath{\Conid{M}\mathbin{:}\Conid{Type}\to \Conid{Type}} equipped with operations of type \ensuremath{\forall \Varid{\alpha}\hsforall \hsdot{\circ }{.\ }\Varid{\alpha}\to \Conid{M}\;\Varid{\alpha}} and
\[
\ensuremath{\forall \Varid{\alpha}\hsforall \;\Varid{\beta}\hsdot{\circ }{.\ }\Conid{P}\to (\Conid{S}\to \Conid{M}\;\Varid{\alpha})\to (\Varid{\alpha}\to \Conid{M}\;\Varid{\beta})\to \Conid{M}\;\Varid{\beta}}.
\]

\subsubsection{An Analogy to Lists}
The contrast between these two views is analogous to the following more
familiar situation.
The type of lists \ensuremath{[\mskip1.5mu \Conid{A}\mskip1.5mu]} of elements of type $A$ has the universal property that
\ensuremath{[\mskip1.5mu \Conid{A}\mskip1.5mu]} together with the empty list $\ensuremath{\Varid{nil}} : 1 \to [A]$ and \ensuremath{\Varid{cons}\mathbin{:}\Conid{A}\to [\mskip1.5mu \Conid{A}\mskip1.5mu]\to [\mskip1.5mu \Conid{A}\mskip1.5mu]} is initial among all types $B$ equipped with functions $1 \to B$ and
$A \to B \to B$.
But the type \ensuremath{[\mskip1.5mu \Conid{A}\mskip1.5mu]} has another universal property: first of all, it can be
equipped with a monoid structure of \ensuremath{\Varid{nil}\mathbin{:}\mathrm{1}\to [\mskip1.5mu \Conid{A}\mskip1.5mu]} and list concatenation \ensuremath{\mathbin{+\!\!+}\mathbin{:}[\mskip1.5mu \Conid{A}\mskip1.5mu]\to [\mskip1.5mu \Conid{A}\mskip1.5mu]\to [\mskip1.5mu \Conid{A}\mskip1.5mu]}, and this monoid with the function $(\lambda
x.\; \ensuremath{\Varid{cons}\;\Varid{x}\;\Varid{nil}}) : A \to [A]$ has the property of being the free monoid over
the type $A$
(which means that for every monoid $M$ with a function $f : A \to M$, there is
a unique monoid homomorphism $h : [A] \to M$ such that $\ensuremath{\Varid{h}\;(\Varid{cons}\;\Varid{x}\;\Varid{nil})} = f\; x$ for
all $x \in A$).
Therefore the same type $[A]$ can be equipped with different algebraic
structures, giving rise to different universal properties.  
It is pointless to ask which of the universal properties is `the correct one' for the type
\ensuremath{[\mskip1.5mu \Conid{A}\mskip1.5mu]} \emph{per se}.
The right question should be -- which class of algebraic structures are we 
interested in when using \ensuremath{[\mskip1.5mu \Conid{A}\mskip1.5mu]}, monoids $M$ with functions 
$A \to M$ or types $B$ with functions $1 \to B$ and $A \to B \to B$?

Similarly, the right question to ask about approaches \hyperref[approach:monad]{(i)} and \hyperref[approach:kont]{(ii)} above should
be -- when programming with scoped effects, are we interested in (i) monads $M$
with operations of type
\ensuremath{\forall \Varid{\alpha}\hsforall \hsdot{\circ }{.\ }\Conid{P}\to (\Conid{S}\to \Conid{M}\;\Varid{\alpha})\to \Conid{M}\;\Varid{\alpha}}
or (ii) type constructors $M : \ensuremath{\Conid{Type}\to \Conid{Type}}$ equipped with operations of type
$\ensuremath{\forall \Varid{\alpha}\hsforall \;\Varid{\beta}\hsdot{\circ }{.\ }\Conid{P}\to (\Conid{S}\to \Conid{M}\;\Varid{\alpha})\to (\Varid{\alpha}\to \Conid{M}\;\Varid{\beta})\to \Conid{M}\;\Varid{\beta}}$
and \ensuremath{\forall \Varid{\alpha}\hsforall \hsdot{\circ }{.\ }\Varid{\alpha}\to \Conid{M}\;\Varid{\alpha}}?

\subsubsection{Sequential Composition is An Operation}
To this question, we advocate for the answer (i).
Our rationale is that sequential composition \emph{ought to be} an operation
for a notion of computation, rather than only a meta-level operation that the
initial object `accidentally' supports.
This view is supported by the fact that in practice it is not uncommon to
consider equations involving both effectful operations and sequential
composition. 
For example, in the study of process algebra \citep{Bergstra_1985}, the
distributivity equations
\[
\spread*{
(p + q); r = (p; r) + (q; r)
\also
p; (q + r) = (p; q) + (p; r)
}
\]
between nondeterministic choice and sequential composition are usually
considered (and most models of concurrency satisfy the former but not the
latter).  
If sequential composition is not an operation in the algebraic theory
of the effect of concurrency, these two equations would not be expressible, and
it would be meaningless to ask whether a handler
of concurrency satisfies these two equations if we followed \hyperref[approach:kont]{approach (ii)}.

The point that we are raising here holds regardless of whether the
programming language formally checks equations on effectful
operations -- even in simply typed programming languages, programmers usually
have equations informally in their minds and reason about effectful programs
with these equations \citep{Gibbons_Hinze_2011}.
Also, our point is not specifically about higher-order operations
either: even for algebraic/generic operations like nondeterministic choice $p +
q$, we may already want to consider equations about the interaction of the
effectful operations with sequential composition.  And this is impossible in
standard algebraic effects
\citep{PlotkinP02} because sequential composition is not an operation in the
theory of the effect, but only an operation that the free algebras determined by
the algebraic theory `happen to have'.  

\subsubsection{A Problem with Laws}
As the final twist of our discussion of higher-order effect handlers 
in this section, we discuss how we should deal with the monadic
laws.
As we advocated above, our handlers of higher-order algebraic effects will
be carried by monads, which are type constructors \ensuremath{\Conid{M}\mathbin{:}\Conid{Type}\to \Conid{Type}} with
operations for returning \ensuremath{\Varid{ret}\mathbin{:}\forall \Varid{\alpha}\hsforall \hsdot{\circ }{.\ }\Varid{\alpha}\to \Conid{M}\;\Varid{\alpha}} 
and
sequential composition \ensuremath{\mbind\mathbin{:}\forall \Varid{\alpha}\hsforall \;\Varid{\beta}\hsdot{\circ }{.\ }\Conid{M}\;\Varid{\alpha}\to (\Varid{\alpha}\to \Conid{M}\;\Varid{\beta})\to \Conid{M}\;\Varid{\beta}} satisfying the monadic laws
\begin{equation}\label{eq:monadic:laws}
\spread{
\ensuremath{\Varid{ret}\;\Varid{a}\mbind\Varid{k}\mathrel{=}\Varid{k}\;\Varid{a}}
\also
\ensuremath{\Varid{m}\mbind\Varid{ret}\mathrel{=}\Varid{m}}
\also
\ensuremath{(\Varid{m}\mbind\Varid{k})\mbind\Varid{k'}\mathrel{=}\Varid{m}\mbind(\lambda \Varid{x}.\;\Varid{k}\;\Varid{x}\mbind\Varid{k'})}
} 
\end{equation}
asserting that \ensuremath{\Varid{ret}} is the left and right identity of \ensuremath{\mbind} and \ensuremath{\mbind} is 
associative.
If our effectful programming language has dependent types -- in particular, 
identity types -- we can demand every handler to come with proofs
of the equations \Cref{eq:monadic:laws}, and the type checker can check the proofs
mechanically.
But if our language does not have dependent types,
it will be impossible for us to mechanically check if
these laws are satisfied. 
In this case our programming language can only ask
handlers to be carried by `raw monads' $\tuple{M, \ensuremath{\Varid{ret}}, \ensuremath{\mbind}}$ that do not
necessarily satisfy these laws.
Therefore, in the absence of dependent types, we as the language designer will
not enforce the programmer to supply lawful monads, but there is a closely
related but different question:

\begin{question}\label{question}
When handlers of computations are carried by
raw monads not necessarily satisfying monadic laws \Cref{eq:monadic:laws},
can we still make \ensuremath{\lhskeyword{let}}-bindings and \ensuremath{\lhskeyword{val}}-returning of the computation
judgements satisfy the following monadic laws judgementally?
\begin{gather}
\spread{
\ensuremath{(\lhskeyword{let}\;\Varid{x}\mathrel{=}\lhskeyword{val}\;\Varid{a}\;\lhskeyword{in}\;\Varid{k})}  \ =\   k[a/x]
\also
\ensuremath{(\lhskeyword{let}\;\Varid{x}\mathrel{=}\Varid{m}\;\lhskeyword{in}\;\lhskeyword{val}\;\Varid{x})}  \ =\   \ensuremath{\Varid{m}}
}
\label{eq:comp:monad:laws:val}
\\
\ensuremath{(\lhskeyword{let}\;\Varid{y}\mathrel{=}(\lhskeyword{let}\;\Varid{x}\mathrel{=}\Varid{m}\;\lhskeyword{in}\;\Varid{k})\;\lhskeyword{in}\;\Varid{k'})} \ =\ \ensuremath{(\lhskeyword{let}\;\Varid{x}\mathrel{=}\Varid{m}\;\lhskeyword{in}\;(\lhskeyword{let}\;\Varid{y}\mathrel{=}\Varid{k}\;\lhskeyword{in}\;\Varid{k'}))}
\label{eq:comp:monad:laws:let}
\end{gather}
\end{question}

It might seem that the answer to this question would be unavoidably
negative, since the computations of our language are handled into raw monads
that do not necessarily satisfy the monadic laws, and syntax
shall only satisfy the equations that are satisfied in \emph{all} semantic
models. 
This would be rather unfortunate because these laws are arguably the most
fundamental algebraic properties of a notion of computation, and they are
needed by programmers to reason about effectful programs and compilers to do
optimisation, for example, to rewrite a computation that invokes no effectful
operations to the form \ensuremath{\lhskeyword{val}\;\Varid{v}} of returning a pure value \ensuremath{\Varid{v}}.

\subsection{Contributions of This Paper}
In this paper we show that the answer to \Cref{question} is actually positive,
provided that we are willing to \emph{not} have the commutativity of handlers
and \ensuremath{\lhskeyword{let}}-bindings:
\[
\big(\ensuremath{\lhskeyword{hdl}\;(\lhskeyword{let}\;\Varid{x}\mathrel{=}\Varid{m}\;\lhskeyword{in}\;\Varid{k})\;\lhskeyword{with}\;\Varid{h}}\big)
\;=\;
\big(\ensuremath{\lhskeyword{let}\;\Varid{x}\mathrel{=}(\lhskeyword{hdl}\;\Varid{m}\;\lhskeyword{with}\;\Varid{h})\;\lhskeyword{in}\;(\lhskeyword{hdl}\;\Varid{k}\;\lhskeyword{with}\;\Varid{h})}\big)
\]
We present a core calculus for higher-order effects and 
handlers that we call System \Fha{} (\Cref{sec:fha}).
This calculus extends \varcitet{Girard1986}{'s}{} System \Fom{} with
higher-order effects in the fine-grain call-by-value style.
Specifically, the signature of a higher-order algebraic effect is given
as a higher-order functor $H : \ensuremath{(\Conid{Type}\to \Conid{Type})\to (\Conid{Type}\to \Conid{Type})}$ 
in \Fha{} following \citet{YangWu23} and 
\citet{vS24}'s categorical frameworks.
Every signature $H$ has a corresponding judgement of computations that supports
\ensuremath{\lhskeyword{let}}-bindings, \ensuremath{\lhskeyword{val}}-returning, invoking $H$-operations, and being handled 
with raw monads equipped with $H$-operations.

As promised, the computation judgements of \Fha{} satisfy the equations
\Cref{eq:comp:monad:laws:val,eq:comp:monad:laws:let}, and we show that the
design of \Fha{} `works' by establishing the following meta-theoretic
properties about \Fha:

\begin{enumerate}
\item
In \Cref{sec:fha:real:model},
a denotational model of \Fha{} is given based on realizability,
establishing the \emph{consistency}  of the equational theory of \Fha{} 
(\Cref{thm:fha:consistency})
and providing a
way
to translate \Fha{} terms to untyped computational models such as
$\lambda$-calculus or Turing machines (\Cref{thm:fha:extraction}).
The key idea of this model is to use
a continuation-passing-style (CPS) transformation to reconcile the mismatch
between lawful computation judgements and lawless handlers in \Cref{question}.
An extension of \Fha{} with general recursion is also considered, and the
realizability model is extended to support recursion (\Cref{thm:fhar:extraction}) using \emph{synthetic domain theory} \citep{Longley_Simpson_1997}.

\item
In \Cref{sec:fha:lr:model}, a logical relation model of \Fha{} is constructed using the method of
\emph{synthetic Tait computability} \citep{sterling:2021:thesis} and a novel
version \emph{$\top\top$-lifting} \citep{Lindley_Stark_2005}.  From this model
we obtain \emph{canonicity} (\Cref{thm:canonicity}) and 
\emph{parametricity} (\Cref{rem:parametricity}) of closed \Fha-terms,
which also imply the \emph{adequacy} of the realizability model of \Fha{} 
(\Cref{cor:adequacy}).
\end{enumerate}

What makes us think that \Fha{} is interesting is that it demonstrates how we
can recover important equations on computations even when the language is not
expressive enough to follow the foundational mathematical theory fully
faithfully.
The retained equational laws are \emph{sound} with respect to a
`compiler' (the realizability model) and are \emph{complete} enough for
reducing closed terms to canonical values.
Therefore these laws are especially useful for compiler writers,
who need equation laws to do optimisations but cannot trust the equations
existing only as code comments.

A methodological character of this paper is its heavy use of modern
type-theoretic and category-theoretic tools to study the language \Fha{}.  
Although this paper is too short to serve as a fully satisfactory introduction
to these tools, we hope that there is still some pedagogical value in this
paper by showing how these tools are used coherently to present and
analyse a polymorphic programming language that is not too simplistic or
complicated. 
We hope this can contribute to making these powerful abstract tools more
accessible to working programming language theorists.

\section{A Core Calculus for Higher-Order Effect Handlers}\label{sec:fha}
In this section we present the type theory that we call System \Fha, an 
extension of System \Fom{} with handlers of higher-order effectful operations.
Instead of defining the calculus by a grammar and typing rules in the traditional way, 
in this paper we use a \emph{logical framework} to present \Fha.
Logical frameworks are type theories designed for defining other type theories, and
they usually provide useful general results for theories definable in the
framework. 
In particular, the logical framework that we will use frees us from
manipulating variables and substitutions manually and provides a notion of
\emph{semantic models} of \Fha{} automatically.

The structure of this section is as follows.
In \Cref{sec:lf}, we briefly introduce the logical framework that we will use.
In \Cref{sec:fhom}, we define \varcitet{Girard1986}{'s}{} System \Fom{} in
the logical framework, which is going to be the basis of our language \Fha{}. 
In \Cref{sec:fha:co}, we add computational judgements to \Fom{}, giving us
\Fha{}.
Finally, in \Cref{sec:gen:rec}, we consider an extension of \Fha{} with 
general recursion.

\subsection{A Logical Framework}\label{sec:lf}

The logical framework (LF) that we will use is the one informally introduced by
\citet{sterling:2021:thesis} in his PhD thesis, which is called \LCCLF{} by
\citet{Yang2025} in his more formal treatment of this framework.
This LF has been used by a number of authors in the study of various type
theories
\citep{sterling2021LRAT,Sterling_Angiuli_2021,Grodin_Niu_Sterling_Harper_2024,Niu2022,Sterling_et_al_2022}.
We aim to be self-contained about the LF in this paper, but our introduction
is unavoidably rather terse and we refer the reader to \citet[Chapter
1]{sterling:2021:thesis} and \citet{Yang2025} if needed.

\begin{typetheory}\label{lang:lf}
The logical framework $\LCCLF$ is a type theory with a universe $\Jdg$ and 
\begin{itemize}
\item
inside the universe $\Jdg$, the type formers of extensional \MLTT{}
(the unit type $1$, $\Sigma$-types, $\Pi$-types, extensional equality types $x
= y$), and

\item
outside the universe $\Jdg$, the unit type $1$, all $\Sigma$-types, and
\emph{restricted} $\Pi$-types $\Pi\; A\; B$ where the domain type $A$ must be in
$\Jdg$ (so $\Pi \; \Jdg \; (\lambda \textunderscore.\;\Jdg)$ will not be a valid type).
\end{itemize}
\end{typetheory}

\begin{notation}
We adopt some concrete syntax similar to Agda
\citep{Agda_Developers_Agda} when working with dependent type theories. 
Dependent function types are written as $(a : A) \to B$ where $a$ may occur
in $B$, or $A \to B$
when $B$ does not depend on $a : A$. 
Iterated $\Sigma$-types are written as records of fields with labels.
Implicit function types $\impfun{a : A} B$ are used when the arguments can
be inferred.
Things whose names are irrelevant are denoted by the wildcard 
`\textunderscore'.
\end{notation}

The way to define an object type theory in the LF is to write a
\emph{signature}, which is a sequence of variable declarations in the LF, where
the type of each variable may depend on the preceding declarations (so
formally, a signature is exactly a context in the LF).
The idea is the \emph{judgements-as-types} principle  as follows
\citep{Harper1993,MartinLof1987}:
\begin{enumerate}
\item
\emph{Judgement forms} of the object theory are declared as LF-functions $A \to
\Jdg$ into the universe $\Jdg$.
For example, the judgement form of a proposition being true, traditionally
written as `$\hypo{}{P\; \textrm{true}}$', is declared as $\ensuremath{\Varid{true}} : \ensuremath{\Varid{prop}} \to
\Jdg$, assuming some $\ensuremath{\Varid{prop}} : \Jdg$ is already declared.

\item
\emph{Inference rules} for object-theory judgements are declared as
LF-functions; for example, the declaration $\ensuremath{\Varid{and\hyp{}intro}\mathbin{:}(\Conid{P},\Conid{Q}\mathbin{:}\Varid{prop})\to \Varid{true}\;\Conid{P}\to \Varid{true}\;\Conid{Q}\to \Varid{true}\;(\Conid{P}\mathrel{\wedge}\Conid{Q})}$ says that the judgement of \ensuremath{\Conid{P}\mathrel{\wedge}\Conid{Q}} being true can be derived from both
$P$ and $Q$ being true.

\item
\emph{Judgemental equalities} of the object-theory can be treated in two ways:
(i) they can be declared as judgements in $\Jdg$ just like other judgements, or
(ii) they can be declared using the (extensional) equality types of the LF. 
Logical frameworks following the former approach are sometimes called
\emph{syntactic logical frameworks}, and those following the latter are 
called \emph{semantic logical frameworks}; see \citet{Harper2016}
and \citet[\S 0.1.2.2]{sterling:2021:thesis} for a comparison between
them.
We will follow the semantic approach, which has the advantage that there is no
need to have the tedious congruence rules for the judgemental equalities
w.r.t.\ all constructs of the object theory, since equality types in the LF are
always respected.
\end{enumerate}
These points will be demonstrated concretely in the example of defining 
Girard's~[\citeyear{Girard1986,Girard1972}] System \Fom{} in the LF below.
Compared to the traditional `gamma-and-turnstile' presentation, there are
three advantages of using \LCCLF{} to present our language \Fha:
\begin{enumerate}
  \item It provides a compact type-theoretic notation to present the rules of 
  \Fha{}, and by using \emph{higher-order abstract syntax} (HOAS),
  standard components of \Fha{} such as contexts and substitutions can be dealt
  with automatically.

  \item It provides a notion of \emph{models} of \Fha{} in any locally
  cartesian closed category (LCCC) $\CatC$.
  By using the internal language of $\CatC$, models of \Fha{} can be defined in
  a type-theoretic manner.

  \item It provides a \emph{classifying category} for \Fha{}, which is an LCCC
  $\CatJdgOf{\Fha}$ such that models of \Fha{} in any LCCC $\CatC$ are
  equivalent to LCCC-functors $\CatJdgOf{\Fha} \to \CatC$. 
  Applying category-theoretic tools to the category $\CatJdgOf{\Fha}$, such as
  Yoneda embedding and Artin gluing, we can do logical-relation proofs for
  \Fha{} in a convenient type-theoretic language.
\end{enumerate}

\subsection{The Signature of System \Fom}\label{sec:fhom}
In the following, we present the signature of \Fha{} in two steps: we first
define Girard's~[\citeyear{Girard1986,Girard1972}] System \Fom{} in the LF 
(\Cref{lang:lf}), and in the next section we bring in computations.

\subsubsection{Kinds}
System \Fom{} has the following declarations for \emph{kinds}:
\begin{align}
&\ensuremath{\Varid{ki}\mathbin{:}\Jdg}&
&\ensuremath{\Varid{el}\mathbin{:}\Varid{ki}\to\Jdg}&
&\ensuremath{\Varid{ty}\mathbin{:}\Varid{ki}}&
&\ensuremath{\text{\textunderscore}{\Rightarrow_k}\text{\textunderscore}\mathbin{:}\Varid{ki}\to\Varid{ki}\to\Varid{ki}}
\tag{\Fom-1}
\label{eq:fha:kinds}
\end{align}
where we have a judgement form \ensuremath{\Varid{ki}} for kinding, a judgement form \ensuremath{\Varid{el}} for
elements of kinds, a base kind \ensuremath{\Varid{ty}\mathbin{:}\Varid{ki}} whose elements will be \emph{types},
and function kinds \ensuremath{\Varid{k}_{\mathrm{1}}\;{\Rightarrow_k}\;\Varid{k}_{\mathrm{2}}}.
These declarations correspond to the following things in the traditional presentation.
The declaration \ensuremath{\Varid{ki}\mathbin{:}\Jdg} corresponds to the judgement form
`$\iskind{\Gamma}{\cdots}$' of something being a kind, and accordingly an
LF-element \ensuremath{\Varid{k}\mathbin{:}\Varid{ki}} corresponds to `$\iskind{\Gamma}{k}$'.
The LF-type \ensuremath{\Varid{el}\;\Varid{k}\mathbin{:}\Jdg} for some \ensuremath{\Varid{k}\mathbin{:}\Varid{ki}} corresponds to the judgement form
`$\hypo{\Gamma}{\cdots : k}$' of something being an element of a kind.
The two declarations \ensuremath{\Varid{ty}} and \ensuremath{\text{\textunderscore}{\Rightarrow_k}\text{\textunderscore}} correspond to two inference rules
for constructing kinds:
\[
\spread*{
\inferrule{ }{\iskind{\Gamma}{\ensuremath{\Varid{ty}}}}
\also
\inferrule{\iskind{\Gamma}{k_1} \\ \iskind{\Gamma}{k_2}}
{\iskind{\Gamma}{\ensuremath{\Varid{k}_{\mathrm{1}}\;{\Rightarrow_k}\;\Varid{k}_{\mathrm{2}}}}}
}
\]


Elements of the base kind $\ensuremath{\Varid{ty}\mathbin{:}\Varid{ki}}$ include a unit type \ensuremath{\Varid{unit}}, a
two-element type \ensuremath{\Varid{bool}}, function types \ensuremath{\Conid{A}\;{\Rightarrow_t}\;\Conid{B}}, and impredicative
polymorphic function types \ensuremath{\allfor\;\Varid{k}\;\Conid{A}} where $k$ can be of any kind:
\begin{equation}
\begin{gathered}
\ensuremath{\Varid{unit}\mathbin{:}\Varid{el}\;\Varid{ty}} \hspace{1.5cm} \ensuremath{\Varid{bool}\mathbin{:}\Varid{el}\;\Varid{ty}} \hspace{1.5cm} \ensuremath{\text{\textunderscore}{\Rightarrow_t}\text{\textunderscore}\mathbin{:}\Varid{el}\;\Varid{ty}\to\Varid{el}\;\Varid{ty}\to\Varid{el}\;\Varid{ty}}  \\
\ensuremath{\allfor\mathbin{:}(\Varid{k}\mathbin{:}\Varid{ki})\to(\Varid{el}\;\Varid{k}\to\Varid{el}\;\Varid{ty})\to\Varid{el}\;\Varid{ty}}
\end{gathered}
\tag{\Fom-2}
\end{equation}
We use the $\forall$ symbol with a bar for the polymorphic function type so that we will not confuse it with meta-level universal quantification later.

Elements of function kinds are specified using \emph{higher-order
abstract syntax} (HOAS) via an isomorphism to functions in the LF:
\begin{equation}\label{eq:fha:el:kfun}
\ensuremath{\Varid{{\Rightarrow_k}\hyp{}iso}\mathbin{:}\{\mskip1.5mu \Varid{k}_{\mathrm{1}},\Varid{k}_{\mathrm{2}}\mathbin{:}\Varid{ki}\mskip1.5mu\}\to\Varid{el}\;(\Varid{k}_{\mathrm{1}}\;{\Rightarrow_k}\;\Varid{k}_{\mathrm{2}})\cong(\Varid{el}\;\Varid{k}_{\mathrm{1}}\to\Varid{el}\;\Varid{k}_{\mathrm{2}})}
\tag{\Fom-3}
\end{equation}
where the type $\cong$ of isomorphisms between two LF-types $A$ and $B$ is
the following record type:
\begin{hscode}\SaveRestoreHook
\column{B}{@{}>{\hspre}l<{\hspost}@{}}%
\column{3}{@{}>{\hspre}l<{\hspost}@{}}%
\column{9}{@{}>{\hspre}l<{\hspost}@{}}%
\column{E}{@{}>{\hspre}l<{\hspost}@{}}%
\>[B]{}\lhskeyword{record}\;\Conid{A}\;{\cong}\;\Conid{B}\;\lhskeyword{where}{}\<[E]%
\\
\>[B]{}\hsindent{3}{}\<[3]%
\>[3]{}\Varid{fwd}{}\<[9]%
\>[9]{}\mathbin{:}\Conid{A}\to \Conid{B}{}\<[E]%
\\
\>[B]{}\hsindent{3}{}\<[3]%
\>[3]{}\Varid{bwd}{}\<[9]%
\>[9]{}\mathbin{:}\Conid{B}\to \Conid{A}{}\<[E]%
\\
\>[B]{}\hsindent{3}{}\<[3]%
\>[3]{}\anonymous {}\<[9]%
\>[9]{}\mathbin{:}(\Varid{a}\mathbin{:}\Conid{A})\to \Varid{bwd}\;(\Varid{fwd}\;\Varid{a})\mathrel{=}\Varid{a}{}\<[E]%
\\
\>[B]{}\hsindent{3}{}\<[3]%
\>[3]{}\anonymous {}\<[9]%
\>[9]{}\mathbin{:}(\Varid{b}\mathbin{:}\Conid{B})\to \Varid{fwd}\;(\Varid{bwd}\;\Varid{b})\mathrel{=}\Varid{b}{}\<[E]%
\ColumnHook
\end{hscode}\resethooks
Let us unpack the declaration \Cref{eq:fha:el:kfun} a bit and see how it 
corresponds to the traditional presentation.
Given two kinds $k_1, k_2 : \ensuremath{\Varid{ki}}$, the record \ensuremath{\Varid{el}\;(\Varid{k}_{\mathrm{1}}\;{\Rightarrow_k}\;\Varid{k}_{\mathrm{2}})\cong(\Varid{el}\;\Varid{k}_{\mathrm{1}}\to\Varid{el}\;\Varid{k}_{\mathrm{2}})} consists of four fields:
the forward-direction function \ensuremath{\Varid{el}\;(\Varid{k}_{\mathrm{1}}\;{\Rightarrow_k}\;\Varid{k}_{\mathrm{2}})\to (\Varid{el}\;\Varid{k}_{\mathrm{1}}\to\Varid{el}\;\Varid{k}_{\mathrm{2}})} 
says that whenever we have an element of the kind \ensuremath{\Varid{k}_{\mathrm{1}}\;{\Rightarrow_k}\;\Varid{k}_{\mathrm{2}}} and an element
of the kind \ensuremath{\Varid{k}_{\mathrm{1}}}, we can construct an element of the kind \ensuremath{\Varid{k}_{\mathrm{2}}}.
This corresponds to the following rule in the traditional presentation:
\[
\inferrule{
\hypo{\Gamma}{F : k_1  \Rightarrow_k k_2 }
\\
\typing{\Gamma}{A : k_1 }
}{
\typing{\Gamma}{F \; A : k_2}
}
\]
The backward direction \ensuremath{(\Varid{el}\;\Varid{k}_{\mathrm{1}}\to\Varid{el}\;\Varid{k}_{\mathrm{2}})\to \Varid{el}\;(\Varid{k}_{\mathrm{1}}\;{\Rightarrow_k}\;\Varid{k}_{\mathrm{2}})}
takes in an LF-function as its argument.
In HOAS, LF-functions correspond to adding new variables to the context
of the object theory, and applications of LF-functions correspond to substitutions in the object theory, 
so the traditional counterpart of the backward direction is
\[
\inferrule{
\hypo{\Gamma, \alpha : k_1}{F : k_2}
}{
\hypo{\Gamma}{\lambda \alpha.\; F : k_1 \Rightarrow_k k_2 }
}
\]
The other two fields of the isomorphism \ensuremath{\Varid{el}\;(\Varid{k}_{\mathrm{1}}\;{\Rightarrow_k}\;\Varid{k}_{\mathrm{2}})\cong(\Varid{el}\;\Varid{k}_{\mathrm{1}}\to\Varid{el}\;\Varid{k}_{\mathrm{2}})} assert that these two directions are mutual inverses, and this is exactly $\eta$- and $\beta$-rules of \ensuremath{\Varid{k}_{\mathrm{1}}\;{\Rightarrow_k}\;\Varid{k}_{\mathrm{2}}}:
\[
\spread*{
\inferrule{
  \hypo{\Gamma}{F : k_1  \Rightarrow_k k_2 }
}{
  \typing{\Gamma}{(\lambda A.\; F \; A) \equiv F : k_1 \Rightarrow_k k_2}
}
\also
\inferrule{
  \hypo{\Gamma, \alpha : k_1}{F : k_2}
  \\  
  \typing{\Gamma}{\alpha : k_1 }
}{
  \typing{\Gamma}{(\lambda \alpha.\; F) \; A  \equiv F[A/\alpha] : k_2}
}
}
\]
In summary, we have used a single LF-declaration \Cref{eq:fha:el:kfun} to
express what would be four rules in the traditional presentation of functions.
We will use this technique a lot in our specification of \Fha. 


\subsubsection{Types}
For terms of types, there is a judgement \ensuremath{\Varid{tm}\mathbin{:}\Varid{el}\;\Varid{ty}\to \Jdg}.  
Terms of (polymorphic) function types and the unit type
are still specified by HOAS:
\begin{gather*}
\spread*{
\ensuremath{\Varid{tm}\mathbin{:}\Varid{el}\;\Varid{ty}\to \Jdg}  \also \ensuremath{\Varid{unit\hyp{}iso}\mathbin{:}\Varid{tm}\;\Varid{unit}\cong\unitty}} \\
\ensuremath{\Varid{{\Rightarrow_t}\hyp{}iso}\mathbin{:}\{\mskip1.5mu \Conid{A},\Conid{B}\mathbin{:}\Varid{el}\;\Varid{ty}\mskip1.5mu\}\to\Varid{tm}\;(\Conid{A}\;{\Rightarrow_t}\;\Conid{B})\cong(\Varid{tm}\;\Conid{A}\to\Varid{tm}\;\Conid{B})} \label{eq:fha:terms} \tag{\Fom-4} \\
\ensuremath{\Varid{\allfor\hyp{}iso}\mathbin{:}\{\mskip1.5mu \Varid{k}\mathbin{:}\anonymous \mskip1.5mu\}\;\{\mskip1.5mu \Conid{A}\mathbin{:}\anonymous \mskip1.5mu\}\to\Varid{tm}\;(\allfor\;\Varid{k}\;\Conid{A})\cong((\Varid{α}\mathbin{:}\Varid{el}\;\Varid{k})\to\Varid{tm}\;(\Conid{A}\;\Varid{α}))}
\end{gather*}
For the two-element type, we only include two terms \ensuremath{\Varid{tt}} and \ensuremath{\Varid{ff}}:
\begin{equation}\label{eq:ttff}
\spreadTag{\Fom-5}{\ensuremath{\Varid{tt}\mathbin{:}\Varid{tm}\;\Varid{bool}} \also \ensuremath{\Varid{ff}\mathbin{:}\Varid{tm}\;\Varid{bool}}}
\end{equation}
This completes the signature of System \Fom{}.
We have set it up in a minimal way for simplicity:
we did not even include an eliminator for \ensuremath{\Varid{bool}}, because having term
constructors is sufficient for stating the canonicity of \ensuremath{\Varid{bool}}
terms in \Cref{thm:canonicity}.  More useful types/kinds, such as a Boolean
type with the correct eliminator, products, and lists, can be added
easily and can be found in \Cref{app:sig}.

\begin{example}
Let us see an example of a program of \Fom{} defined in the LF.
Writing \ensuremath{\Varid{app}} and \ensuremath{\Varid{abs}} for the forward and backward directions of the
isomorphism \ensuremath{\Varid{{\Rightarrow_t}\hyp{}iso}} respectively, and \ensuremath{\Conid{App}} and \ensuremath{\Conid{Abs}} for the two directions
of \ensuremath{\Varid{\allfor\hyp{}iso}}, the Church numeral of $2$ is 
\begin{hscode}\SaveRestoreHook
\column{B}{@{}>{\hspre}l<{\hspost}@{}}%
\column{29}{@{}>{\hspre}l<{\hspost}@{}}%
\column{36}{@{}>{\hspre}l<{\hspost}@{}}%
\column{E}{@{}>{\hspre}l<{\hspost}@{}}%
\>[B]{}\Varid{two}\mathbin{:}\Varid{tm}\;(\allfor\;\Varid{ty}\;(\lambda \Varid{α}.\;{}\<[29]%
\>[29]{}(\Varid{α}\;{\Rightarrow_t}\;\Varid{α})\;{\Rightarrow_t}\;\Varid{α}\;{\Rightarrow_t}\;\Varid{α})){}\<[E]%
\\
\>[B]{}\Varid{two}\mathrel{=}\Conid{Abs}\;(\lambda \Varid{α}.\;\Varid{abs}\;(\lambda \Varid{f}.\;{}\<[36]%
\>[36]{}\Varid{abs}\;(\lambda \Varid{x}.\;\Varid{app}\;\Varid{f}\;(\Varid{app}\;\Varid{f}\;\Varid{x})))){}\<[E]%
\ColumnHook
\end{hscode}\resethooks
\end{example}

\subsubsection{Derived Concepts}

Later we will see that signatures of higher-order effectful operations 
in \Fha{} are given as higher-order endofunctors
\ensuremath{(\Varid{ty}\;{\Rightarrow_k}\;\Varid{ty})\;{\Rightarrow_k}\;(\Varid{ty}\;{\Rightarrow_k}\;\Varid{ty})} 
over \Fom{}-functors \ensuremath{\Varid{ty}\;{\Rightarrow_k}\;\Varid{ty}}, and handlers in \Fha{} always have a monad
structure.
These derived concepts, such as functors and monads, are essentially
the same as the definitions in Haskell, and are collected in 
\Cref{fig:derived:concepts}, which will be ingredients for our computation
judgements in the next step.
Note that because \Fom{} does not have equality types, the equational laws
for functors/monads are not included in these definitions (just like in Haskell),
and they are called \emph{raw functors/monads}.

\begin{figure}
\begin{hscode}\SaveRestoreHook
\column{B}{@{}>{\hspre}l<{\hspost}@{}}%
\column{3}{@{}>{\hspre}l<{\hspost}@{}}%
\column{7}{@{}>{\hspre}l<{\hspost}@{}}%
\column{9}{@{}>{\hspre}l<{\hspost}@{}}%
\column{10}{@{}>{\hspre}l<{\hspost}@{}}%
\column{33}{@{}>{\hspre}l<{\hspost}@{}}%
\column{34}{@{}>{\hspre}l<{\hspost}@{}}%
\column{35}{@{}>{\hspre}l<{\hspost}@{}}%
\column{65}{@{}>{\hspre}l<{\hspost}@{}}%
\column{E}{@{}>{\hspre}l<{\hspost}@{}}%
\>[B]{}\Varid{efty}\mathbin{:}\Varid{ki}{}\<[34]%
\>[34]{}\quad \text{\color{gray}-{}- endo-functions on types}{}\<[E]%
\\
\>[B]{}\Varid{efty}\mathrel{=}(\Varid{ty}\;{\Rightarrow_k}\;\Varid{ty}){}\<[E]%
\\[\blanklineskip]%
\>[B]{}\Varid{fmap\hyp{}ty}\mathbin{:}(\Conid{F}\mathbin{:}\Varid{el}\;\Varid{efty})\to\Varid{el}\;\Varid{ty}{}\<[34]%
\>[34]{}\quad \text{\color{gray}-{}- the type of functor action on functions}{}\<[E]%
\\
\>[B]{}\Varid{fmap\hyp{}ty}\;\Conid{F}\mathrel{=}\allfor\;\Varid{ty}\;(\lambda \Varid{α}.\;\allfor\;\Varid{ty}\;(\lambda \Varid{β}.\;(\Varid{α}\;{\Rightarrow_t}\;\Varid{β})\;{\Rightarrow_t}\;(\Conid{F}\;\Varid{α}\;{\Rightarrow_t}\;\Conid{F}\;\Varid{β}))){}\<[E]%
\\[\blanklineskip]%
\>[B]{}\lhskeyword{record}\;\Conid{RawFunctor}\mathbin{:}\Jdg\;\lhskeyword{where}{}\<[E]%
\\
\>[B]{}\hsindent{3}{}\<[3]%
\>[3]{}\Varid{0}{}\<[7]%
\>[7]{}\mathbin{:}\Varid{el}\;\Varid{efty}{}\<[34]%
\>[34]{}\quad \text{\color{gray}-{}- underlying type function of a functor}{}\<[E]%
\\
\>[B]{}\hsindent{3}{}\<[3]%
\>[3]{}\Varid{fmap}{}\<[7]%
\>[7]{}\mathbin{:}\Varid{tm}\;(\,\Varid{fmap\hyp{}ty}\;\Varid{0}){}\<[34]%
\>[34]{}\quad \text{\color{gray}-{}- action on functions between types}{}\<[E]%
\\[\blanklineskip]%
\>[B]{}\lhskeyword{record}\;\Conid{RawMonad}\mathbin{:}\Jdg\;\lhskeyword{where}{}\<[E]%
\\
\>[B]{}\hsindent{3}{}\<[3]%
\>[3]{}\Varid{0}{}\<[9]%
\>[9]{}\mathbin{:}\Varid{el}\;\Varid{efty}{}\<[34]%
\>[34]{}\quad \text{\color{gray}-{}- underlying type function of a monad}{}\<[E]%
\\
\>[B]{}\hsindent{3}{}\<[3]%
\>[3]{}\Varid{ret}{}\<[9]%
\>[9]{}\mathbin{:}\Varid{tm}\;(\allfor\;\Varid{ty}\;(\lambda \Varid{α}.\;\Varid{α}\;{\Rightarrow_t}\;\Varid{0}\;\Varid{α})){}\<[E]%
\\
\>[B]{}\hsindent{3}{}\<[3]%
\>[3]{}\Varid{bind}{}\<[9]%
\>[9]{}\mathbin{:}\Varid{tm}\;(\allfor\;\Varid{ty}\;(\lambda \Varid{α}.\;\allfor\;\Varid{ty}\;(\lambda \Varid{β}.\;\Varid{0}\;\Varid{α}\;{\Rightarrow_t}\;(\Varid{α}\;{\Rightarrow_t}\;\Varid{0}\;\Varid{β})\;{\Rightarrow_t}\;\Varid{0}\;\Varid{β}))){}\<[E]%
\\[\blanklineskip]%
\>[B]{}\Varid{trans}\mathbin{:}(\Conid{F},\Conid{G}\mathbin{:}\Varid{el}\;\Varid{efty})\to\Varid{el}\;\Varid{ty}{}\<[35]%
\>[35]{}\quad \text{\color{gray}-{}- transformations between endo-functions on types }{}\<[E]%
\\
\>[B]{}\Varid{trans}\;\Conid{F}\;\Conid{G}\mathrel{=}\allfor\;\Varid{ty}\;(\lambda \Varid{α}.\;\Conid{F}\;\Varid{α}\;{\Rightarrow_t}\;\Conid{G}\;\Varid{α}){}\<[E]%
\\[\blanklineskip]%
\>[B]{}\lhskeyword{record}\;\Conid{RawHFunctor}\mathbin{:}\Jdg\;\lhskeyword{where}{}\<[E]%
\\
\>[B]{}\hsindent{3}{}\<[3]%
\>[3]{}\Varid{0}{}\<[10]%
\>[10]{}\mathbin{:}\Varid{el}\;(\Varid{efty}\;{\Rightarrow_k}\;\Varid{efty}){}\<[33]%
\>[33]{}\quad \text{\color{gray}-{}- action on the underlying type functions of functors}{}\<[E]%
\\
\>[B]{}\hsindent{3}{}\<[3]%
\>[3]{}\Varid{hfmap}{}\<[10]%
\>[10]{}\mathbin{:}(\Conid{F}\mathbin{:}\Conid{RawFunctor})\to\Varid{tm}\;(\,\Varid{fmap\hyp{}ty}\;(\Varid{0}\;(\Conid{F}\;.\Varid{0}))){}\<[65]%
\>[65]{}\quad \text{\color{gray}-{}- action on the \ensuremath{\Varid{fmap}} of functors}{}\<[E]%
\\
\>[B]{}\hsindent{3}{}\<[3]%
\>[3]{}\Varid{hmap}{}\<[10]%
\>[10]{}\mathbin{:}(\Conid{F},\Conid{G}\mathbin{:}\Conid{RawFunctor})\to\Varid{tm}\;(\Varid{trans}\;(\Conid{F}\;.\Varid{0})\;(\Conid{G}\;.\Varid{0})){}\<[65]%
\>[65]{}\quad \text{\color{gray}-{}- action on transformations}{}\<[E]%
\\
\>[10]{}\to\Varid{tm}\;(\Varid{trans}\;(\Varid{0}\;(\Conid{F}\;.\Varid{0}))\;(\Varid{0}\;(\Conid{G}\;.\Varid{0}))){}\<[E]%
\ColumnHook
\end{hscode}\resethooks
\caption{Derived concepts in System \Fom}
\label{fig:derived:concepts}
\end{figure}

\begin{notation}\label{notation:omit:iso}
The isomorphisms \ensuremath{\Varid{{\Rightarrow_t}\hyp{}iso}}, \ensuremath{\Varid{{\Rightarrow_k}\hyp{}iso}}, and \ensuremath{\Varid{\allfor\hyp{}iso}} for
constructing/applying (polymorphic) functions easily clutter our notation when
working with raw monads, so we will leave them implicit.
For example, we just write \ensuremath{\Conid{M}\;\Varid{.bind}\;\Varid{\alpha}\;\Varid{\beta}\;\Varid{m}\;(\lambda \Varid{x}.\;\Varid{n})} instead of
the totally unreadable
\ensuremath{\Varid{app}\;(\Varid{app}\;(\Conid{App}\;(\Conid{App}\;(\Conid{M}\;\Varid{.bind})\;\Varid{\alpha})\;\Varid{\beta})\;\Varid{m})\;(\Conid{Abs}\;(\lambda \Varid{x}.\;\Varid{n}))}.

We will also sometimes suppress the field accessor \ensuremath{.\Varid{0}} of raw functors/monads in
\Cref{fig:derived:concepts} for readability, so given $\ensuremath{\Conid{F}} :
\ensuremath{\Conid{RawFunctor}}$ and $\ensuremath{\Conid{X}\mathbin{:}\Varid{el}\;\Varid{ty}}$, we may write \ensuremath{\Conid{F}\;\Conid{X}} for \ensuremath{\Conid{F}\;.\Varid{0}\;\Conid{X}}.
\end{notation}

\begin{example}\label{ex:signature:exception}
The higher-order functor $\Hexc$ for the effects of exception throwing and
catching is
\begin{lgather*}
\Hexc\; .0 : \ensuremath{\Varid{el}\;((\Varid{ty}\;{\Rightarrow_k}\;\Varid{ty})\;{\Rightarrow_k}\;(\Varid{ty}\;{\Rightarrow_k}\;\Varid{ty}))} \\
\Hexc\; .0\; F\; A = \ensuremath{\Varid{unit}} + (F\; A \times F\; A)
\end{lgather*}
with the evident functorial actions on $F$ and $A$, 
and
${(\times)}, {(+)} : \ensuremath{\Varid{el}\;\Varid{ty}\to \Varid{el}\;\Varid{ty}\to \Varid{el}\;\Varid{ty}}$ are the binary product and coproduct
types that are definable using Church encodings (or added directly into \Fha;
see \Cref{app:eff:fam}).
For every (raw) monad $M$, a transformation \ensuremath{\allfor\;\Varid{ty}\;(\lambda \Varid{\alpha}.\;\Hexc\;\Conid{M}\;\Varid{\alpha}\;{\Rightarrow_t}\;\Conid{M}\;\Varid{\alpha})} is the same as two transformations \ensuremath{\allfor\;\Varid{ty}\;(\lambda \Varid{\alpha}.\;\Varid{unit}\;{\Rightarrow_t}\;\Conid{M}\;\Varid{\alpha})} and \ensuremath{\allfor\;\Varid{ty}\;(\lambda \Varid{\alpha}.\;\Conid{M}\;\Varid{\alpha}\;\mathbin{\times_t}\;\Conid{M}\;\Varid{\alpha}\;{\Rightarrow_t}\;\Conid{M}\;\Varid{\alpha})},
which are the operations of throwing and catching on $M$ respectively.
\end{example}

\subsection{Computation Judgements}\label{sec:fha:co}
Now we are ready to add computation judgements to \Fom{} to obtain \Fha.

\subsubsection{Computations}
We follow the \emph{fine-grain call-by-value} (FGCBV) approach
\citep{Lassen_1998,Levy2003}.
For each \ensuremath{\Conid{H}\mathbin{:}\Conid{RawHFunctor}} and \ensuremath{\Conid{A}\mathbin{:}\Varid{el}\;\Varid{ty}}, there
is a judgement \ensuremath{\Varid{co}\;\Conid{H}\;\Conid{A}} for \emph{computations} of \ensuremath{\Conid{A}}-values with effectful
operations specified by \ensuremath{\Conid{H}}:
\begin{equation}\label{eq:fha:co}
\ensuremath{\Varid{co}\mathbin{:}(\Conid{H}\mathbin{:}\Conid{RawHFunctor})\to(\Conid{A}\mathbin{:}\Varid{el}\;\Varid{ty})\to\Jdg}
\tag{\Fha-1}
\end{equation}
The judgement has the following two rules for pure computations and sequential
composition of computations respectively:
\begin{equation}
\begin{aligned}
&\ensuremath{\Varid{val}\mathbin{:}\{\mskip1.5mu \Conid{H},\Conid{A}\mskip1.5mu\}\to\Varid{tm}\;\Conid{A}\to\Varid{co}\;\Conid{H}\;\Conid{A}} \\
&\ensuremath{\Varid{let\hyp{}in}\mathbin{:}\abrace{\Conid{H},\Conid{A},\Conid{B}}\to\Varid{co}\;\Conid{H}\;\Conid{A}\to(\Varid{tm}\;\Conid{A}\to\Varid{co}\;\Conid{H}\;\Conid{B})\to\Varid{co}\;\Conid{H}\;\Conid{B}}
\end{aligned}
\tag{\Fha-2}\label{eq:fha:co:letin}
\end{equation}
The interaction of \ensuremath{\Varid{val}} and \ensuremath{\Varid{let\hyp{}in}} is axiomatised by the
following judgemental equalities, which are exactly the equations
\Cref{eq:comp:monad:laws:val,eq:comp:monad:laws:let} from the introduction and are essentially the monadic
laws:
{\allowdisplaybreaks
\begin{alignat*}{2}
&\ensuremath{\Varid{val\hyp{}let}}&&\ensuremath{\mathbin{:}\{\mskip1.5mu \Conid{H},\Conid{A},\Conid{B}\mskip1.5mu\}\to(\Varid{a}\mathbin{:}\Varid{tm}\;\Conid{A})\to(\Varid{k}\mathbin{:}\Varid{tm}\;\Conid{A}\to\Varid{co}\;\Conid{H}\;\Conid{B})\to\Varid{let\hyp{}in}\;(\Varid{val}\;\Varid{a})\;\Varid{k}=\Varid{k}\;\Varid{a}}\\
&\ensuremath{\Varid{let\hyp{}val}}&&\ensuremath{\mathbin{:}\{\mskip1.5mu \Conid{H},\Conid{A}\mskip1.5mu\}\to(\Varid{m}\mathbin{:}\Varid{co}\;\Conid{H}\;\Conid{A})\to\Varid{let\hyp{}in}\;\Varid{m}\;\Varid{val}=\Varid{m}}    \\
&\ensuremath{\Varid{let\hyp{}assoc}}&&\ensuremath{\mathbin{:}\{\mskip1.5mu \Conid{H},\Conid{A},\Conid{B},\Conid{C}\mskip1.5mu\}\to(\Varid{m₁}\mathbin{:}\Varid{co}\;\Conid{H}\;\Conid{A})} {\tag{\Fha-3}\label{eq:fha:co:laws}}  \\
&\ensuremath{}&&\ensuremath{\to(\Varid{m₂}\mathbin{:}\Varid{tm}\;\Conid{A}\to\Varid{co}\;\Conid{H}\;\Conid{B})\to(\Varid{m₃}\mathbin{:}\Varid{tm}\;\Conid{B}\to\Varid{co}\;\Conid{H}\;\Conid{C})} \\
&\ensuremath{}&&\ensuremath{\to\Varid{let\hyp{}in}\;(\Varid{let\hyp{}in}\;\Varid{m₁}\;\Varid{m₂})\;\Varid{m₃}=\Varid{let\hyp{}in}\;\Varid{m₁}\;(\lambda \Varid{a}.\;\Varid{let\hyp{}in}\;(\Varid{m₂}\;\Varid{a})\;\Varid{m₃})}
\end{alignat*}
}

\subsubsection{Thunks}
We also introduce a new type former \ensuremath{\Varid{th}\;\Conid{H}\;\Conid{A}} for \emph{thunks} of
computations of $A$-values with effects of $H$, 
whose terms are isomorphic to computations:
\begin{equation}\label{el:tmnd}
\spreadTag{\Fha-4}{
\ensuremath{\Varid{th}\mathbin{:}\Conid{RawHFunctor}\to \Varid{el}\;\Varid{ty}\to \Varid{el}\;\Varid{ty}}
\also
\ensuremath{\Varid{th\hyp{}iso}\mathbin{:}\{\mskip1.5mu \Conid{H},\Conid{A}\mskip1.5mu\}\to \Varid{tm}\;(\Varid{th}\;\Conid{H}\;\Conid{A})\;{\cong}\;\Varid{co}\;\Conid{H}\;\Conid{A}}
}
\end{equation}
The two directions of the isomorphism \ensuremath{\Varid{th\hyp{}iso}} will be called \ensuremath{{\Uparrow}}
and \ensuremath{{\Downarrow}} respectively:
\begin{gather*}
\spread*{\ensuremath{{\Uparrow}} : \ensuremath{\Varid{tm}\;(\Varid{th}\;\Conid{H}\;\Conid{A})\to \Varid{co}\;\Conid{H}\;\Conid{A}} \also
\ensuremath{{\Downarrow}} : \ensuremath{\Varid{co}\;\Conid{H}\;\Conid{A}\to \Varid{tm}\;(\Varid{th}\;\Conid{H}\;\Conid{A})}}
\end{gather*}

Thunks can be packed into a raw monad:
\begin{equation*}
\begin{hscode}\SaveRestoreHook
\column{B}{@{}>{\hspre}l<{\hspost}@{}}%
\column{15}{@{}>{\hspre}l<{\hspost}@{}}%
\column{E}{@{}>{\hspre}l<{\hspost}@{}}%
\>[B]{}\Varid{th\hyp{}mnd}\mathbin{:}\Conid{RawHFunctor}\to\Conid{RawMonad}{}\<[E]%
\\
\>[B]{}\Varid{th\hyp{}mnd}\;\Conid{H}\;\Varid{.0}{}\<[15]%
\>[15]{}\mathrel{=}\Varid{th}\;\Conid{H}{}\<[E]%
\\
\>[B]{}\Varid{th\hyp{}mnd}\;\Conid{H}\;\Varid{.ret}{}\<[15]%
\>[15]{}\mathrel{=}\lambda \Varid{\alpha}\;\Varid{x}.\;{\Downarrow}\;(\Varid{val}\;\Varid{x}){}\<[E]%
\\
\>[B]{}\Varid{th\hyp{}mnd}\;\Conid{H}\;\Varid{.bind}{}\<[15]%
\>[15]{}\mathrel{=}\lambda \Varid{\alpha}\;\Varid{\beta}\;\Varid{m}\;\Varid{k}.\;{\Downarrow}\;(\Varid{let\hyp{}in}\;({\Uparrow}\;\Varid{m})\;(\lambda \Varid{a}.\;{\Uparrow}\;(\Varid{k}\;\Varid{a}))){}\<[E]%
\ColumnHook
\end{hscode}\resethooks
\end{equation*}
in which the isomorphisms \ensuremath{\Varid{abs}} and \ensuremath{\Conid{Abs}} for constructing
(polymorphic) function types are left implicit (\Cref{notation:omit:iso}).
Following equations \Cref{eq:fha:co:laws}, \ensuremath{\Varid{th\hyp{}mnd}} satisfies the monad
laws too.


\subsubsection{Operations}
Effectful operations that computations can perform are introduced by
\begin{equation}
\begin{hscode}\SaveRestoreHook
\column{B}{@{}>{\hspre}l<{\hspost}@{}}%
\column{E}{@{}>{\hspre}l<{\hspost}@{}}%
\>[B]{}\Varid{op}\mathbin{:}\{\mskip1.5mu \Conid{H},\Conid{A},\Conid{B}\mskip1.5mu\}\to\Varid{tm}\;(\Conid{H}\;(\Varid{th}\;\Conid{H})\;\Conid{A})\to(\Varid{tm}\;\Conid{A}\to\Varid{co}\;\Conid{H}\;\Conid{B})\to\Varid{co}\;\Conid{H}\;\Conid{B}{}\<[E]%
\ColumnHook
\end{hscode}\resethooks
\tag{\Fha-5}
\label{eq:fha:op}
\end{equation}
The first argument \ensuremath{\Varid{o}\mathbin{:}\Varid{tm}\;(\Conid{H}\;(\Varid{th}\;\Conid{H})\;\Conid{A})}  is the input to an operation
call, such as some parameters or computations that the operation call acts on.
The second argument \ensuremath{\Varid{k}\mathbin{:}\Varid{tm}\;\Conid{A}\to\Varid{co}\;\Conid{H}\;\Conid{B}} of \ensuremath{\Varid{op}} is
the `continuation' of the computation after this operation call, where the
argument \ensuremath{\Varid{tm}\;\Conid{A}} of \ensuremath{\Varid{k}} is the result of the operation call.
The result \ensuremath{\Varid{op}\;\Varid{o}\;\Varid{k}} is understood as the computation that first makes an operation
call with input \ensuremath{\Varid{o}}, which returns an \ensuremath{\Conid{A}}-value, and then continues as \ensuremath{\Varid{k}}.

\begin{example}\label{ex:operation:exception}
Continuing \Cref{ex:signature:exception}, we can use \ensuremath{\Varid{op}} to define the operations of throwing and catching for the computation judgement \ensuremath{\Varid{co}\;\Hexc}:
\[
\spread*{
\begin{hscode}\SaveRestoreHook
\column{B}{@{}>{\hspre}l<{\hspost}@{}}%
\column{E}{@{}>{\hspre}l<{\hspost}@{}}%
\>[B]{}\Varid{throw}\mathbin{:}\{\mskip1.5mu \Conid{A}\mskip1.5mu\}\to \Varid{co}\;\Hexc\;\Conid{A}{}\<[E]%
\\
\>[B]{}\Varid{throw}\mathrel{=}\Varid{op}\;(\Varid{inl}\;\unitel)\;\Varid{val}{}\<[E]%
\ColumnHook
\end{hscode}\resethooks
\also
\begin{hscode}\SaveRestoreHook
\column{B}{@{}>{\hspre}l<{\hspost}@{}}%
\column{E}{@{}>{\hspre}l<{\hspost}@{}}%
\>[B]{}\Varid{catch}\mathbin{:}\{\mskip1.5mu \Conid{A}\mskip1.5mu\}\to \Varid{co}\;\Hexc\;\Conid{A}\to \Varid{co}\;\Hexc\;\Conid{A}\to \Varid{co}\;\Hexc\;\Conid{A}{}\<[E]%
\\
\>[B]{}\Varid{catch}\;\Varid{p}\;\Varid{h}\mathrel{=}\Varid{op}\;(\Varid{inr}\;({\Downarrow}\;\Varid{p},\;{\Downarrow}\;\Varid{h}))\;\Varid{val}{}\<[E]%
\ColumnHook
\end{hscode}\resethooks
}
\]
\end{example}

The interaction of operation calls and sequential composition of computations
is the following, which is similar to the condition for \emph{algebraic operations}
of \citet{Plotkin_Power_2001}:
\begin{equation}
\begin{hscode}\SaveRestoreHook
\column{B}{@{}>{\hspre}l<{\hspost}@{}}%
\column{8}{@{}>{\hspre}l<{\hspost}@{}}%
\column{E}{@{}>{\hspre}l<{\hspost}@{}}%
\>[B]{}\Varid{let\hyp{}op}{}\<[8]%
\>[8]{}\mathbin{:}\{\mskip1.5mu \Conid{H},\Conid{A},\Conid{B},\Conid{C}\mskip1.5mu\}\to(\Varid{p}\mathbin{:}\Varid{tm}\;(\Conid{H}\;(\Varid{th}\;\Conid{H})\;\Conid{A})){}\<[E]%
\\
\>[8]{}\to(\Varid{k}\mathbin{:}\Varid{tm}\;\Conid{A}\to\Varid{co}\;\Conid{H}\;\Conid{B})\to(\Varid{k'}\mathbin{:}\Varid{tm}\;\Conid{B}\to\Varid{co}\;\Conid{H}\;\Conid{C}){}\<[E]%
\\
\>[8]{}\to\Varid{let\hyp{}in}\;(\Varid{op}\;\Varid{p}\;\Varid{k})\;\Varid{k'}=\Varid{op}\;\Varid{p}\;(\lambda \Varid{a}.\;\Varid{let\hyp{}in}\;(\Varid{k}\;\Varid{a})\;\Varid{k'}){}\<[E]%
\ColumnHook
\end{hscode}\resethooks
\tag{\Fha-6}
\label{eq:fha:letop}
\end{equation}
The equation \Cref{eq:fha:letop} implies that every operation call \ensuremath{\Varid{op}\;\Varid{o}\;\Varid{k}} is
equal to \ensuremath{\Varid{let\hyp{}in}\;(\Varid{op}\;\Varid{o}\;\Varid{val})\;\Varid{k}}, we could have alternatively defined
\Cref{eq:fha:op} as \ensuremath{\Varid{op'}\mathbin{:}\{\mskip1.5mu \Conid{H},\Conid{A},\Conid{B}\mskip1.5mu\}\to\Varid{tm}\;(\Conid{H}\;(\Varid{th}\;\Conid{H})\;\Conid{A})\to\Varid{co}\;\Conid{H}\;\Conid{A}} without
the \ensuremath{\Varid{k}} argument, which is
the formulation \Cref{eq:scp:op:call} that we used in the introduction.
This does not make a big technical difference, and we choose the
formulation \Cref{eq:fha:op} as it is closer to the rule presented by
\citet{PlotPret09Hand,PlotPret13Hand} for ordinary algebraic effects.

\subsubsection{Handling}
Now we axiomatise that computations \ensuremath{\Varid{co}\;\Conid{H}\;\Conid{A}} can be \emph{handled}, or \emph{evaluated}, by any raw monad supporting the operations from \ensuremath{\Conid{H}} (for which
we wrote \ensuremath{\lhskeyword{hdl}\;\Varid{p}\;\lhskeyword{with}\;\Varid{h}} in \Cref{sec:intro}).
We define the following structure for monads supporting operations from \ensuremath{\Conid{H}}:
\begin{hscode}\SaveRestoreHook
\column{B}{@{}>{\hspre}l<{\hspost}@{}}%
\column{3}{@{}>{\hspre}l<{\hspost}@{}}%
\column{9}{@{}>{\hspre}l<{\hspost}@{}}%
\column{E}{@{}>{\hspre}l<{\hspost}@{}}%
\>[B]{}\lhskeyword{record}\;\Conid{Handler}\;(\Conid{H}\mathbin{:}\Conid{RawHFunctor})\mathbin{:}\Jdg\,\;\lhskeyword{where}{}\<[E]%
\\
\>[B]{}\hsindent{3}{}\<[3]%
\>[3]{}\lhskeyword{include}\;\Conid{RawMonad}\;\lhskeyword{as}\;\Conid{M}{}\<[E]%
\\
\>[B]{}\hsindent{3}{}\<[3]%
\>[3]{}\Varid{malg}{}\<[9]%
\>[9]{}\mathbin{:}\Varid{tm}\;(\Varid{trans}\;(\Conid{H}\;\Conid{M})\;\Conid{M}){}\<[E]%
\ColumnHook
\end{hscode}\resethooks
where by `\ensuremath{\lhskeyword{include}\;\Conid{RawMonad}\;\lhskeyword{as}\;\Conid{M}}', we mean that \ensuremath{\Conid{Handler}} has all
the fields of the record \ensuremath{\Conid{RawMonad}} from \Cref{fig:derived:concepts} -- namely,
\ensuremath{\Varid{0}}, \ensuremath{\Varid{ret}}, and \ensuremath{\Varid{bind}}.
Moreover, for every \ensuremath{\Varid{h}\mathbin{:}\Conid{Handler}\;\Conid{H}}, there is a projection $h.M : \ensuremath{\Conid{RawMonad}}$.
As usual, \ensuremath{\Varid{h}\;\Varid{.0}\;\Conid{A}} will be abbreviated as \ensuremath{\Varid{h}\;\Conid{A}}.
We then add to \Fha{} the following declaration that evaluates a computation
with effect \ensuremath{\Conid{H}} with a handler of \ensuremath{\Conid{H}}: 
\begin{equation}
\ensuremath{\Varid{hdl}\mathbin{:}\{\mskip1.5mu \Conid{H},\Conid{A}\mskip1.5mu\}\to(\Varid{h}\mathbin{:}\Conid{Handler}\;\Conid{H})\to\Varid{co}\;\Conid{H}\;\Conid{A}\to\Varid{tm}\;(\Varid{h}\;\Conid{A})}
\tag{\Fha-7}
\label{eq:fha:eval}
\end{equation}
The last piece of the signature of \Fha{} is the computation rules for \ensuremath{\Varid{hdl}},
which are similar to the operational semantics of handlers in
conventional algebraic effects
\citep{PlotPret09Hand,PlotPret13Hand}:
\begin{itemize}[wide]
\item
When the computation is a value, it is handled by the \ensuremath{\Varid{ret}} of the monad,
\begin{equation}
\begin{hscode}\SaveRestoreHook
\column{B}{@{}>{\hspre}l<{\hspost}@{}}%
\column{10}{@{}>{\hspre}l<{\hspost}@{}}%
\column{E}{@{}>{\hspre}l<{\hspost}@{}}%
\>[B]{}\Varid{hdl\hyp{}val}{}\<[10]%
\>[10]{}\mathbin{:}\{\mskip1.5mu \Conid{H},\Conid{A}\mskip1.5mu\}\to(\Varid{h}\mathbin{:}\Conid{Handler}\;\Conid{H})\to(\Varid{a}\mathbin{:}\Varid{tm}\;\Conid{A})\to\Varid{hdl}\;\Varid{h}\;(\Varid{val}\;\Varid{a})=\Varid{h}\;\Varid{.ret}\;\Conid{A}\;\Varid{a}{}\<[E]%
\ColumnHook
\end{hscode}\resethooks
\tag{\Fha-8}\label{eq:fha:evalval}
\end{equation}

\item
When the computation is an operation call, it is handled by the corresponding
operation on the monad, with all subterms recursively handled:
\begin{equation}
\begin{hscode}\SaveRestoreHook
\column{B}{@{}>{\hspre}l<{\hspost}@{}}%
\column{9}{@{}>{\hspre}l<{\hspost}@{}}%
\column{12}{@{}>{\hspre}l<{\hspost}@{}}%
\column{16}{@{}>{\hspre}l<{\hspost}@{}}%
\column{17}{@{}>{\hspre}l<{\hspost}@{}}%
\column{23}{@{}>{\hspre}l<{\hspost}@{}}%
\column{E}{@{}>{\hspre}l<{\hspost}@{}}%
\>[B]{}\Varid{hdl\hyp{}op}{}\<[9]%
\>[9]{}\mathbin{:}\{\mskip1.5mu \Conid{H},\Conid{A},\Conid{B}\mskip1.5mu\}\to(\Varid{h}\mathbin{:}\Conid{Handler}\;\Conid{H}){}\<[E]%
\\
\>[9]{}\to(\Varid{p}\mathbin{:}\Varid{tm}\;(\Conid{H}\;(\Varid{th}\;\Conid{H})\;\Conid{A}))\to(\Varid{k}\mathbin{:}\Varid{tm}\;\Conid{A}\to\Varid{co}\;\Conid{H}\;\Conid{B}){}\<[E]%
\\
\>[9]{}\to{}\<[12]%
\>[12]{}\lhskeyword{let}\;{}\<[17]%
\>[17]{}\Conid{T}{}\<[23]%
\>[23]{}\mathrel{=}\Varid{fct\hyp{}of\hyp{}mnd}\;(\Varid{th\hyp{}mnd}\;\Conid{H}){}\<[E]%
\\
\>[17]{}\Conid{M}{}\<[23]%
\>[23]{}\mathrel{=}\Varid{fct\hyp{}of\hyp{}mnd}\;(\Varid{h}\;\Varid{.M}){}\<[E]%
\\
\>[17]{}\Varid{p'}{}\<[23]%
\>[23]{}\mathrel{=}\Conid{H}\;\Varid{.hmap}\;\Conid{T}\;\Conid{M}\;(\lambda \Varid{α}\;\Varid{c}.\;\Varid{hdl}\;\Varid{h}\;({\Uparrow}\;\!\Varid{c}))\;\anonymous \;\Varid{p}{}\<[E]%
\\
\>[12]{}\lhskeyword{in}\;{}\<[16]%
\>[16]{}\Varid{hdl}\;\Varid{h}\;(\Varid{op}\;\Varid{p}\;\Varid{k})=\Varid{h}\;\Varid{.bind}\;\anonymous \;\anonymous \;(\Varid{h}\;\Varid{.malg}\;\anonymous \;\Varid{p'})\;(\lambda \Varid{a}.\;\Varid{hdl}\;\Varid{h}\;(\Varid{k}\;\Varid{a})){}\<[E]%
\ColumnHook
\end{hscode}\resethooks
\tag{\Fha-9}\label{eq:fha:evalop}
\end{equation}
where \ensuremath{\Varid{fct\hyp{}of\hyp{}mnd}} is the canonical functor structure of a monad:
\begin{hscode}\SaveRestoreHook
\column{B}{@{}>{\hspre}l<{\hspost}@{}}%
\column{E}{@{}>{\hspre}l<{\hspost}@{}}%
\>[B]{}\Varid{fct\hyp{}of\hyp{}mnd}\mathbin{:}\Conid{RawMonad}\to\Conid{RawFunctor}{}\<[E]%
\\
\>[B]{}\Varid{fct\hyp{}of\hyp{}mnd}\;\Varid{m}\;.\Varid{0}\mathrel{=}\Varid{m}\;\Varid{.0}{}\<[E]%
\\
\>[B]{}\Varid{fct\hyp{}of\hyp{}mnd}\;\Varid{m}\;.\Varid{fmap}\;\Varid{α}\;\Varid{β}\;\Varid{f}\;\Varid{ma}\mathrel{=}\Varid{m}\;\Varid{.bind}\;\Varid{α}\;\Varid{β}\;\Varid{ma}\;(\lambda \Varid{a}.\;\Varid{m}\;\Varid{.ret}\;\anonymous \;(\,\Varid{f}\;\Varid{a})){}\<[E]%
\ColumnHook
\end{hscode}\resethooks
and \ensuremath{\Varid{p'}} is the handler \ensuremath{\Varid{h}} recursively applied to the argument \ensuremath{\Varid{p}} using the
functorial action of \ensuremath{\Conid{H}}.
\end{itemize}
This completes the signature of \Fha{}. 
The full signature of \Fha{} is collected in \Cref{app:sig}.

\begin{example}\label{ex:hdl:exception:standard}
Continuing \Cref{ex:operation:exception}, we can define handlers of exception
throwing and catching using the \ensuremath{\Conid{Maybe}} monad \ensuremath{\Conid{Maybe}\;\Conid{X}\mathrel{=}\Varid{unit}\mathbin{+}\Conid{X}} (also known
as the \emph{option} monad) \citep{Spivey1990,Moggi89a}.
The standard semantics of throwing and catching on \ensuremath{\Conid{Maybe}\;\Conid{X}} is 
\[
\begin{hscode}\SaveRestoreHook
\column{B}{@{}>{\hspre}l<{\hspost}@{}}%
\column{48}{@{}>{\hspre}l<{\hspost}@{}}%
\column{E}{@{}>{\hspre}l<{\hspost}@{}}%
\>[B]{}\Varid{throw}\mathbin{:}\Varid{tm}\;(\allfor\;\Varid{ty}\;(\lambda \Varid{\alpha}.\;\Varid{unit}\;{\Rightarrow_t}\;\Conid{Maybe}\;\Varid{\alpha})){}\<[E]%
\\
\>[B]{}\Varid{throw}\;\anonymous \;\unitel\mathrel{=}\Varid{inl}\;\unitel{}\<[E]%
\\[\blanklineskip]%
\>[B]{}\Varid{catch}\mathbin{:}\Varid{tm}\;(\allfor\;\Varid{ty}\;(\lambda \Varid{\alpha}.\;\Conid{Maybe}\;\Varid{\alpha}\;\mathbin{\times_t}\;\Conid{Maybe}\;\Varid{\alpha}\;{\Rightarrow_t}\;\Conid{Maybe}\;\Varid{\alpha})){}\<[E]%
\\
\>[B]{}\Varid{catch}\;\anonymous \;\tuple{\Varid{p},\Varid{h}}\mathrel{=}\Varid{case}\;(\lambda \unitel.\;\Varid{h})\;{}\<[48]%
\>[48]{}(\lambda \Varid{x}.\;\Varid{inr}\;\Varid{x})\;\Varid{p}{}\<[E]%
\ColumnHook
\end{hscode}\resethooks
\]
where \ensuremath{\Varid{case}\mathbin{:}(\Varid{tm}\;\Varid{a}\to \Varid{tm}\;\Varid{c})\to (\Varid{tm}\;\Varid{b}\to \Varid{tm}\;\Varid{c})\to \Varid{tm}\;(\Varid{a}\mathbin{+}\Varid{b})\to \Varid{tm}\;\Varid{c}}
is the eliminator of the coproduct type $a+b$ (see \Cref{app:eff:fam} for details).
We define \ensuremath{\hexc\mathbin{:}\Conid{Handler}\;\Hexc} by 
\[
\spread*{
\hexc\; .\ensuremath{\Conid{M}} = \ensuremath{\Conid{Maybe}} 
\also
\text{and}
\also
\hexc \; .\ensuremath{\Varid{malg}} \; \alpha\;o = \ensuremath{\Varid{case}\;(\Varid{throw}\;\Varid{\alpha})\;(\Varid{catch}\;\Varid{\alpha})\;\Varid{o}}.
}
\]
\end{example}

\begin{example}
Continuing the last example, we can also define a non-standard handler of
exceptions such that if \ensuremath{\Varid{p}} in \ensuremath{\Varid{catch}\;\Varid{p}\;\Varid{h}} throws an exception and \ensuremath{\Varid{h}} is
executed, then \ensuremath{\Varid{p}} gets re-run.
This of course only makes sense in the presence of some other effects, such
as mutable state, otherwise re-running \ensuremath{\Varid{p}} after \ensuremath{\Varid{h}} would simply throw the
same exception again.

Let \ensuremath{\Conid{M}\mathbin{:}\Conid{RawMonad}} be any raw monad. The type constructor \ensuremath{\lambda \Conid{A}.\;\Conid{M}\;(\Conid{Maybe}\;\Conid{A})} can be extended to a monad $M_\ensuremath{\Varid{exc}}$ \citep[4.1.1]{Moggi89a}, for which we
can define exception throwing and catching:
\begin{hscode}\SaveRestoreHook
\column{B}{@{}>{\hspre}l<{\hspost}@{}}%
\column{E}{@{}>{\hspre}l<{\hspost}@{}}%
\>[B]{}\Varid{throw'}\mathbin{:}\Varid{tm}\;(\allfor\;\Varid{ty}\;(\lambda \Varid{\alpha}.\;\Varid{unit}\;{\Rightarrow_t}\;\Conid{M}\;(\Conid{Maybe}\;\Varid{\alpha}))){}\<[E]%
\\
\>[B]{}\Varid{throw'}\;\anonymous \;\unitel\mathrel{=}\Conid{M}\;\Varid{.ret}\;\anonymous \;(\Varid{inl}\;\unitel){}\<[E]%
\\[\blanklineskip]%
\>[B]{}\Varid{catch'}\mathbin{:}\Varid{tm}\;(\allfor\;\Varid{ty}\;(\lambda \Varid{\alpha}.\;\Conid{M}\;(\Conid{Maybe}\;\Varid{\alpha})\;\mathbin{\times_t}\;\Conid{M}\;(\Conid{Maybe}\;\Varid{\alpha})\;{\Rightarrow_t}\;\Conid{M}\;(\Conid{Maybe}\;\Varid{\alpha}))){}\<[E]%
\\
\>[B]{}\Varid{catch'}\;\anonymous \;\tuple{\Varid{p},\Varid{h}}\mathrel{=}\Conid{M}\;\Varid{.bind}\;\anonymous \;\anonymous \;\Varid{p}\;(\Varid{case}\;(\lambda \unitel.\;\Conid{M}\;\Varid{.bind}\;\anonymous \;\anonymous \;\Varid{h}\;(\lambda \anonymous .\;\Varid{p}))\;(\lambda \Varid{x}.\;\Conid{M}\;\Varid{.ret}\;\anonymous \;\Varid{x})){}\<[E]%
\ColumnHook
\end{hscode}\resethooks
In the definition of \ensuremath{\Varid{catch'}}, the first argument to \ensuremath{\Varid{case}} corresponds to when
\ensuremath{\Varid{p}} produces an exception.
In this case, we first run \ensuremath{\Varid{h}} (whose return value gets discarded), and
if \ensuremath{\Varid{h}} does not throw an exception, we then run \ensuremath{\Varid{p}} again.
The second time we run \ensuremath{\Varid{p}}, we do not try to catch the exceptions that \ensuremath{\Varid{p}} may
throw, although this is not the only possible behaviour.
For example, if we have general recursion (to be discussed shortly in
\Cref{sec:gen:rec}), we can keep running \ensuremath{\Varid{p}} and \ensuremath{\Varid{h}} until \ensuremath{\Varid{p}} succeeds or \ensuremath{\Varid{h}}
fails.
We can package the monad $M_\ensuremath{\Varid{exc}}$ and the two operations above into
\ensuremath{h_{\Varid{retry}}\mathbin{:}(\Conid{M}\mathbin{:}\Conid{RawMonad})\to \Conid{Handler}\;\Hexc}.
\end{example}

\begin{example}
For any raw monad \ensuremath{\Conid{M}}, the handler \ensuremath{h_{\Varid{retry}}\;\Conid{M}\mathbin{:}\Conid{Handler}\;\Hexc} from the last
example only handles the effect of exceptions.
In many cases, we can also lift the operations that $M$ already has to \ensuremath{\Conid{M}\hsdot{\circ }{.\ }\Conid{Maybe}} in a canonical way \citep{MM10Mon,Jaskelioff09}, giving us a handler of both exceptions and these operations.
Let \ensuremath{\Conid{H}\mathbin{:}\Conid{RawHFunctor}} be any effect signature and \ensuremath{\Varid{alg}\mathbin{:}\Varid{tm}\;(\Varid{trans}\;(\Conid{H}\;\Conid{M})\;\Conid{M})} an operation on $M$.
Suppose we also have some
\[\ensuremath{\Varid{\sigma}_H\mathbin{:}\{\mskip1.5mu \Conid{F},\Conid{G}\mathbin{:}\Conid{RawFunctor}\mskip1.5mu\}\to \Varid{tm}\;(\Varid{trans}\;(\Conid{H}\;(\Conid{F}\hsdot{\circ }{.\ }\Conid{G}))\;(\Conid{H}\;\Conid{F}\hsdot{\circ }{.\ }\Conid{G}))}.\]
Then we can define $\ensuremath{\Varid{alg}}^\sharp : \ensuremath{\Varid{tm}\;(\Varid{trans}\;(\Conid{H}\;(\Conid{M}\hsdot{\circ }{.\ }\Conid{Maybe}))\;(\Conid{M}\hsdot{\circ }{.\ }\Conid{Maybe}))}$ by
\[
\ensuremath{\Varid{alg}}^\sharp \; \alpha \; o = \ensuremath{\Varid{alg}\;(\Conid{Maybe}\;\Varid{\alpha})\;(\Varid{\sigma}_H\;\{\mskip1.5mu \Conid{M}\mskip1.5mu\}\;\{\mskip1.5mu \Conid{Maybe}\mskip1.5mu\}\;\Varid{\alpha}\;\Varid{o})}.
\]
Packaging everything up, we have \ensuremath{\{\mskip1.5mu \Conid{H},\Varid{\sigma}_H\mskip1.5mu\}\to \Conid{Handler}\;\Conid{H}\to \Conid{Handler}\;(\Hexc\mathbin{+}\Conid{H})}, where \ensuremath{\Hexc\mathbin{+}\Conid{H}} is the coproduct of the two raw higher-order functors.
The resulting construction is an example of \emph{modular handlers}
\citep{YangWu23, SPWJ19, YangWu21}.

For a wide range of effect signatures \ensuremath{\Conid{H}}, we have canonical choices of the 
term $\sigma_H$ above.
For example, a generic operation \ensuremath{\Conid{P}\;{\Rightarrow_t}\;\Conid{M}\;\Conid{A}} on a monad $M$ is the same as
$\ensuremath{\Varid{trans}}\; (K_{P,A}\;M)\; M$ where $K_{P,A}$ is the constant higher-order functor
$K_{P,A}\;M\;\alpha = P \ensuremath{\mathbin{\times_t}} (\ensuremath{\Conid{A}\;{\Rightarrow_t}\;\Varid{\alpha}})$.
In this case $\sigma_H$
can just be the identity function.
Similarly, scoped operations will correspond to \ensuremath{\Conid{H}} of the form $\lambda M.\;S \hcomp M$ for some fixed functor $S$.
In this case, $\sigma_H$ can also be the identity function.
\end{example}

\begin{remark}
Due to space limit, we cannot demonstrate many programming examples of \Fha{}
here, and we refer interested readers to the previous work on higher-order
effect handlers \citep{YPWvS2021,Wu2014,vS24,vSPW21,BvTT2024}, whose examples
can be adapted to \Fha{} (or the recursive variant in \Cref{sec:gen:rec}) easily.
\end{remark}

\begin{remark}\label{rem:no:evallet}
We do not include in \Fha{} the equation asserting that \ensuremath{\Varid{hdl}}
also commutes with \ensuremath{\Varid{let\hyp{}in}}:
\begin{hscode}\SaveRestoreHook
\column{B}{@{}>{\hspre}l<{\hspost}@{}}%
\column{9}{@{}>{\hspre}l<{\hspost}@{}}%
\column{E}{@{}>{\hspre}l<{\hspost}@{}}%
\>[B]{}\Varid{hdl\hyp{}let}\mathbin{:}\{\mskip1.5mu \Conid{H},\Conid{A},\Conid{B}\mskip1.5mu\}\to (\Varid{h}\mathbin{:}\Conid{Handler}\;\Conid{H})\to (\Varid{c}\mathbin{:}\Varid{co}\;\Conid{H}\;\Conid{A})\to (\Varid{f}\mathbin{:}\Varid{tm}\;\Conid{A}\to\Varid{co}\;\Conid{H}\;\Conid{B}){}\<[E]%
\\
\>[B]{}\hsindent{9}{}\<[9]%
\>[9]{}\to \Varid{hdl}\;\Varid{h}\;(\Varid{let\hyp{}in}\;\Varid{c}\;\Varid{f})=\Varid{h}\;\Varid{.bind}\;\Conid{A}\;\Conid{B}\;(\Varid{hdl}\;\Varid{h}\;\Varid{c})\;(\lambda \Varid{a}.\;\Varid{hdl}\;\Varid{h}\;(\,\Varid{f}\;\Varid{a})){}\<[E]%
\ColumnHook
\end{hscode}\resethooks
This is because we have chosen to work with \emph{raw} monads that may not 
validate the
monad laws, whereas computations \ensuremath{\Varid{co}\;\Conid{H}\;\Conid{A}} are axiomatised to always satisfy
these laws \Cref{eq:fha:co:laws}. 
Consequently, we can freely re-associate let-bindings in computations but not
in raw monads, so having \ensuremath{\Varid{hdl\hyp{}let}} would result in inconsistency.
Although \ensuremath{\Varid{hdl\hyp{}let}} is left out, later we will prove the \emph{canonicity} of \Fha{} -- handling \emph{closed} elements of computations
never gets stuck.
This is intuitively because in the empty context, every computation is always
equal to a computation without \ensuremath{\Varid{let\hyp{}in}} because of the equations \ensuremath{\Varid{let\hyp{}val}},
\ensuremath{\Varid{let\hyp{}assoc}} and \ensuremath{\Varid{let\hyp{}op}}.
\end{remark}

\begin{remark}\label{el:mod:hdl:in:fha}
We did not include in \Fha{} any built-in support for \emph{type-and-effect
systems} that track the effect operations that a computation may perform
\citep{BauerPretnar2014,KammarPlotkin2012,Lucassen1988}
or support for \emph{modular handlers} \citep{YangWu21,YangWu23} that
organise handlers in a composable way, 
because both of them can be \emph{derived concepts} in \Fha{}, provided that we
extend \Fha{} with some standard type/kind connectives such as products and
lists.
For example, an effect row of \emph{algebraic operations} can be given as
a type-level list \ensuremath{[\mskip1.5mu \Varid{ty}\;\mathbin{\times_k}\;\Varid{ty}\mskip1.5mu]}, where each element \ensuremath{(\Conid{P},\Conid{A})} of the list   
determines an operation receiving an argument of type $P$ and returning 
a value of type $A$.
Every list \ensuremath{[\mskip1.5mu \Varid{ty}\;\mathbin{\times_k}\;\Varid{ty}\mskip1.5mu]} then determines a \ensuremath{\Conid{RawHFunctor}} that can be supplied
to the computation judgement \ensuremath{\Varid{co}}.
In this way, effect tracking is a \emph{user-level library} rather than a
built-in feature of \Fha, and effect polymorphism is just a special case of
the (higher-kinded) polymorphism that \Fha{} already has.
Details of how this is done can be found in \Cref{app:eff:fam}.
\end{remark}

\begin{remark}
We will not discuss type checking for \Fha{} in this paper, as \Fha{}
adds little type-level complexity to \Fom, and we expect the existing
algorithms for type-checking polymorphic calculi \citep{Dunfield2013,
Jones_2007, Daan2008} can be extended to work with \Fha{}.
\end{remark}

\subsection{An Extension of General Recursion}\label{sec:gen:rec}

The \ensuremath{\Varid{hdl}} construct of \Fha{} is a form of \emph{structural recursion}, and
it is also possible to extend \Fha{} with \emph{general recursion}.
We will refer to this extension as \Fhar.
The signature of \Fhar{} extends that of \Fha{} with a new family of judgements
\ensuremath{\Varid{pco}} for \emph{partial computations} that has the same signature as \ensuremath{\Varid{co}}:
\begin{equation}\label{el:fhar:pco}
\ensuremath{\Varid{pco}\mathbin{:}(\Conid{H}\mathbin{:}\Conid{RawHFunctor})\to(\Conid{A}\mathbin{:}\Varid{el}\;\Varid{ty})\to\Jdg}.
\tag{\Fhar-1}
\end{equation}
The original computation judgement \ensuremath{\Varid{co}} is still kept in the language and is
used for total computations as usual.
Most accompanying rules for \ensuremath{\Varid{co}} in \Cref{sec:fha:co} are inherited by \ensuremath{\Varid{pco}}:
\ensuremath{\Varid{val}}, \ensuremath{\Varid{let\hyp{}in}}, \ensuremath{\Varid{th}}, \ensuremath{\Varid{op}}, and all their associated equations.
We shall refer to the copy of them for \ensuremath{\Varid{pco}} by the same names as before,
except for thunks of partial computations, which we call \ensuremath{\Varid{pth}\mathbin{:}\Conid{RawHFunctor}\to \Varid{el}\;\Varid{ty}\to \Varid{el}\;\Varid{ty}}.

The new rules for \ensuremath{\Varid{pco}} are as expected a fixed-point combinator:
\begin{equation}
\begin{aligned}
& \ensuremath{\Varid{fix}\mathbin{:}\{\mskip1.5mu \Conid{H},\Conid{A}\mskip1.5mu\}\to (\Varid{pth}\;\Conid{H}\;\Conid{A}\to \Varid{pco}\;\Conid{H}\;\Conid{A})\to \Varid{pco}\;\Conid{H}\;\Conid{A}} \\
& \ensuremath{\Varid{fix\hyp{}eq}\mathbin{:}\{\mskip1.5mu \Conid{H},\Conid{A},\Varid{f}\mskip1.5mu\}\to \Conid{Y}\;\Varid{f}\mathrel{=}\Varid{f}\;({\Downarrow}\;\Conid{Y}\;\Varid{f})}
\end{aligned}
\tag{\Fhar-2} \label{eq:fhar:y}
\end{equation}

\begin{element}
Another difference between \ensuremath{\Varid{pco}} and \ensuremath{\Varid{co}} is their elimination rules:
\ensuremath{\Varid{hdl}} allows computations \ensuremath{\Varid{co}\;\Conid{H}\;\Conid{A}} to be handled into
any raw monad \ensuremath{\Conid{M}} equipped with an $H$-operation, but naturally, \ensuremath{\Varid{pco}} shall
only be handled into monads \ensuremath{\Conid{M}} that support recursion.
In the current call-by-value setting, the only thing that supports
recursion is \ensuremath{\Varid{pco}}, so we will require that the raw monad $M$ send every type \ensuremath{\Conid{A}\mathbin{:}\Varid{el}\;\Varid{ty}} to thunks of partial computations \ensuremath{\Varid{pth}\;\Conid{H}\;(\Conid{F}\;\Conid{A})} for some \ensuremath{\Conid{H}} and \ensuremath{\Conid{F}\mathbin{:}\Varid{el}\;\Varid{ty}\to \Varid{el}\;\Varid{ty}}:
\begin{hscode}\SaveRestoreHook
\column{B}{@{}>{\hspre}l<{\hspost}@{}}%
\column{3}{@{}>{\hspre}l<{\hspost}@{}}%
\column{7}{@{}>{\hspre}l<{\hspost}@{}}%
\column{E}{@{}>{\hspre}l<{\hspost}@{}}%
\>[B]{}\lhskeyword{record}\;\Conid{HandlerRec}\;(\Conid{H}\mathbin{:}\Conid{RawHFunctor})\mathbin{:}\Jdg\;\lhskeyword{where}{}\<[E]%
\\
\>[B]{}\hsindent{3}{}\<[3]%
\>[3]{}\lhskeyword{include}\;\Conid{Handler}\;\Conid{H}\;\lhskeyword{as}\;\Varid{h}{}\<[E]%
\\
\>[B]{}\hsindent{3}{}\<[3]%
\>[3]{}\Conid{F}{}\<[7]%
\>[7]{}\mathbin{:}\Varid{el}\;\Varid{ty}\to \Varid{el}\;\Varid{ty}{}\<[E]%
\\
\>[B]{}\hsindent{3}{}\<[3]%
\>[3]{}\Varid{eq}{}\<[7]%
\>[7]{}\mathbin{:}(\Conid{A}\mathbin{:}\Varid{el}\;\Varid{ty})\to \Varid{h}\;\Conid{A}\mathrel{=}\Varid{pth}\;\Conid{H}\;(\Conid{F}\;\Conid{A}){}\<[E]%
\ColumnHook
\end{hscode}\resethooks
Here we have formulated the requirement \ensuremath{\Varid{eq}} using the equality type of \LCCLF,
and in an implementation of the type checker for \Fhar{}, the equation \ensuremath{\Varid{eq}}
may be mechanically checked since the kind language of \Fhar{} is normalising.
The language \Fhar{} then has the following declaration:
\begin{equation}\label{eq:fhar:eval}
\ensuremath{\Varid{hdl}\mathbin{:}\{\mskip1.5mu \Conid{H},\Conid{A}\mskip1.5mu\}\to(\Varid{h}\mathbin{:}\Conid{HandlerRec}\;\Conid{H})\to\Varid{pco}\;\Conid{H}\;\Conid{A}\to\Varid{tm}\;(\Varid{h}\;\Conid{A})}
\tag{\Fhar-3}
\end{equation}
together with equations \ensuremath{\Varid{hdl\hyp{}val}} and \ensuremath{\Varid{hdl\hyp{}op}} similar to those of \ensuremath{\Varid{co}} (\ref{eq:fha:evalval}, \ref{eq:fha:evalop}) with \ensuremath{\Varid{co}} replaced by \ensuremath{\Varid{pco}}, \ensuremath{\Varid{th}} replaced by \ensuremath{\Varid{pth}}, and \ensuremath{\Conid{Handler}} replaced by \ensuremath{\Conid{HandlerRec}}.
\end{element}

For discussing its meta-theoretic properties,
it is convenient to include in \Fhar{} the empty type:
\begin{equation}\label{eq:fhar:empty}
\spreadTag{\Fhar-4}{
\ensuremath{\Varid{empty}\mathbin{:}\Varid{el}\;\Varid{ty}}
\also
\ensuremath{\Varid{absurd}\mathbin{:}(\Conid{A}\mathbin{:}\Varid{el}\;\Varid{ty})\to\Varid{tm}\;\Varid{empty}\to\Varid{tm}\;\Conid{A}}
}
\end{equation}
so that we have a judgement \ensuremath{\Varid{pco}\;\Conid{VoidH}} of partial computations without
any other effects, where \ensuremath{\Conid{VoidH}} is the constant raw higher-order
functor: \ensuremath{\Conid{VoidH}\;\anonymous \;\anonymous \mathrel{=}\Varid{empty}}.

\subsection{Semantic Models and the Category of Judgements}
 
We have presented the calculus \Fha{} using \LCCLF, which we hope to be an example demonstrating the compactness and precision of using a
type-theoretic LF to present programming languages.
However, syntactic nicety is not the only advantage of using LFs.
A bigger advantage is that an LF can provide useful general results for
object languages defined in it.
For every signature,
\LCCLF{} provides (1) a notion of semantic \emph{models}, (2) a category of
\emph{judgements}, and (3) an equivalence between models and functors out of
the category of judgements.
These results are established by \citet{Yang2025} in detail, and below we
record the special cases for the signature \Fha{}.

\subsubsection{Semantic Models}
To define the concept of models of \Fha{} in a category, we first need an 
auxiliary concept of categories that can interpret the whole logical framework
$\LCCLF$.

\begin{definition}\label{def:lf:cat}
An \emph{LF-category} is a category $\CatC$ together with an interpretation of
the type formers of \LCCLF{} (\Cref{lang:lf}), i.e.\ the unit type $1$,
$\Sigma$-types, restricted $\Pi$-types, and a universe $\Jdg$ closed under $1$,
$\Sigma$, $\Pi$, and extensional equality types $a = b$.
\end{definition}

We will say `an LF-category $\tuple{\CatC, U}$' where $U$ is the interpretation
of the universe. 
Strictly speaking, the interpretation of other type formers
is also part of the structure, but they are determined uniquely (up to
isomorphisms) by their respective universal properties.

For example,
the category $\Set$ can be made an LF-category, with $\Jdg$ being interpreted as some set-theoretic universe $U$.
A trivial choice of $U$ is just the one-element set $\aset{1}$, which is closed under $1$, $\Sigma$, $\Pi$, $=$.
More generally, the presheaf category $\PrC$ over a (small) 
locally cartesian closed category (LCCC) $\CatC$
can be made an LF-category $\tuple{\PrC, \TyC}$ where the universe $\TyC$ classifies
exactly (the Yoneda embedding of) the objects and morphisms of $\CatC$; see
\citet[\S IV]{Yang2025} for details.

\begin{definition}\label{def:fha:model}
A \emph{model} $\MM$ of \Fha{} in a (small) LCCC $\CatC$ is a morphism $\MM : 1 \to
\ensuremath{\llbracket\Fha\rrbracket_{\CatC}}$ in the presheaf category $\PrC$, where the object \ensuremath{\llbracket\Fha\rrbracket_{\CatC}} is the interpretation of the record type that has all
the declarations of \Fha{} (\ref{eq:fha:kinds} to \ref{eq:ttff}, \ref{eq:fha:co} to \ref{eq:fha:evalop}) as fields, with $\Jdg$ interpreted as $\TyC$.
\begin{equation}\label{eq:fha:model}
\begin{hscode}\SaveRestoreHook
\column{B}{@{}>{\hspre}l<{\hspost}@{}}%
\column{3}{@{}>{\hspre}l<{\hspost}@{}}%
\column{61}{@{}>{\hspre}l<{\hspost}@{}}%
\column{75}{@{}>{\hspre}c<{\hspost}@{}}%
\column{75E}{@{}l@{}}%
\column{E}{@{}>{\hspre}l<{\hspost}@{}}%
\>[B]{}\lhskeyword{record}\;\Fha\;\lhskeyword{where}{}\<[E]%
\\
\>[B]{}\hsindent{3}{}\<[3]%
\>[3]{}\Varid{ki}\mathbin{:}\Jdg;\ \ \Varid{el}\mathbin{:}\Varid{ki}\to \Jdg;\ \ \Varid{ty}\mathbin{:}\Varid{ki};{}\<[61]%
\>[61]{}\ \ {}\<[75]%
\>[75]{}\mathbin{...}{}\<[75E]%
\ColumnHook
\end{hscode}\resethooks
\end{equation}
\end{definition}

\citet{Yang2025} also defined a notion of \emph{model isomorphisms}, so
we have a groupoid $\LFMod{\Fha}{\CatC}$ of \Fha-models and isomorphisms in an
LCCC $\CatC$.
In this paper, model isomorphisms will not play an important role so we omit
their definition here.

\begin{remark}\label{rem:define:mod}
\Cref{def:lf:cat,def:fha:model} may appear as rather opaque to the reader, but
we need not worry about them too much.
For our purposes in this paper, what we need to know about them is that,
for every LCCC $\CatC$, we can define a model of \Fha{} in $\CatC$ by defining
a closed element of the record type \ensuremath{\Fha} \Cref{eq:fha:model} in an internal
language of $\PrC$ with $\Jdg$ replaced by the universe $\TyC$. 

In fact, our semantic domain $\CatC$ usually already has a universe $U$ that can interpret
$\Jdg$.
In this case, we do not have to move to the bigger category $\PrC$. 
To construct a model of \Fha{} in $\CatC$, it is
sufficient to construct a closed element of the record \ensuremath{\Fha} in an internal
language of $\CatC$ itself, with $\Jdg$ replaced by the universe $U$.
We will see several examples of this in the following sections.
\end{remark}

\subsubsection{Category of Judgements}\label{sec:cat:jdg}
In categorical logic, we usually have the category of \emph{types} or 
\emph{contexts} that organises the syntactic entities of a
language as a category.
For languages defined using \LCCLF{}, the category containing the syntactic
entities is the \emph{category of judgements}.

\begin{definition}
The \emph{category of judgements} $\CatJdgOf{\Fha}$ for the language
\Fha{} has (1) \LCCLF-terms $\hypo{\Fha}{A : \Jdg}$ as objects, and (2)
terms $\hypo{\Fha}{f : A \to B}$ as morphisms.
Identities and composition are the evident identity function and function composition in \LCCLF{}.
\end{definition}

The category of judgements is locally cartesian closed.
For every object $A : \Jdg$ of $\CatJdgOf{\Fha}$ (here we omit the
context `$\Fha \vdash$'), the terminal object of the slice category
$\CatJdgOf{\Fha}/A$ is $(\lambda x.\;x) : A \to A$.
Given two objects
$f : B \to A$ and $g : C \to A$ in $\CatJdgOf{\Fha}/A$, their
product is $\lambda p.\;f\;(\pi_1 \; p) : P \to A$ where $P \defeq \Sigma (b : B).\;\Sigma (c : C).\; f\;b = g \; c$, and their
exponential is $\pi_1 : E \to A$ where
\[
\spread*{
E \defeq \Sigma(a : A).\;B_a \to C_a
\also
B_a \defeq \Sigma(b : B).\;f\;b = a
\also
C_a \defeq \Sigma(c:C).\;g\;c = a
}
\]
Among all LCCCs, 
the category $\CatJdgOf{\Fha}$ has the following universal property:

\begin{theorem}[\citet{Yang2025}]\label{thm:fha:functor}
Let $\CatC$ be an LCCC.
The groupoid $\LCCC_{\cong}(\CatJdgOf{\Fha}, \CatC)$ of LCC-functors  and 
natural isomorphisms is equivalent to the groupoid 
of $\Fha$-models in $\CatC$.
\end{theorem}

The practical relevance of this theorem is twofold.
Firstly, it gives us a functor $\sem{\blank}_M : \CatJdgOf{\Fha} \to \CatC$
after we define a model $\MM$ of \Fha{} in $\CatC$.
This functor assigns a meaning of every \Fha-judgement and derivation (not just
the generating ones declared in the signature \Fha{}) 
in the category $\CatC$.
Secondly, it provides a connection between \Fha{} as a \emph{syntactic
signature} and as a \emph{category} $\CatJdgOf{\Fha}$.  
This connection enables us to use categorical tools to study the theory \Fha{}.
For example, in \Cref{sec:fha:lr:model}, we will use a categorical tool known
as \emph{Artin gluing} to prove properties of \Fha. 

\section{The Realizability Model}\label{sec:fha:real:model}

In this section, we present a model of \Fha{} in the category
$\Asm(\bA)$ of \emph{assemblies} over an arbitrary \emph{partial combinatory
algebra} (PCA) $\bA$.
This model serves two purposes: 
(1)~it shows that \Fha{} is consistent in the sense
that \ensuremath{\Varid{tt}} and \ensuremath{\Varid{ff}} are not equal terms of \ensuremath{\Varid{bool}} in the equational theory
of \Fha{},
and (2)~it provides a way to extract executable programs, e.g.\ terms of
$\lambda$-calculus, from well typed terms of \Fha{}, giving us a way to run
terms of \Fha{} without an explicit operational semantics.

\subsection{Assemblies and Their Language}
We will only use the category of assemblies as a black box via a type-theoretic
internal language \Cref{lang:asm}, 
so in principle the reader does not even need to know what an assembly is to
read this section and treat the model presented in this section as a syntactic translation from \Fha{} to another type theory.
We refer readers who are interested in how assemblies work `under the hood' to
tutorials on realizability by \citet{Bauer2022}, \citet{deJong2024},
\citet{Streicher2017}, and the comprehensive book account by
\citet{Oosten_2008}.

A partial combinatory algebra $\tuple{\bA, \cdot}$ is an abstraction for an
untyped model of computation: $\bA$ is a set whose elements serve the dual
purpose of programs and data,
and $\cdot : \bA \times \bA \rightharpoonup \bA$ is a partial binary
operation subject to certain conditions.
The intuition for $n \cdot m$ is applying the program $n$ to the input data $m$.
Notable examples of PCAs include 
(1) $\beta$-equivalence classes of closed $\lambda$-terms together with
$\lambda$-term application and 
(2) the set of natural numbers together with $n \cdot m$ being the result of
running the $n$-th Turing machine with input $m$ (if the $n$-th Turing machine does not halt on $m$, $n \cdot m$ is undefined).
The second example is called \emph{Kleene's first algebra} $\mathbb{K}$.

For every PCA $\tuple{\bA, \cdot}$, the category $\Asm(\bA)$ of assemblies over
$\bA$ (also known as \emph{$\omega$-sets}) is
roughly a category of `computable sets and functions':
an object $\tuple{X, \avert{\blank}_X}$ of $\Asm(\bA)$ is a set $X$ with
a function $\avert{\blank}_X : X \to \Pow(\bA)$ mapping every 
$x \in X$ to a \emph{non-empty} subset $\avert{x}_X$ of $\bA$.
The intuition is that $\avert{x}_X$ is the set of $\bA$-elements that 
\emph{encode} or \emph{realize} $x \in X$.
A morphism $f : \tuple{X, \avert{\blank}_X} \to \tuple{Y, \avert{\blank}_Y}$ 
in $\Asm(\bA)$ is
a set-theoretic function $f : X \to Y$ such that there exists an element 
$r \in \bA$ and for all $x \in X$, $n \in \avert{x}_X$, $r \cdot n$
is defined and $r \cdot n \in \avert{f \; x}_Y$.

The category $\Asm(\bA)$ of assemblies has many pleasant properties, making it a standard
tool for interpreting programming languages, especially those with
impredicative polymorphism.
We will access the nice structure of $\Asm(\bA)$ via the following
type-theoretic internal language.
\begin{typetheory}\label{lang:asm}
The language $\TTASM$ is a dependent type theory with the
following type formers:
\begin{itemize}
\item dependent function types ($\Pi$-types), dependent pair types
($\Sigma$-types), extensional equality types, and inductive types 
(e.g.\ the unit type, the empty type, the natural number type);
\item
three cumulative universes $P : V_1 : V_2$, each closed under the aforementioned type
formers;
\item
$P$ is \emph{impredicative} in the sense that for every type $A$ (not
necessarily in $P$) and every type family $B : A \to P$, the dependent function
type $(x : A) \to B\;x$ is in $P$.
\end{itemize}
\end{typetheory}

The language $\TTASM$ can be interpreted in the category $\Asm(\bA)$
for every non-trivial PCA $\bA$.
Interpretations of similar type theories in the category of assemblies can be
found in \citet[\S 3.4]{Hofmann_1997} and \citet[Chapter 6]{Luo1994}. 
Non-triviality of $\bA$ is needed here for the impredicative universe $P$ to
have types with more than one elements (if $\bA$ is the trivial one-element
PCA, $\Asm(\bA)$ degenerates to the category of sets and $P$ is the set
$\aset{\bot, \top}$ of classical propositions).

\subsection{The Realizability Model of \Fha}\label{sec:real:mod:fha}
As mentioned in \Cref{rem:define:mod},
to define a model of \Fha{} in $\Asm(\bA)$, it is sufficient to 
construct an element of the record type $\sem{\Fha}_{V_2}$ in \TTASM{} that contains all declarations of \Fha{} with $\Jdg$ replaced by the universe $V_2$.
In the following, we construct such a model $\RM : \sem{\Fha}_{V_2}$.

The model of the \Fom-fragment is standard.
Kinds are interpreted as the predicative universe $V_1$ and 
the base kind \ensuremath{\Varid{ty}\mathbin{:}\Varid{ki}} is interpreted as the impredicative universe $P$:
\begin{align*}
&\RM.\ensuremath{\Varid{ki}} : V_2        & &\RM.\ensuremath{\Varid{el}}    : \RM.\ensuremath{\Varid{ki}} \to V_2 &&\RM.\ensuremath{\Varid{ty}} : \RM.\ensuremath{\Varid{ki}}   & & \RM.\ensuremath{\Varid{tm}}: \RM.\ensuremath{\Varid{el}} \; \RM.\ensuremath{\Varid{ty}} \to V_2 \\
&\RM.\ensuremath{\Varid{ki}} = V_1        & &\RM.\ensuremath{\Varid{el}}\;k = k &&\RM.\ensuremath{\Varid{ty}} = P        & & \RM.\ensuremath{\Varid{tm}}\;A = A
\end{align*}
In the definitions of $\RM.\ensuremath{\Varid{el}}$ and $\RM.\ensuremath{\Varid{tm}}$, lifting types along the cumulative 
universes $P : V_1 : V_2$ is implicit:
the argument $k$ of $\RM.\ensuremath{\Varid{el}}$ is in the universe $V_1$, so it can also be used
as an element of the universe $V_2$.

Function kinds $\RM.\ensuremath{\text{\textunderscore}{\Rightarrow_k}\text{\textunderscore}} : \RM.\ensuremath{\Varid{ki}} \to \RM.\ensuremath{\Varid{ki}} \to \RM.\ensuremath{\Varid{ki}}$ are interpreted
by function types in $V_1$, and $\RM.\ensuremath{\Varid{{\Rightarrow_k}\hyp{}iso}}$ is the identity isomorphism.
The unit, Boolean, function types of \ensuremath{\Varid{ty}} in \Fha{} are interpreted as the
corresponding type formers in the universe $P$.
Impredicative polymorphic function types $\ensuremath{\allfor}$ are interpreted as dependent
function types:
\begin{lgather*}
\RM.\ensuremath{\allfor} : (k : \RM.\ensuremath{\Varid{ki}}) \to (\RM.\ensuremath{\Varid{el}} \; k \to \RM.\ensuremath{\Varid{el}}\;\RM.\ensuremath{\Varid{ty}}) \to \RM.\ensuremath{\Varid{el}} \; \RM.\ensuremath{\Varid{ty}} \\
\RM.\ensuremath{\allfor} \; k \; A =  (\alpha : k) \to A\; \alpha
\end{lgather*}
This is well typed because the codomain type $\RM.\ensuremath{\Varid{ty}}$, i.e.\ $P$, is an impredicative universe.

%

The model of the computation judgement $\ensuremath{\Varid{co}}\; H\; A$ is less obvious because of
the mismatch between computations and raw monads in \Fha{}:
computations satisfy the monadic laws strictly (\ensuremath{\Varid{let\hyp{}val}}, \ensuremath{\Varid{val\hyp{}let}}, \ensuremath{\Varid{let\hyp{}assoc}}
from \ref{eq:fha:co:laws}), while raw monads do not.
Consequently, we cannot model \ensuremath{\Varid{co}\;\Conid{H}} as the initial \emph{raw monad} equipped with $H$-operations
because it then would not satisfy the monadic laws.
Conversely, we cannot model it as the initial \emph{monad} equipped
with $H$-operations either because then it cannot be evaluated into raw monads.
Our solution is to model computations by a combination of impredicative
encoding and continuation-passing transformation:
\begin{lgather*}
\RM.\ensuremath{\Varid{co}} : \RM.\ensuremath{\Conid{RawHFunctor}} \to \RM.\ensuremath{\Varid{el}} \; \RM.\ensuremath{\Varid{ty}} \to P\\
\RM.\ensuremath{\Varid{co}} \; H \; A = (h : \RM.\ensuremath{\Conid{Handler}\;\Conid{H}}) \to (B : P) \to (A \to h\;B) \to h \; B  \numberthis\label{el:real:co}
\end{lgather*}
Thunking \ensuremath{\Varid{th}\;\Conid{H}\;\Conid{A}} can be modelled as the identity, because in the model $\RM$,
computations and values live in the same universe $P$.
The computation formers and \ensuremath{\Varid{hdl}} are defined as follows:
\begin{lgather*}
\RM.\ensuremath{\Varid{val}}: \impfun{H, A} A \to \RM.\ensuremath{\Varid{co}}\;H\;A  \\
\RM.\ensuremath{\Varid{val}} \; a = \lambda h\; B \; (r : A \to h\;B).\; r \; a \\[6pt]
\RM.\ensuremath{\Varid{let\hyp{}in}}: \impfun{H, A, B}  \RM.\ensuremath{\Varid{co}}\;H\;A  \to (A \to \RM.\ensuremath{\Varid{co}}\;H\;B) \to \RM.\ensuremath{\Varid{co}}\;H\;B \\
\RM.\ensuremath{\Varid{let\hyp{}in}}\; \imparg{A, B} \; c \; k = \lambda h\; C \; (r : B \to h \; C). \;  
  c \; h \; C\; (\lambda a.\; k \; a \; h \; C \; r)
\end{lgather*}
The model of \ensuremath{\Varid{hdl}} directly follows from the definition of $\RM.\ensuremath{\Varid{co}}$:
\begin{lgather*}
\RM.\ensuremath{\Varid{hdl}} : \impfun{H, A} (h : \RM.\ensuremath{\Conid{Handler}}\; H) \to \RM.\ensuremath{\Varid{co}} \; H \; A \to h \; A \\
\RM.\ensuremath{\Varid{hdl}} \; \imparg{H}\; \imparg{A}\; h \; c = c \; h\; A \; (h.\ensuremath{\Varid{ret}})
\end{lgather*}  
It can be checked the definitions for computations above collectively validate
all the equational laws of \Fha{}. 
Detailed calculations can be found in \Cref{app:real:model:laws}.

\begin{remark}
Our first attempt to model computations \ensuremath{\Varid{co}\;\Conid{H}} was to use the \emph{initial pointed
raw functor with $H$-operations}, which (in the usual lawful setting) can be
equipped with a law-abiding monad structure and can be evaluated into
every raw monad $\MM$ with $H$-operations, since $\MM$ is a pointed functor as
well. 
Unfortunately, it turns out that, in our current lawless setting, this approach
does not validate the equations on computations, unless the raw higher-order 
functor $H$ preserves composition of transformations.
This shows that our usual intuitions about category theory may fail in a subtle
way in `lawless category theory'.
\end{remark}

We have completed the definition of the model $\RM : \sem{\Fha}_{V_2}$ in $\Asm(\bA)$.
An immediate consequence is the \emph{consistency} of the equational theory of System \Fha{}.

\begin{theorem}\label{thm:fha:consistency}
The equational theory of System \Fha{} is \emph{consistent}, in the sense that 
the closed terms \ensuremath{\Varid{tt}} and \ensuremath{\Varid{ff}\mathbin{:}\Varid{bool}} are not judgementally equal.
\end{theorem}

\begin{proof}
Let $\bA$ be any non-trivial PCA, such as Kleene's first algebra.
The interpretations of \ensuremath{\Varid{tt}} and \ensuremath{\Varid{ff}\mathbin{:}\Varid{bool}} in the realizability model $\RM$ are
different morphisms $1 \to 1 + 1$ in $\Asm(\bA)$.
Because interpretation respects judgemental equalities, \ensuremath{\Varid{tt}} and \ensuremath{\Varid{ff}} 
cannot be judgementally equal.
\end{proof}

The realizability model also gives us a way to do \emph{program
extraction}:
\begin{theorem}\label{thm:fha:extraction}
For every closed term \ensuremath{\Varid{t}\mathbin{:}\Varid{tm}\;\Varid{bool}} of \Fha{}, there exists a
$\lambda$-term $\avert{t}$ that normalises to a Church Boolean value.
Moreover, if $t = t'$, $\avert{t}$ and $\avert{t'}$ normalise
to the same value.
\end{theorem}

\begin{proof}
Terms of \Fha{} are interpreted as morphisms in the category $\Asm(\bA)$ of
assemblies, which are realized by elements of the underlying PCA $\mathbb{A}$, in particular, the PCA $\Lambda$ of $\lambda$-terms.
Therefore every closed term \ensuremath{\Varid{t}\mathbin{:}\Varid{tm}\;\Varid{bool}}
can be interpreted as a morphism
$1 \to 1 + 1$ in $\Asm(\Lambda)$, which by the definition of $\Asm$ 
is realised by a $\lambda$-term.
\end{proof}

The interpretation of \Fha{} in the realizability model is clearly constructive, so based on the realizability
model we can implement a compiler that takes in 
well typed \Fha-terms and outputs $\lambda$-terms following
the definition of the model $\RM$ (by simply erasing the type information in $\RM$).
In the next section, we will further see that the realizability model is in fact
\emph{adequate} with respect to the equational theory of \Fha{}: if the
interpretation of a closed Boolean term ${p : \ensuremath{\Varid{tm}\;\Varid{bool}}}$ in the
realizability model is true (resp.\ false), then $p = \ensuremath{\Varid{tt}}$ (resp.\ $p =
\ensuremath{\Varid{ff}}$) in the equational theory of \Fha{}.


\subsubsection*{Modelling General Recursion}
The realizability model of \Fha{} can be extended to a model of \Fhar{} from
\Cref{sec:gen:rec} using \emph{synthetic domain theory}
\citep{rosolini1986continuity,Hyland_first_1991,PhoaThesis}.
We will not go into this here for the lack of space, but interested readers can
see how it is done in \Cref{sec:fhar:real:model}, which also provides a
mini introduction to synthetic domain theory.
Consequences of this model are analogues of
\Cref{thm:fha:consistency} and \Cref{thm:fha:extraction}.

\begin{theorem}
The equational theory of \Fhar{} is consistent: for every \ensuremath{\Conid{H}\mathbin{:}\Conid{RawHFunctor}},
\ensuremath{(\Varid{val}\;\Varid{tt})} and \ensuremath{(\Varid{val}\;\Varid{ff})\mathbin{:}\Varid{pco}\;\Conid{H}\;\Varid{bool}} are not judgementally equal in \Fhar{}.
\end{theorem}

\begin{theorem}\label{thm:fhar:extraction}
For every closed term \ensuremath{\Varid{t}\mathbin{:}\Varid{pco}\;\Conid{VoidH}\;\Varid{bool}} of \Fha{}, there exists a
$\lambda$-term $\avert{t}$ that either diverges or normalises to a Church Boolean value.
Moreover, if $t = t'$, then
$\avert{t}$ and $\avert{t'}$ are Kleene equal.
\end{theorem}

\section{The Synthetic Logical Relation Model}\label{sec:fha:lr:model}

The equational theory of \Fha{} provides the programmer with a set of reasoning
principles to understand the behaviour of \Fha-programs, and also provides
compiler writers with a set of program transformations for optimisation.
One natural question is -- \emph{how complete is the equational theory}?
We answer this question by proving the following theorem in this section.

\begin{theorem}[Canonicity]\label{thm:canonicity}
For every closed Boolean term $b : bool$ of System \Fha{}, either $b = \ensuremath{\Varid{tt}}$ or $b =
\ensuremath{\Varid{ff}}$ holds (but not both) in the equational theory of \Fha.
\end{theorem}

The common proof strategy for showing meta-theoretic properties of programming
languages such as canonicity is the method of \emph{logical relations}
\citep{plotkin_lambda_1973,plotkin_lambda_definability_1980}, also known as
\emph{the computability method} or the \emph{reducibility method} in the
literature
\citep{Tait67,martinlof1975models,martinlof1975AnIntui,Girard1972,Statman1985}.
The high-level idea is to construct a model of a language $L$ such that each
judgement $J$ of $L$ is interpreted as a set of $J$-derivations satisfying
certain properties or equipped with certain data.
Inference rules
of $L$ are then shown to preserve those associated properties or data.
Categorically, such a logical-relation model lives in the \emph{glued category}
of the category of types/judgements of $L$ and the category of sets (or some
other presheaf topos where the meta-theoretic information naturally lives)
\citep{freyd1978proving,fiore2022semantic,Thorsten1995}.

A recent development of the method of logical relations is 
Sterling's \emph{synthetic Tait computability} (STC)
\citep{sterling2021LRAT,sterling:2021:thesis,sterling2022naive},
whose idea is to (1) embed the category of types/judgements in the
(presheaf) topos over it by Yoneda embedding, (2) glue this (presheaf) topos containing 
information of the object language
with
the (presheaf) topos where the meta-theoretic information lives,
which always results in a new (presheaf) topos, and then (3) use a
type-theoretic language to describe the constructions in the glued
topos.
Passing to the presheaf category in the first step is needed so that the
resulting glued category is a topos, where we have a very rich internal language
to describe the construction of the logical relation model.
The advantage of this approach is that after the internal language is set up,
many tedious aspects in a typical logical-relation proof are taken care of
automatically, turning a logical-relation proof into a guided programming
puzzle.

\subsection{The Language of STC for \Fha}
To apply the method of STC to prove canonicity of \Fha{}
(\Cref{thm:canonicity}), we first recall that
in \Cref{sec:cat:jdg}
we defined a category of judgements
$\CatJdgOf\Fha$ for \Fha, whose objects are terms  $J : \Jdg$ and morphisms are functions $f : J \to J'$ under the signature \Fha{} in \LCCLF.

\begin{definition}
The \emph{glued category} $\GFha$ of $\Pr(\CatJdgOf\Fha)$ and $\Set$ along 
the global section functor
$\Hom(1, \blank) : \Pr(\CatJdgOf\Fha) \to \CatSet$
is exactly the \emph{comma category} $\Comma{\Set}{\Hom(1, \blank)}$, whose
  objects are tuples $\tuple{A \in \Pr(\CatJdgOf\Fha),\  P \in
\CatSet,\ p : P \to \Hom(1, A)}$, and
  morphisms $\tuple{A, P, p} \to \tuple{B, Q, q}$ are pairs 
  $\tuple{f : A \to B, g : P \to Q}$ making the following diagram in $\Set$ commute:
  \[\begin{tikzcd}[ampersand replacement=\&, column sep = 2cm]
  P \& Q \\
  {\Hom(1, A)} \& {\Hom(1, B)}
  \arrow["g", from=1-1, to=1-2]
  \arrow["p"', from=1-1, to=2-1]
  \arrow["q", from=1-2, to=2-2]
  \arrow["{\Hom(1, f)}"', from=2-1, to=2-2]
\end{tikzcd}\]
\end{definition}

Usually, the presheaf $A$ of an object $\tuple{A, P, p} \in \GFha$ 
will just be the Yoneda embedding of some
\Fha-judgement $J$. 
In this case $\Hom(1, A)$ is the set of 
closed derivations of the judgement $J$, 
and $p : P \to \Hom(1, A)$ is understood as a
\emph{proof-relevant predicate} over $J$-derivations.
For every $a \in \Hom(1, A)$, the set $\aset{e \in P \mid p(e) = a}$ is the
set of proofs that $a$ satisfies the predicate $p : P \to \Hom(1, A)$.
A morphism $\tuple{f, g} : \tuple{A, P, p} \to \tuple{B, Q, q}$ in $\GFha$ is
then a derivation from $A$ to $B$ that preserves the associated predicates $P$ and $Q$ on $A$ and $B$.

The presheaf category $\Pr(\CatJdgOf\Fha)$ contains a \emph{syntactic model}
$\MM$ of \Fha, where every judgement $J$ is interpreted as its Yoneda embedding
$\yo J$.
A (proof-relevant) logical-relation model of \Fha{} is then a model $\MM^*$ of
\Fha{} in the category $\GFha$ such that 
$\MM^*$ under the first projection $\GFha \to \Pr(\CatJdgOf\Fha)$
is exactly the syntactic model $\MM$.
We will construct our logical-relation model $\MM^*$ via
an internal language of $\GFha$.
In the following, we first introduce this language;
for a proper introduction, we refer the reader to the expositions by
\citet{huang2023synthetictaitcomputabilityhard} and \citet{sterling2022naive}.
The most comprehensive account of STC so far is still
\varcitet{sterling:2021:thesis}{'s} thesis.

\begin{typetheory}
The language \TTSTC{} of \emph{synthetic Tait computability} for \Fha{} is a dependent type theory
with the structure of an elementary topos (\Cref{el:ttstc:topos}), 
universes (\Cref{el:ttstc:universe}), 
a distinguished proposition $\sop$ and a model of \Fha{} under $\sop$
(\Cref{el:ttstc:obj}), 
and glue types (\Cref{el:ttstc:glue:ty}).
\end{typetheory}

First of all, \TTSTC{} has the type formers that axiomatise the
structure of an elementary topos.

\begin{axiom}[\TTSTC-ETopos]\label{el:ttstc:topos}
\TTSTC{} has the following type formers:
\begin{itemize}
  \item a unit type $\unitty$, function types $A \to B$, $\Sigma$-types $\Sigma (a
  : A).\; B$, extensional identity types $a = b$;

  \item a universe $\Omega$ such that 
  (1) it \emph{classifies all propositional types}: if a type $P$ satisfies
  $(p, q : P) \to p = q$, then there is $\classify{P} : \Omega$ with an
  isomorphism $\phi_P : \classify{P} \iso P$;
  (2) it is \emph{extensional}: if $P, Q : \Omega$ and $P \iso Q$, then $P = Q$; and
  (3) it is \emph{proof irrelevant}: if $P : \Omega$ and $p, q : P$, then $p = q$.
\end{itemize}
\end{axiom}

The theory of \emph{elementary toposes}
\citep{borceux1994handbookVol3,MacLane1994} tells us that a great deal of well
behaved logical structures can be defined from the type formers in
\Cref{el:ttstc:topos}, including $\Pi$-types $(a : A) \to B$,
(intuitionistic) logical connectives ($\land$, $\lor$, $\to$, $\top$, $\bot$,
$\exists$,
$\forall$) on $\Omega$, the empty type $0$, coproduct types $A + B$,
quotient types $A/R$.

Given a type $A$ and $P : A \to \Omega$, we define $\aset{x : A \mid
P(x)} \defeq \Sigma (a : A).\; P(a)$ and treat $\aset{x : A \mid P(x)}$ as a
subtype of $A$, eliding the proof of the proposition $P(x)$ and the
pairing/projection operations:
\[
\spread*{
  \inferrule{a : A \\ \_ : P(a)}{a : \aset{x : A \mid P(x)}}
  \also
  \inferrule{a : \aset{x : A \mid P(x)}}{a : A}
  \also
  \inferrule{a : \aset{x : A \mid P(x)}}{\_ : P(a)}
}
\]
Of course, when using this notation, we must ensure that we only use
$a : A$ as an element of $\aset{x : A \mid P(x)}$ when a (necessarily unique)
element of $P(a)$ is available.
Such an informal abuse of notation can be formally justified by \citet{Luo_2013}'s
\emph{coercive subtyping}.

An important special case of subtypes is \emph{extension types}: 
given a type $A$, $\phi : \Omega$ and $a : \phi
\to A$, we define
$
  \extty{A}{\phi}{a} \ \defeq\  \aset{x : A \mid (p : \phi) \to x = a(p)}
$
for the type of the $A$-elements that are strictly equal to the partial element $a$
when $\phi$ holds.

\medskip

Similar to the realizability model, we will again need three universes to model
\Fha{}, which exist provided that the ambient set
theory has enough universes \citep{gratzer2022strict}.
\begin{axiom}[\TTSTC-Universe]\label{el:ttstc:universe}
\TTSTC{} has three cumulative predicative universes $U_0 : U_1 : U_2$, 
each closed under $\Pi$-types, $\Sigma$-types, extensional equality types, and
inductive types.
Moreover, the universe of propositions is in $U_0$, i.e.\ $\Omega :
U_0$.
\end{axiom}

The next ingredient of \TTSTC{}, perhaps the most important one, is the following axiom.
\begin{axiom}[\textsc{\TTSTC{}-Obj}]\label{el:ttstc:obj}
\TTSTC{} has $\sop : \Omega$ and 
$\MM : \impfun{\sop} \sem{\Fha}_{U_0}$.
\end{axiom}
The category $\GFha$ embeds the `object space'
$\Pr(\CatJdgOf\Fha)$ and the `meta space' $\CatSet$ as full subcategories.
In the language \TTSTC,
the proposition $\sop$ serves the purpose of accessing the object space
and the meta space:
$\sop$ is interpreted as an object $\sem{\sop}$ in $\GFha$ 
such that the exponential functor $(\blank)^{\sem{\sop}}$ sends every object
$\tuple{B, Q, q} \in \GFha$ to $\tuple{B, \Hom(1, B), \identity}$.
In other words, the function type $\sop \to A$ in \TTSTC{} erases the
meta-space information in $A$.

For every type $A$, we will write $\Omod A \defeq \impfun{\sop} A$, and if the
function $\Oeta_A \defeq (\lambda a. \; \impabs{z : \sop} a) : A \to \Omod A$
is an isomorphism, we say that the type $A$ is $\Omodal$.
These are types containing no meta-space information and are essentially objects
of $\Pr(\CatJdgOf\Fha)$.
\Cref{el:ttstc:obj} asserts that there is a \Fha-model $\MM$ in the
object space;
this is interpreted as the syntactic model of \Fha{} in $\Pr(\CatJdgOf\Fha)$.

\begin{remark}
When working in \TTSTC, the reader should pay special attention to whether
the proposition $\sop$ is assumed in the current context, because we will have 
types $A$ and $B$ that may look different but are \emph{judgementally
equal} when $\sop$ is assumed.
In this way, $\sop$ is more like a modality for `entering the object space', instead of a traditional mathematical proposition about
which we care whether it is true or false. 
The syntax of \TTSTC{} is chosen to highlight when $\sop$ 
is assumed;
for example, in the type $\extty{\hole{?1}}{\sop}{\hole{?2}}$ the hole \hole{?2} has
$\sop$ assumed.
Moreover, when $\sop$ is assumed, the type of implicit functions $\impfun{\sop}
A$ and the type $A$ can be used interchangeably, since the proposition
$\sop$ has at most one element and it is assumed in the context.
\end{remark}


We have also a modality $\Cmod$ for erasing the object-space information in
types, which may be expressed as a quotient inductive type in \TTSTC{}:
\begin{hscode}\SaveRestoreHook
\column{B}{@{}>{\hspre}l<{\hspost}@{}}%
\column{3}{@{}>{\hspre}l<{\hspost}@{}}%
\column{11}{@{}>{\hspre}l<{\hspost}@{}}%
\column{E}{@{}>{\hspre}l<{\hspost}@{}}%
\>[B]{}\lhskeyword{data}\;\Cmod\Conid{A}\;\lhskeyword{where}{}\<[E]%
\\
\>[B]{}\hsindent{3}{}\<[3]%
\>[3]{}\Ceta_A{}\<[11]%
\>[11]{}\mathbin{:}\Conid{A}\to \Cmod\Conid{A}{}\<[E]%
\\
\>[B]{}\hsindent{3}{}\<[3]%
\>[3]{}\Varid{pt}{}\<[11]%
\>[11]{}\mathbin{:}\{\mskip1.5mu \sop\mskip1.5mu\}\to \Cmod\Conid{A}{}\<[E]%
\\
\>[B]{}\hsindent{3}{}\<[3]%
\>[3]{}\Varid{eq}{}\<[11]%
\>[11]{}\mathbin{:}\{\mskip1.5mu \sop\mskip1.5mu\}\to (\Varid{a}\mathbin{:}\Conid{A})\to \Ceta_A\;\Varid{a}\mathrel{=}\Varid{pt}{}\<[E]%
\ColumnHook
\end{hscode}\resethooks
This quotient inductive type $\Cmod A$ can be explicitly constructed using
quotient and coproduct types, in the same way as constructing pushouts in
$\Set$.
A type $A$ is called $\Cmodal$ if the function $\Ceta_A : A \to
\Cmod A$ is an isomorphism.
In this case, we write $\Ceps_A : \Cmod A \to A$ for the inverse
of $\Ceta_A : A \to \Cmod A$.
The following lemma says that $\Cmod$-modal types have no
object-level information. 

\begin{lemma}\label{lem:cmod:if:unit}
A type $A$ in \TTSTC{} is $\Cmod$-modal iff $\Omod A$ is isomorphic to the unit type $1$.
\end{lemma}
\begin{proof}
Assume $A$ is $\Cmod$-modal.
We define a function $f : 1 \to \Omod A$ by
$f\;\unitel = \impabs{z:\sop} \Ceps_A\;(\ensuremath{\Varid{pt}} \impapp{z})$.
The function $f$ and $\lambda a.\;\unitel : \Omod A \to 1$ form an isomorphism.
To see $f \cdot (\lambda a.\;\unitel) = \identity$, 
for every $a : \Omod A$, we have $f ((\lambda a.\;\unitel)\;a) = f\;\unitel
= \impabs{z:\sop} \Ceps_A\;(\ensuremath{\Varid{pt}} \impapp{z})$, 
which is equal to $\impabs{z:\sop} \Ceps_A\;(\Ceta_A \; a)$
by $\ensuremath{\Varid{eq}\;\Varid{a}\mathbin{:}\Ceta_A\;\Varid{a}\mathrel{=}\Varid{pt}}$, and $\Ceps_A\;(\Ceta_A \; a) = a$ because $\Ceps_A$
is the inverse of $\Ceta_A$,
so we have $f ((\lambda a.\;\unitel)\;a) = \impabs{z : \sop} a = a$.
In the other direction, $(\lambda a.\;\unitel) \cdot f = \identity$ is trivial.

Conversely, assume $\Omod A \iso 1$. 
Let $a$ be the unique element of $\Omod A$.
We can define $\Ceps_A : \Cmod A \to A$
by $\Ceps_A \; (\Ceta_A\; a') = a'$ and $\Ceps_A \; \ensuremath{\Varid{pt}} = a$.
The function $\Ceps_A$ is a mutual inverse of $\Ceta_A$ following the defining
property of the quotient inductive type $\Cmod A$.
\end{proof}

Given an object-space type $A$ and a meta-space type family $B$ indexed by $A$,
we can `glue' them together by $\Sigma(a : A).\;B\;a$.
Under $\sop$, we have $\Sigma(a : A).\;B\;a \cong \Sigma(a : A).1 \cong A$.
The following, and the final, piece of \TTSTC{} allows us to do better,
giving us a type $\gluety{a : A}{B \; a} \cong \Sigma(a : A).\;B\;a$ such that
under $\sop$, $\gluety{a : A}{B\;a}$ is \emph{judgementally equal} in \TTSTC,
not just isomorphic, to $A$.
\begin{axiom}[\textsc{\TTSTC-Glue}]\label{el:ttstc:glue:ty}
\TTSTC{} has \emph{strict glue types} in its universes $U_i$:
\[
  \inferrule{A : \Omod {U_i} \\ B : (\impfun{z : \sop} A \impapp{z}) \to \aset{ X : U_i \mid \Cmodal \; X} }
            { \gluety{a : A}{B \; a} : \extty{U_i}{\sop}{A}}
\]
and isomorphisms between $\gluety{a : A}{B \; a}$
and $\Sigma \; (a : \impfun{\sop} A).\; B\;a$
\[
  \glue : \extty{ \; \Sigma (a : \impfun{\sop}A).\; B\;a \ \ \iso\ \  \gluety{a : A}{B \; a} \; }{\sop}{ \ensuremath{\pi_1\Varid{\hyp{}iso}} }
\]
where $\ensuremath{\pi_1\Varid{\hyp{}iso}} : \impfun{\sop} \Sigma (a : \impfun{\sop}A).\; B\;a \cong A$
is the isomorphism that has the projection $\pi_1 : (\Sigma (a : A).\; B\;a) \to A$
as its forward direction (under $\sop$, $B\;a \cong 1$ so $\pi_1$ is an isomorphism).
\end{axiom}

For all $a : \impfun{\sop}A$ and $b : B\;a$, we will use the following notation
in place of $\glue.\ensuremath{\Varid{fwd}}\;(a,b)$ to signal that the value is strictly equal to $a$
when $\sop$ is assumed:
\[
\gluetm{b}{a} \ \defeq\ \glue.\ensuremath{\Varid{fwd}}\; \tuple{a, b} : \gluety{a : A}{B \; a}.
\]
Given an element $g : \gluety{a : A}{B \; a}$, when
$\sop$ holds $\gluety{a : A}{B \; a}$ is equal to $A$, so we can
directly use $g$ as an element of $A$.
To access the second component of a glued element conveniently, we define
\begin{lgather*}
\unglue : (g : \gluety{a : A}{B \; a}) \to B \; (\impabs{\_:\sop} g) \\
\unglue g = \pi_2 \; (\glue.\ensuremath{\Varid{bwd}} \; g) 
\end{lgather*}


\subsection{Constructing the Logical Relation Model}
Now we construct our logical-relation model to prove \Cref{thm:canonicity}.
Our goal is to define 
\begin{equation}\label{eq:lr:model}
\MM^* : \extty{\sem{\Fha}_{U_2}}{\sop}{\MM}
\end{equation}
such that $\MM^*.\ensuremath{\Varid{bool}}$ encodes the property that we need.
Due to space constraints, this section can only be a digest of the full proof,
focusing on the novel part pertaining to the computation judgements of \Fha{}.
\Cref{app:lr:model} contains a full proof that explains all the definitions 
slowly.

\begin{notation}
For every declaration \ensuremath{\Varid{dec}} in the signature of \Fha,
we will write $\ensuremath{\Varid{dec}}^*$ for $\MM^*.\ensuremath{\Varid{dec}}$ and just \ensuremath{\Varid{dec}} for $\MM.\ensuremath{\Varid{dec}}$.
For example, $\ensuremath{\Varid{ki}}^* : \extty{U_2}{\sop}{\ensuremath{\Varid{ki}}}$
means $\MM^*.\ensuremath{\Varid{ki}} : \extty{U_2}{\sop}{\MM.\ensuremath{\Varid{ki}}}$.
\end{notation}

\begin{remark}
In \Cref{eq:lr:model}, $\MM^*$ is an element of the type
$\sem{\Fha}_{U_2}$ that equals to $S$ under $\sop$, but $S$ has type
$\sem{\Fha}_{U_0}$ under $\sop$ according to \Cref{el:ttstc:obj},
so how does this type check?
This is because the universes $U_i$ are cumulative (\Cref{el:ttstc:universe}) and 
in every LF-signature the universe $\Jdg$ only occurs in strictly positive
positions, so an element $S : \sem{\Fha}_{U_0}$ can be lifted as an element of
$\sem{\Fha}_{U_2}$.
\end{remark}

\subsubsection{Kinds and Types}
Every \Fha-type $A$ is 
interpreted as a \emph{proof-irrelevant} predicate $A^*$ over closed terms of
$A$.
Every kind $k$ is interpreted as a \emph{proof-relevant} predicate $k^*$ over
closed elements of $k$. 
In particular, for the base kind \ensuremath{\Varid{ty}} and every closed \Fha-type $A$,
$\ensuremath{\Varid{ty}}^*(A)$ is the set of all proof-irrelevant predicates over $A$-terms. 
The idea here is essentially the same as \varcitet{Girard_1989}{'s} technique
of \emph{reducibility candidates} in his proof of normalisation for System F,
except that here (1) we employ \emph{proof-relevant} predicates to deal with
higher kinds, and (2) the object language is formulated as an equational theory
rather than a reduction system.

In the language \TTSTC{}, the idea above for kinds is precisely expressed as follows:
\begin{align*}
&\ensuremath{\Varid{ki}}^* : \extty{U_2}{\sop}{\ensuremath{\Varid{ki}}} && \ensuremath{\Varid{el}}^* : \extty{\ensuremath{\Varid{ki}}^* \to U_1}{\sop}{\ensuremath{\Varid{el}}}  \\
& \ensuremath{\Varid{ki}}^* = \gluety{\alpha : \ensuremath{\Varid{ki}}}{\extty{U_1}{\sop}{\ensuremath{\Varid{el}} \; \alpha}} && \ensuremath{\Varid{el}}^* \; g = \unglue g 
\end{align*}
Informally,
an element of $\ensuremath{\Varid{ki}}^*$ is a syntactic kind $\alpha$ together with a `candidate'
$A : \extty{U_1}{\sop}{\ensuremath{\Varid{el}} \; \alpha}$
of the proof-relevant logical predicate  over the kind $\alpha$, where an
element $a : A$ encodes both a syntactic element $(\impabs{\_ : \sop} a) :
\impfun{\sop} \ensuremath{\Varid{el}} \; \alpha$ of the kind $\alpha$
and evidence that $\impabs{\_ : \sop}a$ satisfies the logical predicate $A$.
The definition of $\ensuremath{\Varid{ki}}^*$ uses the glue type (\Cref{el:ttstc:glue:ty}) correctly because 
the generic model $\MM$ (\Cref{el:ttstc:obj}) has type $\impfun{\sop}
\sem{\Fha}_{U_0}$, so the type of $\ensuremath{\Varid{ki}}$, or more explicitly $\impabs{z : \sop} (\MM
\impapp{z}).\ensuremath{\Varid{ki}}$, is $\impfun{\sop} U_0$, i.e.\ $\Omod U_0$.
The type $\extty{U_1}{\sop}{\ensuremath{\Varid{el}} \; \alpha}$ is $\Cmod$-modal because 
when $\sop$ holds, all elements of the extension type $\extty{U_1}{\sop}{\ensuremath{\Varid{el}} \; \alpha}$ are equal to 
$\ensuremath{\Varid{el}} \; \alpha$. 
Hence the type $\extty{U_1}{\sop}{\ensuremath{\Varid{el}} \; \alpha}$ has exactly
one element when $\sop$ holds, and by \Cref{lem:cmod:if:unit}, it is $\Cmod$-modal.
Note also that we have given $\ensuremath{\Varid{el}}^*$ a type
$\extty{\ensuremath{\Varid{ki}}^* \to U_1}{\sop}{\ensuremath{\Varid{el}}}$ more precise than necessary -- it only needs to land in the
larger universe $U_2$, but the universes $U_i$ in \Cref{el:ttstc:universe} are cumulative.

Define the universe of meta-space propositions by $\COmega \defeq \aset{p :
\Omega \mid \Cmodal\;p}$, which corresponds to simply the propositions in the
ambient set theory.
The interpretation of types and terms are
\begin{align*}
&\ensuremath{\Varid{ty}}^* : \extty{\ensuremath{\Varid{ki}}^*}{\sop}{\ensuremath{\Varid{ty}}}  && \ensuremath{\Varid{tm}}^* : \extty{\ensuremath{\Varid{el}}^* \; \ensuremath{\Varid{ty}}^* \to U_0}{\sop}{\ensuremath{\Varid{tm}}}  \\
&\ensuremath{\Varid{ty}}^* = \gluetm{ \gluety{A : \ensuremath{\Varid{el}} \; \ensuremath{\Varid{ty}}}{{(\impfun{\sop} \ensuremath{\Varid{tm}} \; A) \to \COmega}} }{\ensuremath{\Varid{ty}}} & &\ensuremath{\Varid{tm}}^* \; \gluetm{P}{A} = \gluety{t : \ensuremath{\Varid{tm}} \; A}{P \; t}
\end{align*}

The interpretation for the base type \ensuremath{\Varid{bool}} is the predicate $P_\ensuremath{\Varid{can}}(b)$ that 
the closed term $b$ is either equal to \ensuremath{\Varid{tt}\mathbin{:}\Varid{bool}} or \ensuremath{\Varid{ff}\mathbin{:}\Varid{bool}} as follows.
The predicate $P_\ensuremath{\Varid{can}}$ is the only place in the
logical-relation model $\MM^*$ that is specific to canonicity, so we can in fact
replace $P_\ensuremath{\Varid{can}}(b)$ with any other predicates that hold for \ensuremath{\Varid{tt}} and \ensuremath{\Varid{ff}} and
get other logical-relation models.
\begin{align*}
&\ensuremath{\Varid{bool}}^* : \extty{\ensuremath{\Varid{el}}^*\;\ensuremath{\Varid{ty}}^*}{\sop}{\ensuremath{\Varid{bool}}}  &&  \ensuremath{\Varid{tt}}^* : \extty{\ensuremath{\Varid{tm}}^* \; \ensuremath{\Varid{bool}}^*}{\sop}{\ensuremath{\Varid{tt}}}  \\
&\ensuremath{\Varid{bool}}^* = \gluetm{P_\ensuremath{\Varid{can}}}{\ensuremath{\Varid{bool}}} && \ensuremath{\Varid{tt}}^* = \gluetm{\Ceta\; ( \inl\;{\refl} )}{\ensuremath{\Varid{tt}}} \\
& P_\ensuremath{\Varid{can}} : (\impfun{\sop}\ensuremath{\Varid{tm}}\;\ensuremath{\Varid{bool}}) \to \COmega &&\ensuremath{\Varid{ff}}^* : \extty{\ensuremath{\Varid{tm}}^* \; \ensuremath{\Varid{bool}}^*}{\sop}{\ensuremath{\Varid{ff}}} \\
&P_\ensuremath{\Varid{can}}\; b = \Cmod (\impfun{\sop} (b = \ensuremath{\Varid{tt}} \lor b = \ensuremath{\Varid{ff}})) &&\ensuremath{\Varid{ff}}^* = \gluetm{\Ceta\; ( \inr\;{\refl} )}{\ensuremath{\Varid{ff}}}
\end{align*}

The other type/kind formers are interpreted in the standard way, which is
explained in great detail in \Cref{app:lr:model}.
Here, let us only show the interpretation for the polymorphic function type:  
\[
\begin{array}{l}
\ensuremath{\allfor}^* : \extty{(k^* : \ensuremath{\Varid{ki}}^*) \to (\ensuremath{\Varid{el}}^* \; k^* \to \ensuremath{\Varid{el}}^* \; \ensuremath{\Varid{ty}}^*) \to \ensuremath{\Varid{el}}^* \; \ensuremath{\Varid{ty}}^*}{\sop}{\ensuremath{\allfor}} \\
\ensuremath{\allfor}^* \; k^* \; F = 
  \gluetm{ \lambda t.\; 
              \forall (\alpha^* : \ensuremath{\Varid{el}}^* \; k^*).\ \unglue (F \; \alpha^*) \; (\lambda \imparg{\_ : \sop}.\ (t\; \alpha^*))}{\ensuremath{\allfor} \; k^* \; F}
\end{array}
\]
The logical predicate here on closed terms $t : \impfun{\sop}\ensuremath{\Varid{tm}}\;(\ensuremath{\allfor}\;k\;F)$ of the polymorphic function type
is a universal quantification
$\forall (\alpha^* : \ensuremath{\Varid{el}}^*\;k^*)$
over all the `logical predicate candidates' of the kind $k^*$, and we demand 
the result $t \; \alpha^*$ of applying the polymorphic function $t$ to
$\alpha^*$  to satisfy the logical predicate $F\;\alpha^*$.  
It is explained in detail in \Cref{sec:lr:fun:ty} how this definition type
checks.

With the \Fom-fragment of $\MM^*$ defined, we automatically have the
interpretation for the judgements that are derived in \Fom{} (\Cref{fig:derived:concepts}), such as \ensuremath{\Conid{RawHFuncor}} and \ensuremath{\Conid{RawMonad}}. 
We will also denote their interpretation by superscripting an asterisk. 
For example, for the judgement \ensuremath{\Varid{efty}\mathrel{=}(\Varid{ty}\;{\Rightarrow_t}\;\Varid{ty})} from \Cref{fig:derived:concepts},
its interpretation $\ensuremath{\Varid{efty}}^*$ is then $\ensuremath{\Varid{ty}}^*\; \ensuremath{{\Rightarrow_k}}^*\; \ensuremath{\Varid{ty}}^*$.

\subsubsection{Computations}
What remains is to define the computation fragment of \Fha{}
(\Cref{sec:fha:co}) for our logical-relation model $\MM^*$.
Because in \Fha{} the only way to `observe' a computation is to evaluate it
with a raw monad using \ensuremath{\Varid{hdl}}
\Cref{eq:fha:eval}, a natural attempt to define the logical predicate
$P_\ensuremath{\Varid{co}}(c)$ for computations is to assert that the computation $c : \ensuremath{\Varid{co}\;\Conid{H}\;\Conid{A}}$
evaluated with every semantic handler $h^*$ (i.e.\ a syntactic handler $h$
with a proof of $h$ satisfying the proof-relevant logical predicate)
yields a term that satisfies the logical predicate
$h^*.0\;A^*$:
\begin{lgather*}
P_\ensuremath{\Varid{co}}^{\ensuremath{\Varid{wrong}}} : \impfun{H^*, A^*}  (\impfun{\sop} \ensuremath{\Varid{co}} \; H^* \; A^*) \to \COmega \\
P_\ensuremath{\Varid{co}}^{\ensuremath{\Varid{wrong}}} \; c = \forall (h^* : \ensuremath{\Conid{Handler}}^* \; H^*).\ \unglue (h^*.0 \; A^*) \; (\impabs{\_:\sop}\ensuremath{\Varid{hdl}} \; h^* \; c)
\end{lgather*}
However, this definition will not work later when showing that the term constructor 
\ensuremath{\Varid{let\hyp{}in}} \Cref{eq:fha:co:letin} satisfies its logical predicate.
That is, we need to prove
\[
\forall (h^* : \ensuremath{\Conid{Handler}}^* \; H^*).\ \unglue (h^*\ensuremath{\Varid{.0}} \; B^*) \; 
  (\impabs{\_:\sop} \hole{\ensuremath{\Varid{hdl}} \; h^* \; (\ensuremath{\Varid{let\hyp{}in}} \; c \; f)}),
\]
but as mentioned in \Cref{rem:no:evallet}, \Fha{} does not have
the law saying that \ensuremath{\Varid{hdl}} commutes with \ensuremath{\Varid{let\hyp{}in}}, so we have no way to further
simplify the shaded part above to use the assumptions that \ensuremath{\Varid{c}} and \ensuremath{\Varid{f}}
satisfy their logical predicates.

The way to fix this problem is to use the idea of \emph{$\top\top$-lifting}
\citep{Lindley_Stark_2005}:
we strengthen $P_\ensuremath{\Varid{co}}^{\ensuremath{\Varid{wrong}}}$ above to quantify over all `good
continuations' $k$ of the computation $c$, and we demand \ensuremath{\Varid{hdl}\;\Varid{h}\;(\Varid{let\hyp{}in}\;\Varid{c}\;\Varid{k})} to satisfy its logical predicate.
Here a continuation $k$ is `good' if $k$ followed by \ensuremath{\Varid{hdl}} sends input
satisfying its logical predicate to output satisfying its logical predicate,
which can be succinctly expressed by a function $k^* : \ensuremath{\Varid{tm}}^*
\; A^* \to \ensuremath{\Varid{tm}}^* \; (h^* \; R^*)$.
Precisely, the type of `good' continuations accepting $A^*$-values is the
following record:
\begin{lgather*}
\ensuremath{\lhskeyword{record}} \; \ensuremath{\Conid{Con}} \; (H^* : \ensuremath{\Conid{RawHFunctor}}^*) \; (A^* : \ensuremath{\Varid{el}}^* \; \ensuremath{\Varid{ty}}^*) : U_1 \; \ensuremath{\lhskeyword{where}} \\
\hspace{1em}
  h^* : \ensuremath{\Conid{Handler}}^* \; H^* \\
\hspace{1em}
  R^* : \ensuremath{\Varid{el}}^* \; \ensuremath{\Varid{ty}}^* \\
\hspace{1em}
  k : \impfun{\sop} A^* \to \ensuremath{\Varid{co}} \; H^* \; R^* \\
\hspace{1em}
  k^* : \extty{\ensuremath{\Varid{tm}}^* \; A^* \to \ensuremath{\Varid{tm}}^* \; (h^* \; R^*)}{\sop}%
        {\lambda a.\ \ensuremath{\Varid{hdl}} \; h^* \; (k \; a) }
\end{lgather*}
and the correct definition of $P_\ensuremath{\Varid{co}}$ and the interpretation of computation
judgements $\ensuremath{\Varid{co}}^*$ is
\begin{lgather*}
P_\ensuremath{\Varid{co}} : \impfun{H^*, A^*} (\impfun{\sop} \ensuremath{\Varid{co}} \; H^* \; A^*) \to \COmega \\
P_\ensuremath{\Varid{co}} \; c = \forall (K : \ensuremath{\Conid{Con}} \; H^* \; A^*).\
   \unglue (K.h^* \; K.R^*) \; (\impabs{\_ : \sop} 
  \ensuremath{\Varid{hdl}} \; K.h^* \; (\ensuremath{\Varid{let\hyp{}in}} \; c \; K.\ensuremath{\Varid{k}})) \\[4pt]
\ensuremath{\Varid{co}}^* : \extty{\ensuremath{\Conid{HFunctor}}^* \to \ensuremath{\Varid{el}}^* \; \ensuremath{\Varid{ty}}^* \to U_0}{\sop}{\ensuremath{\Varid{co}}} \\
\ensuremath{\Varid{co}}^* \; H^* \; A^* = \gluety{c : \ensuremath{\Varid{co}} \; H^* \; A^* }{P_\ensuremath{\Varid{co}} \; c} 
\end{lgather*}

Based on this definition of $\ensuremath{\Varid{co}}^*$, the interpretation of all constructs of
\Fha{} pertaining to computations can be defined and is shown in detail in
\Cref{sec:fha:co:star}.
Here we only show the case for the term former \ensuremath{\Varid{let\hyp{}in}}, which was previously
problematic:
\begin{lgather*}
\ensuremath{\Varid{let\hyp{}in}}^* : \extty{\impfun{H^*, A^*, B^*} \ensuremath{\Varid{co}}^* \; H^* \; A^*
\to (\ensuremath{\Varid{co}}^* \; H^* \; A^* \to \ensuremath{\Varid{co}}^* \; H^* \; B^*) 
\\ \hspace{5em} \to  \ensuremath{\Varid{co}}^* \; H^* \; B}{\sop}{\ensuremath{\Varid{let\hyp{}in}}} \\
\ensuremath{\Varid{let\hyp{}in}}^* \; c \; f = \gluetm{\lambda (K : \ensuremath{\Conid{Con}} \; H^* \; A^*).\ \unglue c \; K'}{\ensuremath{\Varid{let\hyp{}in}} \; c \; f}
\end{lgather*}
where each field of $K' : \ensuremath{\Conid{Con}} \; H^* \; B^*$ is defined as follows:
\begin{align*}
&K'.h^* = K.h^*  &&  K'.k = \lambda \imparg{\_:\sop}\; a.\ \ensuremath{\Varid{let\hyp{}in}} \; (f \; a) \; K.k \\
&K'.R^* = K.R^*  && K'.k^*  = \lambda a.\ \gluetm{\unglue (f \; a) \; K}{\ensuremath{\Varid{hdl}} \; K'.h^* \; (\ensuremath{\Varid{let\hyp{}in}} \; (f \; a) \; K.\ensuremath{\Varid{k}})}
\end{align*}
The definition of $K'.k^*$ type checks because $f \; a : \ensuremath{\Varid{co}}^*\;H^*\;B^*$, so
$\unglue (f \; a) : P_\ensuremath{\Varid{co}} \; (f \; a)$, so by the definition of $P_\ensuremath{\Varid{co}}$, the type of $\unglue (f \; a) \; K$ is
\begin{lgather*}
\unglue (K.h^* \;\; K.R^*) \; (\lambda \imparg{\_ : \sop}.\ \ensuremath{\Varid{hdl}} \; K.h^* \; (\ensuremath{\Varid{let\hyp{}in}} \; (f \; a) \; K.\ensuremath{\Varid{k}}))
\end{lgather*}
which is indeed the type of proofs that the syntactic component of $K'.k^*$
satisfies the logical predicate of the type $k.h^* \; K.R^*$.

To summarise, we have established the `fundamental lemma'
for the logical relation of \Fha{}.
\begin{lemma}[Fundamental]\label{lem:closed:fundamental}
In the language {\normalfont\TTSTC{}}, 
given any $P : (\impfun{\sop}\MM.\ensuremath{\Varid{tm}}\;\MM.\ensuremath{\Varid{bool}}) \to \COmega$ with
$t : P (\MM.\ensuremath{\Varid{tt}})$ and $f : P (\MM.\ensuremath{\Varid{ff}})$,
there is an
$
\MM^* : \extty{\sem{\Fha}_{U_2}}{\sop}{\MM}
$
such that \[\MM^*.\ensuremath{\Varid{tm}} \; \MM^*.\ensuremath{\Varid{bool}} = \gluety{b : \MM.\ensuremath{\Varid{tm}} \; \MM.\ensuremath{\Varid{bool}}}{P\;b}.\]
\end{lemma}

\subsection{External Closed Term Canonicity}

Finally, we can now prove canonicity (\Cref{thm:canonicity}) using our
logical-relation model $\MM^*$.
\begin{proof}[Proof of \Cref{thm:canonicity}]
The denotation of the \TTSTC-type $\ensuremath{\Varid{tm}}^*\;\ensuremath{\Varid{bool}}^*$ in $\GFha$ is the 
following object (strictly speaking the denotation could be an object
only \emph{isomorphic} to this object, but pretending they are strictly
equal will not cause problems in this proof):
\[
B^* \defeq
\tuple{\yo(\ensuremath{\Varid{tm}\;\Varid{bool}}) \in \Pr(\CatJdgOf\Fha), \;  
      \aset{t : 1 \to \yo(\ensuremath{\Varid{tm}\;\Varid{bool}})) \mid  (t = \yo \; \ensuremath{\Varid{tt}}) \lor (t = \yo \; \ensuremath{\Varid{ff}})},\;
      j }
\]
where $j$ is the inclusion function into $\aset{t : 1 \to \yo(\ensuremath{\Varid{tm}\;\Varid{bool}}))}$.
For every closed term ${b : \ensuremath{\Varid{tm}\;\Varid{bool}}}$ in $\Fha$, its interpretation 
in the model $\MM^*$ is a morphism $1 \to B^*$ in $\GFha$ as follows:
\[
  \begin{tikzcd}[ampersand replacement=\&]
\aset{\unitel} \& {\aset{t \mid  (t = \yo\; \ensuremath{\Varid{tt}}) \lor (t = \yo\; \ensuremath{\Varid{ff}})}} \\
{\aset{\unitel}} \& {\aset{t : 1 \to \yo(\ensuremath{\Varid{tm}\;\Varid{bool}})}}
\arrow[from=1-1, to=2-1]
\arrow["{\lambda \unitel.\; \yo b }"', from=2-1, to=2-2]
\arrow["", from=1-1, to=1-2]
\arrow["{j}", from=1-2, to=2-2]
\end{tikzcd}
\]
The commutativity of this diagram entails $\yo b = \yo \ensuremath{\Varid{tt}}$ or $\yo b = \yo
\ensuremath{\Varid{ff}}$, so $b = \ensuremath{\Varid{tt}}$ or $b = \ensuremath{\Varid{ff}}$ since Yoneda embedding is fully faithful.
Moreover, $b = \ensuremath{\Varid{tt}}$ and $b = \ensuremath{\Varid{ff}}$ cannot be true at the same time because
\ensuremath{\Varid{tt}} and \ensuremath{\Varid{ff}} have different interpretations in the realizability model in
\Cref{sec:real:mod:fha}.
\end{proof}

\begin{corollary}\label{cor:adequacy}
An immediate consequence of canonicity of \Fha{} is that the realizability model in
\Cref{sec:fha:real:model} is \emph{adequate} in the sense that if two closed Boolean
terms $b_1$ and $b_2$ have the same denotation in the realizability model, they
must be judgementally equal.
This is because canonicity says that both $b_1$ and $b_2$ are either equal to $\ensuremath{\Varid{tt}}$ or
$\ensuremath{\Varid{ff}}$, which have different interpretations in the realizability model, so 
$b_1 = b_2$ must be true if their realizability interpretation is the same.
\end{corollary}

\begin{remark}\label{rem:parametricity}
Apart from canonicity, \Cref{lem:closed:fundamental} can be also used to show
other \emph{parametricity} results about terms of \Fha; for example, for a
closed term $t : \ensuremath{\Varid{tm}} \; (\ensuremath{\allfor}\; \ensuremath{\Varid{ty}} \; (\lambda \alpha.\; \alpha
\mathbin{\ensuremath{{\Rightarrow_t}}} \alpha))$, $t$ applied to every closed type $A$ and closed
term $a : \ensuremath{\Varid{tm}\;\Conid{A}}$ is equal to $a$.
Even more pleasantly, we can obtain a binary (or $n$-ary) logical-relation
model of \Fha{} from the seemingly unary logical-predicate model
\Cref{lem:closed:fundamental} by interpreting $\TTSTC$ in a different glued topos,
without modifying the definition of $\MM^*$ at all.
These results are elaborated in \Cref{sec:parametricity}.
\end{remark}

\begin{remark}
We based \Fha{} on fine-grain call-by-value (FGCBV) rather than call-by-push-value (CBPV) since the
theory of higher-order algebraic effects \citep{YangWu23} is based on monads rather
than adjunctions, but CBPV is also possible and 
we sketch the judgements for a CBPV variant of \Fha{} without stack
judgements here.
Again, starting with \Fom{}, instead of \ensuremath{\Varid{co}\mathbin{:}\Conid{RawHFunctor}\to\Varid{el}\;\Varid{ty}\to\Jdg},
we add to \Fom{} a new kind for 
\emph{computation types} and a judgement for \emph{computation terms}:
\[
\spread*{
\ensuremath{\Varid{cty}\mathbin{:}\Conid{RawHFunctor}\to \Varid{ki}}
\also
\ensuremath{\Varid{ctm}\mathbin{:}\{\mskip1.5mu \Conid{H}\mskip1.5mu\}\to \Varid{el}\;(\Varid{cty}\;\Conid{H})\to \Jdg}
}
\]
We then add two new type formers for value-returning computations 
and thunk values:
\[
\spread*{
\ensuremath{\Conid{F}\mathbin{:}(\Conid{H}\mathbin{:}\Conid{RawHFunctor})\to \Varid{el}\;\Varid{ty}\to \Varid{el}\;(\Varid{cty}\;\Conid{H})}
\also
\ensuremath{\Conid{U}\mathbin{:}\{\mskip1.5mu \Conid{H}\mskip1.5mu\}\to \Varid{el}\;(\Varid{cty}\;\Conid{H})\to \Varid{el}\;\Varid{ty}}
}
\]
We have value returning and sequential composition as usual:
\begin{hscode}\SaveRestoreHook
\column{B}{@{}>{\hspre}l<{\hspost}@{}}%
\column{9}{@{}>{\hspre}l<{\hspost}@{}}%
\column{E}{@{}>{\hspre}l<{\hspost}@{}}%
\>[B]{}\Varid{val}{}\<[9]%
\>[9]{}\mathbin{:}\{\mskip1.5mu \Conid{H},\Conid{A}\mskip1.5mu\}\to\Varid{tm}\;\Conid{A}\to\Varid{ctm}\;(\Conid{F}\;\Conid{H}\;\Conid{A}){}\<[E]%
\\
\>[B]{}\Varid{let\hyp{}in}{}\<[9]%
\>[9]{}\mathbin{:}\abrace{\Conid{H},\Conid{A},\Conid{X}}\to\Varid{ctm}\;(\Conid{F}\;\Conid{H}\;\Conid{A})\to(\Varid{tm}\;\Conid{A}\to\Varid{ctm}\;\Conid{X})\to\Varid{ctm}\;\Conid{X}{}\<[E]%
\ColumnHook
\end{hscode}\resethooks
Note that the second argument of \ensuremath{\Varid{let\hyp{}in}} can be an arbitrary computation
type \ensuremath{\Conid{X}\mathbin{:}\Varid{el}\;(\Varid{cty}\;\Conid{H})}
rather than just value-returning computations \ensuremath{\Conid{F}\;\Conid{H}\;\Conid{A}\mathbin{:}\Varid{el}\;(\Varid{cty}\;\Conid{H})}.
Terms of a thunk type are in bijection with the terms of the computation type:
\[
\ensuremath{\Varid{U\hyp{}iso}\mathbin{:}\{\mskip1.5mu \Conid{H}\mskip1.5mu\}\;\{\mskip1.5mu \Conid{X}\mathbin{:}\Varid{el}\;(\Varid{cty}\;\Conid{H})\mskip1.5mu\}\to \Varid{tm}\;(\Conid{U}\;\Conid{X})\;{\cong}\;\Varid{ctm}\;\Conid{X}}
\]
The judgement \ensuremath{\Varid{co}\;\Conid{H}\;\Conid{A}} in \Fha{} then corresponds to \ensuremath{\Varid{ctm}\;(\Conid{F}\;\Conid{H}\;\Conid{A})}
in CBPV.
What we have in CBPV but not in \Fha{} are \emph{function computations}
from a value type \ensuremath{\Conid{A}} to a computation type \ensuremath{\Conid{X}}:
\begin{hscode}\SaveRestoreHook
\column{B}{@{}>{\hspre}l<{\hspost}@{}}%
\column{11}{@{}>{\hspre}l<{\hspost}@{}}%
\column{E}{@{}>{\hspre}l<{\hspost}@{}}%
\>[B]{}\text{\textunderscore}{\Rightarrow_c}\text{\textunderscore}{}\<[11]%
\>[11]{}\mathbin{:}\{\mskip1.5mu \Conid{H}\mskip1.5mu\}\to (\Conid{A}\mathbin{:}\Varid{el}\;\Varid{ty})\to (\Conid{X}\mathbin{:}\Varid{el}\;(\Varid{cty}\;\Conid{H}))\to \Varid{el}\;(\Varid{cty}\;\Conid{H}){}\<[E]%
\\
\>[B]{}\Varid{{\Rightarrow_c}\hyp{}iso}{}\<[11]%
\>[11]{}\mathbin{:}\{\mskip1.5mu \Conid{H},\Conid{A},\Conid{X}\mskip1.5mu\}\to \Varid{ctm}\;(\Conid{A}\;{\Rightarrow_c}\;\Conid{X})\cong(\Varid{tm}\;\Conid{A}\to \Varid{ctm}\;\Conid{X}){}\<[E]%
\ColumnHook
\end{hscode}\resethooks
Finally, we have $H$-operations and evaluation by raw monads with
$H$-operations:
\begin{hscode}\SaveRestoreHook
\column{B}{@{}>{\hspre}l<{\hspost}@{}}%
\column{E}{@{}>{\hspre}l<{\hspost}@{}}%
\>[B]{}\Varid{th}\mathbin{:}\Conid{RawHFunctor}\to \Varid{el}\;\Varid{ty}\to \Varid{el}\;\Varid{ty}{}\<[E]%
\\
\>[B]{}\Varid{th}\;\Conid{H}\;\Conid{A}\mathrel{=}\Conid{U}\;(\Conid{F}\;\Conid{H}\;\Conid{A}){}\<[E]%
\\
\>[B]{}\Varid{op}\mathbin{:}\{\mskip1.5mu \Conid{H},\Conid{A},\Conid{X}\mskip1.5mu\}\to\Varid{tm}\;(\Conid{H}\;(\Varid{th}\;\Conid{H})\;\Conid{A})\to(\Varid{tm}\;\Conid{A}\to\Varid{ctm}\;\Conid{X})\to\Varid{ctm}\;\Conid{X}{}\<[E]%
\\
\>[B]{}\Varid{hdl}\mathbin{:}\{\mskip1.5mu \Conid{H},\Conid{A}\mskip1.5mu\}\to(\Varid{h}\mathbin{:}\Conid{Handler}\;\Conid{H})\to\Varid{ctm}\;(\Conid{F}\;\Conid{H}\;\Conid{A})\to\Varid{tm}\;(\Varid{h}\;\Varid{.0}\;\Conid{A}){}\<[E]%
\ColumnHook
\end{hscode}\resethooks
We see no obvious difficulties in adapting our results
for \Fha{} to this CBPV variant by interpreting $F$ and
$U$ using the Eilenberg-Moore adjunction of the monads that we used to model
\Fha{}. 
\end{remark}

\section{Related Work}
The most related work is the line of research on (higher-order) algebraic
effects and handlers, and we have discussed the position of this paper
within this line of research in \Cref{sec:intro}.
In this section, we discuss some more aspects of related work that were not
discussed in \Cref{sec:intro}.

\smallskip

$\bullet{}$
In the context of handlers of algebraic effects, the paper by \citet{Wu2014}
seems to be the first to consider higher-order operations.
Although the examples of (higher-order) operations considered in this paper
are all what are later called scoped operations, the framework in this paper 
is actually designed for general higher-order effects that can be given as
higher-order functors, similar to \citet{YangWu23}, \citet{vS24}, and the
present paper here, except that \citet{Wu2014} demand the signature
(higher-order) functors to come with a \emph{weave} operation, which
is used for modularly combining handlers of different effects. 
This design of \ensuremath{\Varid{weave}} seems inherently tied to effects similar to
mutable state and has not been further developed since then.

\smallskip

$\bullet{}$
The paper by \citet{Wu2014} is a practically-minded paper presented in Haskell,
and the underlying mathematics for higher-order effects was not clear at the
time.
Therefore in the following years, several authors studied the categorical
foundation (and practical applications) of several \emph{special cases} of
\varcitet{Wu2014}{'s} framework, including \emph{scoped effects} by
\citet{PirogSWJ18} and \citet{YPWvS2021}, \emph{latent effects} by
\citet{vSPW21}, \emph{hefty algebras} by \citet{Bach_Poulsen_vd_Rest_2023}.
All these families can be implemented in our calculus \Fha.

\smallskip

$\bullet{}$
After this trend of diversification, the families of higher-order algebraic
effects are re-unified by \citet{YangWu23} and \citet{vS24}.
\citet{YangWu23} presented (1) a general categorical framework for defining
higher-order algebraic effects (with equations) as algebraic theories of
operations on monoids, and (2) constructions for combining handlers in a
modular way.
In contrast, the present paper studies a \emph{programming language} \Fha{} for
(equation-less) higher-order algebraic effects.
\varcitet{YangWu23}{'s} constructions of modular handlers can be readily
used in \Fha{} but they are not baked in the language.
Another difference is that \citet{YangWu23} worked at the level of abstraction
of \emph{monoids in monoidal categories}, encompassing not just monads but also
applicative functors, graded monads, etc. 

\varcitet{vS24}{'s} work is also a general framework for (equation-less)
higher-order algebraic effects (following \hyperref[approach:kont]{approach
ii}), presented as a Haskell library.
Their paper provides a plethora of interesting concrete examples of
higher-order effects and handlers, which we did not explore in this paper but
in principle can be programmed in \Fha{} too.

\smallskip

\newcommand{\LamSC}{\lambda_{\mathit{sc}}}

$\bullet{}$
The papers discussed above are all about category theory or programming libraries
for higher-order effects.
To our knowledge, the only account of standalone programming languages for
higher-order effect handlers so far is \varcitet{BvTT2024}{'s} work on
$\lambda_{\mathit{sc}}$, \emph{calculus for scoped effects and handlers}.
The key differences between their $\lambda_{\mathit{sc}}$ and our $\Fha$
are that 
(1) $\LamSC$ is designed for algebraic and scoped operations, while $\Fha$ supports arbitrary higher-order operations that can be
given as higher-order functor;
(2) $\LamSC$ uses operations with a continuation argument together with mere type constructors
as handlers (\hyperref[approach:kont]{approach (ii)} in \Cref{sec:intro}), while we use raw monads as handlers;
(3) $\LamSC$ has a baked-in type-and-effect system, while \Fha{} 
leaves it as a user-level construct (\Cref{el:mod:hdl:in:fha}).
Also, the work on $\LamSC$ focused more on the user-facing aspects of the language,
such as type inference, whereas 
we have focused on the meta-theoretic properties of \Fha{} -- an equational
theory validated by the `compiler' (the realizability model), canonicity, 
and parametricity.

\section{Future Prospects}

In this paper, we defined System \Fha{}, an extension
of System \Fom{} with (equation-less) higher-order algebraic effects. 
We gave a denotational model of it using realizability and proved the canonicity
of closed terms using synthetic Tait computability.
A further extension with general recursion was introduced and was modelled
using synthetic domain theory. 
Future work abounds:
\begin{enumerate}[wide=0.5\parindent,listparindent=0pt,topsep=0pt,itemsep=0pt]
\item
We should be able to prove normalisation of open \Fha-terms following the
lines of \citet{sterling:2021:thesis}.
A subtlety is that, in order to interpret \Fha-types, we will need in
\TTSTC{} a universe $U$ closed under both (1) impredicative $\Pi$-types and (2)
the type of normal forms.  
The universe $\COmega$ used in the canonicity proof will not do anymore because
it does not contain the type of normal forms.
To have such a universe $U$, we shall replace the category of sets used in the
gluing proof with the category of assemblies or realizability toposes suitably.

\item
The language \Fha{} is a core calculus. 
Although important features such as effect systems and
modular models may be implemented as libraries,
for practical use they should be supported in a more seamless way, such as
by elaboration or by directly baking it into the language.

\item
Only monadic computations are considered in \Fha{} for simplicity.
Generalising \Fha{} from monads to arbitrary user-defined
monoidal structures should be very useful. 
There does not seem to be a theoretical obstacle, but designing a
user-friendly syntax may be challenging.

\item
Efficiency of implementations is also an interesting aspect.
Note that a continuation-passing style translation for \Fha-computations can be
readily extracted from the realizability model of \Fha, 
but it should be possible to further optimise out all the overhead of effect
handlers for statically known computations and handlers by
\emph{meta-programming}.

\item 
The biggest limitation of \Fha{} is probably that equational axioms on effects
are not formally supported, since we did not include dependent types, in particular
equality/identity types.
Adding intensional identity types \ensuremath{\Conid{Id}} to \Fha{} is straightforward.
With \ensuremath{\Conid{Id}} we can then define in \Fha{} \emph{law-abiding} functors, monads,
higher-order functors, equational systems, etc.  
These allow us to ensure that a user-defined monad to be used with \ensuremath{\Varid{hdl}} must
satisfy the equations associated to the operations.
However, what is challenging is adding equalities from the algebraic theory of
effectful operations to the computation judgements without breaking the
canonicity of the type theory.
This difficulty is the same as that of adding \emph{quotient types} to type
theories without breaking canonicity, which is possible in
\emph{observational type theory} \citep{Pujet2022} and \emph{cubical type theory} \citep{Coquand2018}.
Apart from the difficulty with quotients, having general recursion,
impredicative polymorphism, and dependent types all together is not
straightforward either.

\item
In the design of \Fha{}, monadic laws of computations are retained by
sacrificing the commutativity of \ensuremath{\lhskeyword{let}}-binding and handlers.
We do not think this is the only reasonable choice.
For example, if our programming language has dependent types and we ask for
users to supply monads satisfying the monadic laws \emph{propositionally}, can
we make the computational judgements satisfy all the laws \emph{judgementally}? 
We hope that this paper at least encourages further discussion on this design space.
\end{enumerate}

\bibliography{../references}

\appendix

\section{Complete Signatures of the Languages}\label{app:sig}

\allowdisplaybreaks

This appendix collects the full signatures of the languages in the paper.

\subsection{Signature of System \Fha{}}\label{app:sig:fha}

The following is the signature of System \Fha{} from \Cref{sec:fha} for 
easy reference.

\sethscode{plainhscode}
\renewcommand{\mathindent}{10pt}

\begin{itemize}
\item
Kinds
\begin{hscode}\SaveRestoreHook
\column{B}{@{}>{\hspre}l<{\hspost}@{}}%
\column{9}{@{}>{\hspre}l<{\hspost}@{}}%
\column{E}{@{}>{\hspre}l<{\hspost}@{}}%
\>[B]{}\Varid{ki}{}\<[9]%
\>[9]{}\mathbin{:}\Jdg{}\<[E]%
\\
\>[B]{}\Varid{el}{}\<[9]%
\>[9]{}\mathbin{:}\Varid{ki}\to\Jdg{}\<[E]%
\\
\>[B]{}\Varid{ty}{}\<[9]%
\>[9]{}\mathbin{:}\Varid{ki}{}\<[E]%
\\
\>[B]{}\text{\textunderscore}{\Rightarrow_k}\text{\textunderscore}{}\<[9]%
\>[9]{}\mathbin{:}\Varid{ki}\to\Varid{ki}\to\Varid{ki}{}\<[E]%
\ColumnHook
\end{hscode}\resethooks

\item
Elements of the kind of types
\begin{hscode}\SaveRestoreHook
\column{B}{@{}>{\hspre}l<{\hspost}@{}}%
\column{10}{@{}>{\hspre}l<{\hspost}@{}}%
\column{E}{@{}>{\hspre}l<{\hspost}@{}}%
\>[B]{}\Varid{unit}{}\<[10]%
\>[10]{}\mathbin{:}\Varid{el}\;\Varid{ty}{}\<[E]%
\\
\>[B]{}\Varid{bool}{}\<[10]%
\>[10]{}\mathbin{:}\Varid{el}\;\Varid{ty}{}\<[E]%
\\
\>[B]{}\text{\textunderscore}{\Rightarrow_t}\text{\textunderscore}{}\<[10]%
\>[10]{}\mathbin{:}\Varid{el}\;\Varid{ty}\to\Varid{el}\;\Varid{ty}\to\Varid{el}\;\Varid{ty}{}\<[E]%
\\
\>[B]{}\allfor{}\<[10]%
\>[10]{}\mathbin{:}(\Varid{k}\mathbin{:}\Varid{ki})\to(\Varid{el}\;\Varid{k}\to\Varid{el}\;\Varid{ty})\to\Varid{el}\;\Varid{ty}{}\<[E]%
\ColumnHook
\end{hscode}\resethooks

\item
Elements of function kinds
\begin{hscode}\SaveRestoreHook
\column{B}{@{}>{\hspre}l<{\hspost}@{}}%
\column{3}{@{}>{\hspre}l<{\hspost}@{}}%
\column{9}{@{}>{\hspre}l<{\hspost}@{}}%
\column{10}{@{}>{\hspre}l<{\hspost}@{}}%
\column{25}{@{}>{\hspre}l<{\hspost}@{}}%
\column{E}{@{}>{\hspre}l<{\hspost}@{}}%
\>[B]{}\lhskeyword{record}\;\Conid{A}\;{\cong}\;\Conid{B}\mathbin{:}\Jdg{}\<[25]%
\>[25]{}\;\lhskeyword{where}{}\<[E]%
\\
\>[B]{}\hsindent{3}{}\<[3]%
\>[3]{}\Varid{fwd}{}\<[9]%
\>[9]{}\mathbin{:}\Conid{A}\to \Conid{B}{}\<[E]%
\\
\>[B]{}\hsindent{3}{}\<[3]%
\>[3]{}\Varid{bwd}{}\<[9]%
\>[9]{}\mathbin{:}\Conid{B}\to \Conid{A}{}\<[E]%
\\
\>[B]{}\hsindent{3}{}\<[3]%
\>[3]{}\anonymous {}\<[9]%
\>[9]{}\mathbin{:}(\Varid{a}\mathbin{:}\Conid{A})\to \Varid{bwd}\;(\Varid{fwd}\;\Varid{a})\mathrel{=}\Varid{a}{}\<[E]%
\\
\>[B]{}\hsindent{3}{}\<[3]%
\>[3]{}\anonymous {}\<[9]%
\>[9]{}\mathbin{:}(\Varid{b}\mathbin{:}\Conid{B})\to \Varid{fwd}\;(\Varid{bwd}\;\Varid{b})\mathrel{=}\Varid{b}{}\<[E]%
\\[\blanklineskip]%
\>[B]{}\Varid{{\Rightarrow_k}\hyp{}iso}{}\<[10]%
\>[10]{}\mathbin{:}\{\mskip1.5mu \Conid{A},\Conid{B}\mathbin{:}\Varid{ki}\mskip1.5mu\}\to\Varid{el}\;(\Conid{A}\;{\Rightarrow_k}\;\Conid{B})\cong(\Varid{el}\;\Conid{A}\to\Varid{el}\;\Conid{B}){}\<[E]%
\ColumnHook
\end{hscode}\resethooks

\item
Terms
\begin{hscode}\SaveRestoreHook
\column{B}{@{}>{\hspre}l<{\hspost}@{}}%
\column{10}{@{}>{\hspre}l<{\hspost}@{}}%
\column{E}{@{}>{\hspre}l<{\hspost}@{}}%
\>[B]{}\Varid{tm}{}\<[10]%
\>[10]{}\mathbin{:}\Varid{el}\;\Varid{ty}\to \Jdg{}\<[E]%
\\
\>[B]{}\Varid{unit\hyp{}iso}{}\<[10]%
\>[10]{}\mathbin{:}\Varid{tm}\;\Varid{unit}\cong\unitty{}\<[E]%
\\
\>[B]{}\Varid{{\Rightarrow_t}\hyp{}iso}{}\<[10]%
\>[10]{}\mathbin{:}\{\mskip1.5mu \Conid{A},\Conid{B}\mathbin{:}\Varid{el}\;\Varid{ty}\mskip1.5mu\}\to\Varid{tm}\;(\Conid{A}\;{\Rightarrow_t}\;\Conid{B})\cong(\Varid{tm}\;\Conid{A}\to\Varid{tm}\;\Conid{B}){}\<[E]%
\\
\>[B]{}\Varid{\allfor\hyp{}iso}{}\<[10]%
\>[10]{}\mathbin{:}\{\mskip1.5mu \Varid{k}\mathbin{:}\anonymous \mskip1.5mu\}\;\{\mskip1.5mu \Conid{A}\mathbin{:}\anonymous \mskip1.5mu\}\to\Varid{tm}\;(\allfor\;\Varid{k}\;\Conid{A})\cong((\Varid{α}\mathbin{:}\Varid{el}\;\Varid{k})\to\Varid{tm}\;(\Conid{A}\;\Varid{α})){}\<[E]%
\\
\>[B]{}\Varid{tt}{}\<[10]%
\>[10]{}\mathbin{:}\Varid{tm}\;\Varid{bool}{}\<[E]%
\\
\>[B]{}\Varid{ff}{}\<[10]%
\>[10]{}\mathbin{:}\Varid{tm}\;\Varid{bool}{}\<[E]%
\ColumnHook
\end{hscode}\resethooks

\item
Functors
\begin{hscode}\SaveRestoreHook
\column{B}{@{}>{\hspre}l<{\hspost}@{}}%
\column{3}{@{}>{\hspre}l<{\hspost}@{}}%
\column{E}{@{}>{\hspre}l<{\hspost}@{}}%
\>[B]{}\Varid{efty}\mathbin{:}\Varid{ki}{}\<[E]%
\\
\>[B]{}\Varid{efty}\mathrel{=}(\Varid{ty}\;{\Rightarrow_k}\;\Varid{ty}){}\<[E]%
\\[\blanklineskip]%
\>[B]{}\Varid{fmap\hyp{}ty}\mathbin{:}(\Conid{F}\mathbin{:}\Varid{el}\;\Varid{efty})\to\Varid{el}\;\Varid{ty}{}\<[E]%
\\
\>[B]{}\Varid{fmap\hyp{}ty}\;\Conid{F}\mathrel{=}\allfor\;\Varid{ty}\;(\lambda \Varid{α}.\;\allfor\;\Varid{ty}\;(\lambda \Varid{β}.\;(\Varid{α}\;{\Rightarrow_t}\;\Varid{β})\;{\Rightarrow_t}\;(\Conid{F}\;\Varid{α}\;{\Rightarrow_t}\;\Conid{F}\;\Varid{β}))){}\<[E]%
\\[\blanklineskip]%
\>[B]{}\lhskeyword{record}\;\Conid{RawFunctor}\mathbin{:}\Jdg\;\lhskeyword{where}{}\<[E]%
\\
\>[B]{}\hsindent{3}{}\<[3]%
\>[3]{}\Varid{0}\mathbin{:}\Varid{el}\;\Varid{efty}{}\<[E]%
\\
\>[B]{}\hsindent{3}{}\<[3]%
\>[3]{}\Varid{fmap}\mathbin{:}\Varid{tm}\;(\,\Varid{fmap\hyp{}ty}\;\Varid{0}){}\<[E]%
\ColumnHook
\end{hscode}\resethooks

\item
Monads
\begin{hscode}\SaveRestoreHook
\column{B}{@{}>{\hspre}l<{\hspost}@{}}%
\column{3}{@{}>{\hspre}l<{\hspost}@{}}%
\column{9}{@{}>{\hspre}l<{\hspost}@{}}%
\column{E}{@{}>{\hspre}l<{\hspost}@{}}%
\>[B]{}\lhskeyword{record}\;\Conid{RawMonad}\mathbin{:}\Jdg\;\lhskeyword{where}{}\<[E]%
\\
\>[B]{}\hsindent{3}{}\<[3]%
\>[3]{}\Varid{0}{}\<[9]%
\>[9]{}\mathbin{:}\Varid{el}\;\Varid{efty}{}\<[E]%
\\
\>[B]{}\hsindent{3}{}\<[3]%
\>[3]{}\Varid{ret}{}\<[9]%
\>[9]{}\mathbin{:}\Varid{tm}\;(\allfor\;\Varid{ty}\;(\lambda \Varid{α}.\;\Varid{α}\;{\Rightarrow_t}\;\Varid{0}\;\Varid{α})){}\<[E]%
\\
\>[B]{}\hsindent{3}{}\<[3]%
\>[3]{}\Varid{bind}{}\<[9]%
\>[9]{}\mathbin{:}\Varid{tm}\;(\allfor\;\Varid{ty}\;(\lambda \Varid{α}.\;\allfor\;\Varid{ty}\;(\lambda \Varid{β}.\;\Varid{0}\;\Varid{α}\;{\Rightarrow_t}\;(\Varid{α}\;{\Rightarrow_t}\;\Varid{0}\;\Varid{β})\;{\Rightarrow_t}\;\Varid{0}\;\Varid{β}))){}\<[E]%
\ColumnHook
\end{hscode}\resethooks

\item
Higher-order functors
\begin{hscode}\SaveRestoreHook
\column{B}{@{}>{\hspre}l<{\hspost}@{}}%
\column{3}{@{}>{\hspre}l<{\hspost}@{}}%
\column{10}{@{}>{\hspre}l<{\hspost}@{}}%
\column{E}{@{}>{\hspre}l<{\hspost}@{}}%
\>[B]{}\Varid{htyco}\mathbin{:}\Varid{ki}{}\<[E]%
\\
\>[B]{}\Varid{htyco}\mathrel{=}\Varid{efty}\;{\Rightarrow_k}\;\Varid{efty}{}\<[E]%
\\[\blanklineskip]%
\>[B]{}\Varid{trans}\mathbin{:}(\Conid{F},\Conid{G}\mathbin{:}\Varid{el}\;\Varid{efty})\to\Varid{el}\;\Varid{ty}{}\<[E]%
\\
\>[B]{}\Varid{trans}\;\Conid{F}\;\Conid{G}\mathrel{=}\allfor\;\Varid{ty}\;(\lambda \Varid{α}.\;\Conid{F}\;\Varid{α}\;{\Rightarrow_t}\;\Conid{G}\;\Varid{α}){}\<[E]%
\\[\blanklineskip]%
\>[B]{}\lhskeyword{record}\;\Conid{RawHFunctor}\mathbin{:}\Jdg\;\lhskeyword{where}{}\<[E]%
\\
\>[B]{}\hsindent{3}{}\<[3]%
\>[3]{}\Varid{0}{}\<[10]%
\>[10]{}\mathbin{:}\Varid{el}\;\Varid{htyco}{}\<[E]%
\\
\>[B]{}\hsindent{3}{}\<[3]%
\>[3]{}\Varid{hfmap}{}\<[10]%
\>[10]{}\mathbin{:}(\Conid{F}\mathbin{:}\Conid{RawFunctor})\to\Varid{tm}\;(\,\Varid{fmap\hyp{}ty}\;(\Varid{0}\;(\Conid{F}\;.\Varid{0}))){}\<[E]%
\\
\>[B]{}\hsindent{3}{}\<[3]%
\>[3]{}\Varid{hmap}{}\<[10]%
\>[10]{}\mathbin{:}(\Conid{F},\Conid{G}\mathbin{:}\Conid{RawFunctor})\to\Varid{tm}\;(\Varid{trans}\;(\Conid{F}\;.\Varid{0})\;(\Conid{G}\;.\Varid{0})){}\<[E]%
\\
\>[10]{}\to\Varid{tm}\;(\Varid{trans}\;(\Varid{0}\;(\Conid{F}\;.\Varid{0}))\;(\Varid{0}\;(\Conid{G}\;.\Varid{0}))){}\<[E]%
\ColumnHook
\end{hscode}\resethooks

\item
Computations
\begin{hscode}\SaveRestoreHook
\column{B}{@{}>{\hspre}l<{\hspost}@{}}%
\column{9}{@{}>{\hspre}l<{\hspost}@{}}%
\column{E}{@{}>{\hspre}l<{\hspost}@{}}%
\>[B]{}\Varid{co}{}\<[9]%
\>[9]{}\mathbin{:}(\Conid{H}\mathbin{:}\Conid{RawHFunctor})\to(\Conid{A}\mathbin{:}\Varid{el}\;\Varid{ty})\to\Jdg{}\<[E]%
\\
\>[B]{}\Varid{val}{}\<[9]%
\>[9]{}\mathbin{:}\{\mskip1.5mu \Conid{H},\Conid{A}\mskip1.5mu\}\to\Varid{tm}\;\Conid{A}\to\Varid{co}\;\Conid{H}\;\Conid{A}{}\<[E]%
\\
\>[B]{}\Varid{let\hyp{}in}{}\<[9]%
\>[9]{}\mathbin{:}\abrace{\Conid{H},\Conid{A},\Conid{B}}\to\Varid{co}\;\Conid{H}\;\Conid{A}\to(\Varid{tm}\;\Conid{A}\to\Varid{co}\;\Conid{H}\;\Conid{B})\to\Varid{co}\;\Conid{H}\;\Conid{B}{}\<[E]%
\ColumnHook
\end{hscode}\resethooks

\item
Laws of computations 
\begin{hscode}\SaveRestoreHook
\column{B}{@{}>{\hspre}l<{\hspost}@{}}%
\column{11}{@{}>{\hspre}l<{\hspost}@{}}%
\column{E}{@{}>{\hspre}l<{\hspost}@{}}%
\>[B]{}\Varid{val\hyp{}let}{}\<[11]%
\>[11]{}\mathbin{:}\{\mskip1.5mu \Conid{H},\Conid{A},\Conid{B}\mskip1.5mu\}\to(\Varid{a}\mathbin{:}\Varid{tm}\;\Conid{A})\to(\Varid{k}\mathbin{:}\Varid{tm}\;\Conid{A}\to\Varid{co}\;\Conid{H}\;\Conid{B}){}\<[E]%
\\
\>[11]{}\to\Varid{let\hyp{}in}\;(\Varid{val}\;\Varid{a})\;\Varid{k}=\Varid{k}\;\Varid{a}{}\<[E]%
\\
\>[B]{}\Varid{let\hyp{}val}{}\<[11]%
\>[11]{}\mathbin{:}\{\mskip1.5mu \Conid{H},\Conid{A}\mskip1.5mu\}\to(\Varid{m}\mathbin{:}\Varid{co}\;\Conid{H}\;\Conid{A})\to\Varid{let\hyp{}in}\;\Varid{m}\;\Varid{val}=\Varid{m}{}\<[E]%
\\
\>[B]{}\Varid{let\hyp{}assoc}{}\<[11]%
\>[11]{}\mathbin{:}\{\mskip1.5mu \Conid{H},\Conid{A},\Conid{B},\Conid{C}\mskip1.5mu\}\to(\Varid{m₁}\mathbin{:}\Varid{co}\;\Conid{H}\;\Conid{A}){}\<[E]%
\\
\>[11]{}\to(\Varid{m₂}\mathbin{:}\Varid{tm}\;\Conid{A}\to\Varid{co}\;\Conid{H}\;\Conid{B})\to(\Varid{m₃}\mathbin{:}\Varid{tm}\;\Conid{B}\to\Varid{co}\;\Conid{H}\;\Conid{C}){}\<[E]%
\\
\>[11]{}\to\Varid{let\hyp{}in}\;(\Varid{let\hyp{}in}\;\Varid{m₁}\;\Varid{m₂})\;\Varid{m₃}=\Varid{let\hyp{}in}\;\Varid{m₁}\;(\lambda \Varid{a}.\;\Varid{let\hyp{}in}\;(\Varid{m₂}\;\Varid{a})\;\Varid{m₃}){}\<[E]%
\ColumnHook
\end{hscode}\resethooks

\item
Thunks
\begin{hscode}\SaveRestoreHook
\column{B}{@{}>{\hspre}l<{\hspost}@{}}%
\column{15}{@{}>{\hspre}l<{\hspost}@{}}%
\column{E}{@{}>{\hspre}l<{\hspost}@{}}%
\>[B]{}\Varid{th}\mathbin{:}\Conid{RawHFunctor}\to \Varid{el}\;\Varid{ty}\to \Varid{el}\;\Varid{ty}{}\<[E]%
\\
\>[B]{}\Varid{th\hyp{}iso}\mathbin{:}\{\mskip1.5mu \Conid{H},\Conid{A}\mskip1.5mu\}\to \Varid{tm}\;(\Varid{th}\;\Conid{H}\;\Conid{A})\;{\cong}\;\Varid{co}\;\Conid{H}\;\Conid{A}{}\<[E]%
\\[\blanklineskip]%
\>[B]{}{\Uparrow}\mathbin{:}\{\mskip1.5mu \Conid{H},\Conid{A}\mskip1.5mu\}\to \Varid{tm}\;(\Varid{th}\;\Conid{H}\;\Conid{A})\to \Varid{co}\;\Conid{H}\;\Conid{A}{}\<[E]%
\\
\>[B]{}{\Uparrow}\mathrel{=}\Varid{th\hyp{}iso}\;\Varid{.fwd}{}\<[E]%
\\[\blanklineskip]%
\>[B]{}{\Downarrow}\mathbin{:}\{\mskip1.5mu \Conid{H},\Conid{A}\mskip1.5mu\}\to \Varid{co}\;\Conid{H}\;\Conid{A}\to \Varid{tm}\;(\Varid{th}\;\Conid{H}\;\Conid{A}){}\<[E]%
\\
\>[B]{}{\Downarrow}\mathrel{=}\Varid{th\hyp{}iso}\;\Varid{.bwd}{}\<[E]%
\\[\blanklineskip]%
\>[B]{}\Varid{th\hyp{}mnd}\mathbin{:}\Conid{RawHFunctor}\to\Conid{RawMonad}{}\<[E]%
\\
\>[B]{}\Varid{th\hyp{}mnd}\;\Conid{H}\;\Varid{.0}{}\<[15]%
\>[15]{}\mathrel{=}\Varid{th}\;\Conid{H}{}\<[E]%
\\
\>[B]{}\Varid{th\hyp{}mnd}\;\Conid{H}\;\Varid{.ret}{}\<[15]%
\>[15]{}\mathrel{=}\lambda \Conid{A}\;\Varid{x}.\;{\Downarrow}\;(\Varid{val}\;\Varid{x}){}\<[E]%
\\
\>[B]{}\Varid{th\hyp{}mnd}\;\Conid{H}\;\Varid{.bind}{}\<[15]%
\>[15]{}\mathrel{=}\lambda \Conid{A}\;\Conid{B}\;\Varid{m}\;\Varid{k}.\;{\Downarrow}\;(\Varid{let\hyp{}in}\;(\Varid{force}\;\Varid{m})\;(\lambda \Varid{a}.\;{\Uparrow}\;(\Varid{k}\;\Varid{a}))){}\<[E]%
\ColumnHook
\end{hscode}\resethooks

\item Operations
\begin{hscode}\SaveRestoreHook
\column{B}{@{}>{\hspre}l<{\hspost}@{}}%
\column{8}{@{}>{\hspre}l<{\hspost}@{}}%
\column{E}{@{}>{\hspre}l<{\hspost}@{}}%
\>[B]{}\Varid{op}\mathbin{:}\{\mskip1.5mu \Conid{H},\Conid{A},\Conid{B}\mskip1.5mu\}\to\Varid{tm}\;(\Conid{H}\;(\Varid{th}\;\Conid{H})\;\Conid{A})\to(\Varid{tm}\;\Conid{A}\to\Varid{co}\;\Conid{H}\;\Conid{B})\to\Varid{co}\;\Conid{H}\;\Conid{B}{}\<[E]%
\\[\blanklineskip]%
\>[B]{}\Varid{let\hyp{}op}{}\<[8]%
\>[8]{}\mathbin{:}\{\mskip1.5mu \Conid{H},\Conid{A},\Conid{B},\Conid{C}\mskip1.5mu\}\to(\Varid{p}\mathbin{:}\Varid{tm}\;(\Conid{H}\;(\Varid{th}\;\Conid{H})\;\Conid{A})){}\<[E]%
\\
\>[8]{}\to(\Varid{k}\mathbin{:}\Varid{tm}\;\Conid{A}\to\Varid{co}\;\Conid{H}\;\Conid{B})\to(\Varid{k'}\mathbin{:}\Varid{tm}\;\Conid{B}\to\Varid{co}\;\Conid{H}\;\Conid{C}){}\<[E]%
\\
\>[8]{}\to\Varid{let\hyp{}in}\;(\Varid{op}\;\Varid{p}\;\Varid{k})\;\Varid{k'}=\Varid{op}\;\Varid{p}\;(\lambda \Varid{a}.\;\Varid{let\hyp{}in}\;(\Varid{k}\;\Varid{a})\;\Varid{k'}){}\<[E]%
\ColumnHook
\end{hscode}\resethooks

\item
Monads with algebras
\begin{hscode}\SaveRestoreHook
\column{B}{@{}>{\hspre}l<{\hspost}@{}}%
\column{3}{@{}>{\hspre}l<{\hspost}@{}}%
\column{7}{@{}>{\hspre}l<{\hspost}@{}}%
\column{9}{@{}>{\hspre}l<{\hspost}@{}}%
\column{10}{@{}>{\hspre}l<{\hspost}@{}}%
\column{16}{@{}>{\hspre}l<{\hspost}@{}}%
\column{E}{@{}>{\hspre}l<{\hspost}@{}}%
\>[B]{}\lhskeyword{record}\;\Conid{Handler}\;(\Conid{H}\mathbin{:}\Conid{RawHFunctor})\mathbin{:}\Jdg\,\;\lhskeyword{where}{}\<[E]%
\\
\>[B]{}\hsindent{3}{}\<[3]%
\>[3]{}\lhskeyword{include}\;\Conid{RawMonad}\;\lhskeyword{as}\;\Conid{M}{}\<[E]%
\\
\>[B]{}\hsindent{3}{}\<[3]%
\>[3]{}\Varid{malg}{}\<[9]%
\>[9]{}\mathbin{:}\Varid{tm}\;(\Varid{trans}\;(\Conid{H}\;\Varid{.0}\;\Varid{0})\;\Varid{0}){}\<[E]%
\\[\blanklineskip]%
\>[B]{}\Varid{th\hyp{}alg}\mathbin{:}(\Conid{H}\mathbin{:}\Conid{RawHFunctor})\to\Conid{Handler}\;\Conid{H}{}\<[E]%
\\
\>[B]{}\Varid{th\hyp{}alg}\;{}\<[7]%
\>[7]{}\Conid{H}\;{}\<[10]%
\>[10]{}\Varid{.M}{}\<[16]%
\>[16]{}\mathrel{=}\Varid{th\hyp{}mnd}\;\Conid{H}{}\<[E]%
\\
\>[B]{}\Varid{th\hyp{}alg}\;{}\<[7]%
\>[7]{}\Conid{H}\;\Varid{.malg}{}\<[16]%
\>[16]{}\mathrel{=}\lambda \Varid{α}\;\Varid{o}.\;{\Downarrow}\;(\Varid{op}\;\Varid{o}\;\Varid{val}){}\<[E]%
\ColumnHook
\end{hscode}\resethooks

\item
Evaluation of computations
\begin{hscode}\SaveRestoreHook
\column{B}{@{}>{\hspre}l<{\hspost}@{}}%
\column{9}{@{}>{\hspre}l<{\hspost}@{}}%
\column{10}{@{}>{\hspre}l<{\hspost}@{}}%
\column{12}{@{}>{\hspre}l<{\hspost}@{}}%
\column{16}{@{}>{\hspre}l<{\hspost}@{}}%
\column{17}{@{}>{\hspre}l<{\hspost}@{}}%
\column{23}{@{}>{\hspre}l<{\hspost}@{}}%
\column{E}{@{}>{\hspre}l<{\hspost}@{}}%
\>[B]{}\Varid{hdl}\mathbin{:}\{\mskip1.5mu \Conid{H},\Conid{A}\mskip1.5mu\}\to(\Varid{h}\mathbin{:}\Conid{Handler}\;\Conid{H})\to\Varid{co}\;\Conid{H}\;\Conid{A}\to\Varid{tm}\;(\Varid{h}\;\Varid{.0}\;\Conid{A}){}\<[E]%
\\[\blanklineskip]%
\>[B]{}\Varid{hdl\hyp{}val}{}\<[10]%
\>[10]{}\mathbin{:}\{\mskip1.5mu \Conid{H},\Conid{A}\mskip1.5mu\}\to(\Varid{h}\mathbin{:}\Conid{Handler}\;\Conid{H})\to(\Varid{a}\mathbin{:}\Varid{tm}\;\Conid{A}){}\<[E]%
\\
\>[10]{}\to\Varid{hdl}\;\Varid{h}\;(\Varid{val}\;\Varid{a})=\Varid{h}\;\Varid{.ret}\;\Conid{A}\;\Varid{a}{}\<[E]%
\\[\blanklineskip]%
\>[B]{}\Varid{hdl\hyp{}op}{}\<[9]%
\>[9]{}\mathbin{:}\{\mskip1.5mu \Conid{H},\Conid{A},\Conid{B}\mskip1.5mu\}\to(\Varid{h}\mathbin{:}\Conid{Handler}\;\Conid{H}){}\<[E]%
\\
\>[9]{}\to(\Varid{p}\mathbin{:}\Varid{tm}\;(\Conid{H}\;(\Varid{th}\;\Conid{H})\;\Conid{A}))\to(\Varid{k}\mathbin{:}\Varid{tm}\;\Conid{A}\to\Varid{co}\;\Conid{H}\;\Conid{B}){}\<[E]%
\\
\>[9]{}\to{}\<[12]%
\>[12]{}\lhskeyword{let}\;{}\<[17]%
\>[17]{}\Conid{T}{}\<[23]%
\>[23]{}\mathrel{=}\Varid{fct\hyp{}of\hyp{}mnd}\;(\Varid{th\hyp{}mnd}\;\Conid{H}){}\<[E]%
\\
\>[17]{}\Conid{M}{}\<[23]%
\>[23]{}\mathrel{=}\Varid{fct\hyp{}of\hyp{}mnd}\;(\Varid{h}\;\Varid{.M}){}\<[E]%
\\
\>[17]{}\Varid{p'}{}\<[23]%
\>[23]{}\mathrel{=}\Conid{H}\;\Varid{.hmap}\;\Conid{T}\;\Conid{M}\;(\lambda \Varid{α}\;\Varid{c}.\;\Varid{hdl}\;\Varid{h}\;({\Uparrow}\;\!\Varid{c}))\;\anonymous \;\Varid{p}{}\<[E]%
\\
\>[12]{}\lhskeyword{in}\;{}\<[16]%
\>[16]{}\Varid{hdl}\;\Varid{h}\;(\Varid{op}\;\Varid{p}\;\Varid{k})=\Varid{h}\;\Varid{.bind}\;\anonymous \;\anonymous \;(\Varid{h}\;\Varid{.malg}\;\anonymous \;\Varid{p'})\;(\lambda \Varid{a}.\;\Varid{hdl}\;\Varid{h}\;(\Varid{k}\;\Varid{a})){}\<[E]%
\\[\blanklineskip]%
\>[B]{}\Varid{fct\hyp{}of\hyp{}mnd}\mathbin{:}\Conid{RawMonad}\to\Conid{RawFunctor}{}\<[E]%
\\
\>[B]{}\Varid{fct\hyp{}of\hyp{}mnd}\;\Varid{m}\;.\Varid{0}\mathrel{=}\Varid{m}\;\Varid{.0}{}\<[E]%
\\
\>[B]{}\Varid{fct\hyp{}of\hyp{}mnd}\;\Varid{m}\;.\Varid{fmap}\;\Varid{α}\;\Varid{β}\;\Varid{f}\;\Varid{ma}\mathrel{=}\Varid{m}\;\Varid{.bind}\;\Varid{α}\;\Varid{β}\;\Varid{ma}\;(\lambda \Varid{a}.\;\Varid{m}\;\Varid{.ret}\;\anonymous \;(\,\Varid{f}\;\Varid{a}\,)){}\<[E]%
\ColumnHook
\end{hscode}\resethooks
\end{itemize}

\subsection{Effect Families in \Fha{}}\label{app:eff:fam}
We did not include in System \Fha{} any judgements for \emph{modular
handlers} \citep{YangWu21,YangWu23} or \emph{effect systems}
\citep{BauerPretnar2014,KammarPlotkin2012,Lucassen1988} 
that track the effect operations that a computation may perform,
because both of them can be \emph{derived concepts} in \Fha{}.

\sethscode{autohscode}
Firstly, the judgement for \emph{effect families} \citep{YangWu23} is the following record in
\Fha{}:
\begin{hscode}\SaveRestoreHook
\column{B}{@{}>{\hspre}l<{\hspost}@{}}%
\column{3}{@{}>{\hspre}l<{\hspost}@{}}%
\column{8}{@{}>{\hspre}l<{\hspost}@{}}%
\column{E}{@{}>{\hspre}l<{\hspost}@{}}%
\>[B]{}\lhskeyword{record}\;\Conid{Fam}\mathbin{:}\Jdg\;\lhskeyword{where}{}\<[E]%
\\
\>[B]{}\hsindent{3}{}\<[3]%
\>[3]{}\Varid{eff}{}\<[8]%
\>[8]{}\mathbin{:}\Varid{ki}{}\<[E]%
\\
\>[B]{}\hsindent{3}{}\<[3]%
\>[3]{}\Varid{sig}{}\<[8]%
\>[8]{}\mathbin{:}\Varid{el}\;\Varid{eff}\to\Conid{RawHFunctor}{}\<[E]%
\\
\>[B]{}\hsindent{3}{}\<[3]%
\>[3]{}\Varid{add}{}\<[8]%
\>[8]{}\mathbin{:}\Varid{el}\;\Varid{eff}\to\Varid{el}\;\Varid{eff}\to\Varid{el}\;\Varid{eff}{}\<[E]%
\ColumnHook
\end{hscode}\resethooks
The elements of the kind \ensuremath{\Varid{eff}\mathbin{:}\Varid{ki}} are effect signatures in
this family, each of them determining a higher-order functor via \ensuremath{\Varid{sig}}.
Additionally, there is a way \ensuremath{\Varid{add}} to combine two effects in a family.

Then we have the following definitions for monads and
computations for an effect \ensuremath{\Varid{e}} in a family \ensuremath{\Conid{F}}, which can be viewed as
a generic effect system parameterised by an effect family \ensuremath{\Conid{F}}:
\begin{hscode}\SaveRestoreHook
\column{B}{@{}>{\hspre}l<{\hspost}@{}}%
\column{11}{@{}>{\hspre}l<{\hspost}@{}}%
\column{E}{@{}>{\hspre}l<{\hspost}@{}}%
\>[B]{}\Varid{HandlerEff}\mathbin{:}(\Conid{F}\mathbin{:}\Conid{Fam})\to(\Varid{e}\mathbin{:}\Varid{el}\;(\Conid{F}\;\Varid{.eff}))\to\Jdg{}\<[E]%
\\
\>[B]{}\Varid{HandlerEff}\;\Conid{F}\;\Varid{e}\mathrel{=}\Conid{Handler}\;(\Conid{F}\;\Varid{.sig}\;\Varid{e}){}\<[E]%
\\[\blanklineskip]%
\>[B]{}\Varid{co}[  \text{\textunderscore}{\ni}\text{\textunderscore} ]\,\mathbin{:}{}\<[11]%
\>[11]{}(\Conid{F}\mathbin{:}\Conid{Fam})\to(\Varid{e}\mathbin{:}\Varid{el}\;(\Conid{F}\;\Varid{.eff}))\to\Varid{el}\;\Varid{ty}\to\Jdg{}\<[E]%
\\
\>[B]{}\Varid{co}[\Conid{F}\ni\Varid{e}]\mathrel{=}\Varid{co}\;(\Conid{F}\;\Varid{.sig}\;\Varid{e}){}\<[E]%
\ColumnHook
\end{hscode}\resethooks

A \emph{modular handler} processing the effect \ensuremath{\Varid{e}} in a family \ensuremath{\Conid{F}} and outputting the
effect \ensuremath{\Varid{o}} is the structure
\begin{hscode}\SaveRestoreHook
\column{B}{@{}>{\hspre}l<{\hspost}@{}}%
\column{3}{@{}>{\hspre}l<{\hspost}@{}}%
\column{8}{@{}>{\hspre}l<{\hspost}@{}}%
\column{E}{@{}>{\hspre}l<{\hspost}@{}}%
\>[B]{}\lhskeyword{record}\;\Varid{ModHdl}\;(\Conid{F}\mathbin{:}\Conid{Fam})\;(\Varid{e}\;\Varid{o}\mathbin{:}\Varid{el}\;(\Conid{F}\;\Varid{.eff}))\mathbin{:}\Jdg\;\lhskeyword{where}{}\<[E]%
\\
\>[B]{}\hsindent{3}{}\<[3]%
\>[3]{}\Varid{alg}{}\<[8]%
\>[8]{}\mathbin{:}(\Varid{μ}\mathbin{:}\Varid{el}\;(\Conid{F}\;\Varid{.eff}))\to\Varid{HandlerEff}\;\Conid{F}\;(\Conid{F}\;\Varid{.add}\;\Varid{o}\;\Varid{μ}){}\<[E]%
\\
\>[8]{}\to\Varid{HandlerEff}\;\Conid{F}\;(\Conid{F}\;\Varid{.add}\;\Varid{e}\;\Varid{μ}){}\<[E]%
\\
\>[B]{}\hsindent{3}{}\<[3]%
\>[3]{}\Varid{res}{}\<[8]%
\>[8]{}\mathbin{:}\Varid{el}\;(\Varid{ty}\;{\Rightarrow_k}\;\Varid{ty}){}\<[E]%
\\
\>[B]{}\hsindent{3}{}\<[3]%
\>[3]{}\Varid{run}{}\<[8]%
\>[8]{}\mathbin{:}(\Varid{μ}\mathbin{:}\Varid{el}\;(\Conid{F}\;\Varid{.eff}))\to(\Conid{Mo}\mathbin{:}\Varid{HandlerEff}\;\Conid{F}\;(\Conid{F}\;\Varid{.add}\;\Varid{o}\;\Varid{μ})){}\<[E]%
\\
\>[8]{}\to\Varid{tm}\;(\Varid{trans}\;(\Varid{alg}\;\Varid{μ}\;\Conid{Mo})\;(\lambda \Conid{A}.\;\Conid{Mo}\;(\Varid{res}\;\Conid{A}))){}\<[E]%
\ColumnHook
\end{hscode}\resethooks
Modular handlers as such can be applied to computations \ensuremath{\Varid{co}[\Conid{F}\ni(\Conid{F}\;\Varid{.add}\;\Varid{e}\;\Varid{μ})]\;\Conid{A}}, for all `ambient' effects $\mu$, removing the effect \ensuremath{\Varid{e}} and generating the
effect \ensuremath{\Varid{o}}, yielding computations \ensuremath{\Varid{co}[\Conid{F}\ni(\Conid{F}\;\Varid{.add}\;\Varid{o}\;\Varid{μ})]\;(\Varid{h}\;\Varid{.res}\;\Conid{A})}:
\begin{hscode}\SaveRestoreHook
\column{B}{@{}>{\hspre}l<{\hspost}@{}}%
\column{3}{@{}>{\hspre}l<{\hspost}@{}}%
\column{9}{@{}>{\hspre}l<{\hspost}@{}}%
\column{24}{@{}>{\hspre}l<{\hspost}@{}}%
\column{E}{@{}>{\hspre}l<{\hspost}@{}}%
\>[B]{}\Varid{handle}{}\<[9]%
\>[9]{}\mathbin{:}\{\mskip1.5mu \Conid{F},\Varid{e},\Varid{o},\Varid{μ},\Conid{A}\mskip1.5mu\}\to(\Varid{h}\mathbin{:}\Varid{ModHdl}\;\Conid{F}\;\Varid{e}\;\Varid{o})\to\Varid{co}[\Conid{F}\ni(\Conid{F}\;\Varid{.add}\;\Varid{e}\;\Varid{μ})]\;\Conid{A}{}\<[E]%
\\
\>[9]{}\to\Varid{co}[\Conid{F}\ni(\Conid{F}\;\Varid{.add}\;\Varid{o}\;\Varid{μ})]\;(\Varid{h}\;\Varid{.res}\;\Conid{A}){}\<[E]%
\\[\blanklineskip]%
\>[B]{}\Varid{handle}\;\Varid{h}\;\Varid{c}\mathrel{=}{\Uparrow}\;(\Varid{h}\;\Varid{.run}\;\Varid{μ}\;\Conid{T}\;\Conid{A}\;\Varid{c'})\;\lhskeyword{where}{}\<[E]%
\\[\blanklineskip]%
\>[B]{}\hsindent{3}{}\<[3]%
\>[3]{}\Conid{T}\mathbin{:}\Varid{HandlerEff}\;\Conid{F}\;(\Conid{F}\;\Varid{.add}\;\Varid{o}\;\Varid{μ}){}\<[E]%
\\
\>[B]{}\hsindent{3}{}\<[3]%
\>[3]{}\Conid{T}\mathrel{=}\Varid{th\hyp{}alg}\;(\Conid{F}\;\Varid{.sig}\;(\Conid{F}\;\Varid{.add}\;\Varid{o}\;\Varid{μ})){}\<[E]%
\\[\blanklineskip]%
\>[B]{}\hsindent{3}{}\<[3]%
\>[3]{}\Varid{c'}\mathbin{:}\Varid{tm}\;(\Varid{h}\;\Varid{.alg}\;\Varid{μ}\;\Conid{T}\;{}\<[24]%
\>[24]{}\Conid{A}){}\<[E]%
\\
\>[B]{}\hsindent{3}{}\<[3]%
\>[3]{}\Varid{c'}\mathrel{=}\Varid{hdl}\;(\Varid{h}\;\Varid{.alg}\;\Varid{μ}\;(\Varid{th\hyp{}alg}\;(\Conid{F}\;\Varid{.sig}\;(\Conid{F}\;\Varid{.add}\;\Varid{o}\;\Varid{μ}))))\;\Varid{c}{}\<[E]%
\ColumnHook
\end{hscode}\resethooks

\sethscode{plainhscode}

The complete definition of the effect family \ensuremath{\Varid{algFam}}, together with the needed
standard type connectives, is collected below. 

\begin{itemize}
\item Kind-level and type-level products
\begin{hscode}\SaveRestoreHook
\column{B}{@{}>{\hspre}l<{\hspost}@{}}%
\column{E}{@{}>{\hspre}l<{\hspost}@{}}%
\>[B]{}{\_\times_k\_}\mathbin{:}\Varid{ki}\to\Varid{ki}\to\Varid{ki}{}\<[E]%
\\
\>[B]{}\Varid{\times_k\hyp{}iso}\mathbin{:}\{\mskip1.5mu \Varid{k}\;\Varid{k'}\mathbin{:}\Varid{ki}\mskip1.5mu\}\to\Varid{el}\;(\Varid{k}\;\mathbin{\times_k}\;\Varid{k'})=\Sigma\;(\Varid{el}\;\Varid{k})\;(\lambda \anonymous .\;\Varid{el}\;\Varid{k'}){}\<[E]%
\\[\blanklineskip]%
\>[B]{}{\_\times_t\_}\mathbin{:}\Varid{el}\;\Varid{ty}\to\Varid{el}\;\Varid{ty}\to\Varid{el}\;\Varid{ty}{}\<[E]%
\\
\>[B]{}\Varid{\times_t\hyp{}iso}\mathbin{:}\{\mskip1.5mu \Conid{A}\;\Conid{B}\mathbin{:}\Varid{el}\;\Varid{ty}\mskip1.5mu\}\to\Varid{tm}\;(\Conid{A}\;\mathbin{\times_t}\;\Conid{B})=\Sigma\;(\Varid{tm}\;\Conid{A})\;(\lambda \anonymous .\;\Varid{tm}\;\Conid{B}){}\<[E]%
\ColumnHook
\end{hscode}\resethooks
  
\item The empty type
\begin{hscode}\SaveRestoreHook
\column{B}{@{}>{\hspre}l<{\hspost}@{}}%
\column{E}{@{}>{\hspre}l<{\hspost}@{}}%
\>[B]{}\Varid{empty}\mathbin{:}\Varid{el}\;\Varid{ty}{}\<[E]%
\\[\blanklineskip]%
\>[B]{}\Varid{absurd}\mathbin{:}(\Conid{A}\mathbin{:}\Varid{el}\;\Varid{ty})\to\Varid{tm}\;\Varid{empty}\to\Varid{tm}\;\Conid{A}{}\<[E]%
\\
\>[B]{}\Varid{absurd\hyp{}uniq}\mathbin{:}\{\mskip1.5mu \Conid{A}\mathbin{:}\Varid{el}\;\Varid{ty}\mskip1.5mu\}\to(\Varid{f}\mathbin{:}\Varid{tm}\;\Varid{empty}\to\Varid{tm}\;\Conid{A})\to\Varid{f}=\Varid{absurd}\;\Conid{A}{}\<[E]%
\ColumnHook
\end{hscode}\resethooks

\item Coproducts

We only need type-level coproducts, but for generality
we define coproducts parameterised by judgements \ensuremath{\Conid{U}\mathbin{:}\Jdg} and \ensuremath{\Conid{T}\mathbin{:}\Conid{U}\to \Jdg}:

\begin{hscode}\SaveRestoreHook
\column{B}{@{}>{\hspre}l<{\hspost}@{}}%
\column{3}{@{}>{\hspre}l<{\hspost}@{}}%
\column{10}{@{}>{\hspre}l<{\hspost}@{}}%
\column{11}{@{}>{\hspre}l<{\hspost}@{}}%
\column{12}{@{}>{\hspre}l<{\hspost}@{}}%
\column{13}{@{}>{\hspre}l<{\hspost}@{}}%
\column{E}{@{}>{\hspre}l<{\hspost}@{}}%
\>[B]{}\lhskeyword{record}\;\Varid{coprod\char95 intro}\;(\Conid{U}\mathbin{:}\Jdg)\;(\Conid{T}\mathbin{:}\Conid{U}\to\Jdg)\mathbin{:}\Jdg\;\lhskeyword{where}{}\<[E]%
\\
\>[B]{}\hsindent{3}{}\<[3]%
\>[3]{}{\_{+}\_}{}\<[10]%
\>[10]{}\mathbin{:}\Conid{U}\to\Conid{U}\to\Conid{U}{}\<[E]%
\\
\>[B]{}\hsindent{3}{}\<[3]%
\>[3]{}\Varid{inl}{}\<[10]%
\>[10]{}\mathbin{:}\{\mskip1.5mu \Varid{a},\Varid{b}\mskip1.5mu\}\to\Conid{T}\;\Varid{a}\to\Conid{T}\;(\Varid{a}\mathbin{+}\Varid{b}){}\<[E]%
\\
\>[B]{}\hsindent{3}{}\<[3]%
\>[3]{}\Varid{inr}{}\<[10]%
\>[10]{}\mathbin{:}\{\mskip1.5mu \Varid{a},\Varid{b}\mskip1.5mu\}\to\Conid{T}\;\Varid{b}\to\Conid{T}\;(\Varid{a}\mathbin{+}\Varid{b}){}\<[E]%
\\[\blanklineskip]%
\>[B]{}\lhskeyword{record}\;\Varid{coprod\char95 elim}\;(\Conid{U}\mathbin{:}\Jdg)\;(\Conid{T}\mathbin{:}\Conid{U}\to\Jdg)\;(\Conid{V}\mathbin{:}\Jdg)\;(\Conid{S}\mathbin{:}\Conid{U}\to\Jdg)\;{}\<[E]%
\\
\>[B]{}\hsindent{3}{}\<[3]%
\>[3]{}(\Varid{intr}\mathbin{:}\Varid{coprod\char95 intro}\;\Conid{U}\;\Conid{T})\mathbin{:}\Jdg\;\lhskeyword{where}{}\<[E]%
\\
\>[B]{}\hsindent{3}{}\<[3]%
\>[3]{}\lhskeyword{open}\;\Varid{coprod\char95 intro}\;\Varid{intr}{}\<[E]%
\\[\blanklineskip]%
\>[B]{}\hsindent{3}{}\<[3]%
\>[3]{}\Varid{case}\mathbin{:}\{\mskip1.5mu \Varid{a},\Varid{b},\Varid{c}\mskip1.5mu\}\to(\Conid{T}\;\Varid{a}\to\Conid{S}\;\Varid{c})\to(\Conid{T}\;\Varid{b}\to\Conid{S}\;\Varid{c})\to(\Conid{T}\;(\Varid{a}\mathbin{+}\Varid{b})\to\Conid{S}\;\Varid{c}){}\<[E]%
\\[\blanklineskip]%
\>[B]{}\hsindent{3}{}\<[3]%
\>[3]{}\Varid{caseβl}{}\<[11]%
\>[11]{}\mathbin{:}\{\mskip1.5mu \Varid{a},\Varid{b},\Varid{c}\mskip1.5mu\}\to(\Varid{l}\mathbin{:}\Conid{T}\;\Varid{a}\to\Conid{S}\;\Varid{c})\to(\Varid{r}\mathbin{:}\Conid{T}\;\Varid{b}\to\Conid{S}\;\Varid{c})\to(\Varid{x}\mathbin{:}\Conid{T}\;\Varid{a}){}\<[E]%
\\
\>[11]{}\hsindent{2}{}\<[13]%
\>[13]{}\to\Varid{case}\;\Varid{l}\;\Varid{r}\;(\Varid{inl}\;\Varid{x})=\Varid{l}\;\Varid{x}{}\<[E]%
\\[\blanklineskip]%
\>[B]{}\hsindent{3}{}\<[3]%
\>[3]{}\Varid{caseβr}{}\<[11]%
\>[11]{}\mathbin{:}\{\mskip1.5mu \Varid{a},\Varid{b},\Varid{c}\mskip1.5mu\}\to(\Varid{l}\mathbin{:}\Conid{T}\;\Varid{a}\to\Conid{S}\;\Varid{c})\to(\Varid{r}\mathbin{:}\Conid{T}\;\Varid{b}\to\Conid{S}\;\Varid{c})\to(\Varid{x}\mathbin{:}\Conid{T}\;\Varid{b}){}\<[E]%
\\
\>[11]{}\hsindent{2}{}\<[13]%
\>[13]{}\to\Varid{case}\;\Varid{l}\;\Varid{r}\;(\Varid{inr}\;\Varid{x})=\Varid{r}\;\Varid{x}{}\<[E]%
\\[\blanklineskip]%
\>[B]{}\hsindent{3}{}\<[3]%
\>[3]{}\Varid{caseη}{}\<[10]%
\>[10]{}\mathbin{:}\{\mskip1.5mu \Varid{a},\Varid{b},\Varid{c}\mskip1.5mu\}\to(\Varid{f}\mathbin{:}\Conid{T}\;(\Varid{a}\mathbin{+}\Varid{b})\to\Conid{S}\;\Varid{c}){}\<[E]%
\\
\>[10]{}\hsindent{2}{}\<[12]%
\>[12]{}\to\Varid{case}\;(\lambda \Varid{x}.\;\Varid{f}\;(\Varid{inl}\;\Varid{x}))\;(\lambda \Varid{x}.\;\Varid{f}\;(\Varid{inr}\;\Varid{x}))=\Varid{f}{}\<[E]%
\\[\blanklineskip]%
\>[B]{}\lhskeyword{record}\;\Varid{coprod}\;(\Conid{U}\mathbin{:}\Jdg)\;(\Conid{T}\mathbin{:}\Conid{U}\to\Jdg)\mathbin{:}\Jdg\;\lhskeyword{where}{}\<[E]%
\\
\>[B]{}\hsindent{3}{}\<[3]%
\>[3]{}\Varid{cpintr}\mathbin{:}\Varid{coprod\char95 intro}\;\Conid{U}\;\Conid{T}{}\<[E]%
\\
\>[B]{}\hsindent{3}{}\<[3]%
\>[3]{}\Varid{cpelim}\mathbin{:}\Varid{coprod\char95 elim}\;\Conid{U}\;\Conid{T}\;\Conid{U}\;\Conid{T}\;\Varid{cpintr}{}\<[E]%
\ColumnHook
\end{hscode}\resethooks
We then instantiate with \ensuremath{\Conid{U}\mathrel{=}\Varid{el}\;\Varid{ty}} and \ensuremath{\Conid{T}\mathrel{=}\Varid{tm}} to get type-level coproducts:
\begin{hscode}\SaveRestoreHook
\column{B}{@{}>{\hspre}l<{\hspost}@{}}%
\column{E}{@{}>{\hspre}l<{\hspost}@{}}%
\>[B]{}\Varid{coprodTy}\mathbin{:}\Varid{coprod}\;(\Varid{el}\;\Varid{ty})\;\Varid{tm}{}\<[E]%
\ColumnHook
\end{hscode}\resethooks
In this way, if kind-level coproducts are also needed, they can be easily added
by a declaration \ensuremath{\Varid{coprodKi}\mathbin{:}\Varid{coprod}\;\Varid{ki}\;\Varid{el}}.

\item Booleans

As a special case of coproducts, we have the Boolean type (for which we did not
include an eliminator in the main text for simplicity).
\begin{hscode}\SaveRestoreHook
\column{B}{@{}>{\hspre}l<{\hspost}@{}}%
\column{E}{@{}>{\hspre}l<{\hspost}@{}}%
\>[B]{}\Varid{bool}\mathbin{:}\Varid{el}\;\Varid{ty}{}\<[E]%
\\[\blanklineskip]%
\>[B]{}\Varid{tt}\;\Varid{ff}\mathbin{:}\Varid{tm}\;\Varid{bool}{}\<[E]%
\\[\blanklineskip]%
\>[B]{}\Varid{ite}\mathbin{:}\{\mskip1.5mu \Conid{A}\mathbin{:}\Varid{el}\;\Varid{ty}\mskip1.5mu\}\to \Varid{tm}\;\Varid{bool}\to \Varid{tm}\;\Conid{A}\to \Varid{tm}\;\Conid{A}\to \Varid{tm}\;\Conid{A}{}\<[E]%
\\[\blanklineskip]%
\>[B]{}\Varid{ite\beta{}tt}\mathbin{:}\{\mskip1.5mu \Conid{A}\mathbin{:}\Varid{el}\;\Varid{ty}\mskip1.5mu\}\to (\Varid{a}\;\Varid{b}\mathbin{:}\Varid{tm}\;\Conid{A})\to \Varid{ite}\;\Varid{tt}\;\Varid{a}\;\Varid{b}\mathrel{=}\Varid{a}{}\<[E]%
\\[\blanklineskip]%
\>[B]{}\Varid{ite\beta{}ff}\mathbin{:}\{\mskip1.5mu \Conid{A}\mathbin{:}\Varid{el}\;\Varid{ty}\mskip1.5mu\}\to (\Varid{a}\;\Varid{b}\mathbin{:}\Varid{tm}\;\Conid{A})\to \Varid{ite}\;\Varid{ff}\;\Varid{a}\;\Varid{b}\mathrel{=}\Varid{b}{}\<[E]%
\\[\blanklineskip]%
\>[B]{}\Varid{ite\eta}\mathbin{:}\{\mskip1.5mu \Conid{A}\mathbin{:}\Varid{el}\;\Varid{ty}\mskip1.5mu\}\to (\Varid{f}\mathbin{:}\Varid{tm}\;\Varid{bool}\to \Varid{tm}\;\Conid{A})\to (\lambda \Varid{b}.\;\Varid{ite}\;\Varid{b}\;(\Varid{f}\;\Varid{tt})\;(\Varid{f}\;\Varid{ff}))\mathrel{=}\Varid{f}{}\<[E]%
\ColumnHook
\end{hscode}\resethooks

\item
Kind-level lists with elimination to ML-style signatures

We first define the judgements of lists parameterised by 
the universe $(U, T)$ that the lists live in and the universe $(V, S)$
that the lists can eliminate into:
\begin{hscode}\SaveRestoreHook
\column{B}{@{}>{\hspre}l<{\hspost}@{}}%
\column{3}{@{}>{\hspre}l<{\hspost}@{}}%
\column{9}{@{}>{\hspre}l<{\hspost}@{}}%
\column{10}{@{}>{\hspre}l<{\hspost}@{}}%
\column{13}{@{}>{\hspre}l<{\hspost}@{}}%
\column{14}{@{}>{\hspre}l<{\hspost}@{}}%
\column{58}{@{}>{\hspre}l<{\hspost}@{}}%
\column{E}{@{}>{\hspre}l<{\hspost}@{}}%
\>[B]{}\lhskeyword{record}\;\Conid{ListAlg}\;\{\mskip1.5mu \Conid{U}\mathbin{:}\Jdg\mskip1.5mu\}\;\{\mskip1.5mu \Conid{V}\mathbin{:}\Jdg\mskip1.5mu\}\;(\Conid{T}\mathbin{:}\Conid{U}\to\Jdg)\;(\Conid{S}\mathbin{:}\Conid{V}\to\Jdg)\;{}\<[E]%
\\
\>[B]{}\hsindent{3}{}\<[3]%
\>[3]{}(\Varid{k}\mathbin{:}\Conid{U})\;(\Varid{a}\mathbin{:}\Conid{V})\mathbin{:}\Jdg{}\<[E]%
\\
\>[B]{}\hsindent{3}{}\<[3]%
\>[3]{}\lhskeyword{where}{}\<[E]%
\\
\>[B]{}\hsindent{3}{}\<[3]%
\>[3]{}\Varid{fst}\mathbin{:}\Conid{S}\;\Varid{a}{}\<[E]%
\\
\>[B]{}\hsindent{3}{}\<[3]%
\>[3]{}\Varid{snd}\mathbin{:}\Conid{T}\;\Varid{k}\to\Conid{S}\;\Varid{a}\to\Conid{S}\;\Varid{a}{}\<[E]%
\\[\blanklineskip]%
\>[B]{}\lhskeyword{record}\;\Conid{ListHom}\;\{\mskip1.5mu \Conid{U}\mathbin{:}\Jdg\mskip1.5mu\}\;\{\mskip1.5mu \Conid{V}\mathbin{:}\Jdg\mskip1.5mu\}\;\{\mskip1.5mu \Conid{W}\mathbin{:}\Jdg\mskip1.5mu\}\;{}\<[E]%
\\
\>[B]{}\hsindent{3}{}\<[3]%
\>[3]{}\{\mskip1.5mu \Conid{T}\mathbin{:}\Conid{U}\to\Jdg\mskip1.5mu\}\;\{\mskip1.5mu \Conid{S}\mathbin{:}\Conid{V}\to\Jdg\mskip1.5mu\}\;\{\mskip1.5mu \Conid{R}\mathbin{:}\Conid{W}\to\Jdg\mskip1.5mu\}\;{}\<[E]%
\\
\>[B]{}\hsindent{3}{}\<[3]%
\>[3]{}\{\mskip1.5mu \Varid{k}\mathbin{:}\Conid{U}\mskip1.5mu\}\;\{\mskip1.5mu \Varid{a}\mathbin{:}\Conid{V}\mskip1.5mu\}\;\{\mskip1.5mu \Varid{b}\mathbin{:}\Conid{W}\mskip1.5mu\}\;{}\<[E]%
\\
\>[B]{}\hsindent{3}{}\<[3]%
\>[3]{}(\Varid{alga}\mathbin{:}\Conid{ListAlg}\;\Conid{T}\;\Conid{S}\;\Varid{k}\;\Varid{a})\;(\Varid{algb}\mathbin{:}\Conid{ListAlg}\;\Conid{T}\;\Conid{R}\;\Varid{k}\;\Varid{b})\mathbin{:}\Jdg{}\<[E]%
\\
\>[B]{}\hsindent{3}{}\<[3]%
\>[3]{}\lhskeyword{where}{}\<[E]%
\\[\blanklineskip]%
\>[B]{}\hsindent{3}{}\<[3]%
\>[3]{}\Varid{f}\mathbin{:}\Conid{S}\;\Varid{a}\to\Conid{R}\;\Varid{b}{}\<[E]%
\\
\>[B]{}\hsindent{3}{}\<[3]%
\>[3]{}\Varid{homnil}\mathbin{:}\Varid{f}\;(\Varid{alga}\;\Varid{.fst})=\Varid{algb}\;\Varid{.fst}{}\<[E]%
\\
\>[B]{}\hsindent{3}{}\<[3]%
\>[3]{}\Varid{homcons}\mathbin{:}(\Varid{x}\mathbin{:}\Conid{T}\;\Varid{k})\to(\Varid{a}\mathbin{:}\Conid{S}\;\Varid{a})\to\Varid{f}\;(\Varid{alga}\;\Varid{.snd}\;\Varid{x}\;\Varid{a})=\Varid{algb}\;\Varid{.snd}\;\Varid{x}\;(\Varid{f}\;\Varid{a}){}\<[E]%
\\[\blanklineskip]%
\>[B]{}\lhskeyword{record}\;\Conid{ListIntro}\;(\Conid{U}\mathbin{:}\Jdg)\;(\Conid{T}\mathbin{:}\Conid{U}\to\Jdg)\mathbin{:}\Jdg\;\lhskeyword{where}{}\<[E]%
\\
\>[B]{}\hsindent{3}{}\<[3]%
\>[3]{}\Varid{listc}\mathbin{:}\Conid{U}\to\Conid{U}{}\<[E]%
\\
\>[B]{}\hsindent{3}{}\<[3]%
\>[3]{}\Varid{listcalg}\mathbin{:}\{\mskip1.5mu \Varid{k}\mathbin{:}\Conid{U}\mskip1.5mu\}\to\Conid{ListAlg}\;\Conid{T}\;\Conid{T}\;\Varid{k}\;(\Varid{listc}\;\Varid{k}){}\<[E]%
\\[\blanklineskip]%
\>[B]{}\hsindent{3}{}\<[3]%
\>[3]{}\Varid{nil}\mathbin{:}(\Varid{k}\mathbin{:}\Conid{U})\to\Conid{T}\;(\Varid{listc}\;\Varid{k}){}\<[E]%
\\
\>[B]{}\hsindent{3}{}\<[3]%
\>[3]{}\Varid{nil}\;\Varid{k}\mathrel{=}\Varid{listcalg}\;\Varid{.fst}{}\<[E]%
\\[\blanklineskip]%
\>[B]{}\hsindent{3}{}\<[3]%
\>[3]{}\Varid{cons}\mathbin{:}\{\mskip1.5mu \Varid{k}\mathbin{:}\Conid{U}\mskip1.5mu\}\to\Conid{T}\;\Varid{k}\to\Conid{T}\;(\Varid{listc}\;\Varid{k})\to\Conid{T}\;(\Varid{listc}\;\Varid{k}){}\<[E]%
\\
\>[B]{}\hsindent{3}{}\<[3]%
\>[3]{}\Varid{cons}\;\Varid{x}\;\Varid{xs}\mathrel{=}\Varid{listcalg}\;\Varid{.snd}\;\Varid{x}\;\Varid{xs}{}\<[E]%
\\[\blanklineskip]%
\>[B]{}\lhskeyword{record}\;\Conid{ListElim}\;(\Conid{U}\mathbin{:}\Jdg)\;(\Conid{T}\mathbin{:}\Conid{U}\to\Jdg)\;(\Conid{V}\mathbin{:}\Jdg)\;(\Conid{S}\mathbin{:}\Conid{V}\to\Jdg)\;{}\<[E]%
\\
\>[B]{}\hsindent{3}{}\<[3]%
\>[3]{}(\Varid{intr}\mathbin{:}\Conid{ListIntro}\;\Conid{U}\;\Conid{T})\mathbin{:}\Jdg{}\<[E]%
\\
\>[B]{}\hsindent{3}{}\<[3]%
\>[3]{}\lhskeyword{where}{}\<[E]%
\\[\blanklineskip]%
\>[B]{}\hsindent{3}{}\<[3]%
\>[3]{}\lhskeyword{open}\;\Conid{ListIntro}\;\Varid{intr}{}\<[E]%
\\
\>[B]{}\hsindent{3}{}\<[3]%
\>[3]{}\Varid{fold}{}\<[9]%
\>[9]{}\mathbin{:}\{\mskip1.5mu \Varid{k}\mathbin{:}\Conid{U}\mskip1.5mu\}\to\{\mskip1.5mu \Varid{a}\mathbin{:}\Conid{V}\mskip1.5mu\}\to(\Varid{alga}\mathbin{:}\Conid{ListAlg}\;\Conid{T}\;\Conid{S}\;\Varid{k}\;\Varid{a}){}\<[E]%
\\
\>[9]{}\to\Conid{T}\;(\Varid{listc}\;\Varid{k})\to\Conid{S}\;\Varid{a}{}\<[E]%
\\[\blanklineskip]%
\>[B]{}\hsindent{3}{}\<[3]%
\>[3]{}\Varid{foldβnil}{}\<[13]%
\>[13]{}\mathbin{:}\{\mskip1.5mu \Varid{k},\Varid{a}\mskip1.5mu\}\to(\Varid{alga}\mathbin{:}\Conid{ListAlg}\;\Conid{T}\;\Conid{S}\;\Varid{k}\;\Varid{a}){}\<[E]%
\\
\>[13]{}\to\Varid{fold}\;\Varid{alga}\;(\Varid{nil}\;\Varid{k})=\Varid{alga}\;\Varid{.fst}{}\<[E]%
\\[\blanklineskip]%
\>[B]{}\hsindent{3}{}\<[3]%
\>[3]{}\Varid{foldβcons}{}\<[14]%
\>[14]{}\mathbin{:}\{\mskip1.5mu \Varid{k},\Varid{a}\mskip1.5mu\}\to(\Varid{alga}\mathbin{:}\Conid{ListAlg}\;\Conid{T}\;\Conid{S}\;\Varid{k}\;\Varid{a}){}\<[E]%
\\
\>[14]{}\to(\Varid{x}\mathbin{:}\Conid{T}\;\Varid{k})\;(\Varid{xs}\mathbin{:}\Conid{T}\;(\Varid{listc}\;\Varid{k})){}\<[E]%
\\
\>[14]{}\to\Varid{fold}\;\Varid{alga}\;(\Varid{cons}\;\Varid{x}\;\Varid{xs})=\Varid{alga}\;\Varid{.snd}\;\Varid{x}\;(\Varid{fold}\;\Varid{alga}\;\Varid{xs}){}\<[E]%
\\[\blanklineskip]%
\>[B]{}\hsindent{3}{}\<[3]%
\>[3]{}\Varid{foldη}{}\<[10]%
\>[10]{}\mathbin{:}\{\mskip1.5mu \Varid{k},\Varid{a}\mskip1.5mu\}\to(\Varid{alga}\mathbin{:}\Conid{ListAlg}\;\Conid{T}\;\Conid{S}\;\Varid{k}\;\Varid{a}){}\<[E]%
\\
\>[10]{}\to(\Varid{h}\mathbin{:}\Conid{ListHom}\;\Varid{listcalg}\;\Varid{alga}){}\<[E]%
\\
\>[10]{}\to\Varid{fold}\;\Varid{alga}=\Varid{h}\;\Varid{.f}{}\<[E]%
\\[\blanklineskip]%
\>[B]{}\lhskeyword{record}\;\Conid{List}\;(\Conid{U}\mathbin{:}\Jdg)\;(\Conid{T}\mathbin{:}\Conid{U}\to\Jdg)\;(\Conid{V}\mathbin{:}\Jdg)\;(\Conid{S}\mathbin{:}\Conid{V}\to\Jdg)\mathbin{:}\Jdg{}\<[58]%
\>[58]{}\;\lhskeyword{where}{}\<[E]%
\\
\>[B]{}\hsindent{3}{}\<[3]%
\>[3]{}\Varid{intr}\mathbin{:}\Conid{ListIntro}\;\Conid{U}\;\Conid{T}{}\<[E]%
\\
\>[B]{}\hsindent{3}{}\<[3]%
\>[3]{}\Varid{elim}\mathbin{:}\Conid{ListElim}\;\Conid{U}\;\Conid{T}\;\Conid{V}\;\Conid{S}\;\Varid{intr}{}\<[E]%
\ColumnHook
\end{hscode}\resethooks

We have kind-level lists with declarations
\begin{hscode}\SaveRestoreHook
\column{B}{@{}>{\hspre}l<{\hspost}@{}}%
\column{E}{@{}>{\hspre}l<{\hspost}@{}}%
\>[B]{}\Conid{ListKi}\mathbin{:}\Conid{List}\;\Varid{ki}\;\Varid{el}\;\Varid{ki}\;\Varid{el}{}\<[E]%
\ColumnHook
\end{hscode}\resethooks

We additionally have elimination of kind-level lists to ML-style signatures
\begin{hscode}\SaveRestoreHook
\column{B}{@{}>{\hspre}l<{\hspost}@{}}%
\column{E}{@{}>{\hspre}l<{\hspost}@{}}%
\>[B]{}\Varid{si}\mathbin{:}\Jdg{}\<[E]%
\\
\>[B]{}\Varid{si}\mathrel{=}\Sigma\;\Varid{ki}\;(\lambda \Varid{k}.\;(\Varid{el}\;\Varid{k}\to\Varid{el}\;\Varid{ty})){}\<[E]%
\\[\blanklineskip]%
\>[B]{}\Varid{mo}\mathbin{:}\Varid{si}\to\Jdg{}\<[E]%
\\
\>[B]{}\Varid{mo}\;(\Varid{k},\Varid{t})\mathrel{=}\Sigma\;(\Varid{el}\;\Varid{k})\;(\lambda \Varid{α}.\;\Varid{tm}\;(\Varid{t}\;\Varid{α})){}\<[E]%
\\[\blanklineskip]%
\>[B]{}\Conid{ListKiTyElim}\mathbin{:}\Conid{ListElim}\;\Varid{ki}\;\Varid{el}\;\Varid{si}\;\Varid{mo}\;(\Conid{ListKi}\;\Varid{.intr}){}\<[E]%
\ColumnHook
\end{hscode}\resethooks

In the following we will rename the components of lists as follows:
\begin{hscode}\SaveRestoreHook
\column{B}{@{}>{\hspre}l<{\hspost}@{}}%
\column{3}{@{}>{\hspre}l<{\hspost}@{}}%
\column{5}{@{}>{\hspre}l<{\hspost}@{}}%
\column{6}{@{}>{\hspre}l<{\hspost}@{}}%
\column{16}{@{}>{\hspre}l<{\hspost}@{}}%
\column{17}{@{}>{\hspre}l<{\hspost}@{}}%
\column{28}{@{}>{\hspre}l<{\hspost}@{}}%
\column{33}{@{}>{\hspre}l<{\hspost}@{}}%
\column{39}{@{}>{\hspre}l<{\hspost}@{}}%
\column{44}{@{}>{\hspre}l<{\hspost}@{}}%
\column{E}{@{}>{\hspre}l<{\hspost}@{}}%
\>[B]{}\lhskeyword{open}\;\Conid{List}\;\Conid{ListKi}\;\lhskeyword{renaming}\;{}\<[E]%
\\
\>[B]{}\hsindent{3}{}\<[3]%
\>[3]{}({}\<[6]%
\>[6]{}\Varid{listc}\;{}\<[17]%
\>[17]{}\mapsto\;\Varid{list_k};{}\<[28]%
\>[28]{}\Varid{nil}\;{}\<[39]%
\>[39]{}\mapsto\;\Varid{nil_k};\hspace{1em}\Varid{cons}\;\mapsto\;\Varid{cons_k};{}\<[E]%
\\
\>[6]{}\Varid{fold}\;{}\<[17]%
\>[17]{}\mapsto\;\Varid{fold_k};{}\<[28]%
\>[28]{}\Varid{foldβnil}\;{}\<[39]%
\>[39]{}\mapsto\;\Varid{foldkβni};{}\<[E]%
\\
\>[6]{}\Varid{foldβcons}\;{}\<[17]%
\>[17]{}\mapsto\;\Varid{fold_kβcons};\Varid{foldη}\;{}\<[39]%
\>[39]{}\mapsto\;\Varid{fold_kη}){}\<[E]%
\\[\blanklineskip]%
\>[B]{}\lhskeyword{open}\;\Conid{ListElim}\;\Conid{ListKiTyElim}\;\lhskeyword{renaming}\;{}\<[E]%
\\
\>[B]{}\hsindent{3}{}\<[3]%
\>[3]{}(\Varid{fold}\;{}\<[16]%
\>[16]{}\mapsto\;\Varid{fold_{kt}};{}\<[33]%
\>[33]{}\Varid{foldβnil}\;{}\<[44]%
\>[44]{}\mapsto\;\Varid{fold_{kt}βnil};{}\<[E]%
\\
\>[3]{}\hsindent{2}{}\<[5]%
\>[5]{}\Varid{foldβcons}\;{}\<[16]%
\>[16]{}\mapsto\;\Varid{fold_{kt}βcons};{}\<[33]%
\>[33]{}\Varid{foldη}\;{}\<[44]%
\>[44]{}\mapsto\;\Varid{fold_{kt}η}){}\<[E]%
\ColumnHook
\end{hscode}\resethooks

\item The constantly empty higher-order functor
\begin{hscode}\SaveRestoreHook
\column{B}{@{}>{\hspre}l<{\hspost}@{}}%
\column{3}{@{}>{\hspre}l<{\hspost}@{}}%
\column{12}{@{}>{\hspre}l<{\hspost}@{}}%
\column{17}{@{}>{\hspre}l<{\hspost}@{}}%
\column{23}{@{}>{\hspre}l<{\hspost}@{}}%
\column{E}{@{}>{\hspre}l<{\hspost}@{}}%
\>[B]{}\Conid{VoidH}\mathbin{:}\Conid{RawHFunctor}{}\<[E]%
\\
\>[B]{}\Conid{VoidH}\mathrel{=}\lhskeyword{record}\;{}\<[17]%
\>[17]{}\{\mskip1.5mu \Varid{0}\mathrel{=}\Varid{voidH0}{}\<[E]%
\\
\>[17]{};\Varid{hfmap}\mathrel{=}\Varid{hfmapVoid}{}\<[E]%
\\
\>[17]{};\Varid{hmap}\mathrel{=}\Varid{hmapVoid}\mskip1.5mu\}{}\<[E]%
\\
\>[B]{}\hsindent{3}{}\<[3]%
\>[3]{}\lhskeyword{where}{}\<[E]%
\\[\blanklineskip]%
\>[B]{}\hsindent{3}{}\<[3]%
\>[3]{}\Varid{voidH0}\mathbin{:}\Varid{el}\;\Varid{htyco}{}\<[E]%
\\
\>[B]{}\hsindent{3}{}\<[3]%
\>[3]{}\Varid{voidH0}\;\anonymous \;\anonymous \mathrel{=}\Varid{empty}{}\<[E]%
\\[\blanklineskip]%
\>[B]{}\hsindent{3}{}\<[3]%
\>[3]{}\Varid{hfmapVoid}\mathbin{:}(\Conid{F}\mathbin{:}\Conid{RawFunctor})\to\Varid{tm}\;(\Varid{fmap\hyp{}ty}\;(\Varid{voidH₀}\;(\Varid{0}\;\Conid{F}))){}\<[E]%
\\
\>[B]{}\hsindent{3}{}\<[3]%
\>[3]{}\Varid{hfmapVoid}\;\Conid{F}\;\Varid{α}\;\Varid{β}\;\Varid{f}\;\Varid{x}\mathrel{=}\Varid{x}{}\<[E]%
\\[\blanklineskip]%
\>[B]{}\hsindent{3}{}\<[3]%
\>[3]{}\Varid{hmapVoid}\mathbin{:}(\Conid{F}\;\Conid{G}\mathbin{:}\Conid{RawFunctor})\to\Varid{tm}\;(\Varid{trans}\;(\Varid{0}\;\Conid{F})\;(\Varid{0}\;\Conid{G})){}\<[E]%
\\
\>[3]{}\hsindent{9}{}\<[12]%
\>[12]{}\to\Varid{tm}\;(\Varid{nat}\mathbin{-}\Varid{ty}\;(\Varid{voidH₀}\;(\Varid{0}\;\Conid{F}))\;(\Varid{voidH₀}\;(\Varid{0}\;\Conid{G}))){}\<[E]%
\\
\>[B]{}\hsindent{3}{}\<[3]%
\>[3]{}\Varid{hmapVoid}\;\Conid{F}\;\Conid{G}\;\anonymous \;\Varid{α}\;\Varid{x}{}\<[23]%
\>[23]{}\mathrel{=}\Varid{x}{}\<[E]%
\ColumnHook
\end{hscode}\resethooks

\item Coproduct of higher-order functors
\begin{hscode}\SaveRestoreHook
\column{B}{@{}>{\hspre}l<{\hspost}@{}}%
\column{3}{@{}>{\hspre}l<{\hspost}@{}}%
\column{10}{@{}>{\hspre}l<{\hspost}@{}}%
\column{37}{@{}>{\hspre}l<{\hspost}@{}}%
\column{E}{@{}>{\hspre}l<{\hspost}@{}}%
\>[B]{}\Varid{coprodHF}\mathbin{:}\Conid{RawHFunctor}\to\Conid{RawHFunctor}\to\Conid{RawHFunctor}{}\<[E]%
\\[\blanklineskip]%
\>[B]{}\Varid{coprodHF}\;\Varid{H_1}\;\Varid{H_2}\;\Varid{.0}\mathrel{=}\lambda \Conid{F}\;\Conid{A}.\;{}\<[37]%
\>[37]{}(\Varid{H_1}\;\Varid{.0}\;\Conid{F}\;\Conid{A})\mathbin{+}(\Varid{H_2}\;\Varid{.0}\;\Conid{F}\;\Conid{A}){}\<[E]%
\\[\blanklineskip]%
\>[B]{}\Varid{coprodHF}\;\Varid{H_1}\;\Varid{H_2}\;\Varid{.hfmap}\mathrel{=}\lambda \Conid{F}\;\Varid{α}\;\Varid{β}\;\Varid{f}\;\Varid{x}.{}\<[E]%
\\
\>[B]{}\hsindent{3}{}\<[3]%
\>[3]{}\Varid{case}\;{}\<[10]%
\>[10]{}\{\mskip1.5mu \Varid{c}\mathrel{=}\Varid{H_1}\;\Varid{.0}\;(\Varid{0}\;\Conid{F})\;\Varid{β}\mathbin{+}\Varid{H_2}\;\Varid{.0}\;(\Varid{0}\;\Conid{F})\;\Varid{β}\mskip1.5mu\}\;{}\<[E]%
\\
\>[10]{}(\lambda \Varid{l}.\;\Varid{inl}\;(\Varid{H_1}\;\Varid{.hfmap}\;\Conid{F}\;\Varid{α}\;\Varid{β}\;\Varid{f}\;\Varid{l}))\;{}\<[E]%
\\
\>[10]{}(\lambda \Varid{r}.\;\Varid{inr}\;(\Varid{H_2}\;\Varid{.hfmap}\;\Conid{F}\;\Varid{α}\;\Varid{β}\;\Varid{f}\;\Varid{r}))\;{}\<[E]%
\\
\>[10]{}\Varid{x}{}\<[E]%
\\[\blanklineskip]%
\>[B]{}\Varid{coprodHF}\;\Varid{H_1}\;\Varid{H_2}\;\Varid{.hmap}\mathrel{=}\lambda \Conid{F}\;\Conid{G}\;\Varid{σ}\;\Varid{α}\;\Varid{x}.{}\<[E]%
\\
\>[B]{}\hsindent{3}{}\<[3]%
\>[3]{}\Varid{case}\;{}\<[10]%
\>[10]{}\{\mskip1.5mu \Varid{c}\mathrel{=}\Varid{H_1}\;\Varid{.0}\;(\Varid{0}\;\Conid{G})\;\Varid{α}\mathbin{+}\Varid{H_2}\;\Varid{.0}\;(\Varid{0}\;\Conid{G})\;\Varid{α}\mskip1.5mu\}\;{}\<[E]%
\\
\>[10]{}(\lambda \Varid{l}.\;\Varid{inl}\;(\Varid{H_1}\;\Varid{.hmap}\;\Conid{F}\;\Conid{G}\;\Varid{σ}\;\Varid{α}\;\Varid{l}))\;{}\<[E]%
\\
\>[10]{}(\lambda \Varid{r}.\;\Varid{inr}\;(\Varid{H_2}\;\Varid{.hmap}\;\Conid{F}\;\Conid{G}\;\Varid{σ}\;\Varid{α}\;\Varid{r}))\;{}\<[E]%
\\
\>[10]{}\Varid{x}{}\<[E]%
\ColumnHook
\end{hscode}\resethooks

\item Higher-order functor for an algebraic operation 
\begin{hscode}\SaveRestoreHook
\column{B}{@{}>{\hspre}l<{\hspost}@{}}%
\column{3}{@{}>{\hspre}l<{\hspost}@{}}%
\column{11}{@{}>{\hspre}l<{\hspost}@{}}%
\column{33}{@{}>{\hspre}l<{\hspost}@{}}%
\column{E}{@{}>{\hspre}l<{\hspost}@{}}%
\>[B]{}\Conid{AlgOpHFun}\mathbin{:}\Varid{el}\;\Varid{ty}\to\Varid{el}\;\Varid{ty}\to\Conid{RawHFunctor}{}\<[E]%
\\
\>[B]{}\Conid{AlgOpHFun}\;\Conid{P}\;\Conid{A}\mathrel{=}{}\<[E]%
\\
\>[B]{}\hsindent{3}{}\<[3]%
\>[3]{}\lhskeyword{record}\;{}\<[11]%
\>[11]{}\{\mskip1.5mu \Varid{0}\mathrel{=}\lambda \anonymous \;\Conid{X}.\;{}\<[33]%
\>[33]{}\Conid{P}\;\mathbin{\times_t}\;(\Conid{A}\;{\Rightarrow_t}\;\Conid{X}){}\<[E]%
\\
\>[11]{};\Varid{hfmap}\mathrel{=}\lambda \Conid{F}\;\Varid{α}\;\Varid{β}\;\Varid{f}\;(\Varid{p},\Varid{k}).\;(\Varid{p},(\lambda \Varid{x}.\;\Varid{f}\;(\Varid{k}\;\Varid{x}))){}\<[E]%
\\
\>[11]{};\Varid{hmap}\mathrel{=}\lambda \Conid{F}\;\Conid{G}\;\Varid{σ}\;\Varid{α}\;\Varid{pk}.\;\Varid{pk}\mskip1.5mu\}{}\<[E]%
\ColumnHook
\end{hscode}\resethooks

\item The ML-style signature corresponding to functors
\begin{hscode}\SaveRestoreHook
\column{B}{@{}>{\hspre}l<{\hspost}@{}}%
\column{E}{@{}>{\hspre}l<{\hspost}@{}}%
\>[B]{}\Conid{FctSig}\mathbin{:}\Varid{si}{}\<[E]%
\\
\>[B]{}\Conid{FctSig}\mathrel{=}\Varid{efty},\Varid{fmap\hyp{}ty}{}\<[E]%
\\[\blanklineskip]%
\>[B]{}\Conid{FctToMod}\mathbin{:}\Conid{RawFunctor}\to\Varid{mo}\;\Conid{FctSig}{}\<[E]%
\\
\>[B]{}\Conid{FctToMod}\;\Conid{F}\mathrel{=}(\Varid{0}\;\Conid{F}),(\Varid{fmap}\;\Conid{F}){}\<[E]%
\\[\blanklineskip]%
\>[B]{}\Conid{FctFromMod}\mathbin{:}\Varid{mo}\;\Conid{FctSig}\to\Conid{RawFunctor}{}\<[E]%
\\
\>[B]{}\Conid{FctFromMod}\;(\Conid{F},\Varid{fmap})\mathrel{=}\lhskeyword{record}\;\{\mskip1.5mu \Varid{0}\mathrel{=}\Conid{F};\Varid{fmap}\mathrel{=}\Varid{fmap}\mskip1.5mu\}{}\<[E]%
\ColumnHook
\end{hscode}\resethooks

\item The ML-style signature corresponding of higher-order functors
\begin{hscode}\SaveRestoreHook
\column{B}{@{}>{\hspre}l<{\hspost}@{}}%
\column{3}{@{}>{\hspre}l<{\hspost}@{}}%
\column{7}{@{}>{\hspre}l<{\hspost}@{}}%
\column{8}{@{}>{\hspre}l<{\hspost}@{}}%
\column{12}{@{}>{\hspre}l<{\hspost}@{}}%
\column{13}{@{}>{\hspre}l<{\hspost}@{}}%
\column{17}{@{}>{\hspre}l<{\hspost}@{}}%
\column{26}{@{}>{\hspre}l<{\hspost}@{}}%
\column{E}{@{}>{\hspre}l<{\hspost}@{}}%
\>[B]{}\Conid{HFctSig}\mathbin{:}\Varid{si}{}\<[E]%
\\
\>[B]{}\Conid{HFctSig}\mathrel{=}\Varid{htyco},(\lambda \Conid{H}.\;\Varid{hfmapTy}\;\Conid{H}\;\mathbin{\times_t}\;\Varid{hmapTy}\;\Conid{H})\;\lhskeyword{where}{}\<[E]%
\\
\>[B]{}\hsindent{3}{}\<[3]%
\>[3]{}\Varid{hfmapTy}\mathbin{:}(\Conid{H}\mathbin{:}\Varid{el}\;\Varid{htyco})\to\Varid{el}\;\Varid{ty}{}\<[E]%
\\
\>[B]{}\hsindent{3}{}\<[3]%
\>[3]{}\Varid{hfmapTy}\;\Conid{H}\mathrel{=}\allfor\;\Varid{efty}\;(\lambda \Conid{F}.\;\Varid{fmap\hyp{}ty}\;\Conid{F}\;{\Rightarrow_t}\;\Varid{fmap\hyp{}ty}\;(\Conid{H}\;\Conid{F})){}\<[E]%
\\[\blanklineskip]%
\>[B]{}\hsindent{3}{}\<[3]%
\>[3]{}\Varid{hmapTy}\mathbin{:}(\Conid{H}\mathbin{:}\Varid{el}\;\Varid{htyco})\to\Varid{el}\;\Varid{ty}{}\<[E]%
\\
\>[B]{}\hsindent{3}{}\<[3]%
\>[3]{}\Varid{hmapTy}\;\Conid{H}\mathrel{=}\allfor\;\Varid{efty}\;(\lambda \Conid{F}.\;\Varid{fmap\hyp{}ty}\;\Conid{F}\;{\Rightarrow_t}{}\<[E]%
\\
\>[3]{}\hsindent{4}{}\<[7]%
\>[7]{}\allfor\;\Varid{efty}\;(\lambda \Conid{G}.\;{}\<[26]%
\>[26]{}\Varid{fmap\hyp{}ty}\;\Conid{G}\;{\Rightarrow_t}\;{}\<[E]%
\\
\>[26]{}(\Varid{trans}\;\Conid{F}\;\Conid{G}\;{\Rightarrow_t}\;\Varid{nat\char95 ty}\;(\Conid{H}\;\Conid{F})\;(\Conid{H}\;\Conid{G})))){}\<[E]%
\\[\blanklineskip]%
\>[B]{}\Conid{HFctToMod}\mathbin{:}\Conid{RawHFunctor}\to\Varid{mo}\;\Conid{HFctSig}{}\<[E]%
\\
\>[B]{}\Conid{HFctToMod}\;\Conid{H}\mathrel{=}(\Conid{H}\;\Varid{.0}),{}\<[E]%
\\
\>[B]{}\hsindent{3}{}\<[3]%
\>[3]{}((\lambda \Conid{F}\;\Varid{fmap}.\;\Conid{H}\;\Varid{.hfmap}\;(\Conid{FctFromMod}\;(\Conid{F},\Varid{fmap}))){}\<[E]%
\\
\>[B]{}\hsindent{3}{}\<[3]%
\>[3]{},\lambda {}\<[13]%
\>[13]{}\Conid{F}\;\Varid{fmap₁}\;\Conid{G}\;\Varid{fmap₂}.\;{}\<[E]%
\\
\>[3]{}\hsindent{5}{}\<[8]%
\>[8]{}\Conid{H}\;\Varid{.hmap}\;{}\<[17]%
\>[17]{}(\Conid{FctFromMod}\;(\Conid{F},\Varid{fmap₁}))\;{}\<[E]%
\\
\>[17]{}(\Conid{FctFromMod}\;(\Conid{G},\Varid{fmap₂}))){}\<[E]%
\\[\blanklineskip]%
\>[B]{}\Conid{HFctFromMod}\mathbin{:}\Varid{mo}\;\Conid{HFctSig}\to\Conid{RawHFunctor}{}\<[E]%
\\
\>[B]{}\Conid{HFctFromMod}\;(\Conid{H},(\Varid{hfmap},\Varid{hmap}))\mathrel{=}{}\<[E]%
\\
\>[B]{}\hsindent{3}{}\<[3]%
\>[3]{}\lhskeyword{record}\;{}\<[12]%
\>[12]{}\{\mskip1.5mu \Varid{0}\mathrel{=}\Conid{H}{}\<[E]%
\\
\>[12]{};\Varid{hfmap}\mathrel{=}\lambda \Conid{F}.\;\Varid{hfmap}\;(\Conid{F}\;.\Varid{0})\;(\Conid{F}\;.\Varid{fmap}){}\<[E]%
\\
\>[12]{};\Varid{hmap}\mathrel{=}\lambda \Conid{F}\;\Conid{G}.\;\Varid{hmap}\;(\Conid{F}\;.\Varid{0})\;(\Conid{F}\;.\Varid{fmap})\;(\Conid{G}\;.\Varid{0})\;(\Conid{G}\;.\Varid{fmap})\mskip1.5mu\}{}\<[E]%
\ColumnHook
\end{hscode}\resethooks

\item The family of algebraic operations
\begin{hscode}\SaveRestoreHook
\column{B}{@{}>{\hspre}l<{\hspost}@{}}%
\column{3}{@{}>{\hspre}l<{\hspost}@{}}%
\column{10}{@{}>{\hspre}l<{\hspost}@{}}%
\column{11}{@{}>{\hspre}l<{\hspost}@{}}%
\column{18}{@{}>{\hspre}l<{\hspost}@{}}%
\column{20}{@{}>{\hspre}l<{\hspost}@{}}%
\column{24}{@{}>{\hspre}l<{\hspost}@{}}%
\column{45}{@{}>{\hspre}l<{\hspost}@{}}%
\column{E}{@{}>{\hspre}l<{\hspost}@{}}%
\>[B]{}\Conid{AlgSig}\mathbin{:}\Varid{el}\;(\Varid{list_k}\;(\Varid{ty}\;{\Rightarrow_k}\;\Varid{ty}))\to\Conid{RawHFunctor}{}\<[E]%
\\
\>[B]{}\Conid{AlgSig}\;\Varid{es}\mathrel{=}\Conid{HFctFromMod}{}\<[E]%
\\
\>[B]{}\hsindent{3}{}\<[3]%
\>[3]{}(\Varid{fold_{kt}}\;\{\mskip1.5mu \Varid{a}\mathrel{=}\Conid{HFctSig}\mskip1.5mu\}\;{}\<[E]%
\\
\>[3]{}\hsindent{8}{}\<[11]%
\>[11]{}(\lhskeyword{record}\;{}\<[20]%
\>[20]{}\{\mskip1.5mu \Varid{fst}\mathrel{=}\Conid{HFctToMod}\;\Conid{VoidH}{}\<[E]%
\\
\>[20]{};\Varid{snd}\mathrel{=}\lambda (\Conid{P},\Conid{A})\;\Conid{H}.{}\<[E]%
\\
\>[20]{}\hsindent{4}{}\<[24]%
\>[24]{}\Conid{HFctToMod}\;(\Varid{coprodHF}\;{}\<[45]%
\>[45]{}(\Conid{AlgOpHFun}\;\Conid{P}\;\Conid{A})\;{}\<[E]%
\\
\>[45]{}(\Conid{HFctFromMod}\;\Conid{H}))\mskip1.5mu\})\;{}\<[E]%
\\
\>[3]{}\hsindent{8}{}\<[11]%
\>[11]{}\Varid{es}){}\<[E]%
\\[\blanklineskip]%
\>[B]{}\Varid{ListApp_k}\mathbin{:}\{\mskip1.5mu \Varid{k}\mathbin{:}\Varid{ki}\mskip1.5mu\}\to\Varid{el}\;(\Varid{list_k}\;\Varid{k})\to\Varid{el}\;(\Varid{list_k}\;\Varid{k})\to\Varid{el}\;(\Varid{list_k}\;\Varid{k}){}\<[E]%
\\
\>[B]{}\Varid{ListApp_k}\;\{\mskip1.5mu \Varid{k}\mskip1.5mu\}\;\Varid{x}\;\Varid{y}\mathrel{=}{}\<[E]%
\\
\>[B]{}\hsindent{3}{}\<[3]%
\>[3]{}\Varid{fold_k}\;{}\<[10]%
\>[10]{}\{\mskip1.5mu \Varid{a}\mathrel{=}\Varid{list_k}\;\Varid{k}\mskip1.5mu\}\;{}\<[E]%
\\
\>[10]{}(\lhskeyword{record}\;\{\mskip1.5mu \Varid{fst}\mathrel{=}\Varid{y};\Varid{snd}\mathrel{=}\Varid{cons_k}\mskip1.5mu\})\;\Varid{x}{}\<[E]%
\\[\blanklineskip]%
\>[B]{}\Varid{algFam}\mathbin{:}\Conid{Fam}{}\<[E]%
\\
\>[B]{}\Varid{algFam}\mathrel{=}\lhskeyword{record}\;{}\<[18]%
\>[18]{}\{\mskip1.5mu \Varid{eff}\mathrel{=}\Varid{list_k}\;(\Varid{ty}\;\mathbin{\times_k}\;\Varid{ty}){}\<[E]%
\\
\>[18]{};\Varid{sig}\mathrel{=}\Conid{AlgSig}{}\<[E]%
\\
\>[18]{};\Varid{add}\mathrel{=}\Varid{ListApp_k}\mskip1.5mu\}{}\<[E]%
\ColumnHook
\end{hscode}\resethooks
\end{itemize}

\sethscode{autohscode}

\section{Equations of Computations for the Realizability Model}
\label{app:real:model:laws}

In \Cref{sec:real:mod:fha} we defined a model of \Fha{} in the language of
assemblies (\Cref{lang:asm}).
This appendix shows that the equational laws of \Fha{} are validated by the
definitions in \Cref{sec:real:mod:fha}.

The monadic laws of computations \Cref{eq:fha:co:laws} are satisfied:
\begin{itemize}[itemsep=0pt]
\item
For \ensuremath{\Varid{val\hyp{}let}}, given any $a : \RM.\ensuremath{\Varid{tm}\;\Conid{A}}$ and $k : \RM.\ensuremath{\Varid{tm}\;\Conid{A}} \to \RM.\ensuremath{\Varid{co}\;\Conid{H}\;\Conid{B}}$,
\begin{align*}
    & \ensuremath{\Varid{let\hyp{}in}\;(\Varid{val}\;\Varid{a})\;\Varid{k}} \\
={} & \reason{by definition of $\RM.\ensuremath{\Varid{let\hyp{}in}}$} \\
    & \lambda h \; C \; r.\; \ensuremath{\Varid{val}\;\Varid{a}\;\Varid{h}\;\Conid{C}\;(\lambda \Varid{a}.\;\Varid{k}\;\Varid{a}\;\Varid{h}\;\Conid{C}\;\Varid{r})} \\
={} & \reason{by definition of $\RM.\ensuremath{\Varid{val}}$} \\
    & \lambda h \; C \; r.\; \ensuremath{\Varid{k}\;\Varid{a}\;\Varid{h}\;\Conid{C}\;\Varid{r}} \\
={} & \reason{$\eta$-rule for functions} \\
    & \ensuremath{\Varid{k}\;\Varid{a}}
\end{align*}

\item
The case for \ensuremath{\Varid{let\hyp{}val}} is very similar. 
Given any $c : \RM.\ensuremath{\Varid{co}\;\Conid{H}\;\Conid{A}}$,
\begin{align*}
    & \ensuremath{\Varid{let\hyp{}in}\;\Varid{c}\;\Varid{val}} \\
={} & \reason{by definition of $\RM.\ensuremath{\Varid{let\hyp{}in}}$} \\
    & \lambda h \; C \; r.\; \ensuremath{\Varid{c}\;\Varid{h}\;\Conid{C}\;(\lambda \Varid{a}.\;\Varid{val}\;\Varid{a}\;\Varid{h}\;\Conid{C}\;\Varid{r})} \\
={} & \reason{by definition of $\RM.\ensuremath{\Varid{val}}$} \\
    & \lambda h \; C \; r.\; \ensuremath{\Varid{c}\;\Varid{h}\;\Conid{C}\;(\lambda \Varid{a}.\;\Varid{r}\;\Varid{a})} \\
={} & \reason{$\eta$-rule for functions} \\
    & c
\end{align*}

\item
For \ensuremath{\Varid{let\hyp{}assoc}}, given any $c_1$, $c_2$, and $c_3$, we have
\begin{align*}
    & \ensuremath{\Varid{let\hyp{}in}\;(\Varid{let\hyp{}in}\;\Varid{c}_{\mathrm{1}}\;\Varid{c}_{\mathrm{2}})\;\Varid{c}_{\mathrm{3}}}  \\
={} & \ensuremath{\lambda \Varid{h}\;\Conid{C}\;\Varid{r}.\;(\Varid{let\hyp{}in}\;\Varid{c}_{\mathrm{1}}\;\Varid{c}_{\mathrm{2}})\;\Varid{h}\;\Conid{C}\;(\lambda \Varid{b}.\;\Varid{c}_{\mathrm{3}}\;\Varid{b}\;\Varid{h}\;\Conid{C}\;\Varid{r})} \\
={} & \ensuremath{\lambda \Varid{h}\;\Conid{C}\;\Varid{r}.\;\Varid{c}_{\mathrm{1}}\;\Varid{h}\;\Conid{C}\;(\lambda \Varid{a}.\;\Varid{c}_{\mathrm{2}}\;\Varid{a}\;\Varid{h}\;\Conid{C}\;(\lambda \Varid{b}.\;\Varid{c}_{\mathrm{3}}\;\Varid{b}\;\Varid{h}\;\Conid{C}\;\Varid{r}))} \\
={} & \ensuremath{\lambda \Varid{h}\;\Conid{C}\;\Varid{r}.\;\Varid{c}_{\mathrm{1}}\;\Varid{h}\;\Conid{C}\;(\lambda \Varid{a}.\;\Varid{let\hyp{}in}\;(\Varid{c}_{\mathrm{2}}\;\Varid{a})\;\Varid{c}_{\mathrm{3}})} \\
={} & \ensuremath{\Varid{let\hyp{}in}\;\Varid{c}_{\mathrm{1}}\;(\lambda \Varid{a}.\;\Varid{let\hyp{}in}\;(\Varid{c}_{\mathrm{2}}\;\Varid{a})\;\Varid{c}_{\mathrm{3}})} 
\end{align*}
\end{itemize}

Now we check that the equation \ensuremath{\Varid{hdl\hyp{}val}} \Cref{eq:fha:evalval} is
satisfied: for all $H : \RM.\ensuremath{\Conid{RawHFunctor}}$, $A : \RM.\ensuremath{\Varid{el}} \; \RM.\ensuremath{\Varid{ty}}$, $h :
\RM.\ensuremath{\Conid{Handler}\;\Conid{H}}$ and $a : A$,
\begin{align*}
    & \RM.\ensuremath{\Varid{hdl}} \; h \; (\RM.\ensuremath{\Varid{val}} \; a) \\
={} & \reason{by definition of $\RM.\ensuremath{\Varid{hdl}}$} \\
    & \RM.\ensuremath{\Varid{val}} \; a \; h\; A\; h.\ensuremath{\Varid{ret}} \\
={} & \reason{by definition of $\RM.\ensuremath{\Varid{val}}$} \\
    & (\lambda h\; B\; r.\; r \; a) \; h \; A \; h.\ensuremath{\Varid{ret}}  \\
={} & h.\ensuremath{\Varid{ret}} \; a
\end{align*}

The model of operations is defined as follows:
\begin{lgather*}
\RM.\ensuremath{\Varid{op}} : \impfun{H, A, B}  H \;(\ensuremath{\Varid{th}\;\Conid{H}})\; A \to (A \to \RM.\ensuremath{\Varid{co}} \; H \; B) \to \RM.\ensuremath{\Varid{co}}\;H \; B \\
\RM.\ensuremath{\Varid{op}} \; o \; k = \lambda h \; C \; r.\;\\
\quad h.\ensuremath{\Varid{bind}} \; A \; C \; \\
\Spc (h.\ensuremath{\Varid{malg}} \; A \; (H.\ensuremath{\Varid{hmap}} \; (\ensuremath{\Varid{th}\;\Conid{H}})\; h \; (\RM.\ensuremath{\Varid{hdl}}\;h) \; A \; o)) \\
\Spc (\lambda a.\; k \; a \; h \; C \; r)
\end{lgather*}  
It remains to check that the equations \ensuremath{\Varid{let\hyp{}op}}
\Cref{eq:fha:letop} and \ensuremath{\Varid{hdl\hyp{}op}} \Cref{eq:fha:evalop} are satisfied.

For \ensuremath{\Varid{let\hyp{}op}}, given arbitrary $o : \ensuremath{\Conid{H}\;(\Varid{th}\;\Conid{H})\;\Conid{A}}$, $\ensuremath{\Varid{k}\mathbin{:}\Conid{A}\to \Varid{co}\;\Conid{H}\;\Conid{B}}$, $\ensuremath{\Varid{k'}\mathbin{:}\Conid{B}\to \Varid{co}\;\Conid{H}\;\Conid{C}}$,
\begin{align*}
    & \ensuremath{\Varid{let\hyp{}in}\;(\Varid{op}\;\Varid{o}\;\Varid{k})\;\Varid{k'}} \\
={} & \reason{by definition of $\RM.\ensuremath{\Varid{let\hyp{}in}}$} \\
    & \lambda h \; C\; r.\  \ensuremath{(\Varid{op}\;\Varid{o}\;\Varid{k})\;\Varid{h}\;\Conid{C}} \; (\lambda b.\; \ensuremath{\Varid{k'}}\; b \; h \; C \; r)\\ 
={} & \reasonMult{by definition of $\RM.\ensuremath{\Varid{op}}$ and let \ensuremath{\Varid{o'}} be \\
       $\quad h.\ensuremath{\Varid{malg}} \; \_ \; (H.\ensuremath{\Varid{hmap}} \; \_ \; \_ \; (\RM.\ensuremath{\Varid{hdl}}\;h) \; \_ \; o)$} \numberthis\label{eq:oprime:def} \\
    & h.\ensuremath{\Varid{bind}} \; \_ \; \_ \; \ensuremath{\Varid{o'}} \; (\lambda a.\; k \; a \; h \; \_ \; (\lambda b.\; \ensuremath{\Varid{k'}\;\Varid{b}\;\Varid{h}\;\anonymous \;\Varid{r}})) \\
={} & \reason{by definition of $\RM.\ensuremath{\Varid{let\hyp{}in}\;(\Varid{k}\;\Varid{a})\;\Varid{k'}}$} \\
    & \ensuremath{\Varid{op}\;\Varid{o}\;(\lambda \Varid{a}.\;\Varid{let\hyp{}in}\;(\Varid{k}\;\Varid{a})\;\Varid{k'})}
\end{align*}

For \ensuremath{\Varid{hdl\hyp{}op}}, given any $h : \ensuremath{\Conid{Handler}\;\Conid{H}}$, $o : \ensuremath{\Conid{H}\;(\Varid{th}\;\Conid{H})\;\Conid{A}}$ and $\ensuremath{\Varid{k}\mathbin{:}\Conid{A}\to \Varid{co}\;\Conid{H}\;\Conid{B}}$,
\begin{align*}
    &  \ensuremath{\Varid{hdl}\;\Varid{h}\;(\Varid{op}\;\Varid{o}\;\Varid{k})} \\
={} &  \reason{by definition of $\RM.\ensuremath{\Varid{hdl}}$} \\ 
    & \ensuremath{(\Varid{op}\;\Varid{o}\;\Varid{k})\;\Varid{h}\;\anonymous } \; h.\ensuremath{\Varid{ret}}\\ 
={} & \reason{by definition of $\RM.\ensuremath{\Varid{op}}$ and let \ensuremath{\Varid{o'}} be the same as in \Cref{eq:oprime:def}} \\
    & h.\ensuremath{\Varid{bind}} \; \_ \; \_ \; \ensuremath{\Varid{o'}} \; (\lambda a.\; \ensuremath{\Varid{k}\;\Varid{a}\;\Varid{h}\;\anonymous } \; h.\ensuremath{\Varid{ret}}) \\
={} &  \reason{by definition of $\RM.\ensuremath{\Varid{hdl}} \; h \; \_ \; \ensuremath{\Varid{k}\;\Varid{a}}$} \\ 
    & h.\ensuremath{\Varid{bind}} \; \_ \; \_ \; \ensuremath{\Varid{o'}} \; 
       (\lambda a.\; \ensuremath{\Varid{hdl}\;\Varid{h}\;\anonymous \;(\Varid{k}\;\Varid{a})})
\end{align*}

\section{The Synthetic Logical Relation Model of \Fha}
\label{app:lr:model}

In this appendix, we define the logical relation model of \Fha{} in detail.
Let us start with two useful lemmas that we did not include in the main text.

The first of them says that we can not only glue but also tear types apart.
Given any type $A : U$ in \TTSTC{}, we can tear it to an
object-space fragment $A^\circ$ and a meta-space fragment $A^\bullet$:
\begin{align*}
&A^\circ : \Omod U      &   & A^\bullet : (\impfun{\sop} A) \to \CU \\
&A^\circ = \Oeta_U \; A &   & A^\bullet = \lambda o.\; \extty{A}{\sop}{o} 
\end{align*}
where $\CU \defeq \aset{A : U \mid \Cmodal \; A}$ is the subuniverse of
$\Cmod$-modal types.
The type $\extty{A}{\sop}{a}$ is $\Cmod$-modal because it is a singleton under
$\sop$ (\Cref{lem:cmod:if:unit}).

\begin{lemma}
For every type $A : U$ in \TTSTC{}, there is an isomorphism $A \cong \gluety{o : A^\circ}{A^\bullet\;o}$.
\end{lemma}

\begin{proof}
The two directions of the isomorphism are
\begin{align*}
&\ensuremath{\Varid{fwd}} : A \to \gluety{o : A^\circ}{A^\bullet\;o}  && \ensuremath{\Varid{bwd}} : (\gluety{o : A^\circ}{A^\bullet\;o}) \to A  \\
&\ensuremath{\Varid{fwd}} \; a = \gluetm{a}{\impabs{\_ : \sop} a} && \ensuremath{\Varid{bwd}} \; \gluetm{c}{o} = c
\end{align*}
These two functions are indeed mutual inverses:
for all $a : A$,
\[\ensuremath{\Varid{bwd}}\; (\,\ensuremath{\Varid{fwd}\;\Varid{a}}\,) = \ensuremath{\Varid{bwd}}\;\gluetm{a}{\impabs{\_ : \sop} a} = a;\]
for all $\gluetm{c}{o}$, by definition
$\ensuremath{\Varid{fwd}} \; (\ensuremath{\Varid{bwd}} \; \gluetm{c}{o}) = \gluetm{c}{c}$,
but $c$ has type $A^\bullet \defeq \extty{c}{\sop}{o}$, so $\gluetm{c}{c}
= \gluetm{c}{o}$.
\end{proof}

Since every type of \TTSTC{} is isomorphic to a glue type, 
we can characterise function types of \TTSTC{} more extrinsically,
which explicitises the idea that a map between logical predicates
sends (proofs for) related input to (proofs for) related output.

\begin{lemma}\label{lem:fun:between:glue:tys}
For all universes $U$ of \TTSTC{}, there is an isomorphism \ensuremath{\Varid{\gluetysymbol\hyp{}fun\hyp{}iso}}:
\begin{gather*}
\big(\gluety{a : A}{P \; a}\big) \to \big(\gluety{b : B}{Q \; b}\big) 
\ \cong \ 
\gluety{f : A \to B}{\big( (a : \impfun{\sop} A) \to P \; a \to Q \; (f\; a)\big)}
\end{gather*}
for all $A, B : \Omod{U}$, $P : (\impfun{\sop}A) \to \CU$, and
$Q : (\impfun{\sop} B) \to \CU$, where $\CU$
is the subuniverse $\aset{A : U \mid \Cmodal \; A}$ of $\Cmod$-modal types.
\end{lemma}
\begin{proof}
The two directions of the isomorphism are
\begin{lgather*}
\ensuremath{\Varid{fwd}} \  g = \gluetm{\lambda a\; p.\; \unglue (g \; \gluetm{p}{a}) }{\lambda a.\; g \; a} \\
\ensuremath{\Varid{bwd}} \  \gluetm{h}{f} \  \gluetm{p}{a} = \gluetm{h\;a\;p}{f\; a}
\end{lgather*}
It is routine calculation to check that these two directions are mutual inverses.
\begin{align*}
   & \ensuremath{\Varid{bwd}}\;(\ensuremath{\Varid{fwd}}\;g) \\
=\ & \ensuremath{\Varid{bwd}}\;(\gluetm{\lambda a\; p.\; \unglue (g \; \gluetm{p}{a}) }{\lambda a.\; g \; a}) \\
=\ & \lambda \gluetm{p}{a}.\; \gluetm{\unglue (g \; \gluetm{p}{a})}{ g\;a }\\
=\ & \lambda a^*.\; g\;a^* \\
=\ & g \\[8pt]
   & \ensuremath{\Varid{fwd}} \; (\ensuremath{\Varid{bwd}}\; \gluetm{h}{f}) \\
=\ & \ensuremath{\Varid{fwd}} \; (\lambda \gluetm{p}{a}.\; \gluetm{h\;a\;p}{f\; a}) \\
=\ & \gluetm{ \lambda a\;p.\; \unglue \gluetm{h\;a\;p}{f\; a} }{\lambda a.\; f \; a} \\
=\ & \gluetm{h}{f} \qedhere
\end{align*}
\end{proof}

Now we come back to define the logical relation model:
\begin{equation}
\MM^* : \extty{\sem{\Fha}_{U_2}}{\sop}{\MM}
\end{equation}

\begin{notation}
In the rest of this section, for every declaration \ensuremath{\Varid{dec}} in the signature of \Fha,
we will write $\ensuremath{\Varid{dec}}^*$ for $\MM^*.\ensuremath{\Varid{dec}}$ and just \ensuremath{\Varid{dec}} for $\MM.\ensuremath{\Varid{dec}}$.
For example, $\ensuremath{\Varid{ki}}^* : \extty{U_2}{\sop}{\ensuremath{\Varid{ki}}}$
means $\MM^*.\ensuremath{\Varid{ki}} : \extty{U_2}{\sop}{\ensuremath{\Varid{ki}}}$.
\end{notation}

The logical predicate model of the judgement of kinds \Cref{eq:fha:kinds} is
\begin{lgather}{eq:m:star:ki}
\ensuremath{\Varid{ki}}^* : \extty{U_2}{\sop}{\ensuremath{\Varid{ki}}} \\
\ensuremath{\Varid{ki}}^* = \gluety{\alpha : \ensuremath{\Varid{ki}}}{\extty{U_1}{\sop}{\ensuremath{\Varid{el}} \; \alpha}} 
\end{lgather}
This uses the glue type (\ref{el:ttstc:glue:ty}) correctly because 
the generic model $\MM$ (\ref{el:ttstc:obj}) has type $\impfun{\sop}
\sem{\Fha}_{U_0}$, so the type of $\ensuremath{\Varid{ki}}$, or more explicitly $\impabs{z : \sop} (\MM
\impapp{z}).\ensuremath{\Varid{ki}}$, is $\impfun{\sop} U_0$, i.e.\ $\Omod U_0$.
The type $\extty{U_1}{\sop}{\ensuremath{\Varid{el}} \; \alpha}$ is $\Cmod$-modal because 
when $\sop$ holds, all elements of $\extty{U_1}{\sop}{\ensuremath{\Varid{el}} \; \alpha}$ are equal to 
$\ensuremath{\Varid{el}} \; \alpha$, so the type $\extty{U_1}{\sop}{\ensuremath{\Varid{el}} \; \alpha}$ has exactly
one element so isomorphic to the unit $\unitty$.
By \Cref{lem:cmod:if:unit}, the type $\extty{U_1}{\sop}{\ensuremath{\Varid{el}} \; \alpha}$ is $\Cmod$-modal.

More intuitively,
the definition \Cref{eq:m:star:ki} is the \emph{proof-relevant} logical predicate
for kinds.
A proof for a kind $\alpha : \ensuremath{\Varid{ki}}$ satisfying the
predicate is a type $A : U_1$ that restrict to $\ensuremath{\Varid{el}} \; \alpha$ in the object space.
Such a type $A$ is a `candidate' for the logical predicate for the kind
$\alpha$.
This is the same idea as \emph{reducibility candidates} in \citeauthor{Girard_1989}{'s}
proof of strong normalisation of System F \citep{Girard_1989}.

In accordance, the corresponding $\MM^*.\ensuremath{\Varid{el}}$ is as follows:
\begin{lgather}{eq:m:star:el}
\ensuremath{\Varid{el}}^* : \extty{\ensuremath{\Varid{ki}}^* \to U_1}{\sop}{\ensuremath{\Varid{el}}}  \\
\ensuremath{\Varid{el}}^* \; g = \unglue g
\end{lgather}
Let us more carefully examine how this definition type checks:
the argument $g$ has type $\ensuremath{\Varid{ki}}^* = \gluety{\alpha : \ensuremath{\Varid{ki}}}{\extty{U_1}{\sop}{\ensuremath{\Varid{el}} \; \alpha}}$.
Therefore $\unglue g$ has type $\extty{U_1}{\sop}{\ensuremath{\Varid{el}} \; g}$ (note that
under $\sop$, the type of $g$ is strictly equal to the type $\ensuremath{\Varid{ki}}$, thus
it makes sense to write $\ensuremath{\Varid{el}} \; g$ in a context where $\sop$ holds).
Thus $\ensuremath{\Varid{el}}^*$ is indeed a function $\ensuremath{\Varid{ki}}^* \to U_1$ that strictly restricts
to $\ensuremath{\Varid{el}}$ under $\sop$.

For kind-level functions, we need to define
\[
\ensuremath{\text{\textunderscore}{\Rightarrow_k}\text{\textunderscore}}^* : \extty{\ensuremath{\Varid{ki}}^* \to \ensuremath{\Varid{ki}}^* \to \ensuremath{\Varid{ki}}^*}{\sop}{\ensuremath{\text{\textunderscore}{\Rightarrow_k}\text{\textunderscore}}}
\]
Let us derive the definition step-by-step.
Our goal is to fill the hole $\hole{?0}$ in
\begin{equation*}
\gluetm{A}{\alpha} \;\mathbin{\ensuremath{{\Rightarrow^*_k}}}\; \gluetm{B}{\beta}
 \; = \; \hole{?0 : \extty{\ensuremath{\Varid{ki}}^*}{\sop}{\alpha \mathbin{\ensuremath{{\Rightarrow_k}}} \beta}},
\end{equation*}
where the variables in context have the following types
\begin{equation}\label{eq:ki:ab:types}
\spread{
\alpha, \beta : \impfun{\sop}\ensuremath{\Varid{ki}}
\also
A : \extty{U_1}{\sop}{\ensuremath{\Varid{el}} \; \alpha}
\also
B : \extty{U_1}{\sop}{\ensuremath{\Varid{el}} \; \beta}.
}
\end{equation}
Since $\ensuremath{\Varid{ki}}^*$ is a glue type \Cref{eq:m:star:ki}, we can use the term former of glue types:
\[
\gluetm{A}{\alpha} \mathbin{\ensuremath{{\Rightarrow^*_k}}} \gluetm{B}{\beta}
 \; = \; \gluetm{\hole{?2}}{\hole{?1}}
\]
Since \hole{?0} must restrict to $\alpha \mathbin{\ensuremath{{\Rightarrow_k}}} \beta$ under $\sop$,
\hole{?1} has to be $\alpha \mathbin{\ensuremath{{\Rightarrow_k}}} \beta$:
\[
\gluetm{A}{\alpha} \mathbin{\ensuremath{{\Rightarrow^*_k}}} \gluetm{B}{\beta}
 \; = \; \gluetm{\hole{?2}}{ \alpha \mathbin{\ensuremath{{\Rightarrow_k}}} \beta }
\]
The hole \hole{?2} now has type $\extty{U_1}{\sop}{\ensuremath{\Varid{el}} \; (\alpha \mathbin{\ensuremath{{\Rightarrow_k}}} \beta)}$;
in other words, \hole{?2} is a type in $U_1$ such that it restricts to $\ensuremath{\Varid{el}} \; (\alpha \mathbin{\ensuremath{{\Rightarrow_k}}} \beta)$
when $\sop$ holds.
We again use the glue type to satisfy the restriction:
\[
\hole{?2} \defeq \gluety{f : \ensuremath{\Varid{el}}\;(\alpha \mathbin{\ensuremath{{\Rightarrow_k}}} \beta)}{ \hole{?3} }.
\]
Conceptually, \hole{?2} is the logical predicate for the function kind $\alpha \mathbin{\ensuremath{{\Rightarrow_k}}} \beta$.
Readers experienced with traditional logical relations might expect $\hole{?3}$
to be the proposition asserting that $f$ sends input $a : \ensuremath{\Varid{el}} \; \alpha$
satisfying the logical predicate $A$ to output $f\;a : \ensuremath{\Varid{el}} \; \beta$
satisfying logical predicate $B$.
However, here the predicates $A$ and $B$ are proof-relevant, so the correct definition
of $\hole{?3}$ should be the \emph{type} of functions sending proofs for $a : \ensuremath{\Varid{el}}\;\alpha$ satisfying $A$
to proofs for $f \; a : \ensuremath{\Varid{el}}\;\beta$ satisfying $B$.
This can be concisely expressed in \TTSTC{} as 
\[
\hole{?3} \defeq \extty{A \to B}{\sop}{\ensuremath{\Varid{{\Rightarrow_k}\hyp{}iso}}.\ensuremath{\Varid{fwd}} \; f}
\]
where $\ensuremath{\Varid{{\Rightarrow_k}\hyp{}iso}}$ is the isomorphism in \Fha{} specifying function kinds:
\[\ensuremath{\Varid{{\Rightarrow_k}\hyp{}iso}} : \ensuremath{\Varid{el}} \; (\alpha \mathbin{\ensuremath{{\Rightarrow_k}}} \beta) \cong (\ensuremath{\Varid{el}} \; \alpha \to \ensuremath{\Varid{el}} \; \beta).\]
The function type $A \to B$ in
\TTSTC{} is translated to exponentials in the glued topos $\CatG$, which
takes care of `sending related input to related output' by construction.

For the record, we have completed our initial goal $\ensuremath{\text{\textunderscore}{\Rightarrow_k}\text{\textunderscore}}^*$:
\begin{lgather}{eq:m:star:kfun:def}
\ensuremath{\text{\textunderscore}{\Rightarrow_k}\text{\textunderscore}}^* : \extty{\ensuremath{\Varid{ki}}^* \to \ensuremath{\Varid{ki}}^* \to \ensuremath{\Varid{ki}}^*}{\sop}{\ensuremath{\text{\textunderscore}{\Rightarrow_k}\text{\textunderscore}}}\\
\gluetm{A}{\alpha} \mathbin{\ensuremath{{\Rightarrow^*_k}}} \gluetm{B}{\beta} = \gluetm{ F }{\alpha \mathbin{\ensuremath{{\Rightarrow_k}}} \beta}
\end{lgather}
where $F$ is the logical predicate for the function kind $\alpha \mathbin{\ensuremath{{\Rightarrow_k}}} \beta$:
\begin{equation}\label{eq:kfun:pred}
F \defeq \gluety{f : \ensuremath{\Varid{el}}\;(\alpha \mathbin{\ensuremath{{\Rightarrow_k}}} \beta)}{\extty{A \to B}{\sop}{\ensuremath{\Varid{{\Rightarrow_k}\hyp{}iso}}.\ensuremath{\Varid{fwd}} \; f}}.
\end{equation}

We also need to exhibit the isomorphism $\ensuremath{\Varid{{\Rightarrow_k}\hyp{}iso}}$ $\Cref{eq:fha:el:kfun}$ for $\MM^*$:
\begin{lgather*}
\ensuremath{\Varid{{\Rightarrow_k}\hyp{}iso}}^* : \impfun{\alpha^*, \beta^* : \ensuremath{\Varid{ki}}^*} 
  \extty{\ensuremath{\Varid{el}}^* \; (\alpha^* \mathbin{\ensuremath{{\Rightarrow^*_k}}} \beta^*) \cong (\ensuremath{\Varid{el}}^* \; \alpha^* \to \ensuremath{\Varid{el}}^* \; \beta^*)}{\sop}{\ensuremath{\Varid{{\Rightarrow_k}\hyp{}iso}}}. 
\end{lgather*}
Again by pattern matching the input $\alpha^*$ and $\beta^*$ as $\gluetm{A}{\alpha}$ and $\gluetm{B}{\beta}$ as in \Cref{eq:ki:ab:types},
after expanding out the definition of $\ensuremath{\Varid{el}}^*$, what we need to construct is an isomorphism
$
F \cong A \to B
$
that restricts to $\ensuremath{\Varid{{\Rightarrow_k}\hyp{}iso}}$ under $\sop$, where $F$ is defined as in \Cref{eq:kfun:pred}.
We let the two directions of this isomorphism be
\[
\spread*{
\ensuremath{\Varid{fwd}} \; \gluetm{g}{f} = g
\also
\ensuremath{\Varid{bwd}} \; h = \gluetm{h}{\ensuremath{\Varid{{\Rightarrow_k}\hyp{}iso}}.\ensuremath{\Varid{bwd}} \; h}
}
\]
where $h : A \to B$, $f : \impfun{\sop} \ensuremath{\Varid{el}}\;(\alpha \mathbin{\ensuremath{{\Rightarrow_k}}} \beta)$,
and 
\[g : \extty{A \to B}{\sop}{\ensuremath{\Varid{{\Rightarrow_k}\hyp{}iso}}.\ensuremath{\Varid{fwd}} \; f}.\]
These two functions are mutual inverses because
\[
\ensuremath{\Varid{fwd}}\;(\ensuremath{\Varid{bwd}}\;h) = \ensuremath{\Varid{fwd}} \; (\gluetm{h}{\ensuremath{\Varid{{\Rightarrow_k}\hyp{}iso}}.\ensuremath{\Varid{bwd}} \; h}) = h
\]
and from the other direction,
\[
  \ensuremath{\Varid{bwd}}\;(\,\ensuremath{\Varid{fwd}}\; \gluetm{g}{f}\,) 
= \ensuremath{\Varid{bwd}}\;g 
= \gluetm{g}{\ensuremath{\Varid{{\Rightarrow_k}\hyp{}iso}}.\ensuremath{\Varid{bwd}} \; g};
\]
now by the type of $g$, $g = (\ensuremath{\Varid{{\Rightarrow_k}\hyp{}iso}}.\ensuremath{\Varid{fwd}} \; f)$ under $\sop$, so the 
above further equals
\[
 \gluetm{g}{\ensuremath{\Varid{{\Rightarrow_k}\hyp{}iso}}.\ensuremath{\Varid{bwd}} \; (\ensuremath{\Varid{{\Rightarrow_k}\hyp{}iso}}.\ensuremath{\Varid{fwd}} \; f)} 
=  \gluetm{g}{f}.
\]

The definition \Cref{eq:kfun:pred} of the logical predicate $F$ for function
kinds may look complicated at first, but it has a very intuitive
explanation:
$F$ is basically the same as the type $A \to B$, except that its component in 
the object space, which is equal to $\ensuremath{\Varid{el}} \; \alpha \to \ensuremath{\Varid{el}}\;\beta$,
is swapped for $\ensuremath{\Varid{el}}\;(\alpha \mathbin{\ensuremath{{\Rightarrow_k}}} \beta)$ along the isomorphism
$\ensuremath{\Varid{{\Rightarrow_k}\hyp{}iso}}$, just like in the old days when a component of a personal computer could be
replaced by a compatible part.
This will be a recurring construction in the future, so for every universe $U$
we define
\begin{lgather*}
\ensuremath{\Varid{realign}} : (A : U) \to (B : \impfun{\sop} U) \to (\impfun{\sop} B \cong A)
  \to \extty{U}{\sop}{B} \\
\ensuremath{\Varid{realign}} \; A \; B \; \phi = \gluety{b : B}{\extty{A}{\sop}{\phi.\ensuremath{\Varid{fwd}}\;b}} \\[8pt]
\ensuremath{\Varid{realign\hyp{}iso}} : (A : U) \to (B : \impfun{\sop} U) \to (\phi : \impfun{\sop} B \cong A) \\
\Spc\Spc  \to \extty{\ensuremath{\Varid{realign}} \; A \; B \; \phi \cong A}{\sop}{\phi} \\
(\ensuremath{\Varid{realign\hyp{}iso}} \; A \; B \; \phi).\ensuremath{\Varid{fwd}}\;\gluetm{a}{b} = a \\
(\ensuremath{\Varid{realign\hyp{}iso}} \; A \; B \; \phi).\ensuremath{\Varid{bwd}}\;a = \gluetm{a}{\phi.\ensuremath{\Varid{bwd}}\;a}
\end{lgather*}

This construction is called \emph{realignment} \citep[\S 3.3]{sterling:2021:thesis}
on the universe $U$. In fact, realignment and strict glue types (\Cref{el:ttstc:glue:ty}) are inter-definable:
if we take \ensuremath{\Varid{realign}} and \ensuremath{\Varid{realign\hyp{}iso}} as axioms, we can define strict glue types $\gluety{a : A}{B}$
by realigning the dependent pair type $\Sigma (a : A).\; B$.

Using realignment, the definition \Cref{eq:m:star:kfun:def} can be succinctly expressed as
\[
\gluetm{A}{\alpha} \mathbin{\ensuremath{{\Rightarrow^*_k}}} \gluetm{B}{\beta} = 
\gluetm{\ensuremath{\Varid{realign}}\;(A \to B)\;\ensuremath{\Varid{{\Rightarrow_k}\hyp{}iso}}}{\alpha \mathbin{\ensuremath{{\Rightarrow_k}}} \beta}
\]
and $\ensuremath{\Varid{{\Rightarrow_k}\hyp{}iso}}^*$ is simply $\ensuremath{\Varid{realign\hyp{}iso}} \; (A \to B)\;\ensuremath{\Varid{{\Rightarrow_k}\hyp{}iso}}$.

\begin{element}
We move on to the logical predicates for types and terms.
Similar to function kinds, $\ensuremath{\Varid{ty}}^*$ is $\ensuremath{\Varid{ty}}$ glued together with some additional data:
\begin{lgather*}
\ensuremath{\Varid{ty}}^* : \extty{\ensuremath{\Varid{ki}}^*}{\sop}{\ensuremath{\Varid{ty}}}  \\
\ensuremath{\Varid{ty}}^* = \gluetm{\hole{?0 : \extty{U_1}{\sop}{\ensuremath{\Varid{el}}\ \ensuremath{\Varid{ty}}}}}{\ensuremath{\Varid{ty}}} 
\end{lgather*}%
Since \hole{?0} is a type in $U_1$ that is equal to $\ensuremath{\Varid{el}} \; \ensuremath{\Varid{ty}}$ under
$\sop$, it can be a glue type:
\begin{equation}\label{eq:m:star:ty:goal1}
\ensuremath{\Varid{ty}}^* = \gluetm{ \gluety{A : \ensuremath{\Varid{el}} \; \ensuremath{\Varid{ty}}}{\hole{?1}}}{\ensuremath{\Varid{ty}}} 
\end{equation}
which means that an element of the kind \ensuremath{\Varid{ty}} in the model
$\MM^*$ is a syntactic type $A$ together with the data \hole{?1}.  
It is natural to expect that the data \hole{?1} associated to a type
$A$ is a (candidate of) logical predicate for the type $A$, 
which is just any type that restricts to $\ensuremath{\Varid{tm}} \; A$ under $\sop$:
\begin{lgather*}
\ensuremath{\Varid{ty}}^* = \gluetm{ \gluety{A : \ensuremath{\Varid{el}} \; \ensuremath{\Varid{ty}}}
                           \hole{{\extty{U_0}{\sop}{\ensuremath{\Varid{tm}} \; A}}}}
                  {\ensuremath{\Varid{ty}}} 
                  \tag{$*$}
\end{lgather*}%
mimicking the kind structure \Cref{eq:m:star:ki} that we have seen earlier.
However, this definition will not work when we come to \emph{impredicative}
polymorphic types $\forall \alpha. A$ later, because $U_0$ is not
impredicative in the sense of being closed under $\Pi$-types $\Pi \; A \; B$ for 
\emph{arbitrary} types $A$ that are not necessarily in $U_0$.
\end{element}

\begin{element}
In every topos, we do have an impredicative universe -- the
universe $\Omega$ of \emph{propositions}.
Unfortunately, this universe is `too small' for interpreting \Fha{}-types.
If we have an element $A^* : \extty{\Omega}{\sop}{\ensuremath{\Varid{tm}} \; A}$ for some 
object-space type $A : \ensuremath{\Varid{el}}\;\ensuremath{\Varid{ty}}$,
when $\sop$ holds, $A^*$ is equal to $\ensuremath{\Varid{tm}} \; A$, but $A^*$ is in the universe $\Omega$,
so we have $\impfun{\sop} (a, b : \ensuremath{\Varid{tm}} \; A) \to a = b$, which means that 
the object-space type $A$ has at most one element, and this is clearly not true in general.
\end{element}

\begin{element}
To find a way out, let us recall how traditional logical predicates/relations 
of System F work in, for example, \varcitet{Girard_1989}{'s} normalisation proof.
For every type $A$ of System F, its logical predicate
is a proof-irrelevant predicate on the set of terms of $A$, or equivalently,
a function from terms of $A$ to the set of classical propositions.
Moreover, the logical predicate $P(t)$ of the impredicative polymorphic type $\forall \alpha. \; A$
is defined by `for all types $X$ and all candidate logical predicates $Q$ over
terms of $X$, the term $t \; [X]$ is related by the logical predicate of $A$ with
$\alpha$ replaced by $(X, Q)$'.
This works because classical propositions are impredicative, so we can quantify
over all $X$ and $Q$.
\end{element}

\begin{element}
Mimicking the traditional approach, we first define a universe of \emph{meta-space propositions}
(which are just classical
propositions $\aset{\top, \bot}$ when \TTSTC{} is interpreted in the Artin
gluing of the syntactic category and the category of sets):
\[ 
\COmega \defeq \aset{p : \Omega \mid \Cmodal\;p}.
\]
The universe $\COmega$ inherits all the
connectives that $\Omega$ has, including impredicative quantification.
For example, if $A$ is an arbitrary type and $B : A \to \COmega$, the type
$\forall (x : A). B \; x$ is in $\Omega$, and when $\sop$ holds, $B \; x \cong \unitty$ because a type is $\Cmod$-modal iff it is isomorphic to $\unitty$ under $\sop$
(\Cref{lem:cmod:if:unit}),
so $\forall (x : A). B \; x = \forall (x : A). 1 \cong 1$.
\end{element}

\begin{element}
Using $\COmega$, we fill out the hole \hole{?1} in $\ensuremath{\Varid{ty}}^*$ \Cref{eq:m:star:ty:goal1} by
\begin{lgather}{eq:m:star:ty}
\ensuremath{\Varid{ty}}^* : \extty{\ensuremath{\Varid{ki}}^*}{\sop}{\ensuremath{\Varid{ty}}}  \\
\ensuremath{\Varid{ty}}^* = \gluetm{ \gluety{A : \ensuremath{\Varid{el}} \; \ensuremath{\Varid{ty}}}{\hole{(\impfun{\sop} \; \ensuremath{\Varid{tm}} \; A) \to \COmega}} }{\ensuremath{\Varid{ty}}}
\end{lgather}
That is to say, the candidate of a logical predicate for a type $A$ is given as
a meta-space predicate $P : (\impfun{\sop} \; \ensuremath{\Varid{tm}} \; A) \to \COmega$.

Then $\ensuremath{\Varid{tm}}^*$ glues terms $\ensuremath{\Varid{tm}} \; A$ of an object-space type $A$ with the
predicate $P$:
\begin{lgather}{eq:m:star:tm}
\ensuremath{\Varid{tm}}^* : \extty{\ensuremath{\Varid{el}}^* \; \ensuremath{\Varid{ty}}^* \to U_0}{\sop}{\ensuremath{\Varid{tm}}}  \\
\ensuremath{\Varid{tm}}^* \; \gluetm{P}{A} = \gluety{t : \ensuremath{\Varid{tm}} \; A}{P \; t}
\end{lgather}
That is to say, in the model $\MM^*$, a term of the semantic type $\gluetm{P}{A} : \ensuremath{\Varid{ty}}^*$
is a term $t$ of the syntactic type $A$ that satisfies the meta-space predicate $P$.
\end{element}

\begin{notation}\label{notation:lpre:lprf}
For every $A^* : \ensuremath{\Varid{el}}^* \; \ensuremath{\Varid{ty}}^*$, we define
\begin{lgather*}
\lpre A^* : (\impfun{\sop}\ \ensuremath{\Varid{tm}} \; A^*) \to \COmega  \\
\lpre A^* = \unglue A^* 
\end{lgather*}
to remind us that ungluing a semantic type gives its underlying logical predicate.
Similarly, for every $a^* : \ensuremath{\Varid{tm}}^* \; A^*$, we define
\begin{lgather*}
\lprf a^* :  \lpre A^* \; (\lambda \imparg{\_ : \sop}.\ a^*)  \\
\lprf a^* = \unglue a^*
\end{lgather*}
to remind us that ungluing a semantic term is the proof that the underlying syntactic
term satisfies the corresponding logical predicate. 
\end{notation}

\begin{remark}
For every $A^* : \ensuremath{\Varid{el}}^* \; \ensuremath{\Varid{ty}}^*$,
the type $\ensuremath{\Varid{tm}}^* \; A^*$ satisfies
the property that for every $a : \impfun{\sop} A^*$, there is at most one element
$a^* : \ensuremath{\Varid{tm}}^*\;A^*$ that restricts to $a$ under $\sop$, because the meta-space
component of $\ensuremath{\Varid{tm}}^* \; A^*$ is a (fiberwise) meta-space proposition.
Based on this observation, there is a more intrinsic alternative definition of
$\ensuremath{\Varid{ty}}^*$:
for every universe $U$ of \TTSTC, we can define its \emph{proof-irrelevant} subuniverse $\Uir$ to be
\[
\Uir \defeq \aset{A : U \mid \forall (a : \impfun{\sop} A).\; (x, y : \extty{A}{\sop}{a}) \to (x = y)}.
\]
Then we can define $\ensuremath{\Varid{ty}}^*$ and $\ensuremath{\Varid{tm}}^*$ as simply
\begin{lgather*}
\ensuremath{\Varid{ty}}^* = \gluetm{\gluety{A : \ensuremath{\Varid{el}}\;\ensuremath{\Varid{ty}}}{\hole{\extty{\Uir_0}{\sop}{\ensuremath{\Varid{tm}}\;A}}}}{\ensuremath{\Varid{ty}}} \\[8pt]
\ensuremath{\Varid{tm}}^* \; A^* = \unglue A^*
\end{lgather*}
which directly mirrors the definition of $\ensuremath{\Varid{ki}}^*$ \Cref{eq:m:star:ki} and $\ensuremath{\Varid{el}}^*$ \Cref{eq:m:star:el}.

This alternative definition is in a suitable sense equivalent to the one above (\ref{eq:m:star:ty}, \ref{eq:m:star:tm}) because
for every $A : \impfun{\sop} U$, we have an equivalence
\[
\extty{\Uir_0}{\sop}{A} \ \cong\ 
((\impfun{\sop} \; A) \to \COmega).
\]
when treating them as categories (in fact, preorders) suitably.
We choose to work with $\ensuremath{\Varid{ty}}^*$ \Cref{eq:m:star:ty} in terms
of $\COmega$-valued predicates because it is slightly more convenient for
logical predicates on computation judgements later.
\end{remark}

\subsection{Base Types}

\begin{element}
Since in the theory of \Fha{}, the unit type is specified to be isomorphic
to meta-level unit type \Cref{eq:fha:terms}, we have no choice for the logical predicate for the logical
predicate of the unit type (of \Fha{}) other than the always true predicate:
\begin{lgather*}
\ensuremath{\Varid{unit}}^* : \extty{\ensuremath{\Varid{el}}^* \; \ensuremath{\Varid{ty}}^*}{\sop}{\ensuremath{\Varid{unit}}}  \\
\ensuremath{\Varid{unit}}^* = \gluetm{\lambda(\_ : \impfun{\sop} \ensuremath{\Varid{tm}} \; \ensuremath{\Varid{unit}}). \; \unitty}{\ensuremath{\Varid{unit}}}
\end{lgather*}
Recall that $\ensuremath{\Varid{tm}}^*\;\ensuremath{\Varid{unit}}^*$ computes to $\gluety{t : \ensuremath{\Varid{tm}}\;\ensuremath{\Varid{unit}}}{\unitty}$, we define
\begin{lgather*}
\ensuremath{\Varid{unit\hyp{}iso}}^* : \ensuremath{\Varid{tm}}^*\;\ensuremath{\Varid{unit}}^* \cong 1 \\
\ensuremath{\Varid{unit\hyp{}iso}}^*.\ensuremath{\Varid{fwd}} \; \_ = \unitel \\
\ensuremath{\Varid{unit\hyp{}iso}}^*.\ensuremath{\Varid{bwd}} \; \_= \gluetm{\unitel}{\ensuremath{\Varid{unit\hyp{}iso}}.\ensuremath{\Varid{bwd}}}
\end{lgather*}
This is an isomorphism because $\ensuremath{\Varid{tm}}\;\ensuremath{\Varid{unit}} \cong 1$ by $\ensuremath{\Varid{unit\hyp{}iso}}$.
\end{element}

\begin{element}
The other base type is the weak Boolean type \ensuremath{\Varid{bool}}. 
It is also the type that canonicity is about, so its logical predicate is
specific to canonicity:
\begin{lgather}{eq:m:star:bool}
\ensuremath{\Varid{bool}}^* : \extty{\ensuremath{\Varid{el}}^*\;\ensuremath{\Varid{ty}}^*}{\sop}{\ensuremath{\Varid{bool}}} \\
\ensuremath{\Varid{bool}}^* = \gluetm{P_\ensuremath{\Varid{can}}}{\ensuremath{\Varid{bool}}} \\[8pt]
P_\ensuremath{\Varid{can}} : (\impfun{\sop}\ensuremath{\Varid{tm}}\;\ensuremath{\Varid{bool}}) \to \COmega \\
P_\ensuremath{\Varid{can}}\; b = \Cmod (\impfun{\sop} (b = \ensuremath{\Varid{tt}} \lor b = \ensuremath{\Varid{ff}}))
\end{lgather}
The closed modality $\Cmod$ is needed here to erase the object-space component
of the proposition $\impfun{\sop} (b = \ensuremath{\Varid{tt}} \lor b = \ensuremath{\Varid{ff}})$, turning
it $\Cmod$-modal.
We also need to define the two terms of the weak Boolean types, i.e.\ showing
that the two terms \ensuremath{\Varid{ff}} and \ensuremath{\Varid{tt}} satisfy the logical predicate of \ensuremath{\Varid{bool}}:
\begin{lgather*}
\ensuremath{\Varid{tt}}^* : \extty{\ensuremath{\Varid{tm}}^* \; \ensuremath{\Varid{bool}}^*}{\sop}{\ensuremath{\Varid{tt}}} \\
\ensuremath{\Varid{tt}}^* = \gluetm{\Ceta\; ( \inl\;{\refl} )}{\ensuremath{\Varid{tt}}} \\[8pt]
\ensuremath{\Varid{ff}}^* : \extty{\ensuremath{\Varid{tm}}^* \; \ensuremath{\Varid{bool}}^*}{\sop}{\ensuremath{\Varid{ff}}} \\
\ensuremath{\Varid{ff}}^* = \gluetm{\Ceta\; ( \inr\;{\refl} )}{\ensuremath{\Varid{ff}}}
\end{lgather*}
In the construction of $\MM^*$, the only things that are specific to canonicity are
$P_\ensuremath{\Varid{can}}$, $\ensuremath{\Varid{tt}}^*$ and $\ensuremath{\Varid{ff}}^*$.
They can be changed to anything else without affecting other parts of $\MM^*$
(although there seemingly are not many interesting choices of $P_\ensuremath{\Varid{can}}$).
\end{element}

\subsection{Function Types}\label{sec:lr:fun:ty}

\begin{element}
A function $t : \ensuremath{\Varid{tm}}\ (\ensuremath{\Conid{A}\;{\Rightarrow_t}\;\Conid{B}})$ is related by the logical predicate for
the type $\ensuremath{\Conid{A}\;{\Rightarrow_t}\;\Conid{B}}$ if it maps input $a$ satisfying the logical predicate for $A$ to output $\ensuremath{\Varid{t}\;\Varid{a}}$ satisfying the logical predicate for $B$:
\begin{lgather*}
\ensuremath{\text{\textunderscore}{\Rightarrow_t}\text{\textunderscore}}^* : \extty{\ensuremath{\Varid{el}}^* \; \ensuremath{\Varid{ty}}^* \to \ensuremath{\Varid{el}}^* \; \ensuremath{\Varid{ty}}^* \to \ensuremath{\Varid{el}}^* \; \ensuremath{\Varid{ty}}^*}{\sop}{\ensuremath{\text{\textunderscore}{\Rightarrow_t}\text{\textunderscore}}} \\
\gluetm{P}{A} \mathbin{\ensuremath{{\Rightarrow^*_t}}} \gluetm{Q}{B} = \gluetm{P_\ensuremath{{\Rightarrow_t}}}{A \mathbin{\ensuremath{{\Rightarrow_t}}} B} \numberthis\label{el:m:star:tfun} \\[8pt]
P_\ensuremath{{\Rightarrow_t}} \defeq \lambda t.\;  \forall (a : \impfun{\sop} A).\ P \; a \to Q \; (\impabs{\_ : \sop} t \; a)  
\end{lgather*}
Note that in the expression \ensuremath{\Varid{t}\;\Varid{a}}, we have elided the isomorphism $\ensuremath{\Varid{{\Rightarrow_t}\hyp{}iso}}$ \Cref{eq:fha:terms} between $\ensuremath{\Varid{tm}}\;(A \mathbin{\ensuremath{{\Rightarrow_t}}} B)$ and $\ensuremath{\Varid{tm}}\;A \to \ensuremath{\Varid{tm}}\;B$.

We also need to define an isomorphism,
for all $A^*, B^* : \ensuremath{\Varid{el}}^* \; \ensuremath{\Varid{ty}}^*$,
\[
\ensuremath{\Varid{{\Rightarrow_t}\hyp{}iso}}^* : \ensuremath{\Varid{tm}}^* \; (A^* \mathbin{\ensuremath{{\Rightarrow^*_t}}} B^*) \cong (\ensuremath{\Varid{tm}}^* \; A^*) \to (\ensuremath{\Varid{tm}}^* \; B^*).
\]
Letting $A^* = \gluetm{P}{A}$ and $B^* = \gluetm{Q}{B}$, we compute as follows:
\begin{align*}
   & \ensuremath{\Varid{tm}}^* \; (\gluetm{P}{A} \mathbin{\ensuremath{{\Rightarrow^*_t}}} \gluetm{Q}{B}) \\
=\ & \gluety{t : \ensuremath{\Varid{tm}} \; (A \mathbin{\ensuremath{{\Rightarrow_t}}} B)}{(a : \impfun{\sop} \ensuremath{\Varid{tm}} \; A) \to P \; a \to  Q \; (t \; a)} \\
\cong\ & \reason{by $\ensuremath{\Varid{{\Rightarrow_t}\hyp{}iso}}$ \Cref{eq:fha:terms}} \\
   & \gluety{t : \ensuremath{\Varid{tm}}\; A \to \ensuremath{\Varid{tm}} \; B}{ (a : \impfun{\sop} \ensuremath{\Varid{tm}} \; A) \to P \; a \to  Q \; (t \; a)} \\
\cong\ & \reason{by \ensuremath{\Varid{\gluetysymbol\hyp{}fun\hyp{}iso}} from \Cref{lem:fun:between:glue:tys}}\\
   & \big(\gluety{a : \ensuremath{\Varid{tm}} \; A}{P \; a}\big) \to \big(\gluety{b : \ensuremath{\Varid{tm}} \; B}{Q \; b}\big) \\
=\ & (\ensuremath{\Varid{tm}}^* \; \gluetm{P}{A}) \to (\ensuremath{\Varid{tm}}^* \; \gluetm{Q}{B})
\end{align*}
\end{element}

\begin{element}
The logical predicate for polymorphic functions is 

\begin{equation}\label{el:m:star:p:fun}
\begin{array}{l}
\ensuremath{\allfor}^* : \extty{(k^* : \ensuremath{\Varid{ki}}^*) \to (\ensuremath{\Varid{el}}^* \; k^* \to \ensuremath{\Varid{el}}^* \; \ensuremath{\Varid{ty}}^*) \to \ensuremath{\Varid{el}}^* \; \ensuremath{\Varid{ty}}^*}{\sop}{\ensuremath{\allfor}} \\
\ensuremath{\allfor}^* \; k^* \; F = 
  \gluetm{ \lambda t.\; 
              \forall (\alpha^* : \ensuremath{\Varid{el}}^* \; k^*).\ \lpre (F \; \alpha^*) \; (\lambda \imparg{\_ : \sop}.\ (t\; \alpha^*))}{\ensuremath{\allfor} \; k^* \; F}
\end{array}
\end{equation}
Let us check the type of this definition step-by-step.
The object-space component $\ensuremath{\allfor} \; k^* \; F$ is well typed because
under $\sop$, $\ensuremath{\Varid{ki}}^*$ equals \ensuremath{\Varid{ki}}, so $k^* : \ensuremath{\Varid{ki}}$ under $\sop$, and similarly 
$F : \ensuremath{\Varid{el}\;\Varid{k}\to \Varid{el}\;\Varid{ty}}$ under $\sop$, thus $\ensuremath{\allfor} \; k^* \; F : \ensuremath{\Varid{el}\;\Varid{ty}}$ as expected.

The meta-space component $\lambda t. \cdots$ should be an $\COmega$-valued predicate on $t : \impfun{\sop} \ensuremath{\Varid{tm}} \; (\ensuremath{\allfor} \; k^* \; F)$.
We have $F\;\alpha^* : \ensuremath{\Varid{el}}^* \; \ensuremath{\Varid{ty}}^*$; this type computes to
\[
\gluety{A : \ensuremath{\Varid{el}} \; \ensuremath{\Varid{ty}}}{(\impfun{\sop} \; \ensuremath{\Varid{tm}} \; A) \to \COmega} 
\]
by definitions (\ref{eq:m:star:el}, \ref{eq:m:star:ty}).
Thus $\lpre (F \; \alpha^*)$ has type $(\impfun{\sop} \; \ensuremath{\Varid{tm}} \; (F \; \alpha^*)) \to \COmega$.
On the other hand, $t$ has type $\impfun{\sop} \ensuremath{\Varid{tm}} \; (\ensuremath{\allfor} \; k^* \; F)$,
which is isomorphic to $\impfun{\sop} (\alpha^* : \ensuremath{\Varid{el}}\;k^*)\to \ensuremath{\Varid{tm}} \; (F \; \alpha^*)$ via
\ensuremath{\Varid{\allfor\hyp{}iso}} \Cref{eq:fha:terms}, which we elided in the definition above.
The implicit function 
$\lambda \imparg{\_ : \sop}.\ (t\; \alpha^*)$ then has type 
$\impfun{\sop} \ensuremath{\Varid{tm}}\;(F \; \alpha^*)$, and therefore it can be supplied as an argument
to $\lpre (F \; \alpha^*)$, yielding a proposition in $\COmega$.
The quantification $\forall (\alpha^* : \ensuremath{\Varid{el}}^* \; k^*)$ is allowed because
$\COmega$ is closed under impredicative universal quantification.

We also need to define an isomorphism for $k^* : \ensuremath{\Varid{ki}}^*$ and $F : \ensuremath{\Varid{el}}^* \; k^* \to 
\ensuremath{\Varid{el}}^* \; \ensuremath{\Varid{ty}}^*$
\[
\ensuremath{\Varid{\allfor\hyp{}iso}}^* : \ensuremath{\Varid{tm}}^* \; (\ensuremath{\allfor}^* \; k^* \; F) \cong ((\alpha^* : \ensuremath{\Varid{el}}^* \; k^*) \to \ensuremath{\Varid{tm}}^*\; (F\; \alpha^*)).
\]
This is very similar to $\ensuremath{\Varid{{\Rightarrow_t}\hyp{}iso}}^*$ in \Cref{el:m:star:tfun}, and we give a 
direct definition here:
\begin{lgather*}
\ensuremath{\Varid{fwd}} \; \gluetm{p}{t} = \lambda \alpha^*.\; \gluetm{\ensuremath{\Varid{\allfor\hyp{}iso}}.\ensuremath{\Varid{fwd}}\; t \; \alpha^*}{p \; \alpha^*} \\[8pt]
\ensuremath{\Varid{bwd}} \; h = \gluetm{ \lambda \alpha^*.\; \lprf (h \; \alpha^*) }{\ensuremath{\Varid{\allfor\hyp{}iso}}.\ensuremath{\Varid{bwd}}\;h}
\end{lgather*}
\end{element}

\begin{element}
By this point we have completed the definition of the logical predicates for
the \Fom{}-fragment of \Fha, so the derived concepts in \Cref{fig:derived:concepts}
such as raw functors can be interpreted in $\MM^*$ as well.
For example, we have 
\begin{hscode}\SaveRestoreHook
\column{B}{@{}>{\hspre}l<{\hspost}@{}}%
\column{3}{@{}>{\hspre}l<{\hspost}@{}}%
\column{E}{@{}>{\hspre}l<{\hspost}@{}}%
\>[B]{}\Varid{efty}^*\mathbin{:}\Varid{ki}^*{}\<[E]%
\\
\>[B]{}\Varid{efty}^*\mathrel{=}\Varid{ty}^*\;{\Rightarrow^*_k}\;\Varid{ty}^*{}\<[E]%
\\[\blanklineskip]%
\>[B]{}\Varid{fmap\hyp{}ty}^*\mathbin{:}(\Conid{F}\mathbin{:}\Varid{el}^*\;\Varid{efty}^*)\to\Varid{el}^*\;\Varid{ty}^*{}\<[E]%
\\
\>[B]{}\Varid{fmap\hyp{}ty}^*\;\Conid{F}\mathrel{=}\allfor^*\;\Varid{ty}^*\;(\lambda \Varid{α}.\;\allfor^*\;\Varid{ty}^*\;(\lambda \Varid{β}.\;(\Varid{α}\;{\Rightarrow^*_t}\;\Varid{β})\;{\Rightarrow^*_t}\;(\Conid{F}\;\Varid{α}\;{\Rightarrow^*_t}\;\Conid{F}\;\Varid{β}))){}\<[E]%
\\[\blanklineskip]%
\>[B]{}\lhskeyword{record}\;\Conid{RawFunctor}^*\mathbin{:}U_1\;\lhskeyword{where}{}\<[E]%
\\
\>[B]{}\hsindent{3}{}\<[3]%
\>[3]{}\Varid{0}\mathbin{:}\Varid{el}^*\;\Varid{efty}^*{}\<[E]%
\\
\>[B]{}\hsindent{3}{}\<[3]%
\>[3]{}\Varid{fmap}\mathbin{:}\Varid{tm}^*\;\Varid{fmap\hyp{}ty}^*\;\Varid{0}{}\<[E]%
\ColumnHook
\end{hscode}\resethooks
which is simply the same as the definition of \ensuremath{\Conid{RawFunctor}} in
\Cref{fig:derived:concepts} except that all the judgements of \Fom{} such as \ensuremath{\Varid{ki}}
and \ensuremath{\Varid{ty}} are replaced by their corresponding interpretation in $\MM^*$.
Interpretations of other derived concepts in \Cref{fig:derived:concepts} such as 
$\ensuremath{\Conid{RawMonad}}^*$ and $\ensuremath{\Conid{RawHFunctor}}^*$ can be obtained in this way as well.
\end{element}

\subsection{Computation Judgements}
\begin{element}
What remains is the logical predicates for computation judgements of \Fha{}.
Recall that given \ensuremath{\Conid{H}\mathbin{:}\Conid{RawHFunctor}} and \ensuremath{\Conid{A}\mathbin{:}\Varid{el}\;\Varid{ty}}, the judgements \ensuremath{\Varid{co}\;\Conid{H}\;\Conid{A}}
in \Fha{} roughly axiomatise a monad equipped with \ensuremath{\Conid{H}}-operations -- more precisely, \ensuremath{\Varid{co}\;\Conid{H}\;\Conid{A}} axiomatises the Kleisli category of this monad since
\ensuremath{\Varid{co}\;\Conid{H}\;\Conid{A}} is not an \Fha{}-type but a separate judgement.
\end{element}

\begin{element}
(\emph{First attempt})
%
Since our logical predicates live in an impredicative universe,
a natural idea is to define the logical predicate for \ensuremath{\Varid{co}} by an impredicative
encoding of the initial monad equipped with \ensuremath{\Conid{H}}-operations:
\begin{lgather}{eq:m:star:co:naive}
\ensuremath{\Varid{co}}^* : \extty{\ensuremath{\Conid{HFunctor}}^* \to \ensuremath{\Varid{el}}^* \; \ensuremath{\Varid{ty}}^* \to U_0}{\sop}{\ensuremath{\Varid{co}}} \\
\ensuremath{\Varid{co}}^* \; H^* \; A^* = \gluety{c : \ensuremath{\Varid{co}} \; H^* \; A^* }{P_\ensuremath{\Varid{co}} \; c} \\[4pt]
P_\ensuremath{\Varid{co}} : \impfun{H^*, A^*}  (\impfun{\sop} \ensuremath{\Varid{co}} \; H^* \; A^*) \to \COmega \\
P_\ensuremath{\Varid{co}} \; c = \forall (h : \ensuremath{\Conid{Handler}}^* \; H^*).\ \lpre (h.0 \; A^*) \; (\ensuremath{\Varid{hdl}} \; h \; c)
\end{lgather}
The function $P_\ensuremath{\Varid{co}}$ type-checks as follows:
the type of $h.0$ is $\ensuremath{\Varid{el}}^* \; \ensuremath{\Varid{efty}}^*$, and
\begin{align*}
   & \ensuremath{\Varid{el}}^*\;\ensuremath{\Varid{efty}}^* \\
=\ & \reason{by definition} \\
   & \ensuremath{\Varid{el}}^*\; (\ensuremath{\Varid{ty}}^* \mathbin{\ensuremath{{\Rightarrow^*_k}}} \ensuremath{\Varid{ty}}^*) \\
\cong\ & \reason{by the axiom \Cref{eq:fha:el:kfun}} \\
   & \ensuremath{\Varid{el}}^* \; \ensuremath{\Varid{ty}}^* \to \ensuremath{\Varid{el}}^* \; \ensuremath{\Varid{ty}}^*
\end{align*}
Therefore the type of $h.0 \; A^*$ is $\ensuremath{\Varid{el}}^* \; \ensuremath{\Varid{ty}}^*$, which
is the glue type 
\[
\gluety{A : \ensuremath{\Varid{el}} \; \ensuremath{\Varid{ty}}}{(\impfun{\sop} \; \ensuremath{\Varid{tm}} \; A) \to \COmega}
\]
by (\ref{eq:m:star:el}, \ref{eq:m:star:ty}).
Then $\lpre (h.0 \; A^*)$ has type 
$
(\impfun{\sop} \; \ensuremath{\Varid{tm}} \; (h.0 \; A^*)) \to \COmega
$, i.e.\ it is a meta-space predicate on terms of type $\impfun{\sop} h.0 \; A^*$.
The term $\ensuremath{\Varid{hdl}} \; h \; c$, which in fact is the implicit 
function $\lambda \imparg{\_ : \sop}.\ \ensuremath{\Varid{hdl}} \; h \; c$,
has precisely the type $\impfun{\sop} \; \ensuremath{\Varid{tm}} \; (h.0 \; A^*)$.
Finally, since $\COmega$ is closed under universal quantification $\forall (h :
\ensuremath{\Conid{Handler}}^*\;H^*)$, $P_\ensuremath{\Varid{co}} \; c$ has type $\COmega$, i.e.\ $\extty{\Omega}{\sop}{1}$.
\end{element}

\begin{remark}
Usually, the impredicative encoding of a datatype needs to be `refined' by some
additional equalities to have the correct universal property
\cite{Awodey_2018}.
For example, the impredicative encoding of the coproduct type $A + B$ 
in an impredicative universe $U$ is
\begin{lgather*}
A + B = \sum (\alpha : (X : U) \to (A \to X) \to (B \to X) \to X)\; N \; \alpha \\[4pt]
N \; \alpha = (X, Y : U) \to (f : X \to Y) \to (h : A \to X) \to (k : B \to X)\\
\hspace{5em} \to f \; (\alpha \;X \; h \; k) = \alpha \; Y \; (f \hcomp h) (f \hcomp k)
\end{lgather*}
Without imposing $N$ on $\alpha$, the impredicative encoding would not satisfy
the $\eta$-rule of the coproduct type.
However, our logical predicates land in a universe
\emph{propositions}, where two elements of the same type are automatically equal, so this refinement is unnecessary.
\end{remark}

\begin{element}
However, as we commented in \Cref{rem:no:evallet}, evaluating the sequential
composition
\ensuremath{\Varid{let\hyp{}in}\;\Varid{c}\;\Varid{f}} is \emph{not} compositional because of the discrepancy between computations \ensuremath{\Varid{co}} and
raw monads: \ensuremath{\Varid{co}} satisfies the monadic laws while raw monads do not.
For this reason, with the above definition of $P_\ensuremath{\Varid{co}}$, we will have problems 
with showing the term constructor \ensuremath{\Varid{let\hyp{}in}} satisfying its logical predicate:
\[
\begin{pmboxed}
\defaultcolumn{@{}l@{}}%
\> \ensuremath{\Varid{let\hyp{}in}}^* : \extty{\{H^*, A^*, B^*\} \> \to \ensuremath{\Varid{co}}^* \; H^* \; A^* \to (\ensuremath{\Varid{tm}}^* \; A^* \to \ensuremath{\Varid{co}}^* \; H^* \; B^*) \\ 
\>\> \to \ensuremath{\Varid{co}}^* \; H^* \; B^* }{\sop}{\ensuremath{\Varid{let\hyp{}in}}} \\
\> \ensuremath{\Varid{let\hyp{}in}}^* \; c \; f = \gluetm{\lambda m.\ \hole{?1} }{\ensuremath{\Varid{let\hyp{}in}} \; c \; f}
\end{pmboxed}
\]
where the hole $\hole{?1}$ has type 
$P_\ensuremath{\Varid{co}}\;(\ensuremath{\Varid{let\hyp{}in}} \; c \; f)$, that is, by the definition of $P_\ensuremath{\Varid{co}}$ above,
\[
\forall (h : \ensuremath{\Conid{Handler}}^* \; H^*).\ \lpre (h\ensuremath{\Varid{.0}} \; B^*) \; 
  (\ensuremath{\Varid{hdl}} \; h \; (\ensuremath{\Varid{let\hyp{}in}} \; c \; f)).
\]
Since we do not have the equation \ensuremath{\Varid{hdl\hyp{}let}} in \Cref{rem:no:evallet} to
simplify the computation $\ensuremath{\Varid{hdl}} \; h \; (\ensuremath{\Varid{let\hyp{}in}} \; c \; f)$, we have no way
to fill in the hole \hole{?1} using $c$ and $f$.
\end{element}

\newcommand{\codeindent}{\hspace{1em}}
\begin{element}
To fix this problem, we strengthen $P_\ensuremath{\Varid{co}} \; c$ above to take into account
all possible \emph{continuations} after the computation \ensuremath{\Varid{c}}, which is essentially
the idea of \emph{$\top\top$-lifting}
\citep{Lindley_Stark_2005,Katsumata_2005,Katsumata_Sato_Uustalu_2018}.
We first define a type of continuations accepting $A^*$-values:
\begin{lgather*}
\ensuremath{\lhskeyword{record}} \; \ensuremath{\Conid{Con}} \; (H^* : \ensuremath{\Conid{RawHFunctor}}^*) \; (A^* : \ensuremath{\Varid{el}}^* \; \ensuremath{\Varid{ty}}^*) : U_1 \; \ensuremath{\lhskeyword{where}} \\
\codeindent 
  h^* : \ensuremath{\Conid{Handler}}^* \; H^* \\
\codeindent
  R^* : \ensuremath{\Varid{el}}^* \; \ensuremath{\Varid{ty}}^* \\
\codeindent
  k : \impfun{\sop} A^* \to \ensuremath{\Varid{co}} \; H^* \; R^* \\
\codeindent
  k^* : \extty{\ensuremath{\Varid{tm}}^* \; A^* \to \ensuremath{\Varid{tm}}^* \; (h^*\ensuremath{\Varid{.0}} \; R^*)}{\sop}%
        {\lambda a.\ \ensuremath{\Varid{hdl}} \; h^* \; (k \; a) }
\end{lgather*}
and the strengthened definition of $P_\ensuremath{\Varid{co}}$ is
\begin{lgather}{eq:m:star:co}
P_\ensuremath{\Varid{co}} : \impfun{H^*, A^*} (\impfun{\sop} \ensuremath{\Varid{co}} \; H^* \; A^*) \to \COmega \\
P_\ensuremath{\Varid{co}} \; c = \forall (K : \ensuremath{\Conid{Con}} \; H^* \; A^*).\
   \lpre (K.h^*.0 \;\; K.R^*) \; (\impabs{\_ : \sop} \\
   \hspace{6em} \ensuremath{\Varid{hdl}} \; K.h^* \; (\ensuremath{\Varid{let\hyp{}in}} \; c \; K.\ensuremath{\Varid{k}}))
\end{lgather}
Compared to the earlier version of $P_\ensuremath{\Varid{co}}$ \Cref{eq:m:star:co:naive}, the new
version asserts that the computation $c$ extended with an arbitrary `good'
continuation $k$ and evaluated into a raw monad results in a value satisfying
its logical predicate.
Here a continuation $k$ is `good' if $k$ followed
by \ensuremath{\Varid{hdl}} sends  input satisfying its logical predicate to output satisfying its
logical predicate, which is succinctly expressed by a
function $k^* : \ensuremath{\Varid{tm}}^* \; A^* \to \ensuremath{\Varid{tm}}^* \; (h^*\ensuremath{\Varid{.0}} \; R^*)$, c.f.\ \Cref{lem:fun:between:glue:tys}.
\end{element}

\begin{remark}
The new definition of $P_\ensuremath{\Varid{co}}$ is similar to the model of computations in the
realizability model (\Cref{el:real:co}), except that here we only consider
Kleisli morphisms $k^* : A^* \to h^*.0 \; R^*$ whose object-space component
factors through some $k : \impfun{\sop} A^* \to \ensuremath{\Varid{co}} \; H^* \; R^* $.
The author currently does not know if there could be a conceptual
explanation for such a modified codensity transformation.
\end{remark}

\begin{element}
The logical predicate for thunks is the same as that for computations, modulo
the isomorphism \ensuremath{\Varid{th\hyp{}iso}\mathbin{:}\{\mskip1.5mu \Conid{H},\Conid{A}\mskip1.5mu\}\to \Varid{tm}\;(\Varid{th}\;\Conid{H}\;\Conid{A})\;{\cong}\;\Varid{co}\;\Conid{H}\;\Conid{A}} from \Cref{el:tmnd}:
\begin{lgather*}
\ensuremath{\Varid{th}}^* : \extty{\ensuremath{\Conid{RawHFunctor}}^* \to \ensuremath{\Varid{el}}^* \; \ensuremath{\Varid{ty}}^* \to \ensuremath{\Varid{el}}^* \; \ensuremath{\Varid{ty}}^*}{\sop}{\bT} \\
\ensuremath{\Varid{th}}^*\; H^* \; A^* = \gluetm{\lambda t.\; P_\ensuremath{\Varid{co}}\;(\impabs{\_ : \sop} \ensuremath{{\Uparrow}} \; t)}{\bT \; H^* \; A}
\end{lgather*}
where \ensuremath{{\Uparrow}} is the forward direction of the isomorphism \ensuremath{\Varid{th\hyp{}iso}}.
The isomorphism $\ensuremath{\Varid{th\hyp{}iso}}^* : \ensuremath{\Varid{tm}}^* \; (\bT \; H^* \; A^*) \cong \ensuremath{\Varid{co}}^* \; H^* \;
A^*$ is also straightforward:
\[
\spread*{
\ensuremath{\Varid{fwd}} \; \gluetm{p}{t} = \gluetm{p}{\ensuremath{{\Uparrow}} \; t},
\also
\ensuremath{\Varid{bwd}} \; \gluetm{p}{c} = \gluetm{p}{\ensuremath{{\Downarrow}} \; c}.
}
\]
\end{element}

\subsection{Computation Terms}\label{sec:fha:co:star}

\begin{element}
Finally, we need to prove that the constructors \ensuremath{\Varid{val}}, \ensuremath{\Varid{let\hyp{}in}}, \ensuremath{\Varid{op}} and
the eliminator \ensuremath{\Varid{hdl}} of computations satisfy the logical predicates.
We start with \ensuremath{\Varid{val}}:
\begin{lgather*}
\ensuremath{\Varid{val}}^* : \extty{ \impfun{H^*, A^*}  \ensuremath{\Varid{tm}}^* \; A^* \to \ensuremath{\Varid{co}}^* \; H^* \; A^*  }{\sop}{\ensuremath{\Varid{val}}} \\
\ensuremath{\Varid{val}}^* \; \imparg{H^*} \; \imparg{A^*} \; a = \gluetm{\hole{?1}}{\ensuremath{\Varid{val}} \; a}
\end{lgather*}
where the hole \hole{?1} has type $P_\ensuremath{\Varid{co}}\;(\ensuremath{\Varid{val}} \; a)$, that is, by definition \Cref{eq:m:star:co},
\begin{align*}
& \forall (K : \ensuremath{\Conid{Con}} \; H^* \; A^*).\ \lpre (K.h^*.0 \;\; K.R^*) \; \\  
& \hspace{2em} (\lambda \imparg{\_ : \sop}.\ \ \ensuremath{\Varid{hdl}} \; K.h^* \; (\ensuremath{\Varid{let\hyp{}in}} \; (\ensuremath{\Varid{val}}\;a) \; K.\ensuremath{\Varid{k}})) \\
=\ & \reason{by axiom \ensuremath{\Varid{let\hyp{}val}} \Cref{eq:fha:co:laws}} \\
& \forall (K : \ensuremath{\Conid{Con}} \; H^* \; A^*).\ \lpre (K.h^*.0 \;\; K.R^*) \; \\  
& \hspace{2em} (\lambda \imparg{\_ : \sop}.\ \ \ensuremath{\Varid{hdl}} \; K.h^* \; (K.\ensuremath{\Varid{k}} \; a)) 
\end{align*}

We put $\hole{?1} = \lambda K.\ \lprf (K.k^* \; a)$, which is well typed because $K.k^*$ has type
\begin{equation}\label{eq:k:under:sop}
k^* : \extty{\ensuremath{\Varid{tm}}^*\;A^* \to \ensuremath{\Varid{tm}}^* \; (h^*\ensuremath{\Varid{.0}} \; R^*)}{\sop}{\lambda a.\ \ensuremath{\Varid{hdl}} \; h^* \; (k \; a) }
\end{equation}
so $K.k^* \; a$ has type $\ensuremath{\Varid{tm}}^*\ (K.h^*.0 \; K.R^*)$,
and $\lprf (K.k^* \; a)$ has type 
\begin{align*}
   & \lpre (K.h^*.0 \; K.R^*) \; (\lambda \imparg{\_ : \sop}.\ (K.k^* \; a)) \\
=\ & \reason{by the restriction of $k^*$ under $\sop$ in \Cref{eq:k:under:sop}} \\
   & \lpre (K.h^*.0 \; K.R^*) \; (\lambda \imparg{\_ : \sop}.\ \lambda a.\ \ensuremath{\Varid{hdl}}
   \; K.h^* \; (K.k \; a)) 
\end{align*}
which is the desired type of \hole{?1}.
\end{element}

\begin{element}
The case for \ensuremath{\Varid{let\hyp{}in}} is similar:
\begin{lgather*}
\ensuremath{\Varid{let\hyp{}in}}^* : \extty{\impfun{H^*, A^*, B^*} \ensuremath{\Varid{co}}^* \; H^* \; A^*
\to (\ensuremath{\Varid{co}}^* \; H^* \; A^* \to \ensuremath{\Varid{co}}^* \; H^* \; B^*) 
\\ \hspace{5em} \to  \ensuremath{\Varid{co}}^* \; H^* \; B}{\sop}{\ensuremath{\Varid{let\hyp{}in}}} \\
\ensuremath{\Varid{let\hyp{}in}}^* \; c \; f = \gluetm{\lambda (K : \ensuremath{\Conid{Con}} \; H^* \; A^*).\ \unglue c \; K'}{\ensuremath{\Varid{let\hyp{}in}} \; c \; f}
\end{lgather*}
where each field of $K' : \ensuremath{\Conid{Con}} \; H^* \; B^*$ is defined as follows:
\begin{align*}
K'.h^* &= K.h^* \\
K'.R^* &= K.R^* \\
K'.k &= \lambda \imparg{\_:\sop}\; a.\ \ensuremath{\Varid{let\hyp{}in}} \; (f \; a) \; K.k\\
K'.k^*  &= \lambda a.\ 
\gluetm{\unglue (f \; a) \; K}{\ensuremath{\Varid{hdl}} \; K'.h^* \; (\ensuremath{\Varid{let\hyp{}in}} \; (f \; a) \; K.\ensuremath{\Varid{k}})}
\end{align*}
The last line type checks because $f \; a : \ensuremath{\Varid{co}}^*\;H^*\;B^*$, so
$\unglue (f \; a) : P_\ensuremath{\Varid{co}} \; (f \; a)$, so by definition
\Cref{eq:m:star:co}, the type of $\unglue (f \; a) \; K$ is
\begin{lgather*}
\lpre (K.h^*.0 \;\; K.R^*) \; (\lambda \imparg{\_ : \sop}.\ \ensuremath{\Varid{hdl}} \; K.h^* \; (\ensuremath{\Varid{let\hyp{}in}} \; (f \; a) \; K.\ensuremath{\Varid{k}}))
\end{lgather*}
which is indeed the type of proofs that the syntactic component of $K'.k^*$
satisfies the logical predicate of the type $k.h^*.0 \; K.R^*$.
\end{element}

\begin{element}
The case for \ensuremath{\Varid{op}} is slightly more involved, so let us first show
that \ensuremath{\Varid{hdl}} satisfies the corresponding logical predicate:
\[
\begin{pmboxed}
\defaultcolumn{@{}l@{}}%
\> \ensuremath{\Varid{hdl}}^* : \extty{\{H^*, A^*\} \> \to  (h^* : \ensuremath{\Conid{Handler}}^* \; H^*) \\ 
\>\> \to \ensuremath{\Varid{co}}^* \; H^* \; A^* \to \ensuremath{\Varid{tm}}^* \; (h^*.0 \; A^*)}{\sop}{\ensuremath{\Varid{hdl}}} \\
\> \ensuremath{\Varid{hdl}}^* \; h^* \; c^* = \gluetm{\unglue c^* \; K }{\ensuremath{\Varid{hdl}} \; h^* \; c}
\end{pmboxed}
\]
where the continuation $K : \ensuremath{\Conid{Con}} \; H^* \; A^*$ is defined by
\begin{align*}
&K.h^* = h^*   &&K.R^* = A^* \\
&K.k    = \lambda \imparg{\_ : \sop}.\ \ensuremath{\Varid{val}} &&K.k^*  = h^*.\ensuremath{\Varid{ret}}
\end{align*}
The definition of $K.k^*$ is well typed because the expected type of $K.k^*$ is
\begin{align*}
 & \extty{\ensuremath{\Varid{tm}}^* \; A^* \to \ensuremath{\Varid{tm}}^* \; (h^*\ensuremath{\Varid{.0}} \; R^*)}{\sop}%
        {\lambda a.\ \ensuremath{\Varid{hdl}} \; h^* \; (k \; a) } \\
=\ & \reason{by the definition of $K.k$ above} \\
 \ & \extty{\ensuremath{\Varid{tm}}^* \; A^* \to \ensuremath{\Varid{tm}}^* \; (h^*\ensuremath{\Varid{.0}} \; A^*)}{\sop}%
        {\lambda a.\ \ensuremath{\Varid{hdl}} \; h^* \; (\ensuremath{\Varid{val}} \; a) } \\
=\ & \reason{by axiom \ensuremath{\Varid{hdl\hyp{}val}} \Cref{eq:fha:evalval}} \\
   &  \extty{\ensuremath{\Varid{tm}}^* \; A^* \to \ensuremath{\Varid{tm}}^* \; (h^*\ensuremath{\Varid{.0}} \; A^*)}{\sop}%
        {\lambda a.\ h^*.\ensuremath{\Varid{ret}} \; R^* \; a } 
\end{align*}
\end{element}

\begin{element}
Coming back to \ensuremath{\Varid{op}}, we start with some obvious steps and a hole:
\begin{lgather*}
\ensuremath{\Varid{op}}^* : \extty{\impfun{H^*, A^*, B^*} 
  \ensuremath{\Varid{tm}}^*\;(H^* .0 \; (\bT^* \; H^*) \; A^*) \\\hspace{4em} 
  \to (\ensuremath{\Varid{tm}}^*\; A^* \to \ensuremath{\Varid{co}}^* \; H^* \; B^*)
   \to \ensuremath{\Varid{co}}^* \; H^* \; B^*}{\sop}{\ensuremath{\Varid{op}}} \\
\ensuremath{\Varid{op}}^* \; o \; k = \gluetm{\lambda(K : \ensuremath{\Conid{Con}} \; H^* \; B^*).\ \hole{?1}}{\ensuremath{\Varid{op}} \; o \; k}
\end{lgather*}
where the hole \hole{?1} has type
\begin{align*}
   & \lpre (K.h^*.0 \;\; K.R^*)\; (\impabs{\_ : \sop} \ensuremath{\Varid{hdl}} \;
K.h^* \; (\ensuremath{\Varid{let\hyp{}in}} \; (\ensuremath{\Varid{op}} \; o \; k) \; K.k)) \\
=\ & \reason{by axiom \ensuremath{\Varid{let\hyp{}op}} \Cref{eq:fha:letop}} \\
  & \lpre (K.h^*.0 \;\; K.R^*)\; (\impabs{\_ : \sop} \ensuremath{\Varid{hdl}} \; 
                                    K.h^* \; (\ensuremath{\Varid{op}} \; o \; (\lambda a.\ \ensuremath{\Varid{let\hyp{}in}}\; (k\;a) \; K.k))) \\
=\ & \reason{by axiom \ensuremath{\Varid{hdl\hyp{}op}} \Cref{eq:fha:evalop}} \\
   &  \lpre (K.h^*.0 \;\; K.R^*)\; (\impabs{\_ : \sop} K.h^*.\ensuremath{\Varid{bind}}
   \; o' \; (\lambda a.\ \ensuremath{\Varid{hdl}}\; \_\; (\ensuremath{\Varid{let\hyp{}in}} \; (k\;a) \; K.k) ))
\end{align*}
where $o' : \ensuremath{\Varid{tm}}^* \;(K.h^*.0 \; A^*)$ is the result of evaluating
the operand $o$ inside the higher-order functor $H$ and then applying the operation
on the monad $K.h^*$:
\[
o' \defeq K.h^*.\ensuremath{\Varid{malg}} \; \_ \; (H^*.\ensuremath{\Varid{hmap}} \; \_ \; \_ \; e \; \_ \; o),
\]
and $e : \ensuremath{\Varid{tm}}^* \; (\ensuremath{\Varid{trans}}^* \; (\MM^*.\bT \; H^*) \; K.h^*.0)$ is
$\ensuremath{\Varid{hdl}}^*$ specialised to $K.h^*$:
\[
e \; A^* \; c = \ensuremath{\Varid{hdl}}^* \; K.h^* \; (\ensuremath{{\Uparrow}} \; c).
\]
Now coming back to the hole \hole{?1}, using $k$ we can define
\begin{lgather*}
f : (a : \ensuremath{\Varid{tm}}^* \; A^*) \to \ensuremath{\Varid{tm}}^* \; (K.h^*.0 \;\; K.R^*) \\
f = \gluetm{\unglue (k \; a) \; K}{\ensuremath{\Varid{hdl}}\; \_\; (\ensuremath{\Varid{let\hyp{}in}} \; (k\;a) \; K.k)}
\end{lgather*}
and finally we can put $\hole{?1} = \lprf (K.h^*.\ensuremath{\Varid{bind}} \; o' \; f)$.
\end{element}

\begin{element}
The last bit of our construction of the glued model $\MM^*$ is showing that it
satisfies the equational axioms of \Fha{} pertaining to computations, but this
is easy because our interpretation of computations and terms in $\MM^*$ is
proof \emph{irrelevant}.
For every universe $U$ of \TTSTC{}, there is a subuniverse 
\[
\Uir = \aset{A : U \mid \forall (a : \impfun{\sop} A).\ (x, y : \extty{A}{\sop}{a}) \to x = y}
\]
which classifies \emph{proof-irrelevant} logical predicates in $U$, in the
sense that partial elements $a : \impfun{\sop} A$ of a type $A : \Uir$ 
have unique total extensions (if exist).
\end{element}

\begin{lemma}
For all $A^* : \ensuremath{\Varid{el}}^* \; \ensuremath{\Varid{ty}}^*$ and $H^* : \ensuremath{\Conid{RawHFunctor}}^*$,
the types 
\[
\ensuremath{\Varid{tm}}^* \; A^* : U_0
\SpcAnd
\ensuremath{\Varid{co}}^* \; H^* \; A^* : U_0
\]
are classified by the subuniverse $\Uir_0$.
\end{lemma}
\begin{proof}
Let $A^*$ be $\gluetm{P}{A}$ where $A : \impfun{\sop}\ensuremath{\Varid{el}}\;\ensuremath{\Varid{ty}}$
is an object-space type
and $P : (\impfun{\sop}\ensuremath{\Varid{tm}}\;A) \to \COmega$
is a meta-space predicate.
By definition \Cref{eq:m:star:tm}, $\ensuremath{\Varid{tm}}^* \; A^*$ is the glue type
$\gluety{a : \ensuremath{\Varid{tm}}\;A}{P\; a}$.
Thus a partial element $a$ of $\ensuremath{\Varid{tm}}^* \; A^*$ is exactly an element 
$a : \impfun{\sop} \ensuremath{\Varid{tm}} \; A$.
Given two elements $x, y : \extty{\ensuremath{\Varid{tm}}^* \; A^*}{\sop}{a}$, 
$\unglue x = \unglue y$ since they are elements of the propositional
type $P\; a$, so
$
x 
= \gluetm{\unglue x}{a}
= \gluetm{\unglue y}{a}
= y$.
The case for $\ensuremath{\Varid{co}}^*$ is similar.
\end{proof}
\begin{corollary}
The glued model $\MM^*$ satisfies the equational axioms \ensuremath{\Varid{val\hyp{}let}}, \ensuremath{\Varid{let\hyp{}val}}, \ensuremath{\Varid{let\hyp{}assoc}}
\Cref{eq:fha:co:laws}, \ensuremath{\Varid{let\hyp{}op}} \Cref{eq:fha:letop}, \ensuremath{\Varid{hdl\hyp{}val}}
\Cref{eq:fha:evalval}, \ensuremath{\Varid{hdl\hyp{}op}} \Cref{eq:fha:evalop} of \Fha{}.
\end{corollary}
\begin{proof}
Taking \ensuremath{\Varid{val\hyp{}let}} for example, we need to show
\begin{align*}
\ensuremath{\Varid{val\hyp{}let}}^* : \{H^*, A^*, B^* \}&\to (a : \ensuremath{\Varid{tm}}^*\; A^*) 
  \to (k : \ensuremath{\Varid{tm}}^*\;A \to \ensuremath{\Varid{co}}^*\; H^* \; B^*)\\
  &\to \ensuremath{\Varid{let\hyp{}in}}^* \; (\ensuremath{\Varid{val}}^* \; H^* \; a) \; k = k \; a 
\end{align*}
Since the type $\ensuremath{\Varid{co}}^*\;H^*\;B^*$ is in the universe $\Uir_0$,
it is sufficient to show that 
$\ensuremath{\Varid{let\hyp{}in}}^* \; (\ensuremath{\Varid{val}}^* \; H^* \; a) \; k$ and $k \; a$
are equal under $\sop$,
\begin{align*}
   & \ensuremath{\Varid{let\hyp{}in}}^* \; (\ensuremath{\Varid{val}}^* \; H^* \; a) \; k \\
=\ & \reason{under $\sop$, $\ensuremath{\Varid{let\hyp{}in}}^* = \ensuremath{\Varid{let\hyp{}in}}$ and $\ensuremath{\Varid{val}}^* = \ensuremath{\Varid{val}}$} \\
   & \ensuremath{\Varid{let\hyp{}in}} \; (\ensuremath{\Varid{val}} \; H^* \; a) \; k \\
=\ & \reason{$\MM$ satisfies \ensuremath{\Varid{val\hyp{}let}}} \\
   & k\;a
\end{align*}
The case for other equational axioms are similar.
\end{proof}

\begin{element}
We have completed the construction of $\MM^*$ and thus proved \Cref{lem:closed:fundamental}.
\end{element}

\section{Parametricity and Free Theorems}\label{sec:parametricity}

\begin{element}
An appealing aspect of the synthetic fundamental lemma (\ref{lem:closed:fundamental})
is that it is proved solely in the language \TTSTC{}, thus applicable to any
category $\CatG$ that models \TTSTC{}.
As an instance, we can deduce the \emph{abstraction theorem}
\citep{Reynolds1983types}, also known as \emph{parametricity}
\citep{Wadler1989}, for System \Fha.
\end{element}

\begin{element}
Let $\MM : \CatJdgOf{\Fha} \to \CatC$ be any model of \Fha{} in a small
LCCC $\CatC$.
We can interpret \TTSTC{} in the Artin
gluing $\CatG_\CatC$ of $\PrC$ and $\Set$ along the global section functor $\Hom_{\PrC}(1, \blank) : \PrC \to \Set$, with the object-space model $\MM$ of \TTSTC{} interpreted as the given functor $\MM : \CatJdgOf{\Fha} \to \CatC$ composed with Yoneda
embedding $\yo : \CatC \to \PrC$.

We have a functor $\overline{\MM^*} : \CatJdgOf{\Fha} \to \CatG_\CatC$ by
instantiating the fundamental lemma (\ref{lem:closed:fundamental}) with $P_\ensuremath{\Varid{can}}$ as in \Cref{eq:m:star:bool}.
For every $\ensuremath{\Conid{A}\mathbin{:}\mathrm{1}\to \Varid{el}\;\Varid{ty}} \in \CatJdgOf{\Fha}$, we let $P_A$ be 
$\overline{\MM^*}\; (\ensuremath{\Varid{tm}\;\Conid{A}}) \in \CatG$ viewed as a predicate
(in the ambient meta-theory) on the set $\CatC(1, \MM\; (\ensuremath{\Varid{tm}\;\Conid{A}}))$.
Similarly, for every $\ensuremath{\Conid{K}\mathbin{:}\mathrm{1}\to \Varid{ki}} \in \CatJdgOf{\Fha}$, we let $P_K$ be 
$\overline{\MM^*}\; (\ensuremath{\Varid{el}\;\Conid{K}}) \in \CatG$ viewed as a family of sets
indexed by the set $\CatC(1, \MM \; (\ensuremath{\Varid{el}\;\Conid{K}}))$.
\end{element}

\begin{theorem}[Unary Parametricity]\label{thm:unary:parametricity}
For every \ensuremath{\Conid{A}\mathbin{:}\mathrm{1}\to \Varid{el}\;\Varid{ty}} and \ensuremath{\Varid{t}\mathbin{:}\mathrm{1}\to \Varid{tm}\;\Conid{A}} in $\CatJdgOf{\Fha}$,
$P_A(\MM \; t)$ holds.
Moreover, for every \ensuremath{\Conid{K}\mathbin{:}\mathrm{1}\to \Varid{ki}} and \ensuremath{\Varid{t}\mathbin{:}\mathrm{1}\to \Varid{el}\;\Conid{K}}, there is an element
$t^* \in P_K(\MM\; t)$.
\end{theorem}

\begin{proof}
Given $t : 1 \to \ensuremath{\Varid{tm}\;\Conid{A}} \in \CatJdgOf{\Fha}$, it is mapped by the logical predicate
model $\overline{\MM^*}$ to a morphism $1 \to \MM^*(\ensuremath{\Varid{tm}\;\Conid{A}})$ in $\CatG_\CatC$,
which amounts to a commutative square:
\[
  \begin{tikzcd}[ampersand replacement=\&]
\aset{\unitel} \& {\aset{t : 1 \to \MM (\ensuremath{\Varid{tm}\;\Conid{A}}) \mid P_A(t)}} \\
{\aset{\unitel}} \& {\Hom_{\PrC}(1, \MM (\ensuremath{\Varid{tm}\;\Conid{A}}))}
\arrow["{!}", from=1-1, to=2-1]
\arrow["{\lambda \unitel.\; \MM t }"', from=2-1, to=2-2]
\arrow["t^*", from=1-1, to=1-2]
\arrow["{\subseteq}", from=1-2, to=2-2]
\end{tikzcd}
\]
The commutativity of the square means that $\MM t$ satisfies $P_A$.

The statement for $t : 1 \to \ensuremath{\Varid{el}\;\Conid{K}}$ is essentially the same, with the element $t^*
\in P_K(t)$ given by the top arrow of the diagram.
\end{proof}

\begin{example}
Parametricity is useful for deriving `free theorems' of programming languages
\citep{Wadler1989}.
As a `hello world'-application, we can use parametricity to deduce that
for every closed \Fha{} term $t : \ensuremath{\Varid{tm}} \; (\ensuremath{\allfor}\; \ensuremath{\Varid{ty}} \; (\lambda \alpha.\; \alpha \mathbin{\ensuremath{{\Rightarrow_t}}} \alpha))$, $t$ applied to every closed type $A$ and closed term 
$a : \ensuremath{\Varid{tm}\;\Conid{A}}$ is equal to $a$.


First of all, internal to \TTSTC{}, we prove the following statement:
\[
\begin{pmboxed}
\defaultcolumn{@{}l@{}}%
\>\ensuremath{\Varid{lem}} \>: (t^* : \ensuremath{\Varid{tm}}^* \; (\ensuremath{\allfor}^*\; \ensuremath{\Varid{ty}}^* \; (\lambda \alpha.\; \; \alpha \mathbin{\ensuremath{{\Rightarrow^*_t}}} \alpha))) \\
\>\> \to (A : \impfun{\sop} \ensuremath{\Varid{el}\;\Varid{ty}}) \to (a : \impfun{\sop} \ensuremath{\Varid{tm}}\;A) \\
\>\> \to \Cmod( \impfun{\sop} t^* \; A \; a = a)\\
\>\ensuremath{\Varid{lem}} \; t^* \; A\; a\; = \hole{?0}
\end{pmboxed}
\]
Recall that $\lprf t^*$ is the proof that the object-space component of 
$t^*$ satisfies its logical predicate.
Expanding definitions (\ref{eq:m:star:tm}, \ref{notation:lpre:lprf}, \ref{el:m:star:p:fun}), we have
\[
\lprf t^* : \forall (\alpha^* : \ensuremath{\Varid{el}}^* \; \ensuremath{\Varid{ty}}^*).\ \lpre (\alpha^* \mathbin{\ensuremath{{\Rightarrow^*_t}}} \alpha^*) \; (\lambda \imparg{\_ : \sop}.\ (t^* \; \alpha^*)).
\]
To use $\lprf t^*$, we define a predicate $A^* : \extty{\ensuremath{\Varid{el}}^* \; \ensuremath{\Varid{ty}}^*}{\sop}{A}$ by
\[
A^* \defeq \gluetm{\lambda x.\; \Cmod (\impfun{\sop} x = a)}{A}
\]
for which only the element $a : {\sop} \to \ensuremath{\Varid{tm}\;\Conid{A}}$ is satisfied.
Now we have
\[
\lprf t^* \; A^* : \lpre (A^* \mathbin{\ensuremath{{\Rightarrow^*_t}}} A^*) \; (\lambda \imparg{\_ : \sop}.\ (t^* \; A)).
\]
Expanding the definition of \ensuremath{{\Rightarrow^*_t}} from \Cref{el:m:star:tfun}, we have
\[
\lprf t^* \; A^* : \forall(x : \impfun{\sop} A).\; \Cmod (\impfun{\sop} x = a) \to \Cmod (\impfun{\sop} t^* \; A \; x = a).
\]
The element $a$ is always equal to itself, so we can complete the hole:
\[
\hole{?0} = \lprf t^* \; A^* \; a \; (\Ceta \; \refl_a).
\]

Now we interpret \ensuremath{\Varid{lem}} in the glued topos $\GFha$.
Evaluating the interpretation of \ensuremath{\Varid{lem}} at $t$, $A$, and $a$, we get a global section
of the interpretation of $\Cmod (t \; A \; a = a)$, which implies $t \; A \; a$ and $a$ are equal morphisms $1 \to \ensuremath{\Varid{tm}\;\Conid{A}}$ in $\CatJdgOf{\Fha}$.
\end{example}

\begin{element}
It is also possible to extend the unary parametricity result above to the binary
(or $n$-ary) case.
Following \citet{sterling2021LRAT},
given two models $M_L : \CatJdgOf{\Fha} \to \CatC$ and $M_R : \CatJdgOf{\Fha} \to \CatD$,
we consider the Artin 
gluing $\CatG_{\CatC\CatD}$ of the product category $\PrC \times \Pr{\CatD}$
and the category of sets along the functor 
\[
\tuple{A, B} \mapsto \CatC(1, A) \times \CatD(1, B).
\]
The category $\CatG_{\CatC\CatD}$ is equivalent to the presheaf topos over 
$(\CatC + \CatD)_\top$, and every object in the category $\CatG_{\CatC\CatD}$ is a tuple
\[
\tuple{A \in \PrC, \; B \in \Pr{\CatD},\; P \in \Set,\; l : P \to \Hom(1, A),\;
r : P \to \Hom(1, B)},
\]
i.e.\ a proof-relevant binary relation (also known as a span) over global
elements of the presheaves $A$ and $B$.
The category $\CatG_{\CatC\CatD}$ has two subterminal objects
\[
\spread*{
\sop_L \defeq \tuple{1_{\PrC}, 0, \emptyset, !, !}
\also
\text{and}
\also
\sop_R \defeq \tuple{0, 1_{\PrD}, \emptyset, !, !},
}
\]
which determine two open subtoposes that are equivalent to $\PrC$ and $\PrD$
respectively.
The disjunction of $\sop_L$ and $\sop_R$ is another subterminal object
\[
\sop \defeq \tuple{1_{\PrC}, 1_{\PrD}, \emptyset, !, !},
\]
whose corresponding open subtopos is equivalent to $\Pr{(\CatC + \CatD)}$.
\end{element}

\begin{element}
The type theory $\TTSTC$ can be interpreted in $\CatG_{\CatC\CatD}$ as usual,
with $\sop : \Omega$ interpreted as the subterminal object $\sop$ above.
Moreover, we can extend $\TTSTC$ with the following new constants with the
evident interpretation in $\CatG_{\CatC\CatD}$:
\begin{gather*}
\spread*{
\sop_L : \Omega
\also
\sop_R : \Omega
\also
\_ : \sop_L \lor \sop_R = \sop
\also
\_ : \sop_L \land \sop_R = \emptyty
} \\
\spread*{
M_L : \impfun{\sop_L} \sem{\Fha}_{U_0}
\also
M_R : \impfun{\sop_R} \sem{\Fha}_{U_0}
}\\
\_ : \MM = \impabs{z : \sop} \ensuremath{\lhskeyword{case}\;\Varid{z}\;\lhskeyword{of}}\; \{ 
  \inl\; (\_ : \sop_L) \mapsto  M_L;\ 
  \inr\; (\_ : \sop_R) \mapsto M_R \}
\end{gather*}
We refer to the extended language by \TTSTCLR.
\end{element}

\begin{element}
The synthetic fundamental lemma (\ref{lem:closed:fundamental}) holds in
\TTSTCLR{} without needing any modification, since \TTSTCLR{} only adds new
axioms to \TTSTC.
However, in $\TTSTCLR$ an $\sop$-partial element $\impfun{\sop} A$ of some type
$A$ is now equal to an element of type $\impfun{\sop_L \lor \sop_R} A$, which are equivalently
two partial elements $\impfun{\sop_L} A$ and $\impfun{\sop_R} A$.
Therefore the unary logical predicates in the proof of
\Cref{lem:closed:fundamental} can be now read as binary logical
{relations}.

Specially, we can set both $M_L$ and $M_R$ to be $\IdF : \CatJdgOf\Fha
\to \CatJdgOf\Fha$, and we obtain the binary version of parametricity of
closed \Fha-terms (\Cref{thm:unary:parametricity})
by instantiating the fundamental lemma with the logical 
relation $P$ for \ensuremath{\Varid{bool}} to be equality (this relation cannot be internally defined
in \TTSTCLR {} though, since this relation only makes sense when $M_L = M_R$).
\end{element}

\section{The Realizability Model for General Recursion}\label{sec:fhar:real:model}

\begin{element}
Simply typed $\lambda$-calculi with general recursion famously can be modelled
by variations of \emph{complete partial orders}
from classical domain theory \citep{plotkin_lcf_1977,Scott1993,streicher2006domain}.
However, the language \Fhar{} has impredicative polymorphism, which is 
very tricky to model using classical domain theory, although not impossible 
\citep{Crole_1994,Coquand1989}.
\end{element}

\begin{element}
Alongside a few other reasons, the difficulty of modelling polymorphism in
classical domain theory motivated the development of \emph{synthetic domain
theory} (SDT) \citep{rosolini1986continuity,Hyland_first_1991,PhoaThesis}.
The idea of SDT is to axiomatise `domains', in the general sense of objects that
provide meaning to (recursive) programs,
as special `sets' satisfying certain
properties in the logic of toposes or constructive set theory, so that
every function between those special sets is automatically a `continuous map'
between domains. 
In this way, one can give denotational semantics to recursive programs
in a naive set-theoretic way.
\end{element}

\begin{element}
The exact axiomatisation of SDT varies across authors, but there
are mainly two kinds of models:
realizability toposes \citep{PhoaThesis,Longley_Simpson_1997} and Grothendieck toposes \citep{Fiore_extension_1997,Fiore_Rosolini_1997}.
Since we are already modelling \Fha{} using a realizability model in
\Cref{sec:fha:real:model}, we will stick with the realizability model,
following the ideas of SDT concretely in this model (as opposed to using SDT as
an axiomatic language).
\end{element}

\begin{element}
The rest of this section is a short introduction to SDT based on  \varcitet{Longley_Simpson_1997}{'s} approach
using \emph{well complete objects}, adapted to a type-theoretic language. 
See also \varcitet{Longley_1995}{'s} thesis,
the more general treatment by
\citet{simpson_computational_2004,Simpson_computational_1999}, and the type-theoretic formalisation of SDT using \emph{well complete $\Sigma$-spaces} by
\citet{reus1996program,reus_formalizing_1999} and \citet{reus_general_1999}.
With the machinery of SDT in this section, the interpretation of \Fhar{} will be almost trivial and will be presented in the next section.
\end{element}

\begin{element}
Before going into SDT, let us quickly recall a typical setup of interpreting
recursion in classical domain theory, which we are going to mirror in the
SDT.

A \emph{predomain} (or precisely, an $\omega$-cpo, in this setup) is a partially ordered set $\atuple{A, \sqsubseteq}$  that has suprema $\sqcup_i a_i$ for all $\omega$-chains $a_0 \sqsubseteq a_1 \sqsubseteq a_2 \sqsubseteq a_3 \sqsubseteq \cdots$ in $A$; a predomain need not have a bottom
element.
Morphisms
between predomains are monotone functions preserving those suprema of $\omega$-chains.

A \emph{(Scott-) open set} of a predomain $A$ is a subset $O \subseteq A$ that is
(1) upward closed: for all $x, y \in A$, if $x \in O$ and $x \sqsubseteq y$
then $y \in O$, and 
(2) continuous: for all $\omega$-chains $a_i$ in $A$,
  $\sqcup_i a_i \in O$ iff there exists some $n$ such that $a_n \in O$.
Open sets of a predomain $A$ are in bijection with morphisms $A \to \fS$, where
$\fS$  is the two-element predomain $\aset{\bot \sqsubseteq \top}$, sometimes called the \emph{\Sier{} space} ($\fS$ is the Fraktur letter for \emph{S}).
Namely, every open set $O \subseteq A$ corresponds to the morphism $\chi : A \to \fS$ where $\chi(a) = \top$ if $a \in O$ and $\chi(a) = \bot$ if $a \not\in O$.

The \emph{lifting monad} $L A$ on predomains adjoins a new \emph{bottom element} $\bot$ to $A$, with a monad structure similar to that of the monad
$1 + \blank$ on sets.
Kleisli morphisms of predomains $f : \Gamma \to L A$ 
are in bijection with \emph{partial} morphisms $\tuple{O, \bar{f}} : \Gamma \pto A$, 
each consisting of an open set $O \subseteq \Gamma$ and a (total)
morphism $\bar{f} : O \to A$.

A \emph{domain} $D$ is a predomain with bottom element $\bot_D$, which is the same as 
an Eilenberg-Moore algebra of the lifting monad $L$.
Every endo-morphism $f : D \to D$ on domains then has a least fixed point by taking the
supremum of the chain $\bot_D \sqsubseteq f(\bot_D) \sqsubseteq f(f(\bot_D)) \sqsubseteq \cdots$ in $D$.

Contexts $\Gamma$ and types $\sigma$ of a call-by-value programming language
with recursion are then interpreted as predomains $\sem{\Gamma}, \sem{\sigma}$. Terms $\Gamma\vdash t : \sigma$ are interpreted as morphisms
$\sem{\Gamma} \to L \sem{\sigma}$, i.e.\ partial morphisms between predomains.
\end{element}

\begin{element}
Recall that the internal language of assemblies $\Asm(\bA)$ over a partial
combinator algebra $\bA$, which we used to construct a model of \Fha{} in
\Cref{sec:real:mod:fha}, is an extensional MLTT with three cumulative universes $P : V_1 : V_2$ such that
\begin{itemize}
\item
each closed under the unit type, $\Sigma$, $\Pi$, and inductive types ($W$-types);

\item
for all types $A$ and $a, b : A$, the equality type $a = b$ is
in the universe $P$;

\item
for all types $A$ and $P$-valued type families $B : A \to P$,
$\Pi\; A\; B$ is in $P$.
\end{itemize}
The interpretation of $P$ is the assembly of modest sets (i.e.\ PERs), 
and $V_i$ is the assembly of $U_i$-small assemblies, for universe of sets $U_i$
in the meta-theory.
Details of the interpretation  can be found in 
Reus's thesis \citep[\S 8]{reus1996program}.
\end{element}

\begin{element}
In the following, we will further fix the PCA $\bA$ to be
\emph{Kleene's first algebra} $\mathbb{K}$ \citep{Oosten_2008}, whose elements are natural numbers 
(which intuitively play dual roles as both \emph{data} and \emph{computation} via G\"{o}del codes
of Turing machines), and partial application $n \; m$ is defined to be
$\phi_n(m)$, the possibly divergent result of running the Turing machine
coded by $n$ with input $m$.
We will write $n \; m \diverges$ and $n\;m\converges$ to mean that the partial application diverges and converges respectively.

Specialising $\bA$ to $\mathbb{K}$ is only for providing more intuition, and
interested readers can consult \citet{Longley_Simpson_1997} to see how it can be
done more generally with an arbitrary PCA equipped with a notion of \emph{divergence}.
\end{element}

\begin{element}
A type $A$ is said to be a \emph{proposition} if the type $\isprop \; A \defeq (a, b : A) \to x = y$ is inhabited
\citep{hottbook}.
The subuniverse $P_{-1} \subseteq P$ of \emph{propositional modest sets}
is then defined by
\begin{lgather*}
P_{-1} : V_1 \\  
P_{-1} = \Sigma(A : P).\; \isprop\; A
\end{lgather*}
whose elements are decoded as types by first projection $\fst$, which we will
leave as implicit, as if $P_{-1}$ is a Russell-style universe.

It might be useful to see an external description for the universe $P_{-1}$,
in the sense of universes in categories.
The semantics of the universe $(P_{-1}, \fst)$ is
(isomorphic to) an assembly morphism
$i : \tilde{P}_{-1} \to P_{-1}$, where $P_{-1}$ has an underlying set
containing all \emph{sub-singleton modest sets} $A$, and $\realizes[P_{-1}]{r}{A}$
holds for all $r$ and $A$.  The assembly $\tilde{P}_{-1}$ has an underlying
set containing all modest sets $A$ with exactly one element $a \in
\avert{A}$, and $\realizes[\tilde{P}_{-1}]{r}{A}$ if and only if
$\realizes[A]{r}{a}$.
The morphism $i : \tilde{P}_{-1} \to P_{-1}$ is the inclusion morphism.


From the above explicit description, we can see that the universe $P_{-1}$ satisfies Voevodsky's \emph{propositional resizing axiom}: 
in the language of $\AsmK$,
for every propositional type $A$, there is some $\classify{A} : P_{-1}$
isomorphic to $A$, since a sub-singleton assembly $A$ is 
necessarily a modest set.
\end{element}

\begin{element}
The universe $P_{-1}$ is similar to the universe $\Omega$ of propositions in elementary toposes as axiomatised in \ref{el:ttstc:topos},
and we can define logical connectives $\top$, $\bot$, $\land$, $\lor$, $\forall$, and $\exists$ on $P_{-1}$ in exactly the same way as we do in elementary toposes.

The crucial difference between $P_{-1}$ and the universe $\Omega$ in elementary
toposes is that $P_{-1}$ is \emph{not} univalent: 
given $p, q : P_{-1}$ with $p \cong q$, it is not always the case that $p = q$.
Indeed, two
singleton modest
sets $\tuple{\aset{\unitel}, \realizes[p]{}{}}$
and $\tuple{\aset{\unitel}, \realizes[q]{}{}}$
can have different realizing relations even when their underlying sets are exactly the same.

This seemingly insignificant flaw of $P_{-1}$ has an impact bigger than one may expect on doing mathematics
internal to $\AsmK$;
For one thing, we cannot construct \emph{quotient types} using
$P_{-1}$ in the way that it is usually done in elementary toposes.

We can switch to the realizability topos to use the better-behaved
universe $\Omega$, but we will stay in category of assemblies, as it turns out to be good enough for carrying out 
our development, and more importantly, the simplicity
of $\AsmK$ allows us to give simple external
descriptions of many constructions of SDT, which I found essential when learning SDT for the first time.
\end{element}


\begin{element}
The universe $P_{-1}$ has a subuniverse of \emph{semi-decidable} propositions:
\begin{equation}\label{eq:semidec:subuniv}
P_{-1}^s \defeq 
\aset{p : P_{-1} \mid \exists f : \bN \to 2.\; p \cong (\exists n : \bN.\; f\;n = 0)}.
\end{equation}
Roughly speaking, a proposition in $P_{-1}^s$ is determined by
the (semi-decidable) property of the existence of 
zero points for a computable function $f : \bN \to 2$.

A simple external description of $P_{-1}^s$ is available:
the assembly $P_{-1}^s$ is isomorphic, \emph{up to bi-implication} in $P_{-1}$,  to the assembly
$\fS \defeq \tuple{\aset{\bot, \top}, \realizes[\fS]{}{}}$ where
\[
\spread*{
\realizes[\fS]{r}{\bot} \ifft r \; 0 \diverges
\also
\text{and}
\also
\realizes[\fS]{r}{\top} \ifft r \; 0 \converges.
}
\]

In sketch, the direction $\fS \to P_{-1}^s$ sends $\bot$ and $\top$ to the empty and terminal assemblies respectively, and is realized by the Turing
machine accepting $r$ and returning the computable function $f : \bN \to 2$ that
accepts $n$ and  returns $0$ if and only if running the Turing machine
$r$ halts in $n$ steps.
The other direction $P_{-1}^s \to \fS$ sends an assembly to $\top$
iff it is non-empty, and this is realized by the Turing machine
accepting the code for $f : \bN \to 2$ (and $p \cong \exists n.\; f\;n = 0$)
and returns the machine $r$ that searches for a zero point of $f$ iteratively.
\end{element}

\begin{element}
The type $\fS$ can be viewed as a universe directly: every element $p : \fS$
is decoded as the equality type $p = \top$.
Although $P_{-1}^s$ and $\fS$ are equivalent universes,
we will prefer using the universe $\fS$ over $P_{-1}^s$ because $\fS$ is {univalent}: 
\[
(p, q : \fS) \to (p \cong q) \to (p = q).
\]

The universe $\fS$ is closed under truth $\top : \fS$ and dependent conjunction $\Sigma (p : \fS).\; q(p) : \fS$ for all $p : \fS$ and $q : p \to \fS$.
Therefore, it is a \emph{dominance} \citep{rosolini1986continuity}, which
is the fundamental notion in general SDT \citep{Hyland_first_1991}.
In the present situation, the dominance $\fS$ is moreover closed under falsity $\bot$ and 
countable disjunction $\exists n : \bN.\;p(n)$ for $p : \bN \to \fS$.

As suggested by the notation, the universe $\fS$ of semi-decidable propositions will play the role of the \Sier{} space $\aset{\bot \sqsubseteq \top}$ in classical domain theory.
In the internal language, a (Scott-) open of a type $A$ is defined as a function $O : A \to \fS$, giving rise to a subtype $\aset{a : A \mid  O\;a = \top}$, which we shall usually just write as $O$ when no confusion.
An ($\fS$-) partial function $\Gamma \pto A$ is again an open set
$O$ of $\Gamma$ with a function $O \to A$.
Externally, an open set of an assembly $\tuple{\avert{A}, \realizes[A]{}{}}$
is a subset $O \subseteq A$ such that there is a Turing machine $r$
satisfying that whenever $\realizes[A]{n}{a}$, then $r\;n \converges$ iff $a \in O$.
\end{element}

\begin{element}
Analogous to the lifting monad in classical domain theory, we have a lifting monad
\begin{lgather*}
L : P \to P \\
L \; A = \Sigma (p : \fS).\; (\impfun{p} A)
\end{lgather*}
on modest sets in the internal language of $\AsmK$.
We can actually define $L$ on all types but we shall only need it on $P$.
The monad structure for $L$ is 
\begin{gather*}
\spread*{
\begin{array}{l}
\eta :  A \to L \; A \\
\eta\; a = (\top , a)
\end{array}%
\also%
\begin{array}{l}
\mu : (A \to L\;B) \to L \; B \\
\mu \; (p, a) \; k = (\Sigma (\_ : p).\; \fst \; (k \; a),\  \snd \; (k\; a)
)  
\end{array}%
}
\end{gather*}

An isomorphic external description of the monad $L$ is that it sends every modest set
$\tuple{\avert{A}, \realizes[A]{{}}{}}$ to the modest set $\tuple{1 + \avert{A}, \realizes[L A]{{}}{}}$ where
\begin{align*}
\realizes[L A]{r}{\inl\;\unitel} &\ifft r \; 0 \diverges, \\
\realizes[L A]{r}{\inr\;a} &\ifft {r \; 0 \converges} \;\land\; {\realizes[A]{r \; 0}{a}}.
\end{align*}
That is to say, if a Turing machine $r$ diverges on the input $0$ then it realizes the `bottom element' $\inl\;\unitel$, and otherwise
$r$ realizes $\inr\; a$, for elements $a \in \avert{A}$ that are realized by $r\;0$.
The input $0$ here is completely arbitrary, and the definition will
be isomorphic if $0$ is replaced by any other fixed number or $r$ itself.
The monad structure on $L$ is the same as the one on
$1 + \blank : \Set \to \Set$; see \citet[\S 4]{Longley_Simpson_1997} for details.

Note that $L A$ is not the same as the coproduct $1 + A$ in $\AsmK$.
The latter has the same underlying set $1 + \avert{A}$, but the existence
predicate of $1 + A$ is
\begin{align*}
\realizes[1 + A]{r}{\inl\;\unitel} &\ifft \fst\; r = 0 \\
\realizes[1 + A]{r}{\inr\;a} &\ifft \fst\; r \neq 0 \land \realizes[A]{\snd\; r}{a},
\end{align*}
where $\fst$, $\snd$, and $\tuple{\blank, \blank}$ are some Turing machines implementing projections and pairing of natural numbers $\bN \times \bN \cong \bN$. 
The crucial difference between $L A$ and $1 + A$ is that morphisms of assemblies $f : X \to 1
+ A$ must be realised by Turing machines that can \emph{decide} whether $f(x)$ is
$\inr\;a$ given a realizer of $x$, while morphisms $f : X \to L A$ need only be realised by
Turing machines that \emph{semi-decide} whether $f(x)$ is $\inr\;a$.
Thus $L A$ is the right choice of the lifting monad capturing the idea of
`possibly divergent elements of $A$'.
\end{element}

\begin{element}
As an \emph{endofunctor} on the universe $P$, $L$ has both a final coalgebra $\tuple{\oomega : P, \sigma : \oomega \to L \oomega}$ 
and an initial algebra $\tuple{\omega : P, \tau : L  \omega \to \omega}$.
The following formulae of $\oomega$ and $\omega$ are due to \citet{Jibladze_1997}:
\begin{gather*}
\oomega = \aset{f : \bN \to \fS \mid \forall n.\; f \; (n + 1) \to f \; n} \\
\omega = \aset{f : \oomega \mid \forall p : P_{-1}.\; \big(\forall (n : \bN).\; (f\; n \to p) \to p\big) \to p}
\end{gather*}
which in fact works for any dominance in any elementary topos with a natural number
object (with $P_{-1}$ in the formula of $\omega$ replaced by $\Omega$).

Again, we have simple external descriptions for $\oomega$ and $\omega$ in the case of $\AsmK$.
The carrier $\oomega$ is isomorphic to the assembly
$\tuple{\bN \cup \aset{\infty}, \realizes[\oomega]{}{}}$ where
\begin{align*}
\realizes[\oomega]{r}{n} &\ifft \forall k \in \bN.\; (k < n) \leftrightarrow (r \; k \converges) \\  
\realizes[\oomega]{r}{\infty} &\ifft \forall k \in \bN.\; r \; k \converges
\end{align*}
with $\sigma : \oomega \to L \oomega$ given by $\sigma(0) = \inl \; *$, $\sigma(n+1) = \inr \; n$, and $\sigma(\infty) = \inr \; \infty$.
The type $\omega$ is given by the assembly
$\tuple{\bN, \realizes[\omega]{}{}}$ that
restricts $\oomega$ to the sub-underlying set $\bN$.
The algebra $\tau : L \omega \to \omega$ is simply $\tau(\inl\;*) = 0$ and
$\tau(\inr\;n) = n + 1$.

From the explicit description we see that the assembly $\omega$ is a non-standard representation of natural numbers as an assembly, different from the   
standard representation $\tuple{\bN, \aset{(n, n) \mid n\in\bN}}$ that satisfies the universal property of a
natural number object in $\AsmK$.
In $\omega$, every natural number $n \in \bN$ is represented as a Turing machine that halts exactly for inputs $k$ smaller than $n$.
Since Turing machines are unable to tell whether other
machine halts, assembly maps $\omega \to A$ are constrained to be `continuous' in a sense.
\end{element}

\begin{element}
Let $\kappa : \omega \to \oomega$ be the canonical inclusion morphism 
(given as the unique algebra homomorphism from the initial algebra $\tau : L \omega \to \omega$ to the $L$-algebra $\sigma^{-1} : L\oomega \to \oomega$).
The morphism $\kappa$ plays an important role in synthetic
domain theory:
a morphism $c : \omega \to X$ of assemblies will play the role of an $\omega$-chain 
of elements $c_0 \sqsubseteq c_1 \sqsubseteq \cdots$ in a partial order $X$.
Similarly, a morphism $c^* : \oomega \to X$ is analogous to a chain $c_i$ together with its supremum $\sqcup_i c_i$.
\end{element}

\begin{definition}\label{def:complete:mod:set}
A modest set $X : P$ is called \emph{complete} if the function 
\[(\blank \vcomp \kappa) : (\oomega \to X) \to (\omega \to X)\] is an 
isomorphism, i.e.\ the following proposition holds
\[
\ensuremath{\Varid{complete}} \; X \defeq 
\exists \overline{(\blank)} : (\omega \to X) \to (\oomega \to X).\ (\forall c.\ 
  \bar{c} \vcomp \kappa = c) \land (\forall d.\ \overline{d \vcomp \kappa} = d)
\]
internally in $\AsmK$.
A modest set $X : P$ is called \emph{well complete} if $L X$ is complete.
\end{definition}

\begin{element}
Well complete types will be our synthetic version of predomains:
\begin{lgather*}
\PDom : V_1 \\
\PDom = \Sigma(A : P).\; \ensuremath{\Varid{complete}} \; (L\;A)
\end{lgather*}
They are intuitively modest sets in which an $\omega$-chain of partial
elements has a unique (partial) supremum.
Their crucial difference from predomains in classical domain theory
is that they are just `sets' satisfying a property, rather than sets
carrying additional data (the partial order).
\end{element}

\begin{element}
There are some nuances in the external meaning of (well) completeness.
Firstly, we notice that $\ensuremath{\Varid{complete}} : P \to P_{-1}$ in \Cref{def:complete:mod:set} is a proper \emph{realizability predicate} in the sense that $\ensuremath{\Varid{complete}\;\Conid{X}}$ for
a modest set $X : P$ has non-trivial realizers.
Namely, $\ensuremath{\Varid{complete}\;\Conid{X}}$ is realized by machines sending realizers of $\omega \to X$
to realizers of $\oomega \to X$, in a way that is an inverse to $\blank \vcomp \kappa$.
Secondly, $\ensuremath{\Varid{complete}\;\Conid{X}} : P_{-1}$ makes sense in an arbitrary context
$\Gamma$ in the internal language of $\AsmK$.
Therefore, externally $X$ is not one modest set but
a family of modest sets $\Gamma \to P$ indexed by
an assembly $\Gamma$ of the context.

If we forget about realizers of completeness and consider only \emph{global}
elements $X : P$, then $\ensuremath{\Varid{complete}\;\Conid{X}}$ has a realizer if and only if the morphism
$X^\kappa : X^{\oomega} \to X^\omega$ in $\AsmK$ is an isomorphism. 
The latter condition is precisely the definition of completeness for an object in
$\Asm(\bA)$ by \citet{Longley_Simpson_1997}.
This further means that for all assemblies $\Gamma$ and 
$c : \Gamma \times \omega \to X$, there is a unique $\bar{c} : \Gamma \times \oomega \to X$ making the following
diagram commute:
\[
\begin{tikzcd}[ampersand replacement=\&]
	{\Gamma \times \omega} \& X \\
	{\Gamma \times \oomega}
	\arrow["c", from=1-1, to=1-2]
	\arrow["{\kappa}"', from=1-1, to=2-1]
	\arrow["{\bar{c}}"', from=2-1, to=1-2]
\end{tikzcd}
\]
To see this, by Yoneda embedding, $X$ is complete if and only if 
\[
(X^\kappa \vcomp \blank) : \Hom(\Gamma, X^{\oomega}) \to \Hom(\Gamma, X^\omega)
\]
is an isomorphism, natural in $\Gamma$.
By adjointness, this is equivalent to asking
\[
(\blank \vcomp \Gamma \times \kappa) : \Hom(\Gamma \times \oomega, X) \to \Hom(\Gamma \times \omega, X)
\]
to be a natural isomorphism.
A natural transformation is a natural isomorphism iff every component of it
is an isomorphism, so $X$ is complete if and only if for every $c : \Gamma \times \omega \to X$, there is a unique $\bar{c} : \Gamma \times \oomega \to X$ such
that $\bar{c} \vcomp \Gamma \times \kappa = c$.
\end{element}

\begin{theorem}\label{thm:pdom:closure}
The subuniverse $\PDom \subseteq P$ is closed under liftings $L$, 
the unit type, $\Sigma$-types, $\Pi$-types, equality types, coproducts,
the natural number type $\bN$ in $P$.
Moreover, predomains are also complete (i.e.\ well completeness implies
completeness).
\end{theorem}
\begin{proof}
This is essentially shown by \citet[\S 7]{Longley_Simpson_1997} aside from
the difference between \citeauthor{Longley_Simpson_1997}'s external definition
of completeness and our internal definition.
\citeauthor{Longley_Simpson_1997} defined
completeness of an assembly $X$ as a proposition in the ambient logic ($X^\kappa :
X^{\oomega} \to X^\omega$ being an isomorphism), while our definition is in internal in the language of $\AsmK$, which has non-trivial realizers (Turing machines accepting code of $c : \omega \to X$ and
outputting code of $\oomega \to X$).
Therefore, we have to check that the proofs of the closure properties by \citeauthor{Longley_Simpson_1997} are realizable.
For example, to show $L : P \to P$ restricts to $L : \PDom \to \PDom$, 
we need to check that there is a Turing machine sending 
realizers of $X$ being well complete 
to realizers of $L X$ being well complete.
This is indeed the case by observing that the proofs by \citeauthor{Longley_Simpson_1997} can be carried internally in the language of $\AsmK$.
\end{proof}

\begin{definition}
Mirroring the setup of classical domain theory,
the universe of \emph{domains} is defined as the type of Eilenberg-Moore algebras of the monad $L : \PDom \to \PDom$:
\begin{hscode}\SaveRestoreHook
\column{B}{@{}>{\hspre}l<{\hspost}@{}}%
\column{3}{@{}>{\hspre}l<{\hspost}@{}}%
\column{10}{@{}>{\hspre}l<{\hspost}@{}}%
\column{42}{@{}>{\hspre}l<{\hspost}@{}}%
\column{E}{@{}>{\hspre}l<{\hspost}@{}}%
\>[B]{}\lhskeyword{record}\;\DOM\mathbin{:}\Conid{V}_{\mathrm{1}}\;\lhskeyword{where}{}\<[E]%
\\
\>[B]{}\hsindent{3}{}\<[3]%
\>[3]{}\Conid{A}{}\<[10]%
\>[10]{}\mathbin{:}\PDom{}\<[E]%
\\
\>[B]{}\hsindent{3}{}\<[3]%
\>[3]{}\Varid{\alpha}\mathbin{:}\Conid{L}\;\Conid{A}\to \Conid{A}{}\<[E]%
\\
\>[B]{}\hsindent{3}{}\<[3]%
\>[3]{}\anonymous {}\<[10]%
\>[10]{}\mathbin{:}(\Varid{x}\mathbin{:}\Conid{A})\to \Varid{\alpha}\;(\Varid{\eta^L}\;\Varid{x})\mathrel{=}\Varid{x}{}\<[E]%
\\
\>[B]{}\hsindent{3}{}\<[3]%
\>[3]{}\anonymous {}\<[10]%
\>[10]{}\mathbin{:}(\Varid{x}\mathbin{:}\Conid{L}\;(\Conid{L}\;\Conid{A}))\to \Varid{\alpha}\;(\Varid{\mu^L}\;{}\<[42]%
\>[42]{}\Varid{x})\mathrel{=}\Varid{\alpha}\;(\Conid{L}\;\Varid{\alpha}\;\Varid{x}){}\<[E]%
\ColumnHook
\end{hscode}\resethooks
As usual, given $D : \DOM$, we usually write the type $\fst\; (D.A)$ as just $D$.
\end{definition}

\begin{element}
The crucial property of domains $D : \DOM$ is that they admit fixed points for all endofunctions:
\begin{gather}\label{el:sdt:domain:fixed:point}
\ensuremath{\Varid{fix}} : \impfun{D : \DOM} (f : D \to D) \to D 
\end{gather}
which is defined as follows: 
first we 
define a function $\alpha_f : L \; D \to D$ by $\alpha_f = \alpha \vcomp L f$.
By the initiality of $\tau : L\omega \to \omega$, we have a homomorphism
$c : \omega \to D$.
Then using the completeness of $D$, we have $\bar{c} : \oomega \to D$,
and we let $\ensuremath{\Varid{fix}\;\Varid{f}} \defeq \bar{c} \; \infty$.
It is then the case that $f\;(\,\ensuremath{\Varid{fix}\;\Varid{f}}) = \ensuremath{\Varid{fix}\;\Varid{f}}$
\citep[Theorem 7.3]{reus_general_1999}.
\end{element}

\begin{element}
We note that by \citet[Theorem 5.6]{Longley_Simpson_1997},
Eilenberg-Moore $L$-algebras on a predomain are unique if exist, so it makes 
sense to say `a predomain is a domain' as a proposition.
\end{element}

\subsection{The Interpretation of \Fhar}

\begin{element}
Now we are ready to construct a ($V_2$-small) model of the signature \Fhar{} (\Cref{sec:gen:rec}) in $\AsmK$.
Our goal is to define an element $\rM$
of the record type $\sem{\Fhar}_{V_2}$ containing
all the declarations of \Fhar{} with $\Jdg$ replaced by the universe $V_2$.
\end{element}

\begin{element}
The non-recursive fragment of \Fhar{} will be interpreted in almost the same way
as the model $\RM$ of \Fha{} in \Cref{sec:real:mod:fha}, except that all the occurrences of the
universe $P : V_1$
will be replaced by the subuniverse $\PDom : V_1$.
For example, $\rM.\ensuremath{\Varid{ty}}$ is now $\PDom$ instead of $P$, and the computation
judgement $\rM.\ensuremath{\Varid{co}}$ is now
\begin{lgather*}
\rM.\ensuremath{\Varid{co}} : \rM.\ensuremath{\Conid{RawHFunctor}} \to \rM.\ensuremath{\Varid{el}} \; \rM.\ensuremath{\Varid{ty}} \to \PDom \\
\rM.\ensuremath{\Varid{co}} \; H \; A = (h : \rM.\ensuremath{\Conid{Handler}\;\Conid{H}}) \to (B : \PDom) \to (A \to h\;B) \to h \; B
\end{lgather*}
The constructions in \Cref{sec:real:mod:fha} still work because by \Cref{thm:pdom:closure}, the universe
$\PDom$ is closed under the type formers that we used to interpret \Fha{}, in particular, impredicative $\Pi$-types.
\end{element}

\begin{element}
The empty type from \Cref{eq:fhar:empty} is as expected interpreted as the empty
modest set $0$, which is trivially well complete.
\end{element}

\begin{element}
The interesting thing is modelling partial computations \ensuremath{\Varid{pco}} \Cref{el:fhar:pco}.
In \Cref{sec:gen:rec}, we had a type \ensuremath{\Conid{HandlerRec}} of monads supporting
recursion (and some effectful operations).
However, we cannot simply define $\rM.\ensuremath{\Varid{pco}}$ by
replacing \ensuremath{\Conid{Handler}} in the definition of $\rM.\ensuremath{\Varid{co}}$ above with
\ensuremath{\Conid{HandlerRec}}, because 
the definition of $\ensuremath{\Conid{HandlerRec}}$ depends on $\rM.\ensuremath{\Varid{pth}}$ and thus $\rM.\ensuremath{\Varid{pco}}$.

The type \ensuremath{\Conid{HandlerRec}} ensures that a monad $M$ in \Fhar{} supports recursion by 
requiring the monad $M$ to be partial thunks \emph{syntactically}.
What we need here is a semantic counterpart of monad supporting recursion:
\begin{hscode}\SaveRestoreHook
\column{B}{@{}>{\hspre}l<{\hspost}@{}}%
\column{3}{@{}>{\hspre}l<{\hspost}@{}}%
\column{E}{@{}>{\hspre}l<{\hspost}@{}}%
\>[B]{}\lhskeyword{record}\;\Conid{HandlerL}\;(\Conid{H}\mathbin{:}\Conid{RawHFunctor})\mathbin{:}\Conid{V}_{\mathrm{2}}\;\lhskeyword{where}{}\<[E]%
\\
\>[B]{}\hsindent{3}{}\<[3]%
\>[3]{}\lhskeyword{include}\;\Conid{Handler}\;\Conid{H}\;\lhskeyword{as}\;\Varid{h}{}\<[E]%
\\
\>[B]{}\hsindent{3}{}\<[3]%
\>[3]{}\Varid{dom}\mathbin{:}(\Conid{B}\mathbin{:}\PDom)\to \{\mskip1.5mu \Conid{D}\mathbin{:}\DOM\;\mid\Conid{\Conid{D}.A}\mathrel{=}\Varid{h}\;\Conid{B}\mskip1.5mu\}{}\<[E]%
\ColumnHook
\end{hscode}\resethooks
which requires that \ensuremath{\Varid{h}\;\Conid{B}\mathbin{:}\PDom} is a domain for all \ensuremath{\Conid{B}\mathbin{:}\PDom}. 

The model of partial computations is then
\begin{lgather*}
\rM.\ensuremath{\Varid{pco}} : \rM.\ensuremath{\Conid{RawHFunctor}} \to \rM.\ensuremath{\Varid{el}} \; \rM.\ensuremath{\Varid{ty}} \to \PDom \\
\rM.\ensuremath{\Varid{pco}} \; H \; A = (h : \ensuremath{\Conid{HandlerL}\;\Conid{H}}) \to (B : \PDom) \to (A \to h\;B) \to h \; B
\end{lgather*}
The models of the declarations \ensuremath{\Varid{val}}, \ensuremath{\Varid{let\hyp{}in}}, \ensuremath{\Varid{pth}} and \ensuremath{\Varid{op}} are the same as
those for \ensuremath{\Varid{co}} in \Cref{sec:fha:real:model}, which we shall not repeat here.
\end{element}

\begin{element}
The model for the fixed-point combinator has type
\begin{gather}\label{el:fhar:mod:Y}
\rM.\ensuremath{\Varid{fix}} : \impfun{H, A} (\rM.\ensuremath{\Varid{pth}}\;H\;A \to \rM.\ensuremath{\Varid{pco}}\;H\;A) \to \rM.\ensuremath{\Varid{pco}}\;H\;A
\end{gather}
In the present model, $\ensuremath{\Varid{pth}}\;H\;A$ is simply equal to \ensuremath{\Varid{pco}\;\Conid{H}\;\Conid{A}}, so by 
\cref{el:sdt:domain:fixed:point},
it is sufficient
to show that $\rM.\ensuremath{\Varid{pco}}\;H\;A$ is a domain.
We define the algebra by
\begin{lgather*}
\alpha : \impfun{H,A} L\; (\rM.\ensuremath{\Varid{pco}} \; H \; A) \to \rM.\ensuremath{\Varid{pco}} \; H \; A\\
\alpha \; (p, c) = \lambda h\;B\;k.\; \beta_{h\; B} \;  (p, c\;h\;B\;k)
\end{lgather*}
where $\beta_{h \; B} : L \; (h\;B) \to h\;B$ is $\beta_{h\; B} \defeq (h.\ensuremath{\Varid{dom}}\;B).\alpha$.
The $L$-algebra $\alpha$ is a product of a family of $L$-algebras, so 
it is easy to check that $\alpha$ satisfies the laws:
\begin{align*}
   & \alpha \; (\top,\;c) \\
=\ & \lambda h\;B\;k.\; \beta_{h\;B}\;(\top,\; c\; h \; b\; k) \\
=\ & \reason{$\beta_{h \; B}$ is an Eilenberg-Moore algebra} \\
   & \lambda h\;B\;k.\; c\;h\;b\;k \\
=\ & c
\end{align*}
and similarly for all $(p, (q, c)) : L \; (L\; (\rM.\ensuremath{\Varid{pco}} \; H \; A))$,
{\allowdisplaybreaks
\begin{align*}
   & \alpha \; (L \; \alpha\; (p, (q, c))) \\
=\ & \alpha \; (p, \lambda h\;B\;k.\; \beta_{h\;B} \;  (q, c\;h\;B\;k)) \\
=\ & \lambda h\;B\;k.\; \beta_{h \; B} \; (p, \beta_{h \; B} \; (q, c\;h\;B\;k))\\
=\ & \reason{$\beta_{h \; B}$ is an Eilenberg-Moore algebra}\\
   & \lambda h\;B\;k.\; \beta_{h\;B}\;(\mu^L \; (p, (q, c\;h\;B\;k))) \\
=\ & \lambda h\;B\;k.\; \beta_{h\;B}\;(\Sigma(\_ : p).\;q,\; c\;h\;B\;k) \\
=\ & \alpha\; (\mu^L \; (p, (q, c)))
\end{align*}
}
We have shown that \ensuremath{\Varid{pco}\;\Conid{H}\;\Conid{A}} is a domain, so we can
use \ensuremath{\Varid{fix}} (\cref{el:sdt:domain:fixed:point}) to define 
\[
\rM.\ensuremath{\Varid{fix}} \; f = \ensuremath{\Varid{fix}}\;f.
\]
\end{element}

\begin{element}
Finally, we need to give an interpretation of \ensuremath{\Varid{hdl}} \Cref{eq:fhar:eval}:
\begin{align*}
\rM.\ensuremath{\Varid{hdl}} & : \impfun{H, A} (h : \rM.\ensuremath{\Conid{HandlerRec}}\; H) \\
         & \to \rM.\ensuremath{\Varid{pco}}\; H\; A \to \rM.\ensuremath{\Varid{tm}}\; (h \; A)
\end{align*}
Note that the type of $h$ is $\ensuremath{\Conid{HandlerRec}}$ rather than \ensuremath{\Conid{HandlerL}}.
By the definition of \ensuremath{\Conid{HandlerRec}} in \Cref{sec:gen:rec}, there exists some
$F : \rM.\ensuremath{\Varid{el}}\;\rM.\ensuremath{\Varid{ty}} \to \rM.\ensuremath{\Varid{el}}\;\rM.\ensuremath{\Varid{ty}}$ such that the underlying monad of $h$ maps every $A : \PDom$ to $\rM.\ensuremath{\Varid{pth}}\;H\;(F\;A)$.
By the discussion above around \Cref{el:fhar:mod:Y}, $\rM.\ensuremath{\Varid{pth}}\;H\;(F\;A)$, which
is just $\rM.\ensuremath{\Varid{pco}}\;H\;(F\;A)$, is always a domain.
Therefore we have a conversion function
\[
\sigma : (h : \rM.\ensuremath{\Conid{HandlerRec}}\;H) \to \aset{h' : \ensuremath{\Conid{HandlerL}} \; H \mid h.h = h'.h},
\]
and we define the model of \ensuremath{\Varid{hdl}} to be
\[
\rM.\ensuremath{\Varid{hdl}}\; \imparg{H}\;\imparg{A}\;h \; c = c \; (\sigma\; h)\;A\;h.\ensuremath{\Varid{ret}}
\]
This completes the definition of the model $\rM : \sem{\Fhar}_{V_2}$.
\end{element}

\begin{element}
The realizability model $\rM$ of \Fhar{} gives a way to compute recursive
programs written in \Fhar{} by program extraction.
Since $L\;A$ is a domain, the monad $L : \PDom \to \PDom$ can be extended to 
\[
L' : \aset{L' : \rM.\ensuremath{\Conid{HandlerL}} \; \ensuremath{\Conid{VoidH}}  \mid L'.h = L}
\]
Therefore we have a function
\begin{lgather*}
\ensuremath{\Varid{toL}} : \impfun{A : \PDom}  \rM.\ensuremath{\Varid{pco}} \; \ensuremath{\Conid{VoidH}} \; A \to L\; A\\
\ensuremath{\Varid{toL}}\; c = c \; L' \; A \; \eta^L
\end{lgather*}
Every closed program \ensuremath{\Varid{p}\mathbin{:}\Varid{pco}\;\Conid{VoidH}\;\Varid{bool}} is interpreted as a global element 
of $\rM.\ensuremath{\Varid{pco}}\; \ensuremath{\Conid{VoidH}} \; 2$ in $\AsmK$.
Composing it with \ensuremath{\Varid{toL}}, we then have a global element of the modest 
set $L \; 2$.
The realizer of this element is then a possibly divergent Turing machine $r$ 
that yields a Boolean value if it halts.
\end{element}

\begin{element}
We conjecture that the realizability model of \Fhar{} in this section is
\emph{adequate}.

\begin{conjecture}
For all closed programs $c : \ensuremath{\Varid{pco}\;\Conid{VoidH}\;\Varid{bool}}$
in \Fhar{}, if the morphism $\ensuremath{\Varid{toL}} \vcomp \sem{c} : 1 \to L \; 2$ in $\AsmK$ is 
$\inr\;\ensuremath{\Varid{tt}}$ (or $\inr\;\ensuremath{\Varid{ff}}$), then $c = \ensuremath{\Varid{val}\;\Varid{tt}}$ (or \ensuremath{\Varid{c}\mathrel{=}\Varid{val}\;\Varid{ff}}) in the theory of \Fhar.
\end{conjecture}
This implies that if $\ensuremath{\Varid{toL}} \vcomp \sem{c} = \inl\;\unitel$, then $c$ does
not equal to \ensuremath{\Varid{val}\;\Varid{tt}} or \ensuremath{\Varid{val}\;\Varid{ff}}, otherwise $\ensuremath{\Varid{toL}} \vcomp \sem{c}$ would not
be $\inl \; \unitel$.

We expect adequacy can be proved using synthetic Tait computability (STC) \emph{internally}
in the effective topos $\Eff$, in which we glue the (internal) category $P$
with the  category of $P$-valued presheaves over the category of
judgements of $\Fhar{}$ (constructed internally in $\Eff$).
Such an internal STC argument has been attempted by
\citet{Sterling_et_al_2022} to prove adequacy for a language with \emph{security
levels} and \emph{general recursion} but without impredicative polymorphism, whose
denotational semantics is a sheaf-model of SDT, although there is a currently
unfixed problem in their proof \citep{sterling2023_erratum}.
\end{element}

\end{document}